\documentclass[letterpaper,11pt]{article}
\usepackage{fullpage}
\usepackage[ansinew]{inputenc}
\usepackage{verbatim}
\usepackage{graphicx}
\usepackage{tabularx}
\usepackage{array}
\usepackage{delarray}
\usepackage{dcolumn}
\usepackage{float}
\usepackage{amsmath}
\usepackage{psfrag}
\usepackage{booktabs}
\usepackage{multirow,hhline}
\usepackage{amstext}
\usepackage{amssymb}
\usepackage{fancyhdr}
\usepackage{hyperref}
\usepackage{setspace}
\usepackage{amsthm}
\usepackage{makeidx}
\usepackage{tikz}
\usepackage{lipsum}
\usepackage{url}
\usepackage{adjustbox}
\usepackage{enumitem}
\usepackage{threeparttable}
\usepackage{adjustbox}

\usepackage{xr}
\makeatletter

\newcommand*{\addFileDependency}[1]{
	\typeout{(#1)}
	\@addtofilelist{#1}
	\IfFileExists{#1}{}{\typeout{No file #1.}}
}\makeatother

\RequirePackage[round,authoryear,longnamesfirst]{natbib}
\usepackage{natbib}
\usetikzlibrary{arrows,calc}
\hypersetup{citecolor=true,colorlinks=true,allcolors=blue}
\makeindex
\allowdisplaybreaks
\numberwithin{equation}{section}
\newtheorem{lem}{Lemma}[section]
\newtheorem{thm}{Theorem}[section]
\newtheorem{cor}{Corollary}[section]

\newtheorem{ass}{Assumption}

\renewcommand{\citep}[1]{\citeauthor{#1}, \citeyear{#1}}

\newcommand{\diag}{\text{diag}}

\newcommand{\convP}{\stackrel{p}{\longrightarrow}}
\newcommand{\convD}{\rightsquigarrow}

\newcommand{\N}{\mathcal{N}}
\newcommand{\eps}{\varepsilon}

\renewcommand{\epsilon}{\varepsilon}

\DeclareMathOperator*{\argmin}{arg\,min}

\makeatletter
\newcommand*{\rom}
[1]{\expandafter\@slowromancap\romannumeral #1@}
\makeatother

\title{A Conditional Linear Combination Test with Many Weak Instruments\thanks{We thank the editor, the associate editor, and two referees for extensive comments which have led to many improvements. We are grateful to Andrii Babii, Harold D. Chiang, Keisuke Hirano, Ming Li, Adam Rosen, and all participants at the SH3 conference on econometrics 2023 and the University of Macau Econometrics Seminar for their valuable comments. Wenjie Wang acknowledges the financial support from Singapore Ministry of Education Tier 1 grants RG53/20 and RG104/21. Yichong Zhang acknowledges the financial support from a Lee Kong Chian fellowship and the NSFC under the grant No. 72133002. Any and all errors are our own.\vspace{1.3mm}} 
	\\ \vspace{2mm}
}
\author{Dennis Lim\thanks{Singapore Management University.\ E-mail~address: dennis.lim.2019@phdecons.smu.edu.sg.} \and Wenjie Wang\thanks{Division of Economics, School of Social Sciences, Nanyang Technological University.
		HSS-04-65, 14 Nanyang Drive, Singapore 637332. 
		E-mail address: wang.wj@ntu.edu.sg.}   \and Yichong Zhang\thanks{
		Singapore Management University.\ E-mail~address: yczhang@smu.edu.sg.}
	\date{}
}

\begin{document}
	
	\maketitle

	\begin{abstract}
		We consider a linear combination of jackknife Anderson-Rubin (AR), jackknife Lagrangian multiplier (LM), and orthogonalized jackknife LM tests for inference in IV regressions with many weak instruments and heteroskedasticity. 
		Following I.\cite{Andrews(2016)},
		we choose the weights in the linear combination based on a decision-theoretic rule that is adaptive to the identification strength. Under both weak and strong identifications, the proposed test controls asymptotic size and is admissible among certain class of tests. Under strong identification, our linear combination test has optimal power against local alternatives among the class of invariant or unbiased tests which are constructed based on jackknife AR and LM tests. Simulations and an empirical application to \citeauthor{Angrist-Krueger(1991)}'s (\citeyear{Angrist-Krueger(1991)}) dataset confirm the good power properties of our test.

		\bigskip 
		
		 \bigskip
		
	\end{abstract}
	
	\newpage
	
	\section{Introduction}
	
	Various recent surveys in leading economics journals suggest that weak instruments remain important concerns for empirical practice. For instance, I.\cite{Andrews-Stock-Sun(2019)} survey 230 instrumental variable (IV) regressions from 17 papers published in the \textit{American Economic Review} (AER). They find that many of the first-stage F-statistics (and non-homoskedastic generalizations) are in a range that raises such concerns, and virtually all of these papers report at least one first-stage F with a value smaller than 10.
	Similarly, in \citeauthor{lee2021}'s (\citeyear{lee2021}) survey of 123 AER articles involving IV regressions, 105 out of 847 specifications have first-stage Fs smaller than 10. 
	Moreover, many IV applications involve a large number of instruments. 
	For example, in their seminal paper, \cite{Angrist-Krueger(1991)} study the effect of schooling on wages by interacting three base instruments (dummies for the quarter of birth) with state and year of birth, resulting in 180 instruments. \cite{Hansen-Hausman-Newey(2008)} show that using the 180 instruments gives tighter confidence intervals than using the base instruments even after adjusting for the effect of many instruments. 
	In addition, as pointed out by \cite{MS22},  in empirical papers that employ the ``judge design" (e.g., see \cite{maestas2013}, \cite{sampat2019}, and \cite{dobbie2018}), the number of instruments (the number of judges) is typically proportional to the sample size, and the famous Fama-MacBeth two-pass regression in empirical asset pricing (e.g., see \cite{fama1973}, \cite{shanken1992}, and \cite{anatolyev2022}) is equivalent to IV estimation with the number of instruments proportional to the number of assets. 
	Similarly, \cite{belloni2012} consider an IV application involving more than one hundred instruments for the study of the effect of judicial eminent domain decisions on economic outcomes. 
	\cite{carrasco2015} used many instruments in the estimation of the elasticity of intertemporal substitution in consumption.
	Furthermore, as pointed out by \cite{Goldsmith(2020)}, the shift-share or Bartik 
	instrument (e.g., see \cite{bartik1991} and \cite{blanchard1992}), which has been widely applied in many fields such as labor, public, development, macroeconomics, international trade, and finance, can be considered as a particular way of combining many instruments. For example, in the canonical setting of estimating the labor supply elasticity, the corresponding number of instruments is equal to the number of industries, which is also typically proportional to the sample size. 
	

	In this paper, following the seminal study by I.\cite{Andrews(2016)}, we propose a jackknife conditional linear combination (CLC) test that is robust to weak identification, many instruments, and heteroskedasticity. The proposed test also achieves efficiency under strong identification against local alternatives. The starting point of our analysis is the observation that, under strong identification, an orthogonalized jackknife Lagrangian multiplier (LM) test is the uniformly most powerful (UMP) test against local alternatives among the class of tests that are constructed based on jackknife LM and Anderson-Rubin (AR) tests and are either unbiased or invariant to sign changes. However, the orthogonalized LM test may not have good power under weak identification or against certain fixed alternatives. Therefore, we consider a linear combination of jackknife AR, jackknife LM, and orthogonalized LM tests. Specifically, we follow I.\cite{Andrews(2016)} and determine the linear combination weights by minimizing the maximum power loss, which can be viewed as a maximum regret and 
	is further calibrated based on the limit experiment of interest and a sufficient statistic for the identification strength under many instruments. Then, similar to I.\cite{Andrews(2016)}, we show such a jackknife CLC test is adaptive to the identification strength in the sense that (1) it achieves correct asymptotic size, 
	(2) it is asymptotically and conditionally admissible under weak identification among some class of tests, (3) it converges to the UMP test mentioned above under strong identification against local alternatives,\footnote{We emphasize that the UMP property of our CLC test under strong identification holds within the class of sign-invariant or unbiased tests that are constructed based on jackknife AR and LM tests only. It may be possible to construct more efficient tests using test statistics besides the jackknife AR and LM. How to construct a globally optimal test under strong identification with many IVs and heteroskedastic errors is a topic that remains to be explored in future research. } and (4) it has asymptotic power equal to 1 under strong identification against fixed alternatives.   The properties of jackknife AR, jackknife LM, orthogonalized LM, and our CLC tests are summarized in Table \ref{tab:clc_weights}.  Simulations based on the limit experiment as well as calibrated data confirm the good power properties of our test. Then, we apply the new jackknife CLC test to \citeauthor{Angrist-Krueger(1991)}'s (\citeyear{Angrist-Krueger(1991)}) dataset with the specifications of 180 and 1,530 instruments. We find that, in both specifications, our confidence intervals (CIs) are the shortest 
	among those constructed by weak identification robust tests, namely, the jackknife AR, LM, and CLC tests, and the two-step procedure. Furthermore, our CIs are found to be even shorter than the non-robust Wald test CIs based on the jackknife IV estimator (JIVE) proposed by \cite{Angrist(1999)}, which is in line with the theoretical result that the jackknife CLC test is adaptive to the identification strength and is efficient under strong identification. 
	
	\begin{table}[H]
		\adjustbox{max width=\textwidth}{
			\centering
			\begin{tabular}{l|l|l|l}
				& Weak ID, fixed alternative & Strong ID, local alternative & Strong ID, fixed alternative \\ \hline 
				Jackknife AR & Admissible & Not UMP & Power 1 \\
				Jackknife LM & Admissible & Not UMP & Power 1 \\
				Orthogonalized LM & Admissible & UMP   & Non-monotonic power \\ 
				CLC & Admissible & UMP & Power 1 \\
		\end{tabular}}
		\label{tab:clc_weights}%
		\caption{Power Comparision of the Tests}
	\end{table}%
	

	\textbf{Relation to the literature.} The contributions in the present paper relate to two strands of literature.
	First, it is related to the literature on many instruments; see, for example, \cite{Kunitomo1980}, \cite{morimune1983}, \cite{Bekker(1994)}, \cite{donald2001}, \cite{chamberlain2004}, \cite{Chao-Swanson(2005)}, \cite{stock2005}, \cite{han2006}, D.\cite{Andrews-Stock(2007)}, \cite{Hansen-Hausman-Newey(2008)}, \cite{Newey-Windmeijer(2009)}, \cite{anderson2010}, \cite{kuersteiner2010}, \cite{anatolyev2011}, \cite{belloni2011}, \cite{okui2011}, \cite{belloni2012}, \cite{carrasco2012}, 
	\cite{Chao(2012)}, \cite{Haus2012}, \cite{hansen2014}, \cite{carrasco2015}, \cite{Wang_Kaffo_2016}, \cite{kolesar2018}, \cite{Matsushita2020}, \cite{solvsten2020},  \cite{crudu2021}, and \cite{MS22}, among others. 
	In the context of many instruments and heteroskedasticity, \cite{Chao(2012)} and \cite{Haus2012} provide standard errors for Wald-type inferences that are based on JIVE and jackknifed versions of the limited information maximum likelihood (LIML) and \citeauthor{Fuller(1977)}'s (\citeyear{Fuller(1977)}) estimators (HLIM and HFUL). These estimators are more robust to many instruments than the commonly used two-stage least squares (TSLS) estimator because they can correct the bias caused by the high dimension of IVs.\footnote{Specifically, the rate of growth of the concentration parameter, which measure the overal instrument strength, is denoted as $\mu_n^2$. JIVE, HLIM, and HFUL remain consistent with heteroskedastic errors even when instrument weakness is such that $\mu_n^2$ is slower than the number of instruments $K$, provided that $\mu_n^2/\sqrt{K} \rightarrow \infty$ as the number of observations $n \rightarrow \infty$ (\citealp{Chao(2012), Haus2012}). In contrast, TSLS is less robust to instrument weakness  as it is shown to be consistent only under homoskedasticity if $\mu_n^2/K \rightarrow \infty$ (\citep{Chao-Swanson(2005)}).} In simulations derived from the data in \cite{Angrist-Krueger(1991)}, which is representative of empirical labor studies with many instrument concerns, \citet[Section IV]{Angrist-Frandsen2022} show that such bias-corrected estimators outperform the TSLS that is based on the instruments selected by the least absolute shrinkage and selection operator (LASSO) introduced in \cite{belloni2012} or the random forest-fitted first stage introduced in \cite{athey2019}. 
	Furthermore, under many weak moment asymptotics, \cite{Newey-Windmeijer(2009)} provide new variance estimators for the jackknife GMM and the class of generalized empirical likelihood (GEL) estimators, which includes the continuous updating estimator (CUE) and EL estimator as special cases. In the linear heteroskedastic IV model, consistency and asymptotic normality of CUE require 
	$m^2/n \rightarrow 0$ and $m^3/n \rightarrow 0$, respectively, where $m$ and $n$ denote the number of moment conditions and the sample size (e.g., see p.689 of \cite{Newey-Windmeijer(2009)}). Such conditions are needed to simultaneously control the estimation error for all the elements of the heteroskedasticity consistent weighting matrix. Somewhat stronger rate conditions are required for other GEL estimators.
	
	However, the Wald-type inference methods are invalid under weak identification, which occurs when the ratio of the concentration parameter over the square root of the number of instruments remains bounded as the sample size increases to infinity. In this case,  all the estimators mentioned earlier become inconsistent, and there is no consistent test for the structural parameter of interest (see Section 3 of \cite{MS22}).  For weak identification robust inference under many instruments, D.\cite{Andrews-Stock(2007)} consider
	the AR test, the score test introduced in \cite{Kleibergen(2002)}, 
	and the conditional likelihood ratio test introduced in \cite{Moreira(2003)}.  
	Their IV model is homoskedastic and requires the number of instruments to diverge slower than the cube root of the sample size ($K^3/n \rightarrow 0$, where $K$ denotes the number of instruments). 
	\cite{anatolyev2011} propose a modified AR test that allows for the number of instruments to be proportional to the sample size but still require homoskedastic errors. Recently, \cite{crudu2021} and \cite{MS22} propose jackknifed versions of the AR test in a model with many instruments and heteroskedasticity. Both tests are robust to weak identification, but \citeauthor{MS22}'s (\citeyear{MS22}) jackknife AR test has better power properties due to the use of a cross-fit variance estimator. However, the jackknife AR tests may be inefficient under strong identification. To address this issue, \cite{MS22} also propose a new pre-test for weak identification under many instruments and apply it to form a two-stage testing procedure with a Wald test based on the JIVE introduced in \cite{Angrist(1999)}. The JIVE-Wald test is more efficient than the jackknife AR under strong identification. Therefore, an empirical researcher can employ the jackknife AR if the pre-test suggests weak identification and the JIVE-Wald if the pre-test suggests strong identification. In addition to the jackknife AR, \cite{Matsushita2020} propose a jackknife LM test, which is also robust to weak identification, many instruments, and heteroskedastic errors. However, the jackknife CLC test introduced in our paper is more efficient than 
	the jackknife AR, the jackknife LM, and the two-step test under strong identification and local alternatives, while still being robust to weak identification.  
	
	Second, our paper is related to the literature on weak identification under the framework of a fixed number of instruments or moment conditions, 
	in which various robust inference methods are available for non-homoskedastic errors; see, for example, \cite{Stock-Wright(2000)}, \cite{Kleibergen(2005)}, D.\cite{Andrews-Cheng(2012)}, I.\cite{Andrews(2016)}, I.\cite{Andrews-Mikusheva(2016)}, I.\cite{Andrews(2018)}, \cite{Moreira-Moreira(2019)}, D.\cite{Andrews-Guggenberger(2019)}, 
	and \cite{lee2021}.
	In particular, our jackknife CLC test extends the work of I.\cite{Andrews(2016)} to the framework with many weak instruments. I.\cite{Andrews(2016)} considers the convex combination between the generalized AR statistic (S statistic) introduced by \cite{Stock-Wright(2000)} and the score statistic (K statistic) introduced by \cite{Kleibergen(2005)}. We find that under many weak instruments, the orthogonalized jackknife LM statistic plays a role similar to the K statistic.  
	However, the trade-off between the jackknife AR and orthogonalized LM statistics turns out to be rather different from that between the S and K statistics. As pointed out by I.\cite{Andrews(2016)}, in the case with a fixed number of weak instruments (or moment conditions), the K statistic picks out a particular (random) direction corresponding to the span of a conditioning statistic that measures the identification strength and restricts attention to deviations from the null along this specific direction. In contrast to the K statistic, the S statistic treats all deviations from the null equally. Therefore, the trade-off between the K and S statistics is mainly from the difference in attention to deviation directions. We find that with many weak instruments, the jackknife AR and orthogonalized LM tests do not have such difference in deviation directions. Instead, their trade-off is mostly between local and non-local alternatives. 
	Furthermore, although the standard LM test (without orthogonalization) is not weak identification robust under I.\cite{Andrews(2016)}'s framework, the jackknife LM test is under many instruments.  
	Therefore, we consider a linear combination of jackknife AR, jackknife LM, and orthogonalized jackknife LM tests and find that the resulting CLC test has good power properties in a variety of scenarios.   

\textbf{Notation.} We denote $\mathcal{Z}(\mu)$ as the normal random variable with unit variance and expectation $\mu$ and $[n] = \{1,2,\cdots, n\}.$ We further simplify $\mathcal{Z}(0)$ as $\mathcal{Z}$, which is just a standard normal random variable. We denote $z_{\alpha}$ as the $(1-\alpha)$ quantile of a standard normal random variable and $\mathbb{C}_{\alpha}(a_1,a_2;\rho)$ as the $(1-\alpha)$ quantile of random variable $a_1 \mathcal{Z}_1^2 + a_2(\rho \mathcal{Z}_1 + (1-\rho^2)^{1/2} \mathcal{Z}_2)^2 + (1-a_1 - a_2)\mathcal{Z}_2^2$ where $\mathcal{Z}_1$ and $\mathcal{Z}_2$ are two independent standard normal random variables, $\alpha$ is the significance level, $\rho$ is a constant in $(-1,1)$, and  $a_{1}$ and $a_{2}$ are the weights of the first and second components in the random variable. We further simplify $\mathcal{C}_{0,0;\rho}$ as $\mathbb{C}_{\alpha}$, which is just the $1-\alpha$ quantile of $\mathcal{Z}^2$. We let $\mathbb{C}_{\alpha,\max}(\rho) = \sup_{(a_1,a_2) \in \mathbb{A}_0} \mathbb{C}_{\alpha}(a_1,a_2;\rho)$, where $\mathbb{A}_0 = \{(a_1,a_2) \in [0,1] \times [0,1], a_1+a_2\leq \overline{a}\}$ for some $\overline{a}<1$. We suppress the dependence of $\mathbb{C}_{\alpha,\max}(\rho)$ on $\overline{a}$ for simplicity of notation. The operators $\mathbb{E}^*$ and $\mathbb{P}^*$ are expectation and probability taken conditionally on data, respectively. For example, $\mathbb{E}^*1\{\mathcal{Z}^2(\hat{\mu}) \geq \mathbb{C}_{\alpha}\}$, in which $\hat{\mu}$ is some estimator of the expectation $\mu$ based on data, means the expectation is taken over the normal random variable by treating $\hat{\mu}$ as deterministic. We use $\convD$ to denote convergence in distribution, $U \stackrel{d}{=} V$ to denote that $U$ and $V$ share the same distribution, and $\text{maxeig}(\mathcal{V})$ and $\text{mineig}(\mathcal{V})$ to denote maximum and minimum eigenvalues of a positive semidefinite matrix $\mathcal{V}$. For two sequences of random variables $U_n$ and $V_n$, we write $U_n \stackrel{d}{=}V_n +o_P(1)$ if there exist $\tilde U_n \stackrel{d}{=}U_n$ and $\tilde V_n \stackrel{d}{=}V_n$ such that $\tilde U_n - \tilde V_n = o_P(1)$.

\section{Setup and Limit Problems}\label{sec:limit}
We consider the linear IV regression with a scalar outcome $Y_i$, a scalar endogenous variable $X_i$, and a $K \times 1$ vector of instruments $Z_i$ such that 
\begin{align}
	Y_i = X_i \beta + e_i, \quad  X_i = \Pi_i + V_i, \quad \forall i \in [n],
	\label{eq:1}
\end{align}
where $\Pi_i = \mathbb{E}X_i$ and $\{Z_i\}_{i\in [n]}$ is treated as fixed, following the many-instrument literature. We let $K$ diverge with sample size $n$, allowing for the case that $K$ is of the same order of magnitude as $n$. We further have $\mathbb{E}V_i = 0$ by construction, and $\mathbb{E}e_i = 0$ by IV exogeneity.  We allow $(e_i, V_i)$ to be heteroskedastic across $i$.
Also, following the literature on many instruments (e.g., \cite{MS22}), we assume that there are no controls included in our model as they can be partialled out from $(Y_i,X_i,Z_i)$. We provide more discussions about the effect of partialling out the covariates after Assumption \ref{ass:weak_convergence} below. 


We are interested in testing $\beta = \beta_0$. Let $e_i(\beta_0) = Y_i - X_i \beta_0 = e_i + X_i \Delta$, where $\Delta = \beta- \beta_0$. We collect the transpose of $Z_i$ in each row of $Z$, an $n \times K$ matrix of instruments, and denote $P = Z (Z^\top Z)^{-1}Z^\top$. In addition, Let $Q_{a,b} = \frac{\sum_{i \in [n]}\sum_{j \neq i}a_i P_{ij}b_j}{\sqrt{K}}$ and $\mathcal{C} =Q_{\Pi,\Pi}$. Then, as pointed out by \cite{MS22}, the rescaled $\mathcal{C}$ is the concentration parameter that measures the strength of identification in the heteroskedastic IV model with many instruments. Specifically, the parameter $\beta$ is weakly identified if $\mathcal{C}$ is bounded and strongly identified if $|\mathcal{C}|\rightarrow \infty$. We consider drifting sequence asymptotics so that all quantities are implicitly indexed by the sample size $n$ except specified otherwise. We omit such dependence for notation simplicity.

Throughout the paper, we consider three scenarios: (1) weak identification and fixed alternatives in which $\mathcal{C} \rightarrow \widetilde{\mathcal{C}}$ for some fixed constant $\widetilde{\mathcal{C}} \in \Re$ and $\Delta$ is fixed and bounded, (2) strong identification and local alternatives in which $\mathcal{C} = \widetilde{\mathcal{C}}/d_n $, $\Delta  = \widetilde{\Delta}d_n $, $\widetilde{\mathcal{C}}$ and $\widetilde{\Delta}$ are bounded constants independent of $n$, and $d_n  \rightarrow 0$ is a deterministic sequence, and (3) strong identification and fixed alternatives in which $\mathcal{C} = \widetilde{\mathcal{C}}/d_n $ for the same $\widetilde{\mathcal{C}}$ and $d_n$ defined in case (2) and $\Delta$ is fixed and bounded.\footnote{If we follow the setup in \cite{Chao(2012)} and \cite{Haus2012} and assume $\Pi_i = \mu_n \pi_i/\sqrt{n}$ so that $\infty>C\geq \sum_{i \in [n]}\sum_{j \neq i}\pi_iP_{ij}\pi_j/n \geq c>0$ for some constants $c,C$, then $ \mathcal{C} = \frac{\mu_n^2}{\sqrt{K}}\frac{\sum_{i \in [n]}\sum_{j \neq i}\pi_iP_{ij}\pi_j}{n}$, 
	implying that $d_n = \sqrt{K}/\mu_n^2$. Then, our definition of strong identification ($d_n \rightarrow 0$) is equivalent to that defined in \cite{Chao(2012)} and \cite{Haus2012} ($\mu_n^2/\sqrt{K} \rightarrow \infty$).} Many weak identification robust tests proposed in the literature (namely, the jackknife AR tests proposed by \cite{crudu2021} and \cite{MS22} and the jackknife LM test proposed by \cite{Matsushita2020}) depend on a subset of the following three quantities: $(Q_{e(\beta_0),e(\beta_0)},Q_{X,e(\beta_0)},Q_{X,X})$. Throughout the paper, we maintain the following high-level assumption.

\begin{ass}
	Under both weak and strong identification, the following weak convergence holds:  
	\begin{align}
		\begin{pmatrix}
			& Q_{e,e} \\
			& Q_{X,e} \\
			& Q_{X,X} - \mathcal{C}
		\end{pmatrix} \convD \N\left(\begin{pmatrix}
			0 \\
			0 \\
			0
		\end{pmatrix},\begin{pmatrix}
			\Phi_1 & \Phi_{12} & \Phi_{13} \\
			\Phi_{12} & \Psi & \tau \\
			\Phi_{13} & \tau & \Upsilon
		\end{pmatrix}\right), 
		\label{eq:limittrue}
	\end{align}
	for some $(\Phi_1, \Phi_{12}, \Phi_{13}, \Psi, \tau, \Upsilon)$. 
	\label{ass:weak_convergence}
\end{ass}

Although there are no controls in the model (\ref{eq:1}), we further verify Assumption \ref{ass:weak_convergence} in Section \ref{sec:W0} of the Online Supplement for a proper linear IV regression that includes a fixed  dimension of exogenous control variables, 
which are then partialled out from the original outcome variable, endogenous variable, and instruments.\footnote{
	Here, we focus on the case where the number of exogenous control variables is treated as fixed. In the case where the dimension of the exogenous variables is also large and assumed to diverge to infinity with the sample size,  \cite{CNT23} propose new versions of various jackknife IV estimators and show they are consistent and asymptotically normal under strong identification. We conjecture that it is possible to replace our jackknife construct (i.e. $Q_{a,b}$) by the new version and consider weak identification robust tests and their linear combinations in the same manner as studied in this paper. This is left as a topic for future research.}

Assumption \ref{ass:weak_convergence}  implies that,\footnote{Note that $\begin{pmatrix}
		Q_{e(\beta_0),e(\beta_0)} \\
		Q_{X,e(\beta_0)} \\
		Q_{X,X}
	\end{pmatrix} = \begin{pmatrix}
		1 & 2 \Delta & \Delta^2 \\
		0 & 1 & \Delta \\
		0 & 0 &  1
	\end{pmatrix} \begin{pmatrix}
		Q_{e,e} \\
		Q_{X,e} \\
		Q_{X,X}
	\end{pmatrix}.$} under both strong and weak identification, 
\begin{align}
	\begin{pmatrix}
		& Q_{e(\beta_0),e(\beta_0)} - \Delta^2 \mathcal{C} \\
		& Q_{X,e(\beta_0)} - \Delta \mathcal{C}\\
		& Q_{X,X} - \mathcal{C}
	\end{pmatrix} \stackrel{d}{=} \N\left(\begin{pmatrix}
		0 \\
		0 \\
		0
	\end{pmatrix},\begin{pmatrix}
		\Phi_1(\beta_0) & \Phi_{12}(\beta_0) & \Phi_{13}(\beta_0) \\
		\Phi_{12}(\beta_0) & \Psi(\beta_0) & \tau(\beta_0) \\
		\Phi_{13}(\beta_0) & \tau(\beta_0) & \Upsilon
	\end{pmatrix}\right) + o_p(1), 
	\label{eq:limit}
\end{align}
where 
\begin{align}
	\Phi_1(\beta_0) & = \Delta^4 \Upsilon + 4\Delta^3 \tau + \Delta^2 (4 \Psi + 2 \Phi_{13}) + 4\Delta \Phi_{12} + \Phi_1, \notag \\
	\Phi_{12}(\beta_0) & = \Delta^3 \Upsilon + 3\Delta^2 \tau + \Delta(2 \Psi + \Phi_{13}) + \Phi_{12}, \notag \\
	\Phi_{13}(\beta_0) & = \Delta^2 \Upsilon + 2\Delta  \tau + \Phi_{13}, \notag \\
	\Psi(\beta_0) & = \Delta^2 \Upsilon + 2\Delta \tau + \Psi, \notag \\
	\tau(\beta_0) & = \Delta \Upsilon + \tau. 
	\label{eq:phipsi}
\end{align} 
In particular, under strong identification, we have $Q_{X,X}d_n  \convP \widetilde{\mathcal{C}}$, which has a degenerate distribution. Also, under local alternatives, we have $\Delta = o(1)$ so that 
\begin{align*}
	(\Phi_1(\beta_0),\Phi_{12}(\beta_0), \Phi_{13}(\beta_0),\Psi(\beta_0),\tau(\beta_0)) \rightarrow (\Phi_1,\Phi_{12}, \Phi_{13},\Psi,\tau).
\end{align*}

To describe a feasible version of the test, we assume we have consistent estimates for all the variance components. 

\begin{ass}
	Let $\rho(\beta_0) = \frac{\Phi_{12}(\beta_0)}{\sqrt{\Phi_1(\beta_0)\Psi(\beta_0)}}$, $\widehat{\gamma}(\beta_0) = (\widehat{\Phi}_1(\beta_0), \widehat{\Phi}_{12}(\beta_0), \widehat{\Phi}_{13}(\beta_0), \widehat{\Psi}(\beta_0), \widehat{\tau}(\beta_0), \widehat{\Upsilon}, \widehat{\rho}(\beta_0))$ be an estimator, and $\mathcal{B} \in \Re$ be a compact parameter space. Then, we have $\inf_{\beta_0 \in \mathcal{B}}\Phi_1(\beta_0)>0$, $\inf_{\beta_0 \in \mathcal{B}}\Psi(\beta_0)>0$, $\Upsilon >0$, and  for $\beta_0 \in \mathcal{B}$,
	\begin{align*}
		||\widehat{\gamma}(\beta_0) - \gamma(\beta_0)||_2 = o_p(1),
	\end{align*}
	where $\gamma(\beta_0) \equiv (\Phi_1(\beta_0), \Phi_{12}(\beta_0),\Phi_{13}(\beta_0),\Psi(\beta_0), \tau(\beta_0),\Upsilon,\rho(\beta_0))$.
	\label{ass:variance_est}
\end{ass}

Several remarks on Assumption \ref{ass:variance_est} are in order. First, \cite{Chao(2012)} propose a consistent estimator for $\Psi$ where there is strong identification and many instruments.  
It is possible to compute $\widehat{\gamma}(\beta_0)$ based on \citeauthor{Chao(2012)}'s (\citeyear{Chao(2012)}) estimator with their JIVE-based residuals $\hat{e}_i$ from the structural equation replaced by $e_i(\beta_0)$. 
Under weak identification and $\beta_0 = \beta$, \cite{crudu2021} and \cite{Matsushita-Otsu2021} establish the consistency of such estimators for $\Phi_1(\beta_0)$ and $\Psi(\beta_0)$, respectively.
Similar arguments can be used to show the consistency of the rest of the elements in $\widehat{\gamma}(\beta_0)$ under both weak and strong identification. 
In addition, the consistency can be established under both local and fixed alternatives.
We provide more details in Section \ref{sec:var1} in the Online Supplement. Second, motivated by \cite{KSS2020}, \cite{MS22} propose cross-fit estimators $\widehat{\Phi}_1(\beta_0)$ and $\widehat{\Upsilon}$, which are consistent under both weak and strong identification and lead to better power properties. Following their lead, one can write down the cross-fit estimators for the rest of the elements in $\gamma(\beta_0)$ and show they are consistent.\footnote{For example, \citet[p.22]{MS22} establish the limit of their cross-fit estimator $\widehat{\Psi}$ under weak identification and many instruments when the residual $\hat{e}_i$ from the structural equation is computed based on the JIVE estimator. We can construct $\widehat{\Psi}(\beta_0)$ by replacing $\hat{e}_i$ by $e_i(\beta_0)$. Then, the argument, as theirs with $Q_{X,e}/Q_{X,X}$ replaced by $\Delta$, establishes that $\widehat{\Psi}(\beta_0) \convP \Psi(\beta_0)$.} We provide more details in Section \ref{sec:var2} in the Online Supplement. Note that both \citeauthor{crudu2021}'s (\citeyear{crudu2021}) and \citeauthor{MS22}'s (\citeyear{MS22}) estimators are consistent under heteroskedasticity and allow for $K$ to be of the same order of $n$. Third, the consistency of $\widehat \gamma(\beta_0)$ over the entire parameter space under both strong and weak identifications is more than necessary and maintained mainly for simplicity of presentation. In fact, in order for our jackknife CLC test proposed below to control size under both weak and strong identification, it suffices to require $\widehat{\gamma}(\beta_0)$ to be consistent under the null only. The power analyses in Lemmas \ref{lem:strongID} and \ref{lem:weakID} below, and subsequently, Theorems \ref{thm:weakid} and \ref{thm:strongid}, only require the consistency of $\widehat{\gamma}(\beta_0)$ under strong identification with local alternatives and weak identification with fixed alternatives, respectively.

Under this framework, \cite{crudu2021} and \cite{MS22} consider the jackknife AR test 
\begin{align}
	1\{AR(\beta_0) \geq z_{\alpha}\}, \quad AR(\beta_0) = \frac{Q_{e(\beta_0),e(\beta_0)} }{\widehat{\Phi}_1^{1/2}(\beta_0)},    
	\label{eq:AR}
\end{align}
and \cite{Matsushita2020}  consider the jackknife LM test 
\begin{align}
	1\{LM^2(\beta_0) \geq \mathbb{C}_{\alpha}\}, \quad LM(\beta_0) = \frac{Q_{X,e(\beta_0)} }{\widehat{\Psi}^{1/2}(\beta_0)}.
	\label{eq:LM}
\end{align}
Both tests are robust to weak identification, many instruments, and heteroskedasticity. 
Lemma \ref{lem:strongID}  below characterizes the joint limit distribution of $(AR(\beta_0),LM(\beta_0))^\top$ under strong identification and local alternatives.

\begin{lem}
	Suppose Assumptions \ref{ass:weak_convergence} and \ref{ass:variance_est} hold and we are under strong identification with local alternatives, that is, there exists a deterministic sequence $d_n  \rightarrow 0$ such that $\mathcal{C} = \widetilde{\mathcal{C}}/d_n $ and $\Delta  = \widetilde{\Delta}d_n $, where $\widetilde{\mathcal{C}}$ and $\widetilde{\Delta}$ are bounded constants independent of $n$. Then, we have
	\begin{align*}
		\begin{pmatrix}
			AR(\beta_0) \\
			LM(\beta_0)
		\end{pmatrix} \convD \begin{pmatrix}
			\N_1 \\
			\N_2
		\end{pmatrix} \stackrel{d}{=} \N\left(\begin{pmatrix}
			0 \\
			\frac{\widetilde{\Delta} \widetilde{\mathcal{C}}}{\Psi^{1/2}} 
		\end{pmatrix},\begin{pmatrix}
			1 & \rho  \\
			\rho & 1 
		\end{pmatrix}\right)
	\end{align*}
	where $\rho = \Phi_{12}/\sqrt{\Phi_1\Psi}$.
	\label{lem:strongID}
\end{lem}

Two remarks are in order. First, under strong identification, we consider local alternatives so that $ \beta-\beta_0 \rightarrow 0$. This is why we have $(\Psi(\beta_0),\Phi_1(\beta_0),\Phi_{12}(\beta_0))$ converge to $(\Psi,\Phi_1,\Phi_{12})$, which are just the counterparts of $(\Psi(\beta_0),\Phi_1(\beta_0),\Phi_{12}(\beta_0))$ when $\beta_0$ is replaced by $\beta$. Second, although $AR(\beta_0)$ has zero mean, and hence, no power in this case, it is correlated with $LM(\beta_0)$. It is therefore possible to use $AR(\beta_0)$ to reduce the variance of $LM(\beta_0)$ and obtain a test that is more powerful than the LM test. 

\begin{lem}
	Consider the limit experiment in which researchers observe $(\N_1,\N_2)$ with
	\begin{align*}
		\begin{pmatrix}
			\N_1 \\
			\N_2
		\end{pmatrix} \stackrel{d}{=} \N\left(\begin{pmatrix}
			0 \\
			\theta
		\end{pmatrix},\begin{pmatrix}
			1 & \rho  \\
			\rho & 1 
		\end{pmatrix}\right),  
	\end{align*}
	know the value of $\rho$ and that $\mathbb{E}\N_1 = 0$, and want to test for $\theta=0$ versus the two-sided alternative. In this case, $1\{\N_2^{*2} \geq \mathbb{C}_{\alpha}\}$ is UMP among level-$\alpha$ tests that are either invariant to sign changes or unbiased, where 
	\begin{align*}
		\N_2^* = (1-\rho^2)^{-1/2}(\N_2 - \rho \N_1)    
	\end{align*}
	is the normalized residual from the projection of $\N_2$ on $\N_1$. 
	\label{lem:ump}
\end{lem}

Let the orthogonalized jackknife LM statistic be $LM^*(\beta_0) = (1-\widehat{\rho}(\beta_0)^2)^{-1/2} (LM(\beta_0)- \widehat{\rho}(\beta_0) AR(\beta_0))$. Then, Lemma \ref{lem:strongID} implies, under strong identification and local alternatives,  
\begin{align}
	\begin{pmatrix}
		AR(\beta_0) \\
		LM^*(\beta_0)
	\end{pmatrix} \convD \begin{pmatrix}
		\N_1 \\
		\N_2^*
	\end{pmatrix} \stackrel{d}{=} \N\left(\begin{pmatrix}
		0 \\
		\frac{\widetilde{\Delta} \widetilde{\mathcal{C}}}{[(1-\rho^2)\Psi]^{1/2}} 
	\end{pmatrix},\begin{pmatrix}
		1 & 0 \\
		0 & 1 
	\end{pmatrix}\right).  
	\label{eq:lmstar_str}
\end{align}
Lemma \ref{lem:ump} with $\theta = \widetilde{\Delta} \widetilde{\mathcal{C}}\Psi^{-1/2}$ implies, in this case, that the test $1\{LM^{*2}(\beta_0) \geq \mathbb{C}_{\alpha}\}$ is asymptotically strictly more powerful than the jackknife AR and LM tests based on $AR(\beta_0)$ and $LM(\beta_0)$ against local alternatives as long as $\rho \neq 0$. In addition, under strong identification and local alternatives, \citeauthor{MS22}'s (\citeyear{MS22}) two-step test statistic is asymptotically equivalent to $LM(\beta_0)$, and thus, is less powerful than $LM^*(\beta_0)$ too. 

Next, we compare the behaviors of $AR(\beta_0)$, $LM(\beta_0)$, and $LM^*(\beta_0)$ under strong identification and fixed alternatives. 
\begin{lem}
	Suppose  Assumption \ref{ass:variance_est} holds, $( Q_{e(\beta_0),e(\beta_0)} - \Delta^2 \mathcal{C}, Q_{X,e(\beta_0)} - \Delta \mathcal{C},Q_{X,X} - \mathcal{C})^\top
	= O_p(1)$, and we are under strong  identification so that $d_n \mathcal{C} \rightarrow \widetilde{\mathcal{C}}$ for some $d_n \rightarrow 0$. Then, we have, for any fixed $\Delta \neq 0$, 
	\begin{align*}
		d_n^2 \begin{pmatrix}
			AR^2(\beta_0) \\
			LM^2(\beta_0) \\
			LM^{*2}(\beta_0)
		\end{pmatrix} \convP \begin{pmatrix}
			\Phi_1^{-1}(\beta_0) \Delta^4 \widetilde{\mathcal{C}}^2 \\
			\Psi^{-1}(\beta_0) \Delta^2 \widetilde{\mathcal{C}}^2\\
			(1-\rho^2(\beta_0))^{-1}(\Psi^{-1/2}(\beta_0)  - \rho(\beta_0) \Phi_1^{-1/2}(\beta_0)\Delta)^2 \Delta^2 \widetilde{\mathcal{C}}^2
		\end{pmatrix}.
	\end{align*}
	\label{lem:strongID2}
\end{lem}
Given $d_n \rightarrow 0$ and both $\Phi_1^{-1}(\beta_0) \Delta^4 \widetilde{\mathcal{C}}^2>0$ and $\Phi_1^{-1}(\beta_0) \Delta^2 \widetilde{\mathcal{C}}^2>0$, $AR^2(\beta_0)$ and $LM^2(\beta_0)$ have power $1$ against fixed alternatives asymptotically. By contrast, $LM^{*2}(\beta_0)$ may not have power if $\Delta = \Delta_*(\beta_0) \equiv  \Phi_1^{1/2}(\beta_0)\Psi^{-1/2}(\beta_0)\rho^{-1}(\beta_0)$.


Next, we compare the performance of $AR(\beta_0)$ and $LM^*(\beta_0)$ under weak identification and fixed alternatives. 
\begin{lem}
	Suppose Assumptions \ref{ass:weak_convergence} and \ref{ass:variance_est} hold and we are under weak identification so that 
	$\mathcal{C} \rightarrow \widetilde{\mathcal{C}} \in \Re$. Then, we have, for any fixed $\Delta \neq 0$, 
	\begin{align}
		\begin{pmatrix}
			AR(\beta_0) \\
			LM^*(\beta_0)
		\end{pmatrix} \convD \begin{pmatrix}
			\N_1 \\
			\N_2^*
		\end{pmatrix} \stackrel{d}{=} \N\left(\begin{pmatrix}
			m_1(\Delta) \\
			m_2(\Delta)
		\end{pmatrix},\begin{pmatrix}
			1 & 0 \\
			0 & 1 
		\end{pmatrix}\right),      
		\label{eq:lmstar_wk}
	\end{align}
	where $\rho(\beta_0) = \frac{\Phi_{12}(\beta_0)}{\sqrt{\Psi(\beta_0)\Phi_1(\beta_0) }}$ and 
	\begin{align*}
		\begin{pmatrix}
			m_1(\Delta) \\
			m_2(\Delta)
		\end{pmatrix} & =  \begin{pmatrix}
			\Phi_1^{-1/2}(\beta_0)\Delta^2\widetilde{\mathcal{C}}\\
			(1-\rho^2(\beta_0))^{-1/2}\Psi^{-1/2}(\beta_0)\Delta \widetilde{\mathcal{C}}-\rho(\beta_0) (1-\rho^2(\beta_0))^{-1/2}\Phi_1^{-1/2}(\beta_0)\Delta^2\widetilde{\mathcal{C}}
		\end{pmatrix}.
	\end{align*}
	
	In particular, as $\Delta \rightarrow \infty$, we have 
	\begin{align*}
		m_1(\Delta) & \rightarrow \frac{\widetilde{\mathcal{C}}}{\Upsilon^{1/2}} \quad \text{and} \quad m_2(\Delta) \rightarrow \frac{\widetilde{\mathcal{C}}}{\Upsilon^{1/2}} \frac{\rho_{23}}{(1-\rho_{23}^2)^{1/2}},
	\end{align*}
	where $\rho_{23} = \frac{\tau}{(\Psi \Upsilon)^{1/2}}$ is the correlation between $Q_{X,e}$ and $Q_{X,X}$.\footnote{We suppress the dependence of $m_1(\Delta)$ and $m_2(\Delta)$ on $\gamma(\beta_0)$ and $\widetilde{\mathcal{C}}$ for notation simplicity. } 
	\label{lem:weakID}
\end{lem}

By comparing the means of the normal limit distribution in (\ref{eq:lmstar_wk}), we notice that under weak identification and fixed alternatives, neither $LM^*(\beta_0)$ dominates $AR(\beta_0)$ or vice versa. We also notice from Lemma \ref{lem:weakID} that for testing distant alternatives, the power of  $LM^*(\beta_0)$ is different from $AR(\beta_0)$ by a factor of $\rho_{23}/\sqrt{1-\rho^2_{23}}$, so that it will be lower when $|\rho_{23}| \leq 1/\sqrt{2}$. Under weak identification and homoskedasticity,\footnote{Specifically, we say the data are homoskedastic if the covariance matrices of $(e_i,V_i)$ are constant across $i$.} we have $\rho_{23} = \rho = \Phi_{12}/\sqrt{\Psi \Phi_1}$. Therefore, although the test 
$1\{LM^{*2}(\beta_0) \geq \mathbb{C}_{\alpha}\}$
has a power advantage under strong identification against local alternatives, it may lack power under weak identification against distant alternatives if the degree of endogeneity is low. Furthermore, $LM^*(\beta_0)$ may not have power if $\Delta = \Delta_*(\beta_0)$.  

In the current setting with many instruments, $AR(\beta_0)$ and $LM^*(\beta_0)$ play roles similar to that of \citeauthor{Stock-Wright(2000)}'s (\citeyear{Stock-Wright(2000)}) S and \citeauthor{Kleibergen(2005)}'s (\citeyear{Kleibergen(2005)}) K statistics in I.\citeauthor{Andrews(2016)}'s (\citeyear{Andrews(2016)}) setting, respectively. In the fixed number of IVs case, the power trade-off between S and K statistics is based on the direction of deviations from the null. However, as shown in Lemma \ref{lem:weakID} (the case with weak identification and fixed alternatives), the deviations of $AR(\beta_0)$ and $LM^*(\beta_0)$ from the null do not have such a difference in direction under the many-instrument setting because $\widetilde {\mathcal{C}}$ is just a scalar. Instead, their power trade-off is between local and non-local alternatives. This is in stark contrast to the setting in I.\cite{Andrews(2016)}.


To achieve the advantages of $AR(\beta_0)$, $LM(\beta_0)$, and $LM^*(\beta_0)$ in all three scenarios above, we need to combine them in a way that is adaptive to the identification strength. Following I.\cite{Andrews(2016)}, we consider the linear combination of $AR^2(\beta_0)$, $LM^2(\beta_0)$,  and $LM^{*2}(\beta_0)$. Recall that $(\N_1,\N_2^*)$ are the limits of $(AR(\beta_0),LM^{*}(\beta_0))$ in either strong or weak identification. See \eqref{eq:lmstar_str} and \eqref{eq:lmstar_wk} for their expressions in these two cases. Then, in the limit experiment, the linear combination test can be written as 
\begin{align}
	\phi_{a_1,a_2,\infty}    = 1\{a_1 \N_1^2 + a_2 (\tilde \rho \N_1 + (1-\tilde \rho^2)^{1/2}\N_2^*)^2+ (1-a_1-a_2)\N_2^{*2} \geq \mathbb{C}_{\alpha}(a_1,a_2;\tilde \rho) \},
	\label{eq:phi1}
\end{align}
where $(a_1,a_2) \in \mathbb{A}_0$ are the combination weights,  $\N_1 \sim \mathcal{Z}(\theta_1)$, and $\N_2^{*} \sim \mathcal{Z}(\theta_2)$; the mean parameters $\theta_1$ and $\theta_2$  are defined in Lemmas \ref{lem:strongID} and \ref{lem:weakID} for strong and weak identification, respectively; and $\tilde \rho$ is the limit of $\widehat{\rho}(\beta_0)$.\footnote{Under fixed alternatives, $\tilde \rho = \rho(\beta_0)$; under local alternatives, $\tilde \rho = \rho$.} Let the eigenvalue decomposition of the matrix $\begin{pmatrix}
	a_1 + a_2 \tilde \rho^2 & a_2 \tilde \rho (1- \tilde\rho^2)^{1/2} \\
	a_2\tilde \rho (1-\tilde \rho^2)^{1/2} & 1-a_1 -a_2 \tilde \rho^2
\end{pmatrix}$ be
\begin{align}
	\begin{pmatrix}
		a_1 + a_2 \tilde \rho^2 & a_2 \tilde \rho (1- \tilde \rho^2)^{1/2} \\
		a_2 \tilde \rho (1- \tilde \rho^2)^{1/2} & 1-a_1 -a_2 \tilde \rho^2
	\end{pmatrix} = \mathcal{U}\begin{pmatrix}
		\nu_1(a_1,a_2) & 0 \\
		0 & \nu_2(a_1,a_2)
	\end{pmatrix}   \mathcal{U}^\top
	\label{eq:eig}
\end{align}
where, by construction, $\nu_1(a_1,a_2)\geq \nu_2(a_1,a_2) \geq 0$ and $\mathcal{U}$ is a $2 \times 2$ unitary matrix. We highlight the dependence of eigenvalues $(\nu_1,\nu_2)$ on the weights $(a_1,a_2)$. The dependence of $\mathcal{U}$ on $(a_1,a_2)$ is suppressed for notation simplicity. Then, we have 
\begin{align*}
	a_1 \N_1^2 + a_2 ( \tilde \rho \N_1 + (1- \tilde \rho^2)^{1/2}\N_2^*)^2+ (1-a_1-a_2)\N_2^{*2} = \nu_1(a_1,a_2) \widetilde{\N}_1^2 + \nu_2(a_1,a_2) \widetilde{\N}_2^2
\end{align*}
and $\phi_{a_1,a_2,\infty} = 1\{\nu_1(a_1,a_2) \widetilde{\N}_1^2 + \nu_2(a_1,a_2) \widetilde{\N}_2^2 \geq \mathbb{C}_{\alpha}(a_1,a_2; \tilde \rho)\}$, 
where 
\begin{align}
	\begin{pmatrix}
		\widetilde{\N}_1 \\
		\widetilde{\N}_2 
	\end{pmatrix}    = \mathcal{U}^\top \begin{pmatrix}
		\N_1 \\
		\N_2^*
	\end{pmatrix}
	\label{eq:Nstar}
\end{align}
and $\widetilde{\N}_1$ and $\widetilde{\N}_2$ are independent normal random variables with unit variance. This implies that $\phi_{a_1,a_2,\infty}$ can be viewed as a linear combination test of two independent chi-squared random variables with one degree of freedom, and those two chi-squared random variables are obtained by properly rotating $\N_1$ and $\N_2^*$ (i.e., the limits of $AR(\beta_0)$ and $LM^*(\beta_0)$). 

Theorem \ref{thm:admissible} states the key properties of $\phi_{a_1, a_2, \infty}$ under the limit experiment. 



	
\begin{thm}
	\begin{enumerate}[label=(\roman*)]
		\item Suppose we are under weak identification and fixed alternatives and let $\N_1 \sim \mathcal{Z}(\theta_1)$, $\N_2^{*} \sim \mathcal{Z}(\theta_2)$, and they are independent, where $\theta_1 = m_1(\Delta)$ and $\theta_2 = m_2(\Delta)$ as in \eqref{eq:lmstar_wk}. We consider the test of $H_0: \theta_1 = \theta_2 = 0$ against $H_1: \theta_1 \neq 0$ or $\theta_2 \neq 0$. Let $\Phi_{\alpha}$ denote the class of size-$\alpha$ tests for $H_0: \theta_1 = \theta_2 = 0$ constructed based on $(\widetilde{\N}_1^2,\widetilde{\N}_2^2)$ defined in \eqref{eq:Nstar}. Then, for any $(a_1,a_2) \in \mathbb{A}_0$, $\phi_{a_1,a_2,\infty}$ defined in \eqref{eq:phi1} is an admissible test within $\Phi_{\alpha}$.  In addition, let $(\widetilde \theta_1, \widetilde \theta_2) = ( \theta_1, \theta_2) \mathcal{U}$. If $(\widetilde \theta_1^2, \widetilde \theta_2^2) = b \cdot (\nu_1(a_1,a_2),\nu_2(a_1,a_2))$ for some positive constant $b$, then for any test $\phi \in \Phi_\alpha$, there exists some $\overline{b}>0$ such that for any $0 < b < \overline{b}$, we have $\mathbb{E}\phi \leq \mathbb{E}\phi_{a_1,a_2,\infty}$.
		\item Suppose we are under strong identification and local alternatives and 
		\begin{align*}
			\begin{pmatrix}
				\N_1 \\
				\N_2
			\end{pmatrix}    \stackrel{d}{=}\N \left( \begin{pmatrix}
				0 \\
				\theta 
			\end{pmatrix}, \begin{pmatrix}
				1 & \rho \\
				\rho & 1
			\end{pmatrix}  \right),
		\end{align*}
		where $\theta =\frac{\widetilde{\Delta}\widetilde{\mathcal{C}}}{\Psi^{1/2}} $. We consider the test of $H_0: \theta = 0$ against $H_1: \theta \neq 0$. Then, $\phi_{a_1,a_2,\infty}$ defined in \eqref{eq:phi1} is UMP among the class of level-$\alpha$ tests that are constructed based on  $(\N_1,\N_2)$ and invariant to the sign change if and only if $a_1  = 0$ and $a_2\rho = 0$. In this case, this test is also UMP among the class of unbiased level-$\alpha$ tests that are constructed based on  $(\N_1,\N_2)$.
		\item Suppose Assumption \ref{ass:variance_est} holds, $( Q_{e(\beta_0),e(\beta_0)} - \Delta^2 \mathcal{C}, Q_{X,e(\beta_0)} - \Delta \mathcal{C},Q_{X,X} - \mathcal{C})^\top
		= O_p(1)$, and we are under strong identification with fixed alternatives. If $1 \geq a_{1,n} \geq  \frac{\tilde{q}\Phi_1(\beta_0)}{\mathcal{C}^2\Delta_*^4(\beta_0)}$ for some constant $\tilde{q} > \mathbb{C}_{\alpha,max}(\rho(\beta_0))$ and $(a_{1,n},a_{2,n}) \in \mathbb{A}_0$, where $\Delta_*(\beta_0) = \Phi_1^{1/2}(\beta_0)\Psi^{-1/2}(\beta_0)\rho^{-1}(\beta_0)$,
		then 
		\begin{align*}
			1\{a_{1,n} AR^2(\beta_0) + a_{2,n}LM^2(\beta_0)+ (1-a_{1,n} - a_{2,n})LM^{*2}(\beta_0) \geq \mathbb{C}_{\alpha}(a_{1,n},a_{2,n};\widehat{\rho}(\beta_0))\} \convP 1.
		\end{align*}
	\end{enumerate}
	\label{thm:admissible}
\end{thm}


Several remarks are in order. First, unlike the one-sided jackknife AR test proposed by \cite{MS22}, we construct the jackknife CLC test based on $AR^2(\beta_0)$ for several reasons. First, under weak identification, when the concentration parameter $\mathcal{C}$, and thus, $m_1(\Delta)$ defined in Lemma \ref{lem:weakID} is nonnegative, the one-sided test has good power. However, even in this case, the power curves simulation in Section \ref{sec:sim1} shows that our jackknife CLC test is more powerful than the one-sided AR test in most scenarios. Second, our jackknife CLC test will have good power even when $\mathcal{C}$ is negative.\footnote{We note that $\mathcal{C} = \frac{\sum_{i \in [n]} \sum_{ j \neq i} \Pi_i P_{ij}\Pi_j}{\sqrt{K}} = \frac{ \sum_{i \in [n]}(1-P_{ii})\Pi_i^2 - \Pi^\top M \Pi}{\sqrt{K}}$, where $M = I-P$. If $\Pi^\top M \Pi$ and $\sum_{i \in [n]}P_{ii}\Pi_i^2$ are sufficiently large, $\mathcal{C}$ can be negative. \cite{MS22} further assume that $\Pi^\top M \Pi \leq \frac{C \Pi^\top \Pi}{K}$ for some constant $C>0$, which implies that $\mathcal{C}>0$.} Third, we show below that under strong identification and local alternatives, our jackknife CLC test converges to the UMP test $1\{\N^{*2}_2>\mathbb{C}_{\alpha}\}$ whereas both the one- and two-sided tests based on $AR(\beta_0)$ have no power, as shown in Lemma \ref{lem:strongID}. Fourth, under strong identification and fixed alternatives, our jackknife CLC test has asymptotic power equal to 1, as shown in Lemma \ref{lem:strongID2} and Theorem \ref{thm:strong_fixed} below. In this case, using the one-sided jackknife AR test cannot further improve the power. Fifth, combining $LM^{*2}(\beta_0)$ with $AR^2(\beta_0)$ (and $LM^2(\beta_0)$), rather than $AR(\beta_0)$, can substantially mitigate the impact of power loss of $LM^*(\beta_0)$ at $\Delta_*(\beta_0)$, as shown in the numerical investigation in Section \ref{sec:sim}.


Second, Theorem \ref{thm:admissible}(i) implies that $\phi_{a_1,a_2,\infty}$ is admissible among tests that are also quadratic functions of $\N_1$ and $\N_2^*$ with the same rotation $\mathcal{U}$ but different eigenvalues $(\tilde \nu_1,\tilde \nu_2)$; that is,
\begin{align*}
	(\N_1,\N_2^*) \mathcal{U} \begin{pmatrix}
		\tilde \nu_1 & 0 \\
		0 & \tilde \nu_2
	\end{pmatrix} \mathcal{U}^\top \begin{pmatrix}
		\N_1 \\
		\N_2^*
	\end{pmatrix}.    
\end{align*}
Specifically, in the special case with $a_2 = 0$ (i.e., we put zero weight on $LM^2(\beta_0)$), the rotation matrix $\mathcal{U} = I_2$ and $\phi_{a_1,0,\infty}$ is admissible among level-$\alpha$ tests based on the test statistics of the form $a_1\N_1^2 + (1-a_1)\N_2^{*2}$ for $a_1 \in [0,1]$, 
which is similar to the result for the linear combination of S and K statistics in I.\cite{Andrews(2016)}.

Third, similar to I.\citet[Theorem 2.1]{Andrews(2016)}, Theorem \ref{thm:admissible}(i) also shows that our linear combination test is 
optimal against certain alternatives under weak identification.  
Additionally, in the case with $a_2 = 0$,
the power optimality result in \ref{thm:admissible}(i) also carries over to $\phi_{a_1, 0, \infty}$ among level-$\alpha$ tests of the form $a_1 \mathcal{N}_1^2 + (1-a_1)\mathcal{N}^{*2}_2$ for $a_1 \in [0,1]$.

Fourth, when $a_1 = 0$ and $a_2\rho  = 0$ and under strong identification and local alternatives, we have $\phi_{a_1,a_2,\infty} = 1\{\N_2^{*2} \geq \mathbb{C}_{\alpha}\}$, which is both the UMP invariant and unbiased test. When $\rho=0$ and under local alternatives,
$a_2\mathcal{N}_2^{*2}$ in the second and third terms of $\phi_{a_1,a_2,\infty}$ cancels out, implying that $\phi_{a_1,a_2,\infty} = 1\{\N_2^{*2} \geq \mathbb{C}_{\alpha}\}$ as long as $a_1 = 0$. 

Fifth, we note that both the rotation matrix $\mathcal{U}$ and the eigenvalues $\nu_1$ and $\nu_2$ in \eqref{eq:eig} are functions of $(a_1,a_2)$. We choose this specific parametrization so that $\phi_{a_1,a_2,\infty}$ can be written as a linear combination of $AR^2(\beta_0)$, $LM^2(\beta_0)$, and $LM^{*2}(\beta_0)$. It is possible to use other parametrizations to combine $AR(\beta_0)$ and $LM^*(\beta_0)$. For example, let 
$$\mathcal{O}(\zeta) = \begin{pmatrix}
	\cos(\zeta) & -\sin(\zeta) \\
	\sin(\zeta) & \cos(\zeta)
\end{pmatrix}$$ be a rotation matrix with angle $\zeta$ and $\begin{pmatrix}
	AR^\dagger(\beta_0,\zeta) \\
	LM^\dagger(\beta_0,\zeta) 
\end{pmatrix} = \mathcal{O}(\zeta) \begin{pmatrix}
	AR(\beta_0) \\
	LM^*(\beta_0) 
\end{pmatrix}$. Then, in the limit experiment, the linear combination test statistic can be written as
\begin{align}\label{eq: new-para}
	a \N_1^{\dagger 2} + (1-a)\N_2^{\dagger 2},
\end{align}
where $(\N^{\dagger}_1,\N^{\dagger}_2)$ are the limits of $(AR^\dagger(\beta_0,\zeta) ,LM^\dagger(\beta_0,\zeta) )$ under either weak or strong identification. In the following, we will use a minimax procedure to determine the optimal weights $(a_1,a_2)$ for our jackknife CLC test $\phi_{a_1,a_2,\infty}$. Similarly, we can use this procedure to select the value of $a$ and $\zeta$ for the new parametrization in (\ref{eq: new-para}). Under strong identification and local alternatives, Lemma \ref{lem:ump} shows that the test $1\{LM^{*2}(\beta_0) \geq \mathbb{C}_\alpha\}$ is the most powerful test against local alternatives. This is achieved by our jackknife CLC test $\phi_{a_1,a_2,\infty}$ with $a_1=0$ and $a_2\rho =0$. In this case, the new parametrization does not bring any additional power. 

\section{A Conditional Linear Combination Test}
In this section, we determine the weights $(a_1,a_2)$ in the jackknife CLC test via a minimax procedure. Under weak identification, the limit test statistic of the jackknife CLC test with weights $(a_1,a_2)$ is 
\begin{align}
	\phi_{a_1,a_2,\infty} = 1
	\begin{Bmatrix}
		a_1 \mathcal{Z}_1^2(m_1(\Delta)) + a_2(\rho(\beta_0) \mathcal{Z}_1(m_1(\Delta)) + (1-\rho^2(\beta_0))^{1/2} \mathcal{Z}_2(m_2(\Delta)))^2 \\
		+ (1-a_1 -a_2)\mathcal{Z}_2^2(m_2(\Delta)) \geq \mathbb{C}_{\alpha}(a_1,a_2;\rho(\beta_0))
	\end{Bmatrix},   
	\label{eq:phi_infty1}
\end{align}
where $m_1(\Delta)$ and $m_2(\Delta)$ are defined in Lemma \ref{lem:weakID}, and $\mathcal{Z}_1(\cdot)$ and $\mathcal{Z}_2(\cdot)$ are independent. In this case, we can be explicit and write $\phi_{a_1,a_2,\infty} = \phi_{a_1,a_2,\infty}(\Delta)$. However, the limit power of the jackknife CLC test will typically remain unknown as the true parameter $\beta$ (and hence $\Delta$) is unknown. To overcome this issue, we follow I.\cite{Andrews(2016)} and calibrate the power, i.e,  $\mathbb{E}\phi_{a_1,a_2,\infty}(\delta)$, where $\delta$ ranges over all possible values that $\Delta$ can potentially take; we define $\phi_{a_1,a_2,\infty}(\delta)$ as well as the range of potential values of $\Delta$  below. 

Let $\widehat{D} = Q_{X,X} - (Q_{e(\beta_0),e(\beta_0)},Q_{X,e(\beta_0)}) \begin{pmatrix}
	\widehat{\Phi}_1(\beta_0) & \widehat{\Phi}_{12}(\beta_0) \\
	\widehat{\Phi}_{12}(\beta_0) & \widehat{\Psi}(\beta_0)
\end{pmatrix}^{-1} \begin{pmatrix}
	\widehat{\Phi}_{13}(\beta_0)\\
	\widehat{\tau}(\beta_0) 
\end{pmatrix}$ be the residual from the projection of $Q_{X,X}$ on $(Q_{e(\beta_0),e(\beta_0)},Q_{X,e(\beta_0)})$. By \eqref{eq:limit}, under weak identification, 
\begin{align*}
	\widehat{D} = D + o_p(1), \quad D \stackrel{d}{=} \N(\mu_D,\sigma_D^2),
\end{align*}
where 
\begin{align*}
	\mu_D & = \widetilde{\mathcal{C}}\left[1 - (\Delta^2,\Delta)\left(\begin{pmatrix}
		\Phi_1(\beta_0) & \Phi_{12}(\beta_0) \\
		\Phi_{12}(\beta_0) & \Psi(\beta_0)
	\end{pmatrix}^{-1} \begin{pmatrix}
		\Phi_{13}(\beta_0) \\
		\tau(\beta_0)
	\end{pmatrix}\right) \right] \quad \text{and}\\ 
	\sigma_D^2 & = \Upsilon - \left((\Phi_{13}(\beta_0),\tau(\beta_0))\begin{pmatrix}
		\Phi_1(\beta_0) & \Phi_{12}(\beta_0) \\
		\Phi_{12}(\beta_0) & \Psi(\beta_0)
	\end{pmatrix}^{-1} \begin{pmatrix}
		\Phi_{13}(\beta_0) \\
		\tau(\beta_0) 
	\end{pmatrix}  \right).
\end{align*}
We note that $\widehat{D}$ is a sufficient statistic for $\mu_D$, which contains information about the concentration parameter $\mathcal{C}$ and is asymptotically independent of $AR(\beta_0)$, $LM(\beta_0)$, and hence $LM^*(\beta_0)$. 

Under weak identification, we observe that $m_1(\Delta)$ and $m_2(\Delta)$ in Lemma \ref{lem:weakID} can be written as 
\begin{align}
	\begin{pmatrix}
		m_1(\Delta) \\
		m_2(\Delta)
	\end{pmatrix}  =  \begin{pmatrix}
		C_{1}(\Delta) \\
		C_{2}(\Delta)
	\end{pmatrix} \mu_D,
	\label{eq:C12}
\end{align}
where
\begin{align}
	\begin{pmatrix}
		C_{1}(\Delta) \\
		C_{2}(\Delta)
	\end{pmatrix} & \equiv \begin{pmatrix}
		\Phi_1^{-1/2}(\beta_0)\Delta^2 \\
		(1-\rho^2(\beta_0))^{-1/2}(\Psi^{-1/2}(\beta_0)\Delta -\rho(\beta_0) \Phi_1^{-1/2}(\beta_0)\Delta^2)
	\end{pmatrix} \notag \\
	& \times \left[1 - (\Delta^2,\Delta)\left(\begin{pmatrix}
		\Phi_1(\beta_0) & \Phi_{12}(\beta_0) \\
		\Phi_{12}(\beta_0) & \Psi(\beta_0)
	\end{pmatrix}^{-1} \begin{pmatrix}
		\Phi_{13}(\beta_0) \\
		\tau(\beta_0)
	\end{pmatrix}\right) \right]^{-1}.
	\label{eq:CDelta}
\end{align}
By \eqref{eq:C12}, we see that $  \phi_{a_1,a_2,\infty}=\phi_{a_1,a_2,\infty}(\Delta)$ defined in \eqref{eq:phi1} can be written as
\begin{align*}
	1\begin{Bmatrix}
		a_1 \mathcal{Z}_1^2(C_1(\Delta)\mu_D) + a_2(\rho(\beta_0) \mathcal{Z}_1(C_1(\Delta)\mu_D) + (1-\rho^2(\beta_0))^{1/2} \mathcal{Z}_2(C_2(\Delta)\mu_D))^2 \\
		+ (1-a_1 -a_2)\mathcal{Z}_2^2(C_2(\Delta)\mu_D) \geq \mathbb{C}_{\alpha}(a_1,a_2;\rho(\beta_0))
	\end{Bmatrix}.
\end{align*}

This motivates the definition that 
\begin{align}
	\phi_{a_1,a_2,\infty}(\delta) =   1\begin{Bmatrix}
		a_1 \mathcal{Z}_1^2(C_1(\delta)\mu_D) + a_2(\rho(\beta_0) \mathcal{Z}_1(C_1(\delta)\mu_D) + (1-\rho^2(\beta_0))^{1/2} \mathcal{Z}_2(C_2(\delta)\mu_D))^2 \\
		+ (1-a_1 -a_2)\mathcal{Z}_2^2(C_2(\delta)\mu_D) \geq \mathbb{C}_{\alpha}(a_1,a_2;\rho(\beta_0))
	\end{Bmatrix}.
	\label{eq:phi}
\end{align}
To emphasize the dependence of $\phi_{a_1,a_2,\infty}(\delta)$ on $\mu_D$ and $\gamma(\beta_0)$, we further write $\phi_{a_1,a_2,\infty}(\delta)$ as $\phi_{a_1,a_2,\infty}(\delta,\mu_D,\gamma(\beta_0))$. 

The range of values that $\Delta$ can take is defined as $\mathcal{D}(\beta_0) = \{\delta: \delta+\beta_0 \in \mathcal{B}\}$, where $\mathcal{B}$ is the parameter space. For instance, in their empirical application of returns to education, \cite{MS22} assume that $\beta$ (i.e., the return to education) ranges from -0.5 to 0.5, with $\mathcal{B} = [-0.5,0.5]$. We adopt the same practice in our simulations based on calibrated data in Section \ref{sec:sim2} and empirical application in Section \ref{sec: empirical}. Specifying the parameter space is almost inevitable for any weak-identification-robust inference method, but additional simulation results in Section \ref{sec:add_sim} of the Online Supplement show that our method is insensitive to the choice of parameter space.

Following the lead of I.\cite{Andrews(2016)}, we define the highest attainable power for each $\delta \in \mathcal{D}(\beta_0)$ as 
$\mathcal{P}_{\delta,\mu_D} = \sup_{(a_1,a_2) \in \mathbb{A}(\mu_D,\gamma(\beta_0))} \mathbb{E}\phi_{a_1,a_2,\infty}(\delta,\mu_D,\gamma(\beta_0))$, 
which means that
\begin{align*}
	\mathcal{P}_{\delta,\mu_D} - \mathbb{E}\phi_{a_1,a_2,\infty}(\delta,\mu_D,\gamma(\beta_0))
\end{align*}
is the power loss when the weights are set as $(a_1,a_2)$. Here we denote the domain of $(a_1,a_2)$ as $\mathbb{A}(\mu_D,\gamma(\beta_0))$ and define it as $\mathbb{A}(\mu_D,\gamma(\beta_0)) = \{(a_1,a_2) \in \mathbb{A}_0, a_1 \in [\underline{a}(\mu_D,\gamma(\beta_0)),1] \}$ where $\mathbb{A}_0 = \{(a_1,a_2) \in [0,1] \times [0,1], a_1+a_2\leq \overline{a}\}$ for some $\overline{a}<1$, 
\begin{align}\label{eq:a_underline}
	\underline{a}(\mu_D,\gamma(\beta_0)) = \min\left(p_1, \frac{p_2 \mathbb{C}_{\alpha,\max}(\rho(\beta_0)) \Phi_1(\beta_0) c_{\mathcal{B}}(\beta_0) }{ \Delta_*^4(\beta_0) \mu_D^2} \right),
\end{align}
the two tuning parameters $(p_1,p_2) = (0.01,1.1)$, 
$\Delta_*(\beta_0) = \Phi_1^{1/2}(\beta_0)\Psi^{-1/2}(\beta_0)\rho^{-1}(\beta_0)$ as defined after Lemma \ref{lem:strongID2}, and 
\begin{align*}
	c_{\mathcal{B}}(\beta_0) = \sup_{\delta \in \mathcal{D}(\beta_0)}    \left[1 - (\delta^2,\delta)\left(\begin{pmatrix}
		\Phi_1(\beta_0) & \Phi_{12}(\beta_0) \\
		\Phi_{12}(\beta_0) & \Psi(\beta_0)
	\end{pmatrix}^{-1} \begin{pmatrix}
		\Phi_{13}(\beta_0) \\
		\tau(\beta_0)
	\end{pmatrix}\right) \right]^2.
\end{align*}

The maximum power loss over $\delta \in \mathcal{D}(\beta_0)$ can be viewed as a maximum regret. 
Then, we choose $(a_1,a_2)$ that minimizes the maximum regret; that is, 
\begin{align}
	(a_1(\mu_D,\gamma(\beta_0)),a_2(\mu_D,\gamma(\beta_0))) \in \argmin_{(a_1,a_2) \in \mathbb{A}(\mu_D,\gamma(\beta_0))} \sup_{\delta \in  \mathcal{D}(\beta_0)}(\mathcal{P}_{\delta,\mu_D} -     \mathbb{E}\phi_{a_1,a_2,\infty}(\delta,\mu_D,\gamma(\beta_0))).
	\label{eq:atrue}
\end{align}

Four remarks on the domain of $(a_1,a_2)$ (i.e., $\mathbb{A}(\mu_D,\gamma(\beta_0))$) are in order. First, the lower bound $\underline{a}(\mu_D,\gamma(\beta_0))$ is motivated by Theorem \ref{thm:admissible}(iii). Specifically, we require $p_1 \in (0,1)$ and close to 0 and $p_2>1$. In the Online Supplement, we provide a detailed report on the finite sample performance of our CLC test for both simulation designs analyzed in Section \ref{sec:sim} and the empirical application in Section \ref{sec: empirical}, where we consider different values of $p_1$ and $p_2$. The results indicate that our test's finite sample performance is not affected by the specific values chosen for $(p_1,p_2)$, as all the results are very close to those reported in the main paper. Second, under weak identification, $\mu_D$ is bounded, and $\frac{1.1 \mathbb{C}_{\alpha,\max}(\rho(\beta_0)) \Phi_1(\beta_0) c_{\mathcal{B}}(\beta_0) }{ \Delta_*^4(\beta_0) \mu_D^2}$ may be larger than $0.01$. In this case, we have $\mathbb{A}(\mu_D,\gamma(\beta_0)) = \{(a_1,a_2) \in \mathbb{A}_0, a_1 \in [0.01,1]\}$. Third, under strong identification and local alternatives, $\frac{1.1 \mathbb{C}_{\alpha,\max}(\rho(\beta_0)) \Phi_1(\beta_0) c_{\mathcal{B}}(\beta_0) }{ \Delta_*^4(\beta_0) \mu_D^2}$ will converge to zero so that 
\begin{align*}
	\mathbb{A}(\mu_D,\gamma(\beta_0)) = \left\{(a_1,a_2) \in \mathbb{A}_0, a_1 \in \left[\frac{1.1 \mathbb{C}_{\alpha,\max}(\rho(\beta_0)) \Phi_1(\beta_0) c_{\mathcal{B}}(\beta_0) }{ \Delta_*^4(\beta_0) \mu_D^2},1\right]\right\}.
\end{align*}
We show in Theorem \ref{thm:strongid} below that in this case, the minimax jackknife CLC test converges to $1\{\N_2^{*2} \geq \mathbb{C}_{\alpha}\}$ defined in Lemma \ref{lem:ump}, which is the UMP invariant and unbiased test.
Furthermore, the minimax $a_1$ satisfies the requirement in Theorem \ref{thm:admissible}(iii) with $\tilde{q} = 1.1\mathbb{C}_{\alpha,\max}(\rho(\beta_0))$ so that under strong identification, our CLC test has asymptotic power 1 against fixed alternatives, as shown in Theorem \ref{thm:strong_fixed}.  Fourth, we require $\overline{a}<1$ for some technical reason. In our simulations, we have not observed the minimax $a_1+a_2$ reaching the upper bound. Therefore, setting the upper bound to $\overline{a}$ or $1$ does not have any numerical impact.

Since we cannot observe the values of $\mu_D$ and $\gamma(\beta_0)$ in practice, we adopt the plug-in method described in Section 6 of I.\cite{Andrews(2016)}. Specifically, we replace $\gamma(\beta_0)$ with its consistent estimator $\widehat{\gamma}(\beta_0)$ as specified in Assumption \ref{ass:variance_est}. 
To obtain a proxy of $\mu_D$,\footnote{In fact, as $\phi_{a_1,a_2,\infty}(\delta,\mu_D,\gamma(\beta_0))$ only depends on $\mu_D^2$, we aim to find a good estimator for $\mu_D^2$.} we define 
\begin{align*}
	\widehat{\sigma}_D =     \left(\widehat{\Upsilon} - (\widehat{\Phi}_{13}(\beta_0),\widehat{\tau}(\beta_0))\begin{pmatrix}
		\widehat{\Phi}_1(\beta_0) & \widehat{\Phi}_{12}(\beta_0) \\
		\widehat{\Phi}_{12}(\beta_0) & \widehat{\Psi}(\beta_0)
	\end{pmatrix}^{-1} \begin{pmatrix}
		\widehat{\Phi}_{13}(\beta_0) \\
		\widehat{\tau}(\beta_0) 
	\end{pmatrix}\right)^{1/2}, 
\end{align*}
which is a function of $\widehat{\gamma}(\beta_0)$ and a consistent estimator of $\sigma_D$ by Assumption \ref{ass:variance_est}. 
Then, under weak identification, we have $\widehat{D}^2/\widehat{\sigma}_D^2 = D^2 /\sigma_D^2 + o_p(1) \stackrel{d}{=} \mathcal{Z}^2(\mu_D/\sigma_D) + o_p(1)$ and $D^2 /\sigma_D^2$ is a sufficient statistic for $\mu_D^2$. Let $\widehat{r} = \widehat{D}^2/\widehat{\sigma}_D^2$. We consider two estimators for $\mu_D$ as functions of $\widehat{D}$ and $\widehat{\sigma}_D$, namely, $f_{pp}(\widehat{D},\widehat{\gamma}(\beta_0)) = \widehat{\sigma}_D \sqrt{\widehat{r}_{pp}}$ and $f_{krs}(\widehat{D},\widehat{\gamma}(\beta_0)) = \widehat{\sigma}_D \sqrt{\widehat{r}_{krs}}$, where $\widehat{r}_{pp} = \max(\widehat{r}-1,0)$ and
\begin{align*}
	\widehat{r}_{krs} =  \widehat{r}-1 + \exp\left(-\frac{\widehat{r}}{2}\right)\left(\sum_{j=0}^\infty\left( -\frac{\widehat{r}}{2}\right)^j \frac{1}{j!(1+2j)}  \right)^{-1}.
\end{align*}
Specifically, \cite{krs93} show that $\widehat{r}_{krs}$ is positive as long as $\widehat{r}>0$ and $\widehat{r} \geq \widehat{r}_{krs} \geq \widehat{r}-1$. It is also possible to consider the MLE based on a single observation $\widehat{D}^2/\widehat{\sigma}_D^2$. However, such an estimator is harder to use because it does not have a closed-form expression. 

In practice, we estimate $\mathbb{E}\phi_{a_1,a_2,\infty}(\delta,\mu_D,\gamma(\beta_0))$ by $ \mathbb{E}^*\phi_{a_1,a_2,s}(\delta,\widehat{D},\widehat{\gamma}(\beta_0))$ for $s \in \{pp,krs\}$, where 
\begin{align}
	& \phi_{a_1,a_2,s}(\delta,\widehat{D},\widehat{\gamma}(\beta_0)) \notag \\
	& =  1\begin{Bmatrix}
		a_1 \mathcal{Z}_1^2(\widehat{C}_1(\delta)f_s(\widehat{D},\widehat{\gamma}(\beta_0))) \\
		+ a_2\left[\widehat \rho(\beta_0) \mathcal{Z}_1(\widehat{C}_1(\delta)f_s(\widehat{D},\widehat{\gamma}(\beta_0))) + (1-\widehat \rho^2(\beta_0))^{1/2} \mathcal{Z}_2(\widehat{C}_2(\delta)f_s(\widehat{D},\widehat{\gamma}(\beta_0)))\right]^2 \\
		+ (1-a_1 -a_2)\mathcal{Z}_2^2(\widehat{C}_2(\delta)f_s(\widehat{D},\widehat{\gamma}(\beta_0)) \geq \mathbb{C}_{\alpha}(a_1,a_2;\widehat{\rho}(\beta_0))
	\end{Bmatrix},
	\label{eq:phias}
\end{align}
and $(\widehat{C}_{1}(\delta),\widehat{C}_{2}(\delta))$ are similarly defined as $(C_1(\delta),C_2(\delta))$ in \eqref{eq:CDelta} with $\gamma(\beta_0)$ replaced by $\widehat{\gamma}(\beta_0)$; that is, 
\begin{align*}
	\begin{pmatrix}
		\widehat{C}_{1}(\delta) \\
		\widehat{C}_{2}(\delta)
	\end{pmatrix} & \equiv \begin{pmatrix}
		\widehat{\Phi}_1^{-1/2}(\beta_0)\delta^2 \\
		(1-\widehat{\rho}^2(\beta_0))^{-1/2}(\widehat{\Psi}^{-1/2}(\beta_0)\delta -\widehat{\rho}(\beta_0) \widehat{\Phi}_1^{-1/2}(\beta_0)\delta^2)
	\end{pmatrix} \notag \\
	& \times \left[1 - (\delta^2,\delta)\left(\begin{pmatrix}
		\widehat{\Phi}_1(\beta_0) & \widehat{\Phi}_{12}(\beta_0) \\
		\widehat{\Phi}_{12}(\beta_0) & \widehat{\Psi}(\beta_0)
	\end{pmatrix}^{-1} \begin{pmatrix}
		\widehat{\Phi}_{13}(\beta_0) \\
		\widehat{\tau}(\beta_0)
	\end{pmatrix}\right) \right]^{-1}.    
\end{align*}

Let $\mathcal{P}_{\delta,s}(\widehat{D},\widehat{\gamma}(\beta_0)) = \sup_{(a_1,a_2) \in \mathbb{A}(f_{s}(\widehat{D},\widehat{\gamma}(\beta_0)),\widehat{\gamma}(\beta_0)) } \mathbb{E}^*\phi_{a_1,a_2,s}(\delta,\widehat{D},\widehat{\gamma}(\beta_0))$. Then, for $s \in \{pp,krs\}$, we can estimate $a(\mu_D,\gamma(\beta_0))$ in \eqref{eq:atrue} by $\mathcal{A}_s(\widehat{D},\widehat{\gamma}(\beta_0)) = (\mathcal{A}_{1,s}(\widehat{D},\widehat{\gamma}(\beta_0)),\mathcal{A}_{2,s}(\widehat{D},\widehat{\gamma}(\beta_0)))$ defined as 
\begin{align}
	\mathcal{A}_s(\widehat{D},\widehat{\gamma}(\beta_0)) \in  \argmin_{(a_1,a_2) \in \mathbb{A}(f_{s}(\widehat{D},\widehat{\gamma}(\beta_0)),\widehat{\gamma}(\beta_0))} \sup_{\delta \in  \mathcal{D}(\beta_0)}(\mathcal{P}_{\delta,s}(\widehat{D},\widehat{\gamma}(\beta_0)) -     \mathbb{E}^*\phi_{a_1,a_2,s}(\delta,\widehat{D},\widehat{\gamma}(\beta_0))),
	\label{eq:afeasible}
\end{align}
where $\phi_{a_1,a_2,s}(\delta,\widehat{D},\widehat{\gamma}(\beta_0))$ is defined in \eqref{eq:phias},
\begin{align*}
	\mathbb{A}(f_{s}(\widehat{D},\widehat{\gamma}(\beta_0)),\widehat{\gamma}(\beta_0)) = \{(a_1,a_2) \in \mathbb{A}_0, a_1 \in [\underline{a}(f_{s}(\widehat{D},\widehat{\gamma}(\beta_0)),\widehat{\gamma}(\beta_0)),\overline{a}]\},    
\end{align*}
\begin{align*}
	\underline{a}(f_{s}(\widehat{D},\widehat{\gamma}(\beta_0)),\widehat{\gamma}(\beta_0)) =  \min\left(0.01, \frac{1.1 \mathbb{C}_{\alpha,\max}(\widehat{\rho}(\beta_0)) \widehat{\Phi}_1(\beta_0) \widehat{c}_{\mathcal{B}}(\beta_0)}{ \widehat{\Delta}_*^4(\beta_0) f_{s}^2(\widehat{D},\widehat{\gamma}(\beta_0))} \right),
\end{align*}
\begin{align*}
	\widehat{c}_{\mathcal{B}}(\beta_0) = \sup_{\delta \in \mathcal{D}(\beta_0)}    \left[1 - (\delta^2,\delta)\left(\begin{pmatrix}
		\widehat{\Phi}_1(\beta_0) & \widehat{\Phi}_{12}(\beta_0) \\
		\widehat{\Phi}_{12}(\beta_0) & \widehat{\Psi}(\beta_0)
	\end{pmatrix}^{-1} \begin{pmatrix}
		\widehat{\Phi}_{13}(\beta_0) \\
		\widehat{\tau}(\beta_0)
	\end{pmatrix}\right) \right]^2, 
\end{align*}
and $\widehat{\Delta}_*(\beta_0) = \widehat{\Phi}_1^{1/2}(\beta_0)\widehat{\Psi}^{-1/2}(\beta_0) \widehat{\rho}^{-1}(\beta_0)$. Then, the feasible jackknife CLC test is, for  $s \in \{pp,krs\}$, 
\begin{align}
	\widehat{\phi}_{\mathcal{A}_s(\widehat{D},\widehat{\gamma}(\beta_0))} =  1 \begin{Bmatrix}
		\mathcal{A}_{1,s}(\widehat{D},\widehat{\gamma}(\beta_0)) AR^{2}(\beta_0) + \mathcal{A}_{2,s}(\widehat{D},\widehat{\gamma}(\beta_0)) LM^{2}(\beta_0) \\ + (1-\mathcal{A}_{1,s}(\widehat{D},\widehat{\gamma}(\beta_0))-\mathcal{A}_{2,s}(\widehat{D},\widehat{\gamma}(\beta_0))) LM^{*2}(\beta_0) \geq \mathbb{C}_{\alpha}(\mathcal{A}_s(\widehat{D},\widehat{\gamma}(\beta_0));\widehat{\rho}(\beta_0))
	\end{Bmatrix}.
	\label{eq:phihat}
\end{align}

\section{Asymptotic Properties}
We first consider the asymptotic properties of the jackknife CLC test under weak identification and fixed alternatives, in which $\mathcal{C}\rightarrow \widetilde{\mathcal{C}}$ and $\Delta$ is treated as fixed so that we have
\begin{align*}
	\widehat{D} \convD D \stackrel{d}{=} \N(\mu_D,\sigma_D^2).
\end{align*}
We see from \eqref{eq:atrue} and \eqref{eq:afeasible} that $\mathcal{A}_s(d,r) = (a_1(f_s(d,r),r),a_2(f_s(d,r),r))$ is a function of $(d,r) \in \Re \times \Gamma$, where $\Gamma$ is the parameter space for $\gamma(\beta_0)$ and $s \in \{pp,krs\}$. We make the following assumption on  $\mathcal{A}_s(\cdot)$. 

\begin{ass}
	Let $\mathcal{S}_s$ be the set of discontinuities of $\mathcal{A}_s(\cdot,\gamma(\beta_0)): \Re \mapsto [0,1] \times [0,1]$. Then, we assume $\mathcal{A}_s(d,r)$ is continuous in $r$ for any $d \in \Re/\mathcal{S}_s$, and the Lebesgue measure of $\mathcal{S}_s$ is zero for $s \in \{pp,krs\}$. 
	\label{ass:a}
\end{ass}

Assumption \ref{ass:a} is a technical condition that allows us to apply the continuous mapping theorem. It is mild because $\mathcal{A}_s(\cdot)$ is allowed to be discontinuous in its first argument. In practice, we can approximate $\mathcal{A}_s(\cdot)$ by a step function defined over a grid of $d$ so that there is a finite number of discontinuities. The continuity of $\mathcal{A}_s(\cdot)$ in its second argument is due to the smoothness of the bivariate normal PDF with respect to the covariance matrix. Therefore, in this case, Assumption \ref{ass:a} holds automatically. 

\begin{thm}
	Suppose we are under weak identification and fixed alternatives and that Assumptions \ref{ass:weak_convergence}--\ref{ass:a} hold. Then,  for $s\in  \{pp,krs\}$, 
	\begin{align*}
		\mathcal{A}_s(\widehat{D},\widehat{\gamma}(\beta_0)) & \convD \mathcal{A}_s(D,\gamma(\beta_0))=  (a_1(f_s(D,\gamma(\beta_0)),\gamma(\beta_0)),a_2(f_s(D,\gamma(\beta_0)),\gamma(\beta_0)))
	\end{align*}
	and\footnote{We assume that $\frac{C}{0} = +\infty$ if $C>0$ and $\min(C,+\infty) = C.$} 
	\begin{align*}
		\mathbb{E}\widehat{\phi}_{\mathcal{A}_s(\widehat{D},\widehat{\gamma}(\beta_0))} \rightarrow \mathbb{E}\phi_{a_1(f_s(D,\gamma(\beta_0)),\gamma(\beta_0)),a_2(f_s(D,\gamma(\beta_0)),\gamma(\beta_0)),\infty}(\Delta,\mu_D,\gamma(\beta_0)),
	\end{align*}
	where $\phi_{a_1,a_2,\infty}(\delta)$ is defined in \eqref{eq:phi} and $a_l(f_s(D,\gamma(\beta_0)),\gamma(\beta_0))$ is interpreted as $a_l(\mu_D,\gamma(\beta_0))$ defined in \eqref{eq:atrue} with $\mu_D$ replaced by $f_s(D,\gamma(\beta_0))$ for $l = 1,2$. 
	
	In addition, let $BL_1$ be the class of functions $h(\cdot)$ of $D$ that is bounded and Lipschitz with Lipschitz constant 1. Then, if the null hypothesis holds such that $\Delta = 0$, we have
	\begin{align*}
		\mathbb{E}(\widehat{\phi}_{\mathcal{A}_s(\widehat{D},\widehat{\gamma}(\beta_0))} - \alpha)h(\widehat{D}) \rightarrow 0, \quad \forall h \in BL_1.
	\end{align*}
	\label{thm:weakid}
\end{thm}

Several remarks on Theorem \ref{thm:weakid} are in order. First, we see that the power of our jackknife CLC test is  $\mathbb{E}\phi_{\mathcal{A}_s(D,\gamma(\beta_0)),\infty}(\Delta,\mu_D,\gamma(\beta_0))$, which does not exactly match the minimax power $$\mathbb{E}\phi_{a_1(\mu_D,\gamma(\beta_0)),a_2(\mu_D,\gamma(\beta_0)),\infty}(\Delta,\mu_D,\gamma(\beta_0))$$ in the limit problem. This is because under weak identification, it is impossible to consistently estimate $\mu_D$, or equivalently, the concentration parameter. A similar result holds under weak identification with a fixed number of moment conditions in I.\cite{Andrews(2016)}. The best we can do is to approximate $\mu_D$ by reasonable estimators based on $D$ such as $f_{pp}(D,\gamma(\beta_0))$ and $f_{krs}(D,\gamma(\beta_0))$. Second, Theorem \ref{thm:weakid} implies that our jackknife CLC test controls size asymptotically conditionally on $\widehat{D}$, and thus, unconditionally. Last, according to Theorem \ref{thm:weakid}, the CLC test's asymptotic power, with weights $(a_1,a_2)$ chosen through the minimax procedure, is equivalent to the limit experiment's asymptotic power when the weights are $\mathcal{A}_s(D,\gamma(\beta_0))$, which is a function of $D$. As $D$ is independent of the normal random variables in $\phi_{a_1,a_2,\infty}(\delta)$ in \eqref{eq:phi}, the two optimality results stated in Theorem \ref{thm:admissible}(i) also hold asymptotically, conditional on $\widehat{D}$. To make this statement precise, we define the eigenvalue decomposition 
\begin{align}
	& \begin{pmatrix}
		\mathcal{A}_{1,s}(\widehat{D},\widehat{\gamma}(\beta_0)) + \mathcal{A}_{2,s}(\widehat{D},\widehat{\gamma}(\beta_0)) \hat \rho^2(\beta_0) & \mathcal{A}_{2,s}(\widehat{D},\widehat{\gamma}(\beta_0)) \hat \rho(\beta_0) (1- \hat \rho^2(\beta_0))^{1/2} \\
		\mathcal{A}_{2,s}(\widehat{D},\widehat{\gamma}(\beta_0)) \hat \rho(\beta_0) (1- \hat \rho^2(\beta_0))^{1/2} & 1-\mathcal{A}_{1,s}(\widehat{D},\widehat{\gamma}(\beta_0)) - \mathcal{A}_{2,s}(\widehat{D},\widehat{\gamma}(\beta_0)) \hat \rho^2(\beta_0)
	\end{pmatrix} \notag \\
	& = \mathcal{U}_s(\widehat{D},\widehat{\gamma}(\beta_0))\begin{pmatrix}
		\nu_{1,s}(\widehat{D},\widehat{\gamma}(\beta_0)) & 0 \\
		0 & \nu_{2,s}(\widehat{D},\widehat{\gamma}(\beta_0))
	\end{pmatrix}   \mathcal{U}_s(\widehat{D},\widehat{\gamma}(\beta_0))^\top. 
	\label{eq:eig'}
\end{align}
Define a class of tests 
\begin{align*}
	\Phi_{\alpha} = \begin{Bmatrix}
		\tilde \phi(\mathcal Z_1^2,\mathcal Z_2^2,d,r): \mathbb{E}\tilde \phi(\mathcal{Z}_1^2,\mathcal{Z}_2^2, d,  r) \leq \alpha, \text{ for any } (d,r) \in \Re \times \Gamma, \\
		\tilde \phi(\mathcal Z_1^2,\mathcal Z_2^2,d,r) \text{ is continuous in $r$}, \\
		\text{the discontinuities of $\tilde \phi(\mathcal Z_1^2,\mathcal Z_2^2,d,r)$ w.r.t. }\\
		\text{the first three arguments have zero Lebesgue measure}
	\end{Bmatrix},
\end{align*}
where $(\mathcal{Z}_1,\mathcal{Z}_2)$ are two independent standard normal random variables. Further define, for $s \in \{pp,krs\}$, 
\begin{align*}
	\begin{pmatrix}
		\widetilde {AR}_s(\beta_0) \\
		\widetilde {LM}^{*}_s(\beta_0)
	\end{pmatrix} =   \mathcal{U}_s(\widehat{D},\widehat{\gamma}(\beta_0))^\top  \begin{pmatrix}
		AR(\beta_0) \\
		LM^{*}(\beta_0)
	\end{pmatrix}.
\end{align*}

\begin{ass}
	Suppose $\mathcal{U}_s(d,r)$ is continuous in $r$ and 
	the set of discontinuities of $\mathcal{U}_s(\cdot)$ w.r.t. its first argument has zero Lebesgue measure.
	\label{ass:V}
\end{ass}

\begin{cor}
	Suppose we are under weak identification and fixed alternatives and that Assumptions \ref{ass:weak_convergence}--\ref{ass:V} hold. Let $\tilde \phi(\cdot) \in \Phi_\alpha$ and for any $d \in \Re$, denote $(\theta_1,\theta_2) = (m_1(\Delta),m_2(\Delta))\mathcal{U}_s(d,\gamma(\beta_0))$. Then, the following two optimality results hold. 
	\begin{enumerate}[label=(\roman*)]
		\item If for some $d \in \Re$ and $s \in \{pp,krs\}$, we have  
		\begin{align*}
			\lim_{\eps \rightarrow 0} \lim_{n\rightarrow \infty} &  \frac{\mathbb{E} \tilde \phi(\widetilde {AR}_s^2(\beta_0),
				\widetilde {LM}^{*2}_s(\beta_0),\widehat{D},\widehat{\gamma}(\beta_0)) 1\{ |\widehat D -d| \leq \eps \} }{\mathbb{E}  1\{ |\widehat D -d| \leq \eps \} } \\
			& \geq     \lim_{\eps \rightarrow 0} \lim_{n\rightarrow \infty}   \frac{\mathbb{E} \hat \phi_{\mathcal{A}_s(\widehat D, \widehat \gamma(\beta_0))}1\{ |\widehat D -d| \leq \eps \} }{\mathbb{E}  1\{ |\widehat D -d| \leq \eps \} },
		\end{align*} 
		for all $(\theta_1,\theta_2) \in \Re^2$, 
		then 
		\begin{align*}
			\lim_{\eps \rightarrow 0} \lim_{n\rightarrow \infty}  & \frac{\mathbb{E} \tilde \phi(\widetilde {AR}_s^2(\beta_0),
				\widetilde {LM}^{*2}_s(\beta_0),\widehat{D},\widehat{\gamma}(\beta_0)) 1\{ |\widehat D -d| \leq \eps \} }{\mathbb{E}  1\{ |\widehat D -d| \leq \eps \} } \\
			& =     \lim_{\eps \rightarrow 0} \lim_{n\rightarrow \infty}   \frac{\mathbb{E} \hat \phi_{\mathcal{A}_s(\widehat D, \widehat \gamma(\beta_0))}1\{ |\widehat D -d| \leq \eps \} }{\mathbb{E}  1\{ |\widehat D -d| \leq \eps \} },
		\end{align*} 
		for all $(\theta_1,\theta_2) \in \Re^2$.
		\item If $( \theta_1^2,  \theta_2^2) = b \cdot (\nu_1(d,\gamma(\beta_0)),\nu_2(d,\gamma(\beta_0)))$ for some positive constant $b$, then there exists $\overline{b}>0$ such that if $0<b<\overline{b}$, we have
		\begin{align*}
			\lim_{\eps \rightarrow 0} \lim_{n\rightarrow \infty}  & \frac{\mathbb{E} \tilde \phi(\widetilde {AR}_s^2(\beta_0),
				\widetilde {LM}^{*2}_s(\beta_0),\widehat{D},\widehat{\gamma}(\beta_0)) 1\{ |\widehat D -d| \leq \eps \} }{\mathbb{E}  1\{ |\widehat D -d| \leq \eps \} } \\
			& \leq    \lim_{\eps \rightarrow 0} \lim_{n\rightarrow \infty}   \frac{\mathbb{E} \hat \phi_{\mathcal{A}_s(\widehat D, \widehat \gamma(\beta_0))}1\{ |\widehat D -d| \leq \eps \} }{\mathbb{E}  1\{ |\widehat D -d| \leq \eps \} },
		\end{align*} 
	\end{enumerate}
	\label{cor:weakid}
\end{cor}

Corollary \ref{cor:weakid} shows that under weak identification and fixed alternatives, our jackknife CLC test is asymptotically admissible and optimal against certain alternatives conditional on $\widehat D$.

Next, we consider the performance of $\widehat{\phi}_{\mathcal{A}_s(\widehat{D},\widehat{\gamma}(\beta_0))} $ defined in \eqref{eq:phihat} under strong identification and local alternatives. To precisely state the optimality result, we further consider the class of level-$\alpha$ tests against $\theta =0$ v.s. the two-sided alternative that are constructed based on one observation of $(\N_1,\N_2)$, where $\theta = \widetilde \Delta \widetilde {\mathcal{C}} \Psi^{-1/2}$ and 
\begin{align*}
	\begin{pmatrix}
		\N_1 \\
		\N_2
	\end{pmatrix}    \stackrel{d}{=}\N \left( \begin{pmatrix}
		0 \\
		\theta 
	\end{pmatrix}, \begin{pmatrix}
		1 & \rho \\
		\rho & 1
	\end{pmatrix}  \right),
\end{align*}
Specifically, denote 
\begin{align*}
	\Phi_{\alpha}^I = \begin{Bmatrix}
		\phi(\cdot): \mathbb{E}\phi(\N_1,\N_2) \leq \alpha \quad \text{under the null}, \\
		\phi(\N_1,\N_2) = \phi(\N_1,-\N_2), \\
		\text{the discontinuities of $\phi(\cdot)$ has zero Lebesgue measure} 
	\end{Bmatrix}    
\end{align*}
and 
\begin{align*}
	\Phi_{\alpha}^U = \begin{Bmatrix}
		\phi(\cdot): \mathbb{E}\phi(\N_1,\N_2) \leq \alpha \quad \text{under the null}, \\
		\mathbb{E}\phi(\N_1,\N_2) \geq \alpha \quad \text{under the alternative}, \\
		\text{the discontinuities of $\phi(\cdot)$ has zero Lebesgue measure} 
	\end{Bmatrix}    
\end{align*}
as the classes of sign-invariant and unbiased tests, respectively. 


\begin{thm}
	Suppose that Assumptions \ref{ass:weak_convergence} and \ref{ass:variance_est} hold. 
	Further suppose that we are under strong identification and local alternatives as described in Lemma \ref{lem:strongID}. Then,  for $s \in \{pp,krs\}$, we have
	\begin{align*}
		\mathcal{A}_{1,s}(\widehat{D},\widehat{\gamma}(\beta_0)) \convP 0, \quad \mathcal{A}_{2,s}(\widehat{D},\widehat{\gamma}(\beta_0))\rho \convP 0, \quad \text{and} \quad   \widehat{\phi}_{\mathcal{A}_s(\widehat{D},\widehat{\gamma}(\beta_0))}  \convD 1\{\N_2^{*2} \geq \mathbb{C}_{\alpha}\},  
	\end{align*}
	where $\N_2^{*} \stackrel{d}{=} \N\left(\frac{\widetilde{\Delta} \widetilde{\mathcal{C}}}{[(1-\rho^2)\Psi]^{1/2}},1\right)$. 
	In addition, suppose $\breve{\phi}_n = \phi(AR(\beta_0),LM(\beta_0)) +o_P(1)$ for some $\phi \in \Phi_{\alpha}^I \cup \Phi_{\alpha}^U$ and the sequence $\{\breve{\phi}_n \}_{n \geq 1}$ is uniformly integrable. Then,  we have
	\begin{align*}
		\lim_{n \rightarrow \infty}    \mathbb{E}\widehat{\phi}_{\mathcal{A}_s(\widehat{D},\widehat{\gamma}(\beta_0))} = \sup_{\phi \in \Phi_{\alpha}^I \cup \Phi_{\alpha}^U } \lim_{n \rightarrow \infty} \mathbb{E}\phi(AR(\beta_0),LM(\beta_0)) \geq  \lim_{n \rightarrow \infty} \mathbb{E}\breve{\phi}_n. 
	\end{align*}
	\label{thm:strongid}
\end{thm}

Five remarks are in order. First, under strong identification, $\mu_D$, and thus, $D$ approaches infinity, and so does our estimator $\widehat D$. This is how our estimator $\widehat D$ can detect the identification strength. 
In addition, we show in the proof of Theorem \ref{thm:strongid}
that under strong identification, the calibrated power gap $\mathcal{P}_{\delta,s}(\widehat{D},\widehat{\gamma}(\beta_0)) -     \mathbb{E}^*\phi_{a_1,a_2,s}(\delta,\widehat{D},\widehat{\gamma}(\beta_0))$ is maximized when $\delta$ is in the region of local alternatives. However, in this region, as shown by Lemma \ref{lem:ump}, the maximum power gap can achieve zero if all the weights are put on $LM^*(\beta_0)$, which leads to the first result in Theorem \ref{thm:strongid}. 
Second, our jackknife CLC test is adaptive to identification strength. In practice, econometricians do not know whether the true value $\beta$ is close to the null $\beta_0$. Therefore, our jackknife CLC test calibrates power across all possible values of $\delta$ (i.e., $\delta \in \mathcal{D}(\beta_0)$), which include both local and fixed alternatives. Yet, Theorem \ref{thm:strongid} shows that the minimax procedure can produce the most powerful test as if it is known that $\beta$ belongs to the region of local alternatives. Third, Theorem \ref{thm:strongid} shows that under strong identification and local alternatives, our jackknife CLC test converges to the UMP level-$\alpha$ test that is either invariant to the sign change or unbiased and constructed based on $AR(\beta_0)$ and $LM(\beta_0)$. Therefore, it is more powerful than the jackknife AR and LM tests. Fourth, under strong identification and local alternatives, the JIVE-based Wald test proposed by \cite{Chao(2012)} is asymptotically equivalent to the jackknife LM test, which implies that the jackknife AR and JIVE-Wald-based two-step test in \cite{MS22} is also dominated by the jackknife CLC test. Fifth, consider the HLIM based Wald test statistic proposed by \cite{Haus2012}, which is denoted as $W_h(\beta_0)$. In Section \ref{sec:HLIM} in the Online Supplement, we show that, under local alternative and strong identification, 
\begin{align*}
	W_h(\beta_0) =  \frac{\Psi^{1/2}}{\Psi_h^{1/2}}LM(\beta_0)- \frac{\tilde{\rho}\Phi_1^{1/2} }{\Psi_h^{1/2}}AR(\beta_0) + o_P(1), 
\end{align*}
where $\tilde{\rho} = plim_{n\rightarrow \infty} X^\top e(\beta_0)/(e(\beta_0)^\top e(\beta_0))$ and $\Psi_h = \Psi - 2 \tilde \rho \Phi_{12} + \tilde \rho^2 \Phi_1$ is the corresponding asymptotic variance.  Then, by letting $\breve{\phi}_n = 1\{W_h^2(\beta_0) \geq \mathbb{C}_{\alpha}\}$ and 
\begin{align*}
	\phi(AR(\beta_0),LM(\beta_0)) = 1\left\{\left[\frac{\Psi^{1/2}}{\Psi_h^{1/2}}LM(\beta_0)- \frac{\tilde{\rho}\Phi_1^{1/2} }{\Psi_h^{1/2}}AR(\beta_0)\right]^2 \geq \mathbb{C}_{\alpha}\right\},    
\end{align*}
Theorem \ref{thm:strongid} implies our jackknife CLC test is more powerful than the HLIM based Wald test under strong identification against local alternatives. In fact, by direct calculation, we can see that, for $\theta = \widetilde \Delta \widetilde{\mathcal{C}}\Psi^{-1/2}$,   
\begin{align*}
	\frac{\Psi^{1/2}}{\Psi_h^{1/2}}LM(\beta_0)- \frac{\tilde{\rho}\Phi_1^{1/2} }{\Psi_h^{1/2}}AR(\beta_0) \convD \mathcal{Z}(\tilde{\theta}),\quad \text{where} \quad \tilde \theta^2 = \frac{\theta^2}{1-\rho^2 + \left(\tilde \rho \Phi_1^{1/2}\Psi^{-1/2} - \rho\right)^2} \leq \frac{\theta^2}{(1-\rho^2)}. 
\end{align*}
The noncentrality parameter for the HLIM based Wald test is weakly smaller than that of the CLC test, which explains the power comparison. The equality holds if $\tilde \rho \Phi_1^{1/2}\Psi^{-1/2} = \rho$, which further holds in the special case of many weak IVs and homoskedasticity in the sense that $\Pi^\top \Pi/K = o(1)$ and $\mathbb{E}(V_i,e_i)^\top (V_i,e_i)$ does not vary across $i$. 

Combining Theorems \ref{thm:weakid} and \ref{thm:strongid}, we can show the uniform size control of our jackknife CLC test no matter the identification is strong or weak. Let $\lambda_n \in \Lambda_n$ be the data generating process of $n$ observations of $(e,V,Z)$.
Under $\lambda_n$, the covariance matrix of $(Q_{e,e},Q_{X,e},Q_{X,X})$ is denoted as $\mathbb V_n$. 
We impose the following restriction on the sequence of classes of DGPs ($\{\Lambda_n \}_{n \geq 1}$):\footnote{In \eqref{eq:Lambda}, we focus on the model without exogenous control variables. The independence and moment conditions for $(e_i,V_i)$  are sufficient for Assumption \ref{ass:weak_convergence}. We further verify in Section \ref{sec:W0} of the Online Supplement that the joint asymptotic normality (Assumption 1) holds in the case with exogenous controls. } 
\begin{align}\label{eq:Lambda}
\begin{pmatrix}
	&  \{ V_i, e_i \}_{i \in [n]} \; \text{are independent}, \; \mathbb{E}e_i = \mathbb{E}V_i = 0,\\
	& \max_i \mathbb{E} e_i^4 + \max_i \mathbb{E} V_i^4 \leq C_1 <\infty, \\
	&       \mathcal{C}_n = \frac{1}{\sqrt{K}}\sum_{i \in [n]}\sum_{j \neq i} \Pi_i P_{ij} \Pi_j \in \Re, \\
	& P_{ii} \leq C_2 < 1, \\
	& 0 < \kappa_1 \leq \text{mineig}(\mathbb V_n) \leq \text{maxeig}(\mathbb V_n) \leq \kappa_2 < \infty, \\
	& \text{where $C_1$, $C_2$, $\kappa_1$, and $\kappa_2$ are some fixed constants},  \\
	& \text{and Assumption \ref{ass:variance_est} holds for $\beta_0=\beta$.}
\end{pmatrix}
\end{align}
In Sections \ref{sec:var1} and \ref{sec:var2} of the Online Supplement, we further verify that Assumption \ref{ass:variance_est} holds, respectively,  
for the standard variance estimators, which follow the construction in \cite{crudu2021}, and the cross-fit variance estimators, which follow \cite{MS22}. 
Theorem \ref{thm:uniform_size} shows that our jackknife CLC test has correct asymptotic size, under similar arguments as those in \cite{ACG2020} and I.\cite{Andrews(2016)}.

\begin{thm}
Suppose Assumption \ref{ass:a} holds, $\{\Lambda_n \}_{n \geq 1}$ satisfies \eqref{eq:Lambda}, and we are under the null hypothesis that $\beta_0 = \beta$. Then, we have
\begin{align*}
	\liminf_{n \rightarrow \infty} \inf_{\lambda_n \in \Lambda_n} \mathbb{E}_{\lambda_n}(\widehat{\phi}_{\mathcal{A}_s(\widehat{D},\widehat{\gamma}(\beta_0))}) =     \limsup_{n \rightarrow \infty} \sup_{\lambda_n \in \Lambda_n} \mathbb{E}_{\lambda_n}(\widehat{\phi}_{\mathcal{A}_s(\widehat{D},\widehat{\gamma}(\beta_0))}) = \alpha.
\end{align*}
\label{thm:uniform_size}
\end{thm}

Last, we show that, under strong identification, the jackknife CLC test $\widehat{\phi}_{\mathcal{A}_s(\widehat{D},\widehat{\gamma}(\beta_0))}$ defined in \eqref{eq:phihat} has asymptotic power 1 against fixed alternatives. 
\begin{thm}
Suppose Assumption \ref{ass:variance_est} holds, and $( Q_{e(\beta_0),e(\beta_0)} - \Delta^2 \mathcal{C}, Q_{X,e(\beta_0)} - \Delta \mathcal{C},Q_{X,X} - \mathcal{C})^\top
= O_p(1)$. Further suppose that we are under strong identification with fixed alternatives so that $\Delta = \beta - \beta_0$ is nonzero and fixed. Then, we have
\begin{align*}
	\widehat{\phi}_{\mathcal{A}_s(\widehat{D},\widehat{\gamma}(\beta_0))} \convP 1.     
\end{align*}
\label{thm:strong_fixed}
\end{thm}

\section{Simulation}
\label{sec:sim}
\subsection{Power Curve Simulation for the Limit Problem}
\label{sec:sim1}

In this section, we present simulation results to compare the power performance of various tests under the limit problem described in Section \ref{sec:limit}. We consider  the following tests with a nominal rate of $5\%$: (i) our jackknife CLC test, where $\mu_D$ is estimated using either $pp$ or $krs$ method, (ii) the one-sided jackknife AR test defined in (\ref{eq:AR}), (iii) the jackknife LM test defined in (\ref{eq:LM}), 
and (iv) the test that is based on the orthogonalized jackknife LM statistic $LM^{*2}(\beta_0)$ defined in this paper. We conduct 5,000 simulation replications to obtain stable simulation results. 

We set the parameter space for $\beta$ as $\mathcal{B} = [-6/\mathcal{C},6/\mathcal{C}]$, where $\mathcal{C} = 3$ and $6$. 
The choice of parameter space follows that in I.\citet[Section 7.2]{Andrews(2016)}. We set $\beta_0 = 0$, and the values of the covariance matrix in (\ref{eq:limittrue}) are set as follows:
$\Phi_1 = \Psi = \Upsilon =1$, and $\Phi_{12} = \Phi_{13} = \tau = \rho$, 
where $\rho \in \{0.2, 0.4, 0.7, 0.9\}$. We then compute $\gamma(\beta_0)$ based on \eqref{eq:phipsi} as $\beta$ ranges over $\mathcal{B}$ and generate $AR(\beta_0)$ and $LM(\beta_0)$ based on  \eqref{eq:limit}. Last, we implement our CLC test purely based on $AR(\beta_0)$, $LM(\beta_0)$, $\gamma(\beta_0)$, and $\mathcal{B}$ without assuming the knowledge of $(\mathcal{C},\beta)$. We have tried to simulate under alternative settings of the covariance matrix, and the obtained patterns of the power behavior are very similar.  



Figures \ref{limit_fig1}--\ref{limit_fig4} plot the power curves for $\rho=0.2, 0.4, 0.7,$ and $0.9$. In each figure, we report the results under both $\mathcal{C}=3$ and $6$. 
We observe that overall, the two jackknife CLC tests have the best power properties in terms of minimizing the maximum regret. Especially when the identification is relatively strong ($\mathcal{C}=6$) and/or the degree of endogeneity is not very low ($\rho =0.4, 0.7$, or $0.9$), the jackknife CLC tests outperform their AR and LM counterparts by a large margin.  
In addition, we notice that when $\mathcal{C}=3$, for some parameter values $LM^*(\beta_0)$ can suffer from substantial declines in power relative to the other tests, 
which is in line with our theoretical predictions. 
By contrast, our jackknife CLC tests are able to guard against such substantial power loss because of the adaptive nature of their minimax procedure.
In Section \ref{sec: further_simu_limit} of the Online Supplement, we further report power curves for alternative values of the tuning parameters $(p_1,p_2)$ in (\ref{eq:a_underline}) and  of $\mathcal{C}$, and find that the overall patterns remain very similar.


\begin{figure}[h]
\centering
\includegraphics[width=0.9\textwidth,height = 6cm]{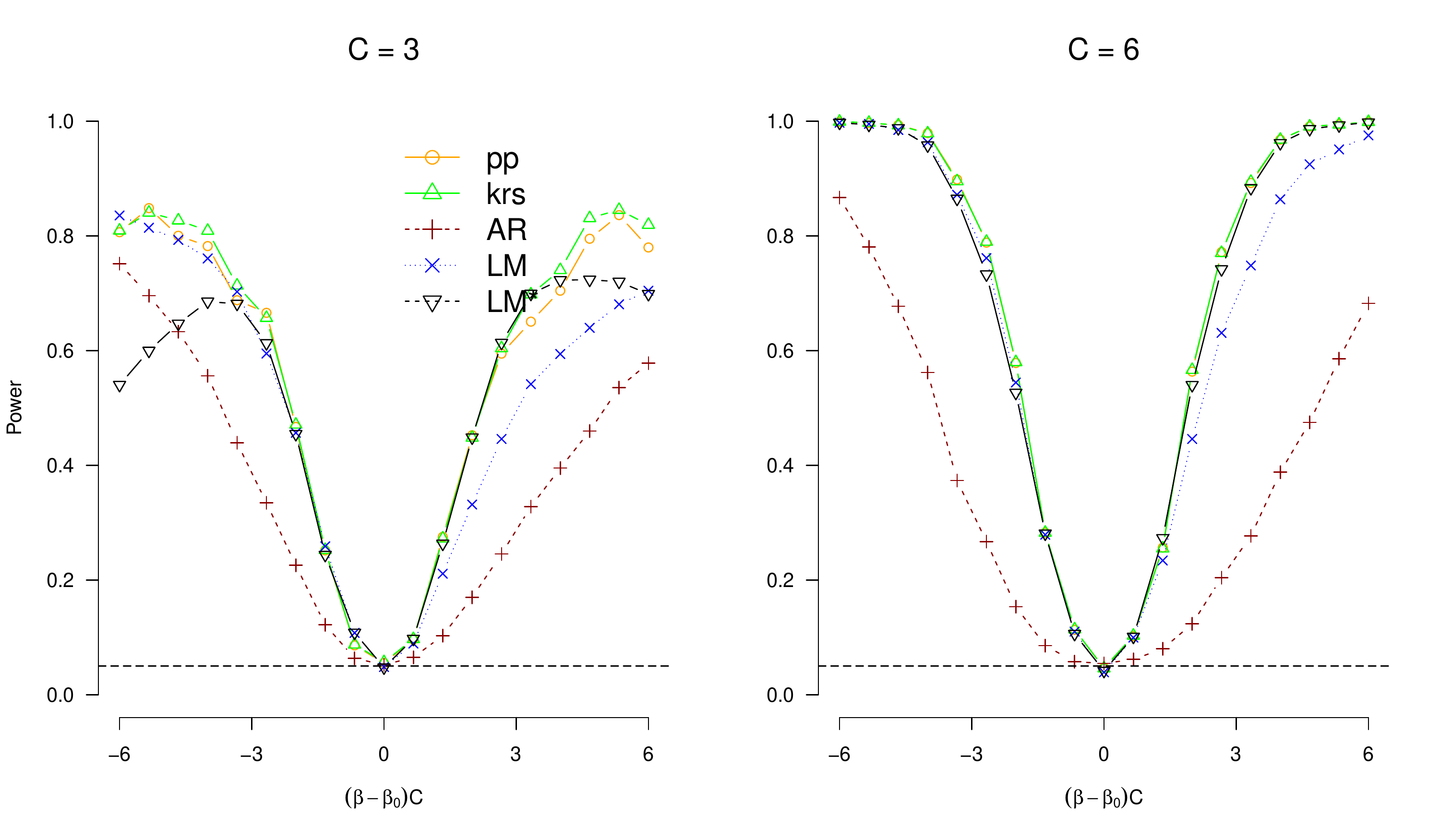}
\caption{Power Curve for $\rho = 0.2$}
\label{limit_fig1}
\end{figure}

\begin{figure}[h]
\centering
\includegraphics[width=0.9\textwidth,height = 5.85cm]{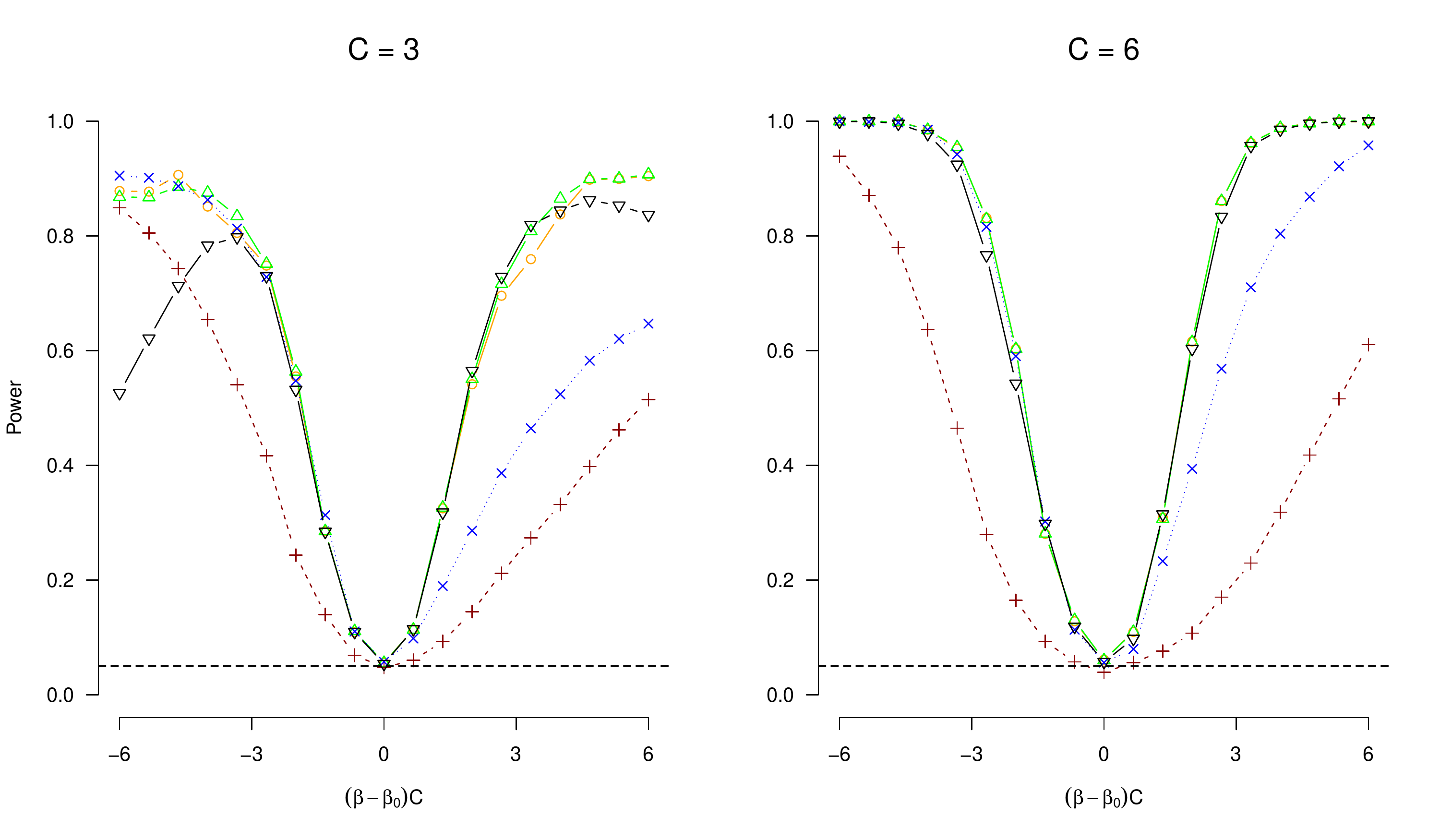}
\caption{Power Curve for $\rho=0.4$}
\label{limit_fig2}
\end{figure}

\begin{figure}[h]
\centering
\includegraphics[width=0.9\textwidth,height = 5.85cm]{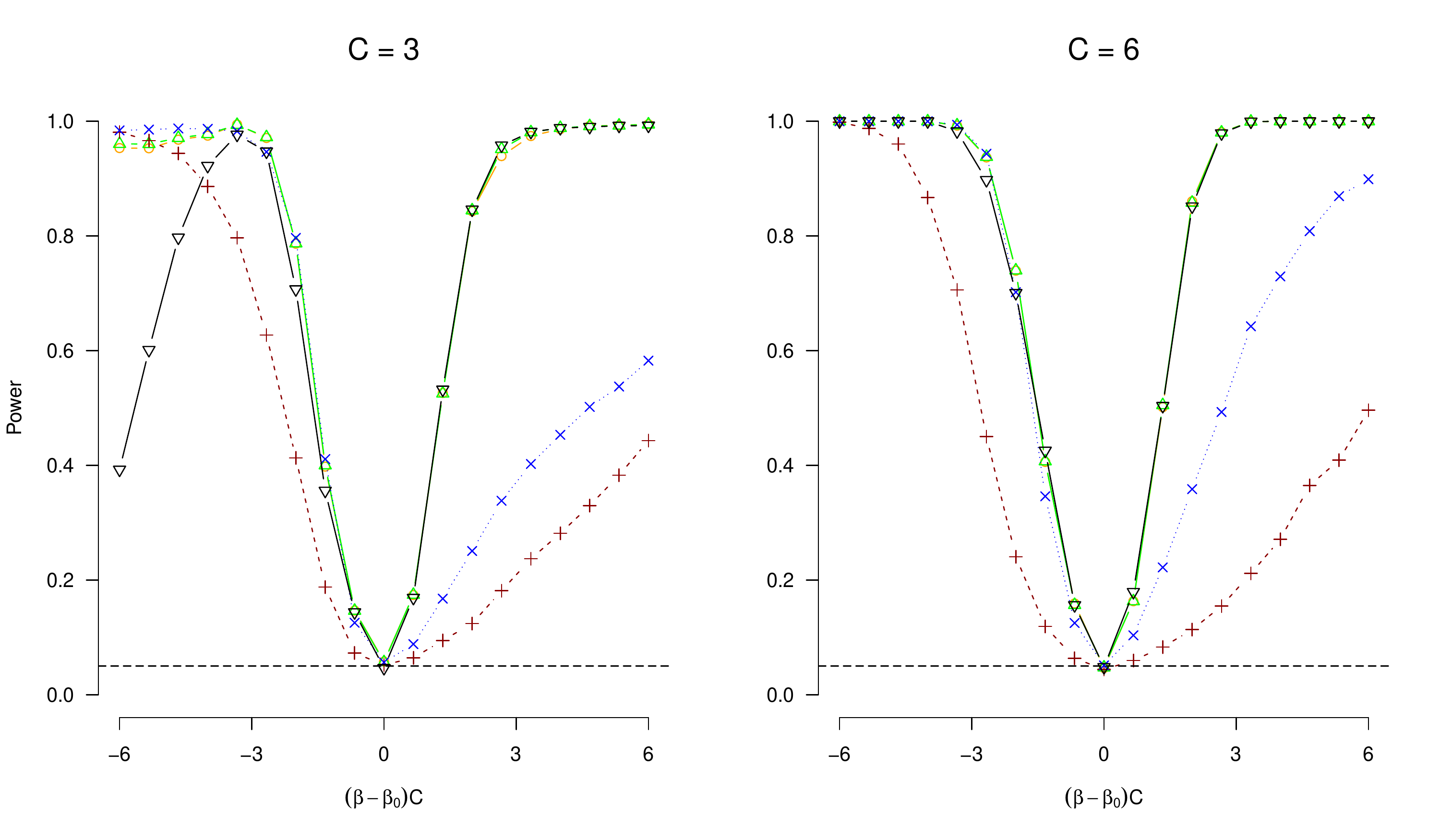}
\caption{Power Curve for $\rho=0.7$}
\label{limit_fig3}
\end{figure}

\begin{figure}[h]
\centering
\includegraphics[width=0.9\textwidth,height = 5.85cm]{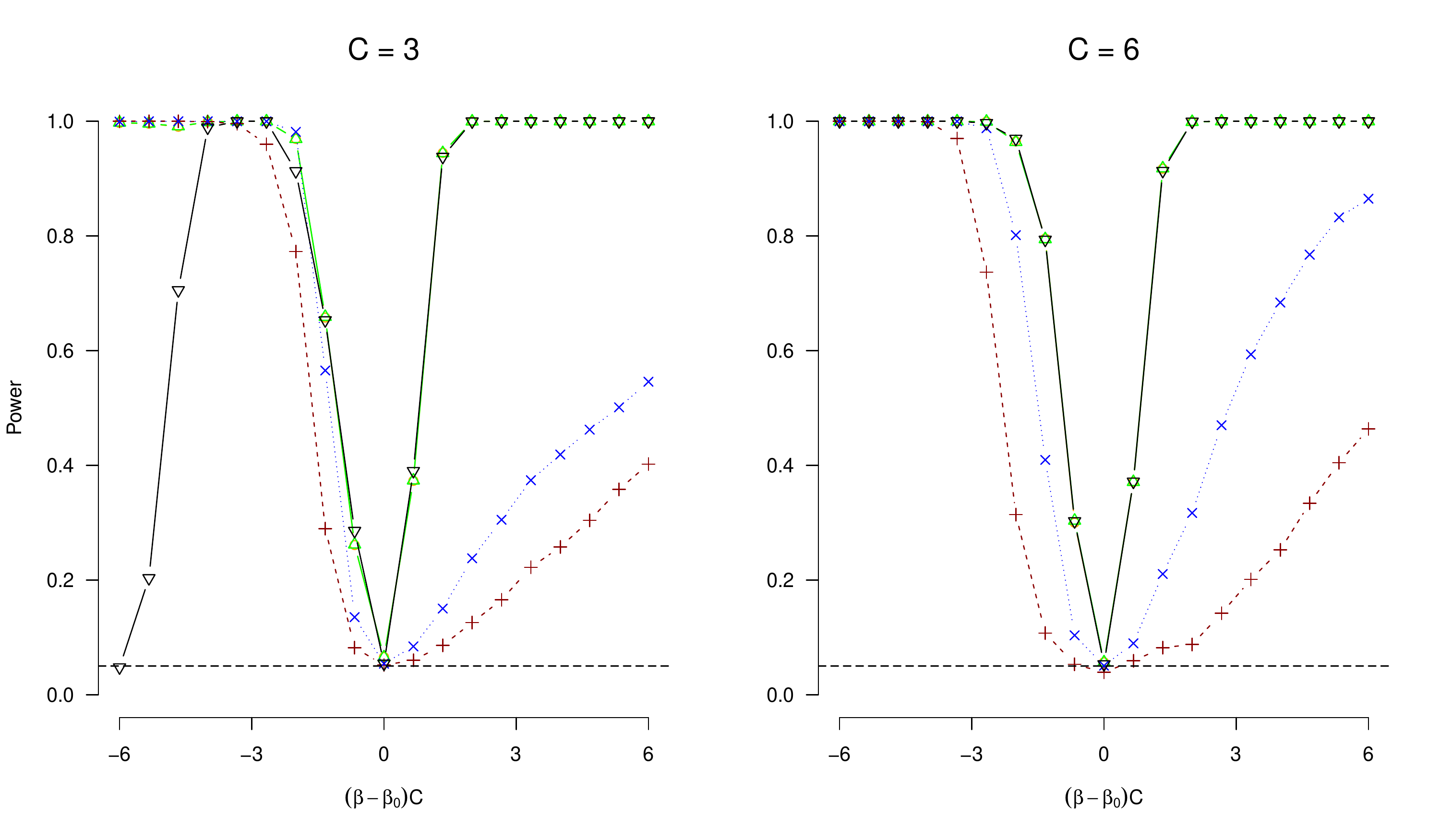}
\caption{Power Curve for $\rho=0.9$}
\label{limit_fig4}
\end{figure}

\subsection{Simulation Based on Calibrated Data}
\label{sec:sim2}
We follow the approach of \cite{Angrist-Frandsen2022} and \cite{MS22} and use a data generating process (DGP) calibrated based on the  1980 census  dataset from \cite{Angrist-Krueger(1991)}. We define the instruments as  
$$\tilde Z_i = \big( (1 \{Q_i = q, C_i = c\})_{q \in \{2,3,4\}, c \in \{31,\cdots,39\} }, (1 \{Q_i = q, P_i = p \})_{q \in \{2,3,4\}, p \in \{\text{51 states}\} } \big),$$
where $Q_i, C_i, P_i$ are individual $i$'s quarter of birth (QOB), year of birth (YOB) and place of birth (POB), respectively, so that there are 180 instruments. Note that the dummy with $q = 1$ and $c = 30$ is omitted in $Z_i$. We denote  $\tilde Y_i$ as income, $\tilde X_i$ as the highest grade completed,  and $\tilde W_i$ as the full set of YOB-POB interactions; that is,
$$\tilde W_i = \big( 1 \{C_i = c, P_i = p\}_{c \in \{30,...,39\}, p \in \{\text{51 states}\} }  \big),$$
which is a $510 \times 1$ matrix. 

As in \cite{Angrist-Frandsen2022}, using the full 1980 sample (consisting of 329,509 individuals), we first obtain the average $\tilde X_i$ for each QOB-YOB-POB cell; we call this $\bar{s}(q,c,p)$. Next we use LIML to estimate the structural parameters in the following linear IV regression: 
$$\tilde Y_i = \tilde X_i \beta_X + \tilde W_i^\top \beta_W +e_i,$$
$$\tilde X_i = \tilde Z_i^\top \Gamma_Z + \tilde W_i^\top \Gamma_W + V_i,$$
where $\tilde X$ is endogenous and instrumented by $\tilde Z_i$ and $\tilde W_i$ is the exogenous control variable. 
Denote the LIML estimate for $\beta_{X,W} \equiv (\beta_X^\top, \beta_W^\top)^\top$ as $\widehat\beta_{LIML}^\top = (\widehat\beta_{LIML, X}^\top, \widehat\beta_{LIML,W}^\top)$.
We let $\widehat{y}(C_i,P_i)  = \tilde W_i^\top \widehat{\beta}_{LIML,W}$ and 
$$\omega(Q_i,C_i,P_i) = \tilde Y_i - \tilde X_i \widehat{\beta}_{LIML,X} - \tilde W_i^\top \widehat{\beta}_{LIML,W}.$$

Based on the LIML estimate and the calibrated $\omega(Q_i,C_i,P_i)$, we simulate the following two DGPs:
\begin{enumerate}
\item DGP 1: 
\begin{align}
	\widetilde{y}_i = \bar{y} + \beta \widetilde{s}_i + \omega(Q_i,C_i,P_i)(\nu_i +\kappa_2 \xi_i)   \label{eq:DGP1} 
\end{align}
$$\widetilde{s}_i \sim Poisson(\mu_i),$$
where $\beta$ is the parameter of interest, $\nu_i$ and $\xi_i$ are independent standard normal, $\bar{y} = \frac{1}{n}\sum_{i=1}^n \widehat{y}(C_i,P_i)$, $\mu_i \equiv max\{1,\gamma_0+ \gamma_Z^\top \tilde Z_i + \kappa_1 \nu_i \}$, and $\gamma_0 +\gamma^\top_Z \tilde Z_i$ is the projection of $\bar{s}_i(q,c,p)$ onto a constant and $\tilde Z_i$. We set $\kappa_1 = 1.7$ and $\kappa_2 = 0.1$ as in \cite{MS22}.
\item DGP 2: Same as DGP 1 except that $\kappa_1 = 2.7$ and 
$$\widetilde{s}_i \sim \lfloor Poisson(2\mu_i)/2 \rfloor.$$
\end{enumerate}
We consider sample sizes of  0.5\%, 1\%, and 1.5\% of the full sample size. Upon obtaining $n$ observations, we exclude instruments with $\sum_{i=1}^n \tilde Z_{ij} <5$. This results in three different sample sizes: small, medium, and large, with 1,648, 3,296, and 4,943 observations, respectively. The number of instruments also varies across sample sizes, with 119, 142, and 150 instruments for small, medium, and large samples, respectively. Our DGP 1 is exactly the same as that in \cite{MS22}, with the correlation parameter of $\rho =0.41$. DGP 2 has a higher correlation parameter of $\rho =0.7$. The identification strength increases with the sample size. For DGP 1, the concentration parameters  $\mathcal{C}/\Upsilon^{1/2}$ for small, medium, and large samples are 2.15, 3.62, and 4.85, respectively. For DGP 2, they are 2.38, 3.97, 5.28, respectively. 

We emphasize that following \cite{Angrist-Frandsen2022} and \cite{MS22}, we only use $\tilde W_i$ to compute the LIML estimator and calibrate $\omega(Q_i,C_i,P_i)$, but do not use it to generate new data. Therefore, for the simulated data, the outcome variable is $\tilde{y}_i$, the endogenous variable is $\tilde s_i$, the IV $\tilde Z_i$ is viewed to be fixed, and the exogenous control variable is just an intercept. We then denote the demeaned versions of $\tilde{y}_i$, $\tilde{s}_i$, and $\tilde Z_i$ as $Y_i$, $X_i$, and $Z_i$, respectively, in \eqref{eq:1} and implement various inference methods described below. Following \cite{MS22}, we test the null hypothesis that $\beta = \beta_0$ for $\beta_0 =0.1$ while varying the true value $\beta \in \mathcal{B}$. The parameter space is set as $\mathcal{B} = [-0.5,0.5]$, which is consistent with the choice of parameter space for the empirical application below. The results below are based on 1,000 simulation repetitions. We provide more details about the implementation in Section \ref{sec:imp_sim} in the Online Supplement. We set $(p_1,p_2) = (0.01,1.1)$ in \eqref{eq:a_underline}. Additional simulation results using other choices of $(p_1,p_2)$ and $\mathcal{B}$ are reported in Section \ref{sec:add_sim_2} in the Online Supplement. All of them are very close to what we report here. 

We compare the following tests with a nominal rate of $5\%$:
\begin{enumerate}
\item pp: our jackknife CLC test when $\mu_D$ is estimated by the method $pp$. 
\item krs: our jackknife CLC test when $\mu_D$ is estimated by the method $krs$. 
\item AR: the one-sided jackknife AR test with the cross-fit variance estimator proposed by \cite{MS22}. 
\item LM\_CF: \citeauthor{Matsushita-Otsu2021}'s (\citeyear{Matsushita-Otsu2021}) jackknife LM test, but with a cross-fit variance estimator (details are given in Section \ref{sec:var2} in the Online Supplement).
\item 2-step: \citeauthor{MS22}'s (\citeyear{MS22}) two-step estimator in which the overall size is set at $5\%$.
\item LM$^*$: LM$^*$ test defined in this paper. 
\item LM\_MO: \citeauthor{Matsushita-Otsu2021}'s (\citeyear{Matsushita-Otsu2021}) original jackknife LM test. 
\end{enumerate}

\begin{figure}[h]
\centering
\includegraphics[width=1\textwidth,height = 5.85cm]{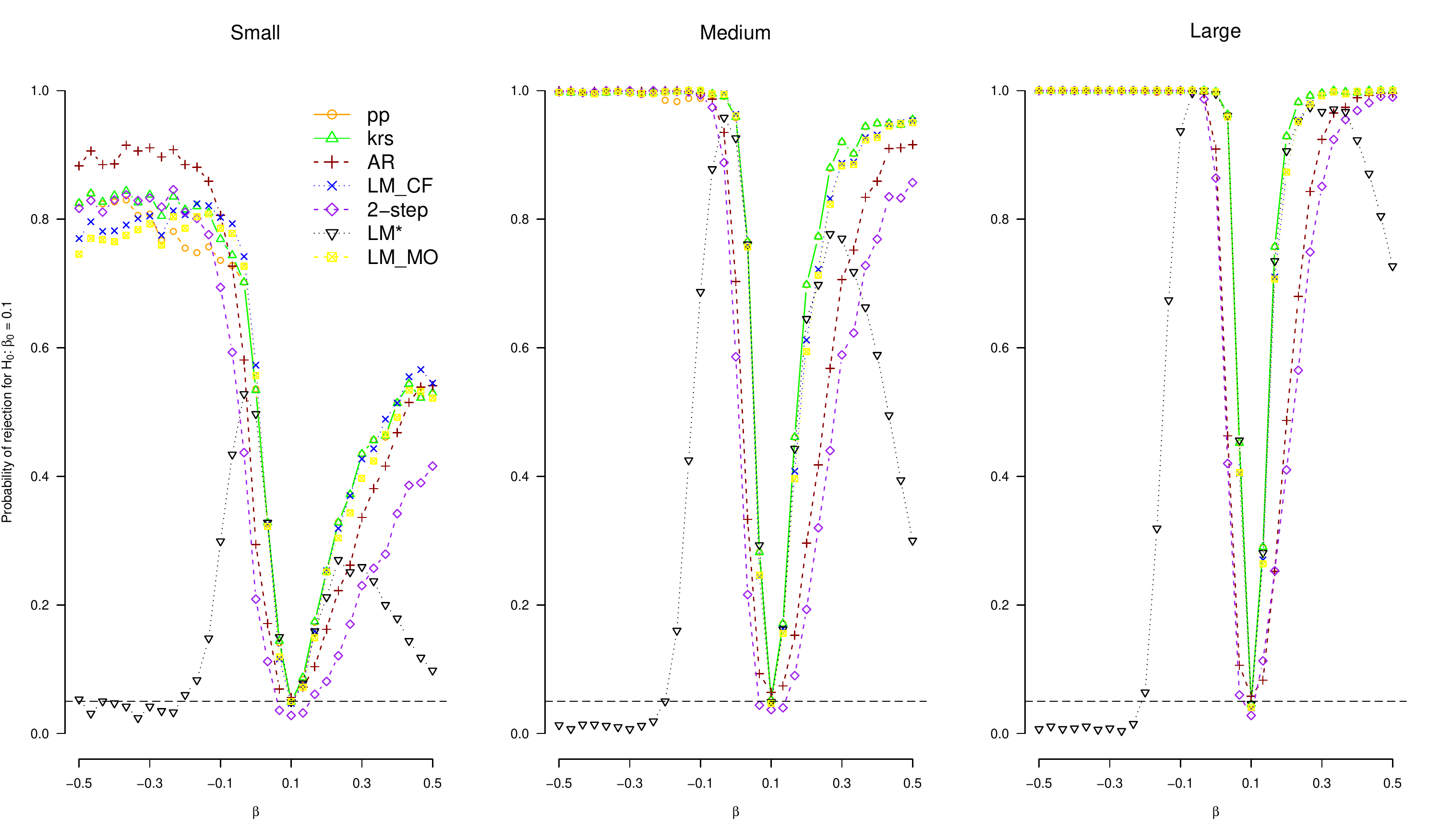}
\caption{Power Curve for DGP 1}
\label{fig1}
\end{figure}

\begin{figure}[h]
\centering
\includegraphics[width=1\textwidth,height = 5.85cm]{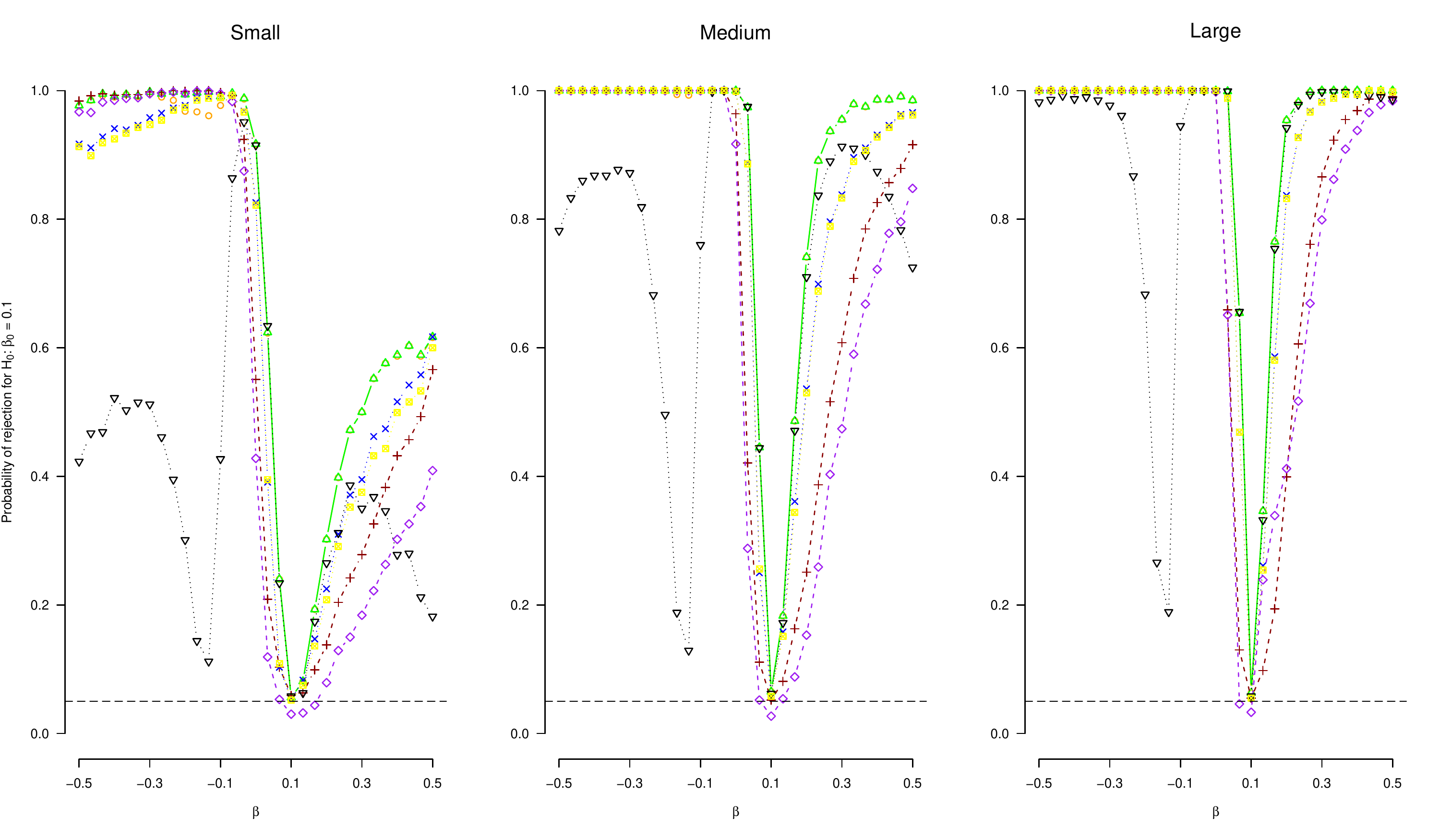}
\caption{Power Curve for DGP 2}
\label{fig2}
\end{figure}

Figures \ref{fig1} and \ref{fig2} plot the power curves of the aforementioned tests. We can make four observations. First, all methods control size well because they are all weak identification robust. Second, the performance of the jackknife CLC test with $krs$ is slightly better than that with $pp$, which is consistent with the power curve simulation in Section \ref{sec:sim1}. Third, in DGP 1 with a small sample size, the power of the jackknife AR test is at most about 9.2\% higher than that of the $krs$ test when $\beta$ is around -0.3. However, for alternatives close to the null (e.g., when $\beta$ is around 0), the power of the $krs$ test is 24\% higher, which implies that the power of the $krs$ test is still better than that for the jackknife AR test in the minimax sense. The power of the jackknife LM tests is similar to that of the $krs$ test in DGP 1 with a small sample size.
Fourth, for the rest of the scenarios, the power of the $krs$ test is the highest in most regions of the parameter space. The power of the jackknife AR and LM is at most 0.7\% higher than that of the $krs$ test at some point. For DGP 1 with medium and large sample sizes, the maximum power gaps between our $krs$ test and the jackknife LM are about 8.6\% and 5.6\%, and about 43.2\% and 50\% compared with the jackknife AR. Furthermore, they are 23.3\%, 19.5\%, and 18.5\% compared with the jackknife LM for DGP 2 with small, medium, and large sample sizes, respectively, and about 41.5\%, 55.3\%, and 55.85\% compared with the jackknife AR.

Figures \ref{fig3} and \ref{fig4} show the average values of $(a_1,a_2)$, which represents the weights assigned to $AR(\beta_0)$ and $LM(\beta_0)$ in our CLC tests, under DGPs 1 and 2, respectively. The weight assigned to $LM^*(\beta_0)$ is simply $1-a_1-a_2$. As shown in Table \ref{tab:clc_weights}, under weak identification and fixed alternatives, there is no clear winner among $AR(\beta_0)$, $LM(\beta_0)$, and $LM^*(\beta_0)$, and thus, our CLC test assigns weights to all the three tests. However, under strong identification and local alternative, $LM^*(\beta_0)$ is the UMP test and should carry all the weights, which means $a_1+a_2$ should be minimum. 
On the other hand, under strong identification and for some fixed alternatives, $LM^*(\beta_0)$ may lack power while both $AR(\beta_0)$ and $LM(\beta_0)$ have power 1. In this case, as long as we do not assign all weights on $LM^*(\beta_0)$, our CLC test should also have power 1. We observe that our simulation results are consistent with these theoretical predictions. First, when $\beta_0$ is close to the null 0.1, both $a_1$ and $a_2$ are small, indicating that most of the weights are put on $LM^*(\beta_0)$. Second, we observe from Figures \ref{fig1} and \ref{fig2} that the power of $LM^*(\beta_0)$ drops rapidly when $\beta$ is smaller than around zero. Therefore, our CLC test assigns more weights on $AR(\beta_0)$ and $LM(\beta_0)$. Third, for distant alternatives, significant weights are assigned to $AR(\beta_0)$ and $LM(\beta_0)$, which ensures the good power of our CLC test. Additionally, we note that the weights assigned to $AR(\beta_0)$ ($a_1$) are higher on the left side of the parameter space relative to the right, since $AR(\beta_0)$ is more powerful on the left.  

\begin{figure}[h]
\centering
\includegraphics[width=1\textwidth,height = 5.85cm]{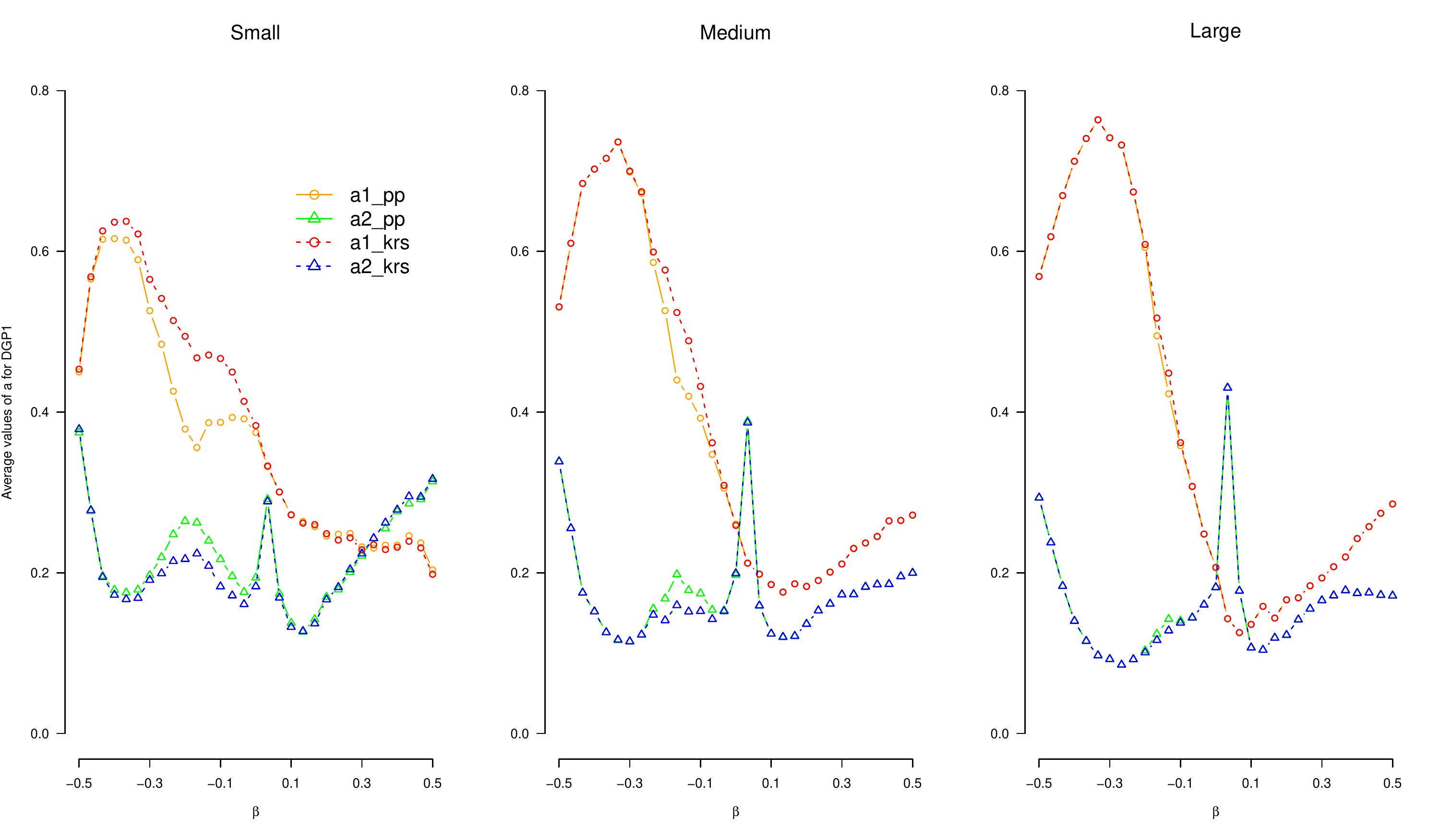}
\caption{Average Values of $a$ for DGP 1}
\label{fig3}
\end{figure}

\begin{figure}[h]
\centering
\includegraphics[width=1\textwidth,height = 5.85cm]{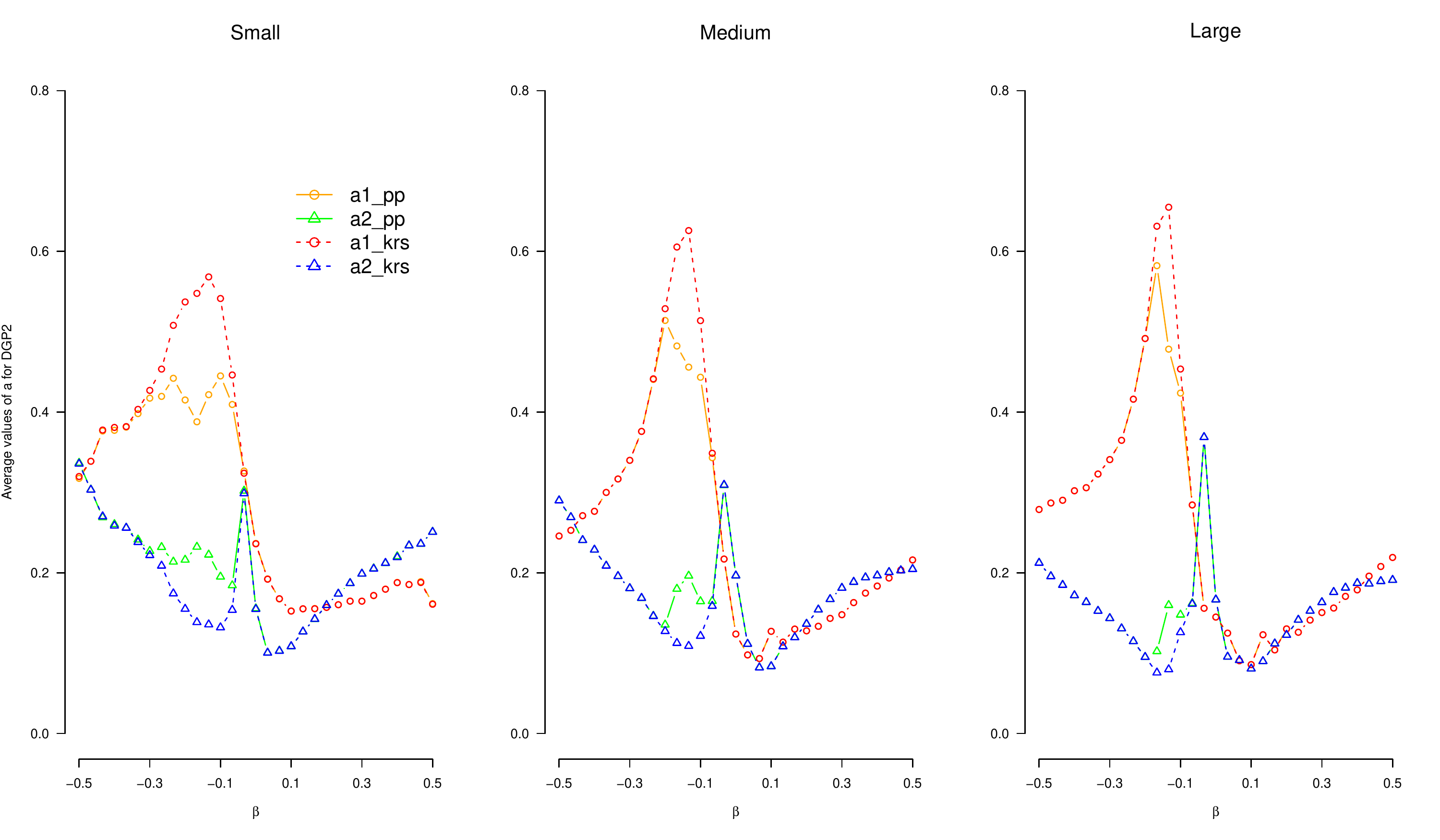}
\caption{Average Values of $a$ for DGP 2}
\label{fig4}
\end{figure}

\section{Empirical Application}\label{sec: empirical}

In this section, we consider the linear IV regressions with the specification underlying \citet[Table VII, column (6)]{Angrist-Krueger(1991)}, using the full original dataset.\footnote{The dataset can be downloaded from MIT Economics, Angrist Data Archive, https://economics.mit.edu/faculty/angrist/data1/data/angkru1991.} The outcome variable $Y$ and endogenous variable $X$ are log weekly wages and schooling, respectively. We follow \cite{Angrist-Krueger(1991)} and focus on two specifications with 180 and 1,530 instruments. The 180 instruments consist of 30 quarter and
year of birth interactions (QOB-YOB) and 150 quarter and place of birth interactions
(QOB-POB).  The second specification includes full
interactions among QOB-YOB-POB, resulting in 1,530 instruments. The exogenous control variables have been partialled out from the outcome, endogenous variables, and IVs. Further details on the empirical application can be found in Section \ref{sec:imp_app} in the Online Supplement. 
The considered tests are similar to those in the previous section. 
The jackknife AR test is defined in \eqref{eq:AR} with $\widehat{\Phi}_1$ being the cross-fit estimator in \cite{MS22}. The jackknife LM test is defined in (\ref{eq:LM}) with the cross-fit estimator for $\Psi(\beta_0)$. The $pp$ and $krs$ tests are our jackknife CLC tests. The two-step procedure is given by \citet[Section 5]{MS22}. Specifically, the researcher accepts the null if $\widetilde{F}>9.98$ and $Wald(\beta_0) < \mathbb{C}_{0.02}$\footnote{$\widetilde{F} = Q_{X,X}/\widehat{\Upsilon}$, where $\widehat{\Upsilon}$ is the cross-fit estimator. $Wald(\beta_0)$ is defined as $\left(\frac{\hat{\beta} - \beta_0}{\hat{V}}\right)^2$, where $\hat{\beta}$ is the JIVE estimator and $\hat{V}$ is a cross-fit estimator of the asymptotic variance of $\hat{\beta}$. We refer interested readers to \citet[Section 5]{MS22} for more details.} or if $\widetilde{F}\leq 9.98$ and $AR(\beta_0) < \emph{z}_{0.02}$. In the case of 180 instruments, because $\widetilde{F} = 13.42 > 9.98$, the lower and upper bounds of the 95\% confidence interval (CI) for the two-step procedure correspond respectively to the minimum and maximum of the set $\{\beta_0 \in \Re: Wald(\beta_0)<\mathbb{C}_{0.02}\}$; similarly, for the 1,530 instruments, as $\widetilde{F} = 6.32 \leq 9.98$, the lower and upper bounds of the CI for the two-step procedure correspond respectively to the minimum and maximum of the set $\{\beta_0 \in \Re: AR(\beta_0)<\emph{z}_{0.02} \}$. We also report the 95\% Wald test CI based on the JIVE estimator, denoted as JIVE-t. Table \ref{tab:CI} reports the 95\% CIs by inverting the corresponding 5\% tests mentioned above for the parameter space $\mathcal{B} = [-0.5,0.5]$. Note all CIs except JIVE-t are robust to weak identification. As $\widetilde{F}$'s are higher than $4.14$ in both cases, the JIVE-t (5\%) has the \citeauthor{Stock-Yogo(2005b)} (\citeyear{Stock-Yogo(2005b)})-type guarantee with at most a 5\% size distortion (i.e., the overall size is less than 10\%). We set $(p_1,p_2)$ in \eqref{eq:a_underline} as $(0.01,1.1)$. The empirical results with other choices of $(p_1,p_2)$ and $\mathcal{B}$ are reported in Section \ref{sec:add_app} in the Online Supplement. All of them are very close to what we report here.

\begin{table}[h]
\adjustbox{max width=\textwidth}{%
	\centering
	\begin{tabular}{c| c c c c c c c } 
		\hline &   jackknife AR & jackknife LM  & JIVE-t & Two-step & pp & krs  \\  
		& (5\%) & (5\%) &  (5\%) & (5\%) & (5\%) & (5\%)   \\ [0.5ex] 
		\hline
		180 IVs &   [0.008,0.201] & [0.067,0.135] & [0.066,0.132]  &  [0.059,0.139] & [0.067,0.128] & [0.067,0.128] \\ 
		\hline
		1530 IVs & [-0.035,0.22] & [0.036,0.138] & [0.035,0.133] &  [-0.051,0.242] & [0.037,0.133] & [0.037,0.133] \\ [1ex] 
		\hline
	\end{tabular}
}
\caption{\textbf{Confidence Intervals}}
\small Notes: The $\widetilde{F}$'s for 180 and 1,530 instruments are 13.42 and 6.32, respectively. The grid-search used for our confidence interval was over 10,000 equidistant grid-points for $\beta_0 \in [-0.5,0.5]$.  Our jackknife AR confidence interval for 1530 instruments differs from that in \cite{MS22} because they used year-of-birth 1930-1938 dummies for the QOB-YOB-POB interactions, whereas we used 1930-1939 dummies. 
More details are provided in Section \ref{sec:imp_app} in the Online Supplement.
\label{tab:CI}
\end{table}

Table \ref{tab:CI} highlights that the CIs generated by our jackknife CLC tests are the shortest among all the weak identification robust CIs (i.e., pp, krs, jackknife AR, jackknife LM, and two-step). Furthermore, the jackknife CLC CIs are $7.6\%$ and $2.0\%$ shorter than the non-robust JIVE-t CIs with 180 and 1,530 instruments, respectively, which is in line with our theoretical result that the CLC tests are adaptive to the identification strength and efficient under strong identification.

\newpage
\appendix

\section{Exogenous Control Variables}
\label{sec:W0}
Suppose we observe $\{\tilde Y_i, \tilde X_i, \tilde Z_i, W_i\}_{i \in [n]}$, where 
\begin{align*}
	\tilde Y_i & = \tilde X_i \beta+ W_i^\top \gamma + \tilde e_i, \quad   \tilde X_i = \tilde \Pi_i + \tilde V_i,
\end{align*}
$\tilde X_i \in \Re$, $\tilde Z_i \in \Re^K$, $W_i \in \Re^{d}$, $\tilde \Pi_i = \mathbb{E} \tilde{X}_i$, and $(\tilde Z_i, W_i)_{i \in [n]}$ are treated as fixed. We allow $K$ to diverge to infinity with $n$ while $d$ is fixed. We then have $\mathbb{E}\tilde e_i = \mathbb{E} \tilde V_i = 0$. Denote $P_W = W (W^\top W)^{-1} W^\top$ and $M_W = I_n - P_W$ be the projection and residual matrices based on $W$, respectively, where $I_n$ is the $n\times n$ identity matrix and $W = (W_1,W_2,\cdots,W_n)^\top \in \Re^{n \times d}$. Further denote $\tilde Y,\tilde X, \tilde e,\tilde \Pi, \tilde V$ as matrices with their $i$th row being $\tilde Y_i, \tilde X_i, \tilde e_i,\tilde \Pi_i, \tilde V_i$, respectively. Then, we have
\begin{align*}
	Y_i = X_i \beta+e_i, \quad X_i = \Pi_i + V_i,
\end{align*}
where $Y= M_W \tilde Y$, $X= M_W \tilde X$, $V= M_W \tilde V$, $e= M_W \tilde e$, $\Pi= M_W \tilde \Pi$, and $Z= M_W \tilde Z$. We still denote $P$ as the projection matrix constructed by $Z$. The next theorem shows Assumption \ref{ass:weak_convergence} holds. 
\begin{thm}
	Suppose 
	$\{\tilde V_i, \tilde e_i\}_{i \in [n]}$ are independent, $\max_i \mathbb{E}\tilde e_i^4 + \max_i \mathbb{E}\tilde V_i^4 \leq C<\infty$, $\max_i ||W_i||_2 \leq C<\infty$, $\Pi^\top \Pi/K = O(1)$, and $0<c \leq \text{mineig}(W^\top W/n) \leq \text{maxeig}(W^\top W/n) \leq C<\infty$, for some constants $c,C$. Then, Assumption \ref{ass:weak_convergence} holds and $Q_{e,e} = Q_{\tilde{e},\tilde{e}} + o_P(1)$. If in addition, $p_n^2 \frac{\Pi^\top \Pi}{K} = o(1)$ with $p_n = \max_{i}P_{ii}$, then we have $Q_{X,e} = Q_{\overline{X},\tilde{e}} + o_P(1)$ and $Q_{X,X} = Q_{\overline{X},\overline{X}}+ o_P(1)$, where $\overline{X}_i = \Pi_i+\tilde V_i$.
	\label{thm:W}
\end{thm}
Theorem \ref{thm:W} shows Assumption \ref{ass:weak_convergence} still holds if $(Y_i,X_i,Z_i)$ are defined after partialing out the fixed dimensional control variables $W_i$. It further provides a sufficient condition under which the effect of partialling-out on the sampling error is asymptotically negligible, i.e., the asymptotic covariance matrix remains the same after partialing out $W_i$.
To interpret the sufficient condition, we consider the balanced design in which $p_n$ is of order $K/n$. If $K/n = o(1)$ and $\Pi^\top \Pi/n = O(1)$, then the sufficient condition holds because 
\begin{align*}
	p_n^2 \Pi^\top \Pi/K = O\left( \frac{\Pi^\top \Pi}{n} \frac{K}{n}\right) = o(1).
\end{align*}
On the other hand, if $K \asymp n$, the sufficient condition requires $\Pi^\top \Pi/K = o(1)$, which can hold under both weak identification ($\Pi^\top \Pi/\sqrt{K} = O(1)$) and strong identification $(\Pi^\top \Pi/\sqrt{K} \rightarrow \infty)$. We further emphasize that, even if $K \asymp n$ and $\Pi^\top \Pi/K \asymp 1$ so that the sufficient condition does not hold, Assumption \ref{ass:weak_convergence} still holds. It is just that partialing out the exogenous control variable will have a non-negligible effect on the asymptotic covariance of $(Q_{e,e},Q_{X,e}, Q_{X,X} - Q_{\Pi,\Pi})$.

\section{Verifying Assumption \ref{ass:variance_est}}
\label{sec:var}
\subsection{Standard Estimators}
\label{sec:var1}

In this section, we maintain Assumption \ref{ass:K}, which is stated below and just \citet[Assumption 1]{MS22}. 
\begin{ass}
	Suppose $\{ V_i, e_i\}_{i \in [n]}$ are independent and $\mathbb{E} e_i = \mathbb{E} V_i = 0$. 
	Suppose $P$ is an $n \times n$ projection matrix of rank $K$, $K\rightarrow \infty$ as $n \rightarrow \infty$ and there exists a constant $\delta$ such that $P_{ii}\leq \delta <1$.
	\label{ass:K}
\end{ass}

Following the results in \cite{Chao(2012)} and \cite{MS22}, we can show that under either weak or strong identification, Assumption \ref{ass:weak_convergence} in the paper holds:  
\begin{align}
	\begin{pmatrix}
		& Q_{e,e} \\
		& Q_{X,e} \\
		& Q_{X,X} - \mathcal{C}
	\end{pmatrix} \convD \N\left(\begin{pmatrix}
		0 \\
		0 \\
		0
	\end{pmatrix},\begin{pmatrix}
		\Phi_1 & \Phi_{12} & \Phi_{13} \\
		\Phi_{12} & \Psi & \tau \\
		\Phi_{13} & \tau & \Upsilon
	\end{pmatrix}\right), 
	\label{eq:limittrue'}
\end{align}
where  $\sigma_i^2 = \mathbb{E}e_i^2$, $\eta_i^2 = \mathbb{E}V_i^2$, $\gamma_i = \mathbb{E}e_iV_i$, $\omega_i = \sum_{j \neq i}P_{ij}\Pi_j$, 
\begin{align*}
	\Phi_1 & = \lim_{n \rightarrow \infty} \frac{2}{K}\sum_{i \in [n]}\sum_{j \neq i}P_{ij}^2 \sigma_i^2\sigma_j^2,\\
	\Phi_{12}& = \lim_{n \rightarrow \infty}\frac{1}{K}\sum_{i \in [n]}\sum_{j \neq i}P_{ij}^2 (\gamma_j \sigma_i^2 + \gamma_i\sigma_j^2), \\
	\Phi_{13} & = \lim_{n \rightarrow \infty}\frac{2}{K}\sum_{i \in [n]}\sum_{j \neq i}P_{ij}^2 \gamma_i\gamma_j, \\
	\Psi & = \lim_{n \rightarrow \infty}\left[\frac{1}{K}\sum_{i \in [n]}\sum_{j \neq i}P_{ij}^2 (\eta_i^2 \sigma_j^2 + \gamma_i \gamma_j) + \frac{1}{K}\sum_{i \in [n]} \omega_i^2 \sigma_i^2\right], \\
	\tau & = \lim_{n \rightarrow \infty} \left[ \frac{2}{K}\sum_{i \in [n]}\sum_{j \neq i}P_{ij}^2 \eta_i^2 \gamma_j + \frac{2}{K}\sum_{i \in [n]} \omega_i^2 \gamma_i \right], \quad \text{and}\\
	\Upsilon & = \lim_{n \rightarrow \infty} \left[ \frac{2}{K}\sum_{i \in [n]}\sum_{j \neq i}P_{ij}^2\eta_i^2 \eta_j^2 + \frac{4}{K}\sum_{i \in [n]} \omega_i^2 \eta_i^2 \right].
\end{align*}

We note that the standard estimators of the above variance components proposed by \cite{crudu2021} are equal to \citeauthor{Chao(2012)}'s (\citeyear{Chao(2012)}) 
estimators with their residual $\hat{e}_i$ replaced by $e_i(\beta_0)$. Specifically, let 
\begin{align*}
	\widehat{\Phi}_1(\beta_0) & =  \frac{2}{K}\sum_{i \in [n]}\sum_{j \neq i}P_{ij}^2 e_i^2(\beta_0)e_j^2(\beta_0),\\
	\widehat{\Phi}_{12}(\beta_0)& = \frac{1}{K}\sum_{i \in [n]}\sum_{j \neq i}P_{ij}^2 (X_je_j(\beta_0) e_i^2(\beta_0) + X_ie_i(\beta_0)e_j^2(\beta_0)), \\
	\widehat{\Phi}_{13}(\beta_0) & = \frac{2}{K}\sum_{i \in [n]}\sum_{j \neq i}P_{ij}^2 X_ie_i(\beta_0)X_je_j(\beta_0), \\
	\widehat{\Psi}(\beta_0) & = \frac{1}{K}\sum_{i \in [n]}(\sum_{j \neq i}P_{ij}X_j)^2 e_i^2(\beta_0) + \frac{1}{K} \sum_{i \in [n]}\sum_{j \neq i}P_{ij}^2 X_ie_i(\beta_0) X_je_j(\beta_0), \\
	\widehat{\tau}(\beta_0) & = \frac{1}{K}\sum_{i \in [n]}(\sum_{j \neq i}P_{ij}X_j)^2X_ie_i(\beta_0) + \frac{1}{K}\sum_{i \in [n]}\sum_{j \neq i}P_{ij}^2 X_i^2 X_je_j(\beta_0), \quad \text{and}\\
	\widehat{\Upsilon} & = \frac{2}{K}\sum_{i \in [n]}\sum_{j \neq i}P_{ij}^2X_i^2 X_j^2.
\end{align*}

\begin{ass}
	Suppose $\max_{i \in [n]}|\Pi_i| \leq C$, 
	$\frac{\Pi^\top \Pi}{K} = o(1)$,
	and  $\max_i \mathbb{E} e_i^6 + \max_i \mathbb{E} V_i^6 <\infty$.
	\label{ass:reg1}
\end{ass}

Two remarks on Assumption \ref{ass:reg1} are in order. First, $\max_{i \in [n]}|\Pi_i| \leq C$ is mild because $\Pi_i = \mathbb{E}X_i$. Second, Assumption \ref{ass:reg1} allows for weak identification when $\Pi^\top \Pi/\sqrt{K} \rightarrow c$ for a constant $c$. 
It also allows for strong identification when $\Pi^\top \Pi/\sqrt{K}\rightarrow \infty$
and $\Pi^\top \Pi/ K \rightarrow 0$. The restriction that $\Pi^\top \Pi/ K \rightarrow 0$
is needed because Assumption \ref{ass:variance_est} includes the case of fixed alternatives (i.e., fixed $\Delta \neq 0$), which is not considered in \cite{crudu2021} and \cite{Chao(2012)}.
Furthermore, our results include $\widehat{\tau}(\beta_0)$
and $\widehat{\Upsilon}$, which are not considered in
\cite{crudu2021} and \cite{Chao(2012)}, 
and the consistency of these terms require $\Pi^\top \Pi/ K \rightarrow 0$.

\begin{thm}
	Suppose Assumptions \ref{ass:K} and \ref{ass:reg1} hold. Then Assumption \ref{ass:variance_est} holds for \citeauthor{crudu2021}'s (\citeyear{crudu2021}) estimators defined above. 
	\label{thm:var1}
\end{thm}

\subsection{Cross-Fit Estimators }
\label{sec:var2}
Let $M = I-P$, $M_{ij}$ be the $(i,j)$ element of $M$, $M_i$ be the $i$th row of $M$, and $\widetilde{P}_{ij}^2 = \frac{P_{ij}^2}{M_{ii}M_{jj} + M_{ij}^2}$. Then, \cite{MS22} consider the cross-fit estimators for $\Phi_1(\beta_0)$, $\Psi(\beta_0)$, and $\Upsilon$ defined as 
\begin{align*}
	\widehat{\Phi}_1(\beta_0) & =  \frac{2}{K}\sum_{i \in [n]}\sum_{j \neq i} \widetilde{P}_{ij}^2 [e_i(\beta_0) M_i e(\beta_0)] [e_j(\beta_0) M_j e(\beta_0)], \\
	\widehat{\Psi}(\beta_0) & = \frac{1}{K}\left[\sum_{i \in [n]} (\sum_{j \neq i} P_{ij}X_j )^2\frac{e_i(\beta_0)M_i e(\beta_0)}{M_{ii}} + \sum_{i \in [n]}\sum_{j \neq i} \widetilde{P}_{ij}^2M_iX e_i(\beta_0) M_j X e_j(\beta_0)\right], \quad \text{and}\\
	\widehat{\Upsilon} & = \frac{2}{K}\sum_{i \in [n]}\sum_{j \neq i} \widetilde{P}_{ij}^2 [X_i(\beta_0) M_i X] [X_j(\beta_0) M_j X],
\end{align*}
where $X$ and $e(\beta_0)$ are the column vectors that collect all $X_i$ and $e_i(\beta_0)$, respectively. Following their lead, we can construct the cross-fit estimators for the rest three elements in $\gamma(\beta_0)$ as follows: 
\begin{align*}
	\widehat{\Phi}_{12}(\beta_0)& = \frac{1}{K}\sum_{i \in [n]}\sum_{j \neq i}\widetilde{P}_{ij}^2 (M_jX e_j(\beta_0) e_i(\beta_0)M_i e(\beta_0) + M_i X e_i(\beta_0)e_j(\beta_0) M_j e(\beta_0)), \\
	\widehat{\Phi}_{13}(\beta_0) & = \frac{2}{K}\sum_{i \in [n]}\sum_{j \neq i}\widetilde{P}_{ij}^2 M_i X e_i(\beta_0) M_j Xe_j(\beta_0), \quad \text{and}\\
	\widehat{\tau}(\beta_0) & = \frac{1}{K}\sum_{i \in [n]}\sum_{j \neq i} \widetilde{P}_{ij}^2 (X_i M_iX)(M_jX e_j(\beta_0)) + 
	\frac{1}{K}\sum_{i \in [n]} ( \sum_{j \neq i} P_{ij}X_j)^2\left(\frac{e_i(\beta_0)M_i X}{2M_{ii}}+\frac{X_iM_i e(\beta_0)}{2M_{ii}}\right),
\end{align*}

\begin{ass}
	Suppose Assumption \ref{ass:reg1} holds. 
	Further suppose that  $\Pi^\top M \Pi \leq \frac{C \Pi^\top \Pi}{K}$ for some constant $C>0$.
	\label{ass:reg2}
\end{ass}

Compared with the assumptions in \cite{MS22}, Assumption \ref{ass:reg2} further requires that 
$\max_{i \in [n]}|\Pi_i| \leq C$. 
However, for all the above cross-fit estimators to be consistent, we only need $\Pi^{\top}\Pi/K \rightarrow 0$, which is weaker than that assumed in \cite{MS22} (e.g., Theorems 5 in their paper require $\Pi^{\top}\Pi/K^{2/3} \rightarrow 0$).
\begin{lem}
	Suppose Assumptions \ref{ass:K} and \ref{ass:reg2} hold. 
	Then, Lemmas 2, 3, S3.1, S3.2 in \cite{MS22} hold.
	\label{lem:SA}
\end{lem}

\begin{thm}
	Suppose Assumptions \ref{ass:K}
	and \ref{ass:reg2} hold. Then, Assumption \ref{ass:variance_est} holds for \citeauthor{MS22}'s (\citeyear{MS22}) cross-fit estimators defined above. 
	\label{thm:var2}
\end{thm}

\section{Details for Simulations Based on Calibrated Data}
\label{sec:imp_sim}
The DGP contains only the intercept as the control variable. Therefore, we implement our jackknife CLC test on the demeaned version of $(\tilde{y}_i, \tilde{s}_i, \tilde Z_i)$. The parameter space is $\mathcal{B} = [-0.5,0.5]$. We test the null hypothesis that $\beta = \beta_0$ for $\beta_0 =0.1$ while varying the true value $\beta$ over 31 equal-spaced grids over $\mathcal{B}$. The grids for $\delta$ is the grid for $\beta$ minus $\beta_0$. We generate grids of $(a_1,a_2)$ as $a_1 = \sin^2(t_1)$ and $a_2 = \cos^2(t_1)\sin^2(t_2)$ with $t_1$ taking values over 16 equal-spaced grids over $[\underline{a}^{1/2}(f_{s}(\widehat{D},\widehat{\gamma}(\beta_0)),\pi/2]$ and $t_2$ taking values over 16 equal-spaced grids over $[0,\pi/2]$. We gauge $\mathbb{E}^*\phi_{a_1,a_2,s}(\delta,\widehat{D},\widehat{\gamma}(\beta_0))$  via a Monte Carlo integration with $R=2000$ draws of independent standard normal random variables. In practice, it is rare but possible that  $\mathcal{A}_s(\widehat{D},\widehat{\gamma}(\beta_0))$ defined in \eqref{eq:afeasible} is not unique. To increase numerical stability, we follow I.\cite{Andrews(2016)} and allow for some slackness in the minimization. Let $\mathcal{G}_a$ be the grid of $(a_1,a_2)$ mentioned above,  $\widehat{Q}(a_1,a_2) = \sup_{\delta \in  \mathcal{D}(\beta_0)}(\mathcal{P}_{\delta,s}(\widehat{D},\widehat{\gamma}(\beta_0)) -     \mathbb{E}^*\phi_{a_1,a_2,s}(\delta,\widehat{D},\widehat{\gamma}(\beta_0)))$, $\widehat{Q}_{\min} = \min_{(a_1,a_2) \in \mathcal{G}_a}\widehat{Q}(a_1,a_2) + 1/n$, where $n$ is the sample size, and 
\begin{align*}
	\Xi = \{(a_1,a_2) \in \mathcal{G}_a: \widehat{Q}(a) \leq \widehat{Q}_{\min} + (\widehat{Q}_{\min}(1-\widehat{Q}_{\min}))^{1/2}(2\log(\log(R)))^{1/2}R^{-1/2}\}.    
\end{align*}
The slackness term in the definition of $\Xi$ is due to the law of the iterated logarithm for sum of Bernoulli random variables and captures the randomness of the Monte Carlo integration. Suppose there are $L$ elements in $\Xi$ with an ascending order w.r.t. $(t_1,t_2)$, which are denoted as $\{(a_{1,l},a_{2,l})\}_{l = 1}^L$. We then define  $\mathcal{A}_s(\widehat{D},\widehat{\gamma}(\beta_0))$ as $(a_{1,\lfloor L/2 \rfloor}, a_{2,\lfloor L/2 \rfloor})$. We use the cross-fit estimators defined in Section \ref{sec:var2} throughout the simulation. 

\section{Details for Empirical Application}\label{sec:imp_app}
We consider the 1980s census of 329,509 men born in 1930-1939 based on \citeauthor{Angrist-Krueger(1991)}'s (\citeyear{Angrist-Krueger(1991)}) dataset. The model for \textbf{180 instruments} follows \cite{MS22}, which can be written explicitly as
\begin{align*}
	ln \; W_i &= Constant + H_i^\top \zeta + \sum_{c=30}^{38} YOB_{i,c} \xi_c + \sum_{s\neq 56} POB_{i,s} \eta_s + \beta E_i +  \gamma_i \\
	E_i & = Constant + H_i^\top \lambda + \sum_{c=30}^{38} YOB_{i,c}\mu_c + \sum_{s\neq 56} POB_{i,s} \alpha_s \\
	& + \sum_{j = 1}^{3} \sum_{s \neq 56} QOB_{i,j} POB_{i,s} \delta_{c,s} 
	+ \sum_{j=1}^3 \sum_{c=30}^{39}  QOB_{i,j}YOB_{i,c} \theta_{j,c} + \epsilon_i,
\end{align*}
where $W_i$ is the weekly wage, $E_i$ is the education of the $i$-th individual, $H_i$ is a vector of  covariates,\footnote{The covariates we consider are: RACE, MARRIED, SMSA, NEWENG, MIDATL, ENOCENT, WNOCENT, SOATL, ESOCENT, WSOCENT, and MT.} $YOB_{i,c}$ is a dummy variable indicating whether the individual was born in year $c = \{30,31,...,39 \}$, while $QOB_{i,j}$ is a dummy variable indicating whether the individual was born in quarter-of-birth $j \in \{1,2,3,4\}$. $POB_{i,s}$ is the dummy variable indicating whether the individual was born in state $s \in \{51\: states\}$.\footnote{The state numbers are from 1 to 56, excluding (3,7,14,43,52), corresponding to U.S. state codes.} Both $\gamma_i$ and $\varepsilon_i$ are the error terms. The coefficient $\beta$ is the return to education. We vary this $\beta$ across 10,000 equidistant grid-points from -0.5 to 0.5 (i.e., $\beta \in \{-0.5,-4.9999,-4.9998,...,0,...,4.9999,0.5\}$) and solve for the range of $\beta$ where the null hypothesis cannot be rejected. Specifically, we can write the above model as 
\begin{align*}
	ln \; W_i &= C_i \Gamma + \beta E_i +  \gamma_i \\
	E_i & = C_i\tau + Z_i\Theta + \epsilon_i,
\end{align*}
where $C_i$ is a (329,509$\times$71)-matrix of controls containing the first four terms on the right-hand of the first equation, while $Z_i$ is the (329,509$\times$180)-matrix of instruments containing the first two terms in the third line. We can then partial out the controls $C_i$ by multiplying each equation by the residual matrix $I - C(C^\top C)^{-1}C^\top$ to obtain a form analogous to that in the main text:
\begin{align*}
	Y_i &= X_i \beta + e_i, \\
	X_i &= \Pi_i + v_i.
\end{align*}
Then, at each grid-point we take $\beta_0 = \beta$ and compute $AR(\beta_0),\; LM(\beta_0),\; Wald(\beta_0),\; \widehat{\phi}_{\mathcal{A}_{pp}(\widehat{D},\widehat{\gamma}(\beta_0))}$ and $\widehat{\phi}_{\mathcal{A}_{krs}(\widehat{D},\widehat{\gamma}(\beta_0))}$. We reject the chosen value of $\beta_0$ for $AR(\beta_0)$ if it exceeds the one-sided 5\%-quantile of the standard normal (i.e., reject if $AR(\beta_0)>z_{0.05}$). If $LM(\beta_0)^2 > \mathbb{C}_{0.05}$, we reject the chosen $\beta_0$ for Jackknife LM. If $Wald(\beta_0) > \mathbb{C}_{0.05}$, we reject for JIVE-t. If $\widehat{\phi}_{\mathcal{A}_{s}(\widehat{D},\widehat{\gamma}(\beta_0))} > \mathbb{C}_{0.05}(\mathcal{A}_{s}(\widehat{D},\widehat{\gamma}(\beta_0));\widehat{\rho}(\beta_0))$ for $s \in \{pp,krs\}$, we reject accordingly. The two-step procedure depends on the value of $\widetilde{F}$. If $\widetilde{F}>9.98$, we reject if $Wald(\beta_0)>\mathbb{C}_{0.02}$; otherwise if $\widetilde F \leq 9.98$, we reject if $AR(\beta_0) > z_{0.02}$.

\vspace{2mm}
The model for \textbf{1,530 instruments} can be written explicitly as
\begin{align*}
	ln \; W_i &= Constant + H_i^\top \zeta + \sum_{c=30}^{38} YOB_{i,c} \xi_c + \sum_{s\neq 56} POB_{i,s} \eta_s + \beta E_i +  \gamma_i. \\
	E_i & = Constant + H_i^\top \lambda + \sum_{c=30}^{38} YOB_{i,c}\mu_c + \sum_{s\neq 56} POB_{i,s} \alpha_s 
	\\
	&+ \sum_{j=1}^3 \sum_{c=30}^{39} \sum_{s\in \{51\; states \}} QOB_{i,j} YOB_{i,c} POB_{i,s} \delta_{j,c,s}. 
\end{align*}
The main difference between this 1,530-instrument specification and the 180-instrument one is that we now have QOB-YOB-POB interactions as our instruments, compared with QOB-YOB and QOB-POB interactions in the case of 180 instruments. Note that in both cases, only quarter-of-birth 1--3 are used; quarter 4 is omitted in order to avoid multicollinearity.

\section{Proof of Lemma \ref{lem:strongID}}
\label{sec:pf_lem_strongID}
Under strong identification, by \eqref{eq:limit} and Assumption \ref{ass:variance_est}, we have
\begin{align*}
	\begin{pmatrix}
		1 & 0 & 0 \\
		0 & 1 & 0 \\
		0 & 0 & d_n 
	\end{pmatrix}\begin{pmatrix}
		Q_{e,e} \\
		Q_{X,e} \\
		Q_{X,X} 
	\end{pmatrix} \convD \N\left(\begin{pmatrix}
		0 \\
		0 \\
		\widetilde{\mathcal{C}}
	\end{pmatrix},\begin{pmatrix}
		\Phi_1 & \Phi_{12} & 0 \\
		\Phi_{12} & \Psi & 0 \\
		0 & 0 & 0
	\end{pmatrix}\right), 
\end{align*}

In addition, we note that $e_i(\beta_0) = e_i + X_i \Delta$ with $\Delta = d_n  \widetilde{\Delta} \rightarrow 0$. Therefore, we have 
\begin{align*}
	Q_{e(\beta_0),e(\beta_0)} & = Q_{e,e} + 2\Delta Q_{X,e} + \Delta^2 Q_{X,X} = Q_{e,e} + o_p(1), \\
	Q_{X,e(\beta_0)} & = Q_{X,e} + \Delta Q_{X,X} = Q_{X,e} + \widetilde{\mathcal{C}} \widetilde{\Delta} + o_p(1), \\
	\widehat{\Phi}_1^{1/2}(\beta_0) & \convP \Phi_1^{1/2}, \quad \text{and} \quad \widehat{\Psi}^{1/2}(\beta_0) \convP \Psi^{1/2}. 
\end{align*}
This implies 
\begin{align*}
	\begin{pmatrix}
		AR(\beta_0) \\
		LM(\beta_0)
	\end{pmatrix} = \begin{pmatrix}
		Q_{e(\beta_0),e(\beta_0)}/\widehat{\Phi}_1^{1/2}(\beta_0) \\
		Q_{X,e(\beta_0)}/\widehat{\Psi}^{1/2}(\beta_0) \\
	\end{pmatrix}\convD \N\left(\begin{pmatrix}
		0 \\
		\frac{\widetilde{\mathcal{C}}\widetilde{\Delta}}{\Psi^{1/2}} 
	\end{pmatrix},\begin{pmatrix}
		1 & \rho  \\
		\rho & 1 
	\end{pmatrix}\right).  
\end{align*}

\section{Proof of Lemma \ref{lem:ump}}
Recall $\N^*_2 = (1-\rho^2)^{-1/2} (\N_2- \rho\N_1)$ and 
\begin{align*}
	\begin{pmatrix}
		\N_1 \\
		\N^*_2
	\end{pmatrix} \stackrel{d}{=} \N\left(\begin{pmatrix}
		0 \\
		\frac{\theta}{(1-\rho^2)^{1/2}} 
	\end{pmatrix},\begin{pmatrix}
		1 & 0 \\
		0 & 1 
	\end{pmatrix}\right).
\end{align*}
Because $\rho$ is known, it suffices to construct the uniformly most powerful invariant test based on observations $(\N_1,\N^*_2)$. As the null and alternative are invariant to sign changes, the maximum invariant is $(\N_1,\N^{*2}_2)$. Then, \citet[Theorem 6.2.1]{LR06} implies the invariant test should be based on the maximum invariant. Note $(\N_1,\N^{*2}_2)$ are independent, $\N_1$ follows a standard normal distribution, and $\N^*_2$ follows a noncentral chi-square distribution with one degree of freedom and noncentrality parameter $\lambda = \frac{\theta^2}{1-\rho^2}$.
Therefore, by the Neyman-Pearson's Lemma (\citet[Theorem 3.2.1]{LR06}), the most powerful test based on observations $(\N_1,\N^{*2}_2)$ is the likelihood ratio test where the likelihood ratio function evaluated at  $(\N_1= \ell_1,\N^{*2}_2 = \ell_2)$  depends on $\ell_2$ only and can be written as 
\begin{align*}
	LR\left( \ell_2;\lambda \right) = - \frac{\lambda}{2} + \log \left( \frac{\exp(\sqrt{\lambda \ell_2}) + \exp(-\sqrt{\lambda \ell_2}) }{2}\right) \end{align*}

In addition, we note that $LR\left( \ell_2;\lambda \right)$ is monotone increasing in $\ell_2$ for any $\lambda \geq 0$ and $\ell_2 \geq 0$. Therefore, \citet[Theorem 3.4.1]{LR06} implies the likelihood ratio test is equivalent to $1\{\N^{*2}_2 \geq \mathbb{C}_{\alpha}\}$, which is uniformly most powerful among tests for $\lambda = 0$ v.s. $\lambda>0$ and based on observations $(\N_1,\N^{*2}_2)$ only. This means it is also the uniformly most powerful test that is invariant to sign changes.  

In addition, the joint density of $(\N_1,\N_2)$ is 
\begin{align*}
	& (2\pi)^{-1} (1-\rho^2)^{-1/2} \exp\left(-\frac{1}{2}\left(\frac{\N_1^2}{1-\rho^2} - \frac{2\rho \N_1 \N_2}{1-\rho^2} + \frac{\N_2^2}{1-\rho^2} \right) \right) \exp\left(\theta \frac{\rho \N_1 - \N_2}{1-\rho^2}\right) \exp\left( \frac{\theta^2}{1-\rho^2} \right) \\
	& \equiv C(\theta) \exp(\theta \N_2^*) h(\N_1,\N_2),
\end{align*}
where $C(\theta) = (2\pi)^{-1} (1-\rho^2)^{-1/2}\exp\left( \frac{\theta^2}{1-\rho^2} \right)$ and $h(\N_1,\N_2) = \exp\left(-\frac{1}{2}\left(\frac{\N_1^2}{1-\rho^2} - \frac{2\rho \N_1 \N_2}{1-\rho^2} + \frac{\N_2^2}{1-\rho^2} \right) \right)$. Note that $\N_2^*$ is symmetric around $0$ under the null.  
By \citet[Section 4.2]{LR06}, $1\{\N_2^{*2} \geq \mathbb{C}_\alpha\}$ is the UMP unbiased level-$\alpha$ test. 

\section{Proof of Lemma \ref{lem:strongID2}}
Under strong identification and fixed alternatives, because $( Q_{e(\beta_0),e(\beta_0)} - \Delta^2 \mathcal{C}, Q_{X,e(\beta_0)} - \Delta \mathcal{C},Q_{X,X} - \mathcal{C})^\top
= O_p(1)$, we have
\begin{align*}
	\begin{pmatrix}
		d_n AR(\beta_0) \\
		d_n LM(\beta_0) \\
	\end{pmatrix} \convP 
	\begin{pmatrix}
		\frac{\Delta^2 \widetilde{\mathcal{C}}}{\Phi_1^{1/2}(\beta_0)}\\
		\frac{\Delta \widetilde{\mathcal{C}}}{\Psi^{1/2}(\beta_0)}
	\end{pmatrix}.
\end{align*}
This implies
\begin{align*}
	d_n LM^*(\beta_0)  \convP \frac{1}{(1-\rho^2(\beta_0))^{1/2}}\left(\frac{\Delta \widetilde{\mathcal{C}}}{\Psi^{1/2}(\beta_0)} - \frac{\rho(\beta_0)\Delta^2 \widetilde{\mathcal{C}}}{\Phi_1^{1/2}(\beta_0)}\right),
\end{align*}
which leads to the desired result.

\section{Proof of Lemma \ref{lem:weakID}}
Under weak identification, \eqref{eq:limit} implies  
\begin{align*}
	\begin{pmatrix}
		Q_{e(\beta_0),e(\beta_0)}\\
		Q_{X,e(\beta_0)}
	\end{pmatrix} = 
	\begin{pmatrix}
		Q_{e,e} + 2\Delta Q_{X,e} + \Delta^2 Q_{X,X}\\
		Q_{X,e} + \Delta Q_{X,X}
	\end{pmatrix} \convD \N\left(\begin{pmatrix}
		\Delta^2 \widetilde{\mathcal{C}} \\
		\Delta \widetilde{\mathcal{C}}
	\end{pmatrix}, \begin{pmatrix}
		\Phi_1(\beta_0) & \Phi_{12}(\beta_0) \\
		\Phi_{12}(\beta_0) & \Psi(\beta_0)
	\end{pmatrix}   \right),
\end{align*}
which leads to the first result. 

For the second result, it is obvious that $m_1(\Delta) \rightarrow \widetilde{\mathcal{C}}\Upsilon^{-1/2}$. In addition, we have

\begin{align*}
	m_2(\Delta) &= \frac{\widetilde{\mathcal{C}} \left(\Delta \Phi_1(\beta_0) -  \Delta^2 \Phi_{12}(\beta_0) \right)}{(\Phi_1(\beta_0)(\Phi_1(\beta_0) \Psi(\beta_0) - \Phi^2_{12}(\beta_0)))^{1/2}} \\
	& \rightarrow \frac{ \tau \widetilde{\mathcal{C}}}{(\Upsilon(\Upsilon \Psi - \tau^2))^{1/2}} = \frac{\widetilde{\mathcal{C}}}{\Upsilon^{1/2}} \frac{\rho_{23}}{(1-\rho_{23}^2)^{1/2}}, 
\end{align*}
where we use the fact that 
\begin{align*}
	& \Phi_1(\beta_0)/\Delta^4 \rightarrow \Upsilon, \\
	& (\Phi_1(\beta_0) \Psi(\beta_0) - \Phi^2_{12}(\beta_0))/\Delta^4 \rightarrow \Upsilon \Psi - \tau^2, \\
	&     \frac{\Phi_1(\beta_0) - \Delta \Phi_{12}(\beta_0)}{\Delta^3} \rightarrow \tau.
\end{align*}

\section{Proof of Theorem \ref{thm:admissible}}
The first statement in Theorem \ref{thm:admissible}(i) is a direct consequence of \citet[Theorem 2.1]{marden1982} because the acceptance region $\mathcal{A} = \{(A,B): \nu_1 A^2 + \nu_2 B^2 \leq \mathbb{C}_\alpha(a_1,a_2;\rho(\beta_0))\}$ is closed, convex, and monotone decreasing in the sense that if $(A,B) \in \mathcal{A}$ and $A' \leq A$, $B' \leq B$, then $(A',B') \in \mathcal{A}$. The second statement in Theorem \ref{thm:admissible}(i) follows \citet[Theorem 2.1]{Andrews(2016)}, which is a direct consequence of results in \cite{MS76} and \cite{KP78}. 

For Theorem \ref{thm:admissible}(ii), we note that $\tilde \rho = \rho$ under local alternatives and  
\begin{align*}
	\phi_{a_1,a_2,\infty} = 1\left\{ (a_1 + a_2  \rho^2)\N_1^2 + 2 a_2 \rho(1 - \rho^2)^{1/2}\N_1 \N_2^* + (1- a_1 - a_2 \rho^2)\N_2^{*2} \geq \mathbb{C}_{\alpha}(a_1,a_2; \rho) \right\}.    
\end{align*}
The ``if" part of Theorem \ref{thm:admissible}(ii) is a direct consequence of Lemma \ref{lem:ump}. The ``only if" part of Theorem  \ref{thm:admissible}(ii) is a direct consequence of the necessary part of \citet[Theorem 3.2.1]{LR06}. Specifically, given $\N_1$ and $\N_2^*$ are independent, the ``only if" part requires $a_1 + a_2\rho^2 = 0$, which implies $a_1=0$ and $a_2 \rho = 0$.  

For Theorem \ref{thm:admissible}(iii), we consider two cases of fixed alternatives: (1) $\Delta \neq \Phi_1^{1/2}(\beta_0)\Psi^{-1/2}(\beta_0)\rho^{-1}(\beta_0)$ and (2) $\Delta = \Phi_1^{1/2}(\beta_0)\Psi^{-1/2}(\beta_0)\rho^{-1}(\beta_0)$. In Case (1), by Lemma \ref{lem:strongID2}, the limits of $d_n^2AR^2(\beta_0)$, $d_n^2LM^{2}(\beta_0)$, $d_n^2LM^{*2}(\beta_0)$ are all positive, which implies that for all $(a_{1,n},a_{2,n}) \in \mathbb{A}_0$,
\begin{align*}
	1\{a_{1,n} AR^2(\beta_0) + a_{2,n}LM^2(\beta_0) + (1-a_{1,n}-a_{2,n})LM^{*2}(\beta_0)  \geq \mathbb{C}_{\alpha}(a_{1,n},a_{2,n};\widehat{\rho}(\beta_0))\} \convP 1.     
\end{align*}
In Case (2), we have 
\begin{align*}
	& \mathbb{P}\left(a_{1,n} AR^2(\beta_0) + a_{2,n}LM^2(\beta_0) + (1-a_{1,n}-a_{2,n})LM^{*2}(\beta_0)  \geq \mathbb{C}_{\alpha}(a_{1,n},a_{2,n};\widehat{\rho}(\beta_0)) \right) \\
	& \geq \mathbb{P}\left(\frac{\tilde{q}\Psi^{2}(\beta_0)\rho^{4}(\beta_0) }{\widetilde{\mathcal{C}}^2\Phi_1(\beta_0)} d_n^2 AR^2(\beta_0)   \geq \mathbb{C}_{\alpha}(a_{1,n},a_{2,n};\widehat{\rho}(\beta_0) \right) \\
	& \geq \mathbb{P}\left(  \tilde{q} + o_p(1)   \geq \mathbb{C}_{\alpha,\max}(\rho(\beta_0)) \right) \rightarrow 1,
\end{align*}
where the first inequality follows from the restriction on $a_{1.n}$ and the facts that $LM^2(\beta_0)\geq 0$ and $LM^{*2}(\beta_0) \geq  0$, the second inequality follows from $d_n^2 AR^2(\beta_0) \convP \Phi_1^{-1}(\beta_0)\Delta_*^4(\beta_0) \tilde{\mathcal{C}}^2$ (by Lemma \ref{lem:strongID2}) and $\widehat{\rho}(\beta_0) \convP \rho(\beta_0)$, and the last convergence follows from the fact that $\tilde{q}> \mathbb{C}_{\alpha,\max}(\rho(\beta_0))$. This concludes the proof.

\section{Proof of Theorem \ref{thm:weakid}}
We are under weak identification. By Lemma \ref{lem:weakID} and Assumption \ref{ass:variance_est}, we have
\begin{align*}
	\begin{pmatrix}
		AR(\beta_0) \\
		LM^*(\beta_0) \\
		\widehat{D} 
	\end{pmatrix} \convD \N\left(\begin{pmatrix}
		m_1(\Delta) \\
		m_2(\Delta) \\
		\mu_D
	\end{pmatrix}, \begin{pmatrix}
		1 & 0 & 0 \\
		0 & 1 & 0 \\
		0 & 0 & \sigma_D^2
	\end{pmatrix} \right).
\end{align*}
This implies $(AR(\beta_0),LM^*(\beta_0),\widehat{D} )$ are asymptotically independent. By Assumption \ref{ass:a}, we have
\begin{align*}
	(AR^2(\beta_0), LM^{*2}(\beta_0),\mathcal{A}_s(\widehat{D},\widehat{\gamma}(\beta_0))) \convD (\mathcal{Z}^2(m_1(\Delta)),\mathcal{Z}^2(m_2(\Delta)),\mathcal{A}_s(D,\gamma(\beta_0)))    
\end{align*}
where the two normal random variables are independent and independent of $D$, and by definition, $\mathcal{A}_s(D,\gamma(\beta_0))) = (a_1(f_s(D,\gamma(\beta_0)),\gamma(\beta_0)),a_2(f_s(D,\gamma(\beta_0)),\gamma(\beta_0)))$. In addition, we have $\widehat{\rho}(\beta_0) \convP \rho(\beta_0)$. By the bounded convergence theorem, this further implies
\begin{align}
	\mathbb{E}\widehat{\phi}_{\mathcal{A}_s(\widehat{D},\widehat{\gamma}(\beta_0))} \rightarrow \mathbb{E}\phi_{a_1(f_s(D,\gamma(\beta_0)),\gamma(\beta_0)),a_2(f_s(D,\gamma(\beta_0)),\gamma(\beta_0)),\infty}(\Delta,\mu_D,\gamma(\beta_0)).
	\label{eq:phihat_conv}    
\end{align}

In addition, suppose the null holds so that $\Delta = 0$. This implies $m_1(\Delta) = m_2(\Delta) = 0$. Then,  we have
\begin{align*}
	(\widehat{\phi}_{\mathcal{A}_s(\widehat{D},\widehat{\gamma}(\beta_0))}-\alpha)f(\widehat{D}) \convD (\phi_{a_1(f_s(D,\gamma(\beta_0)),\gamma(\beta_0)),a_2(f_s(D,\gamma(\beta_0)),\gamma(\beta_0)),\infty}(0,\mu_D,\gamma(\beta_0)) - \alpha)f(D),
\end{align*}
where
\begin{align*}
	& \phi_{a_1(f_s(D,\gamma(\beta_0)),\gamma(\beta_0)),a_2(f_s(D,\gamma(\beta_0)),\gamma(\beta_0)),\infty}(0,\mu_D,\gamma(\beta_0)) \\
	& = 1\begin{Bmatrix}
		a_1(f_s(D,\gamma(\beta_0)),\gamma(\beta_0)) \mathcal{Z}_1^2 + a_2(f_s(D,\gamma(\beta_0)),\gamma(\beta_0)) (\rho(\beta_0) \mathcal{Z}_1 + (1-\rho^2(\beta_0))^{1/2} \mathcal{Z}_2) \\
		(1 -a_1(f_s(D,\gamma(\beta_0)),\gamma(\beta_0))-a_2(f_s(D,\gamma(\beta_0)),\gamma(\beta_0)))\mathcal{Z}_2^2  \\
		\geq \mathbb{C}_\alpha(a_1(f_s(D,\gamma(\beta_0)),\gamma(\beta_0)),a_2(f_s(D,\gamma(\beta_0)),\gamma(\beta_0));\rho(\beta_0))
	\end{Bmatrix},
\end{align*}
$\mathcal{Z}_1$ and $\mathcal{Z}_2$ are independent standard normals, and they are independent of $D$. 
Then, by the definition of $\mathbb{C}_\alpha(\cdot)$, we have
\begin{align*}
	& \mathbb{E}(\widehat{\phi}_{\mathcal{A}_s(\widehat{D},\widehat{\gamma}(\beta_0))}-\alpha)h(\widehat{D}) \rightarrow \mathbb{E} \left[\mathbb{E} \left(\phi_{a(f_s(D,\gamma(\beta_0)),\gamma(\beta_0)),\infty}(0,\mu_D,\gamma(\beta_0)) - \alpha|D\right)h(D)\right] = 0. 
\end{align*}

\section{Proof of Corollary \ref{cor:weakid}}
By the continuous mapping theorem, we have 
\begin{align*}
	\lim_{n\rightarrow \infty}   \frac{\mathbb{E} \hat \phi_{\mathcal{A}_s(\widehat D, \widehat \gamma(\beta_0))}1\{ |\widehat D -d| \leq \eps \} }{\mathbb{E}  1\{ |\widehat D -d| \leq \eps \} } =  \frac{\mathbb{E}(\phi_{a_1(f_s(D,\gamma(\beta_0)),\gamma(\beta_0)),a_2(f_s(D,\gamma(\beta_0)),\gamma(\beta_0)),\infty}1\{|D-d|\leq \eps\})}{\mathbb{E} 1\{|D-d|\leq \eps)\}},
\end{align*}
and 
\begin{align*}
	& \lim_{\eps \rightarrow 0}\frac{\mathbb{E}(\phi_{a_1(f_s(D,\gamma(\beta_0)),\gamma(\beta_0)),a_2(f_s(D,\gamma(\beta_0)),\gamma(\beta_0)),\infty}1\{|D-d|\leq \eps\})}{\mathbb{E} 1\{|D-d|\leq \eps)\}} \\
	& = \mathbb{E}(\phi_{a_1(f_s(D,\gamma(\beta_0)),\gamma(\beta_0)),a_2(f_s(D,\gamma(\beta_0)),\gamma(\beta_0)),\infty}|D=d),
\end{align*}
where, by construction, we have 
\begin{align*}
	& \phi_{a_1(f_s(D,\gamma(\beta_0)),\gamma(\beta_0)),a_2(f_s(D,\gamma(\beta_0)),\gamma(\beta_0)),\infty} \\
	&= 1\{\nu_{1,s}(D,\gamma(\beta_0))\widetilde{\N}_1^2 + \nu_{2,s}(D,\gamma(\beta_0))\widetilde{\N}_2^2 \geq \widetilde{\mathbb{C}}_{\alpha}(\nu_{1,s}(D,\gamma(\beta_0)),\nu_{2,s}(D,\gamma(\beta_0)) )\}
\end{align*}
and
\begin{align*}
	(\widetilde{\N}_1, \widetilde{\N}_2) =  (\mathcal{Z}_1( m_1(\Delta)), \mathcal{Z}_2(m_2(\Delta)))\mathcal{U}_s(D,\gamma(\beta_0)).
\end{align*}

Similarly, we can show 
\begin{align*}
	\lim_{\eps \rightarrow 0} \lim_{n\rightarrow \infty}   \frac{\mathbb{E} \tilde \phi(\widetilde {AR}_s^2(\beta_0),
		\widetilde {LM}^{*2}_s(\beta_0),\widehat{D},\widehat{\gamma}(\beta_0)) 1\{ |\widehat D -d| \leq \eps \} }{\mathbb{E}  1\{ |\widehat D -d| \leq \eps \} } =  \mathbb{E}(\tilde \phi(\widetilde{\N}_1^2, \widetilde{\N}_2^2,D,\gamma(\beta_0))|D=d).
\end{align*}

Therefore, conditional on $D=d$, $\phi_{a_1(f_s(D,\gamma(\beta_0)),\gamma(\beta_0)),a_2(f_s(D,\gamma(\beta_0)),\gamma(\beta_0)),\infty}$ is a linear combination of 
$(\widetilde{\N}_1^2,\widetilde{\N}_2^2)$ with weights $(\nu_{1,s}(d,\gamma(\beta_0)),\nu_{2,s}(d,\gamma(\beta_0)))$, and $\widetilde \N_1$ and $\widetilde \N_2$ are two independent normal random variables with unit variance and expectations $\theta_1$ and $\theta_2$, respectively. Under the null, we have $(\theta_1,\theta_2) = (0,0)$, which, by definition of $\tilde \phi(\cdot)$, implies 
\begin{align*}
	\mathbb{E}(\tilde \phi(\widetilde{\N}_1^2, \widetilde{\N}_2^2,D,\gamma(\beta_0))|D=d) \leq \alpha.
\end{align*}
Therefore, $\tilde \phi(\widetilde{\N}_1^2, \widetilde{\N}_2^2,D,\gamma(\beta_0))$ is a level-$\alpha$ test. Then, the two optimality results follow Theorem \ref{thm:admissible}(i).

\section{Proof of Theorem \ref{thm:strongid}}

Denote $c_{\mathcal{B}}=c_{\mathcal{B}}(\beta)$ and $\Delta_* = \Delta_*(\beta)$. By Assumption \ref{ass:variance_est}, $\Phi_1>0$, which implies $|\Delta_*|>0$. Under strong identification and local alternatives, we have $\Delta \rightarrow 0$, $c_{\mathcal{B}}(\beta_0)\rightarrow c_{\mathcal{B}}$, $\Delta_*(\beta_0) \rightarrow \Delta_*$, $\mathbb{C}_{\alpha,\max}(\rho(\beta_0)) \rightarrow \mathbb{C}_{\alpha,\max}(\rho)$, and
\begin{align*}
	\begin{pmatrix}
		AR(\beta_0) \\
		LM^*(\beta_0) \\
		d_n  \widehat{D}
	\end{pmatrix} \convD \N\left(\begin{pmatrix}
		0 \\
		\frac{\widetilde{\Delta} \widetilde{\mathcal{C}}}{((1-\rho^2)\Psi)^{1/2}} \\
		\widetilde{\mathcal{C}}
	\end{pmatrix},\begin{pmatrix}
		1 & 0 & 0   \\
		0 & 1 & 0 \\
		0 & 0 & 0
	\end{pmatrix}\right).  
\end{align*}
This implies 
$d_n  \widehat{\sigma}_D \sqrt{\widehat{r}} = d_n  \widehat{D}  \convP \widetilde{\mathcal{C}}$, which further implies $d_n  f_{pp}(\widehat{D},\widehat{\gamma}(\beta_0)) \convP \widetilde{\mathcal{C}}$. For $f_{krs}(\widehat{D},\widehat{\gamma}(\beta_0))$, we note that 
\begin{align*}
	\max(\widehat{r} - 1,0)  \leq   \widehat{r}_{krs} \leq \widehat{r}.
\end{align*}
Therefore, we also have $f_{krs}(\widehat{D},\widehat{\gamma}(\beta_0)) d_n  \convP \widetilde{\mathcal{C}}$. Let $\mathcal{E}_n(\eps) = \{||\widehat{\gamma}(\beta_0) - \gamma(\beta_0)|| + |\delta_n \widehat{D}- \widetilde{\mathcal{C}}| \leq \eps\}$. Then, for an arbitrary $\eps>0$, we have $\mathbb{P}(\mathcal{E}_n(\eps)) \geq 1-\eps$ when $n$ is sufficiently large. 

Denote $\delta = d_n \widetilde{\delta}$. We have
\begin{align*}
	\mathcal{A}_s(\widehat{D},\widehat{\gamma}(\beta_0)) \in  \argmin_{(a_1,a_2) \in \mathbb{A}(f_{s}(\widehat{D},\widehat{\gamma}(\beta_0)),\widehat{\gamma}(\beta_0))} \sup_{\widetilde{\delta} \in  \widetilde{\mathcal{D}}_n}\left(\mathcal{P}_{d_n\widetilde{\delta},s}(\widehat{D},\widehat{\gamma}(\beta_0)) -     \mathbb{E}^* \phi_{a_1,a_2,s}(d_n\widetilde{\delta},\widehat{D},\widehat{\gamma}(\beta_0))\right),
\end{align*}
where 
$\widetilde{\mathcal{D}}_n = \{\widetilde{\delta}: d_n\widetilde{\delta}  \in \mathcal{D}(\beta_0)\}$. Let 
\begin{align*}
	& Q_n(a_1,a_2,\widetilde{\delta}) = \mathcal{P}_{d_n\widetilde{\delta},s}(\widehat{D},\widehat{\gamma}(\beta_0)) -     \mathbb{E}^*\phi_{a_1,a_2,s}(d_n\widetilde{\delta},\widehat{D},\widehat{\gamma}(\beta_0)) \quad \text{and}  \\
	& Q(a_1,a_2,\widetilde{\delta}) =
	\mathbb{E}1\{\mathcal{Z}_2^2(     (1-\rho^2)^{-1/2} \Psi^{-1/2} \widetilde{\delta} \widetilde{\mathcal{C}}) \geq \mathbb{C}_{\alpha}\} \notag \\
	& - \mathbb{E} 1\begin{Bmatrix}
		a_1 \mathcal{Z}_1^2 + a_2\left(\rho \mathcal{Z}_1 + (1-\rho^2)^{1/2}\mathcal{Z}_2 (     (1-\rho^2)^{-1/2} \Psi^{-1/2} \widetilde{\delta} \widetilde{\mathcal{C}}) \right)^2 \\
		+ (1-a_1 - a_2)\mathcal{Z}_2^2(     (1-\rho^2)^{-1/2} \Psi^{-1/2} \widetilde{\delta} \widetilde{\mathcal{C}}) \geq \mathbb{C}_{\alpha}(a_1,a_2;\rho) 
	\end{Bmatrix},
\end{align*}
where $\mathcal{Z}_1$ is standard normal, $\mathcal{Z}_2(     (1-\rho^2)^{-1/2} \Psi^{-1/2} \widetilde{\delta} \widetilde{\mathcal{C}})$ is normal with mean $(1-\rho^2)^{-1/2} \Psi^{-1/2} \widetilde{\delta} \widetilde{\mathcal{C}}$ and unit variance, and $\mathcal{Z}_1$ and $\mathcal{Z}_2(\cdot)$ are independent.
Then, we aim to show that 
\begin{align}
	& \sup_{(a_1,a_2) \in \mathbb{A}(f_{s}(\widehat{D},\widehat{\gamma}(\beta_0)),\widehat{\gamma}(\beta_0)), \widetilde{\delta} \in \widetilde{\mathcal{D}}_{n}} \biggl|Q_n(a_1,a_2,\widetilde{\delta}) - Q(a_1,a_2,\widetilde{\delta})  \biggr| \convP 0.
	\label{eq:Q0}
\end{align}

We divide $\widetilde{\mathcal{D}}_n$ into three parts: 
\begin{align*}
	\widetilde{\mathcal{D}}_{n,1}(\eps) & = \{\widetilde{\delta} \in \widetilde{\mathcal{D}}_n, |\widetilde{\delta}| \leq M_1(\eps)\}, \notag \\
	\widetilde{\mathcal{D}}_{n,2}(\eps) & = \left\{\widetilde{\delta} \in \widetilde{\mathcal{D}}_n, \left|\frac{d_n \widetilde{\delta}}{\widehat{\Delta}_*(\beta_0)}-1 \right| \leq \eps \right\}, \quad \text{and} \notag \\
	\widetilde{\mathcal{D}}_{n,3}(\eps) & = \widetilde{\mathcal{D}}_{n}\cap\widetilde{\mathcal{D}}_{n,1}^c(\eps) \cap \widetilde{\mathcal{D}}_{n,2}^c(\eps), 
\end{align*}
where $M_1(\eps)$ is a large constant so that 
\begin{align}\label{eq:def_M_1_eps}
	\mathbb{P}\left( (1-\overline{a}) \mathcal{Z}^2\left( \frac{M_1(\eps)\eps |\widetilde{\mathcal{C}}| }{ (2(1-\rho^2) \Psi c_{\mathcal{B}})^{1/2}}\right) \geq \mathbb{C}_{\alpha,\max}(\rho)+1\right) = 1-\eps.
\end{align}
When $n$ is sufficiently large and $\eps$ is sufficiently small, on $\mathcal{E}_n(\eps)$, there exists a constant $c$ such that 
\begin{align}
	& |\widehat{\Delta}_*(\beta_0)-\Delta_*| \leq  c\eps, \quad \inf_{\widetilde{\delta} \in \widetilde{\mathcal{D}}_{n,2}(\eps) }|d_n \widetilde{\delta}| \geq (1-\eps)(|\Delta_*| - c\eps), \notag  \\
	& |\widehat{\Phi}_1(\beta_0) - \Phi_1| \leq c\eps, \quad |d_n^2f^2_s(\widehat{D},\widehat{\gamma}(\beta_0)) - \widetilde{\mathcal{C}}^2|\leq  c\eps, \notag \\
	& \sup_{\widetilde{\delta} \in \widetilde{\mathcal{D}}_{n,2}(\eps) }\left[1- (d_n^2\widetilde{\delta}^2, d_n\widetilde{\delta}) \left(\begin{pmatrix}
		\widehat{\Phi}_1(\beta_0) & \widehat{\Phi}_{12}(\beta_0) \\
		\widehat{\Phi}_{12}(\beta_0) & \widehat{\Psi}(\beta_0)
	\end{pmatrix}^{-1} \begin{pmatrix}
		\widehat{\Phi}_{13}(\beta_0) \\
		\widehat{\tau}(\beta_0)
	\end{pmatrix}\right) \right]^{2} \notag \\
	& \leq   \left[1- (\Delta_*^2, \Delta_*) \left(\begin{pmatrix}
		\Phi_1 & \Phi_{12} \\
		\Phi_{12} & \Psi
	\end{pmatrix}^{-1} \begin{pmatrix}
		\Phi_{13} \\
		\tau
	\end{pmatrix}\right) \right]^2 + c\eps \leq c_{\mathcal{B}}+c\eps, \notag \\
	& |\widehat{c}_{\mathcal{B}}(\beta_0)-c_{\mathcal{B}}| \leq c\eps. 
	\label{eq:mathcalE_1}
\end{align}
This further implies 
\begin{align*}
	\widetilde{\mathcal{D}}_{n,1}(\eps) \cap \widetilde{\mathcal{D}}_{n,2}(\eps)  = \emptyset.  
\end{align*}

Recall $\phi_{a_1,a_2,s}(\delta,\widehat{D},\widehat{\gamma}(\beta_0))$ defined in \eqref{eq:phias}. With $\delta$ replaced by $d_n \widetilde{\delta}$ and when $\widetilde{\delta} \in \widetilde{\mathcal{D}}_{n,1}(\eps)$, we have 
\begin{align*}
	\begin{pmatrix}
		d_n^{-1}\widehat{C}_{1}(d_n \widetilde{\delta}) \\
		d_n^{-1}\widehat{C}_{2}(d_n \widetilde{\delta})
	\end{pmatrix} (d_n f_s(\widehat{D},\widehat{\gamma}(\beta_0))) \convP \begin{pmatrix}
		0 \\
		(1-\rho^2)^{-1/2} \Psi^{-1/2} \widetilde{\delta} \widetilde{\mathcal{C}} 
	\end{pmatrix},
\end{align*}
Therefore, uniformly over $(a_1,a_2) \in \mathbb{A}_0$ and $\widetilde{\delta} \in \widetilde{\mathcal{D}}_{n,1}(\eps)$ and conditional on data, we have 
\begin{align*}
	\phi_{a_1,a_2,s}(d_n \widetilde{\delta}, \widehat{D},\widehat{\gamma}(\beta_0)) \convD     1\begin{Bmatrix}
		a_1 \mathcal{Z}_1^2 + a_2\left(\rho \mathcal{Z}_1 + (1-\rho^2)^{1/2}\mathcal{Z}_2(     (1-\rho^2)^{-1/2} \Psi^{-1/2} \widetilde{\delta} \widetilde{\mathcal{C}}) \right)^2 \\
		+ (1-a_1 - a_2)\mathcal{Z}_2^2(     (1-\rho^2)^{-1/2} \Psi^{-1/2} \widetilde{\delta} \widetilde{\mathcal{C}}) \geq \mathbb{C}_{\alpha}(a_1,a_2;\rho) 
	\end{Bmatrix}.
\end{align*}
This implies
\begin{align*}
	& \sup_{(a_1,a_2) \in \mathbb{A}_0, \widetilde{\delta} \in \widetilde{\mathcal{D}}_{n,1}(\eps)}    \bigg|\mathbb{E}^*\phi_{a_1,a_2,s}(d_n \widetilde{\delta}, \widehat{D},\widehat{\gamma}(\beta_0))  \\
	& - \mathbb{E} 1\begin{Bmatrix}
		a_1 \mathcal{Z}_1^2 + a_2\left(\rho \mathcal{Z}_1 + (1-\rho^2)^{1/2}\mathcal{Z}_2(     (1-\rho^2)^{-1/2} \Psi^{-1/2} \widetilde{\delta} \widetilde{\mathcal{C}}) \right)^2 \\
		+ (1-a_1 - a_2)\mathcal{Z}_2^2(     (1-\rho^2)^{-1/2} \Psi^{-1/2} \widetilde{\delta} \widetilde{\mathcal{C}}) \geq \mathbb{C}_{\alpha}(a_1,a_2;\rho) 
	\end{Bmatrix} \bigg| \convP 0.
\end{align*}

In addition, by Lemma \ref{lem:ump}, for any $\widetilde{\delta}$, $\mathbb{E} 1\begin{Bmatrix}
	a_1 \mathcal{Z}_1^2 + a_2\left(\rho \mathcal{Z}_1 + (1-\rho^2)^{1/2}\mathcal{Z}_2(     (1-\rho^2)^{-1/2} \Psi^{-1/2} \widetilde{\delta} \widetilde{\mathcal{C}}) \right)^2 \\
	+ (1-a_1 - a_2)\mathcal{Z}_2^2(     (1-\rho^2)^{-1/2} \Psi^{-1/2} \widetilde{\delta} \widetilde{\mathcal{C}}) \geq \mathbb{C}_{\alpha}(a_1,a_2;\rho) 
\end{Bmatrix}$ is maximized at $a_1=0$ and $a_2\rho =0$. This implies
\begin{align*}
	&  \sup_{\widetilde{\delta} \in \widetilde{\mathcal{D}}_{n,1}(\eps)}  |\mathcal{P}_{d_n\widetilde{\delta},s}(\widehat{D},\widehat{\gamma}(\beta_0)) -\mathbb{E}1\{\mathcal{Z}_2^2(     (1-\rho^2)^{-1/2} \Psi^{-1/2} \widetilde{\delta} \widetilde{\mathcal{C}}) \geq \mathbb{C}_{\alpha}\}| \\
	&  = \sup_{\widetilde{\delta} \in \widetilde{\mathcal{D}}_{n,1}(\eps)}|\sup_{(a_1,a_2) \in \mathbb{A}(f_{s}(\widehat{D},\widehat{\gamma}(\beta_0)),\widehat{\gamma}(\beta_0))}\mathbb{E}^*\phi_{a_1,a_2,s}(d_n \widetilde{\delta}, \widehat{D},\widehat{\gamma}(\beta_0))  - \mathbb{E}1\{\mathcal{Z}_2^2(     (1-\rho^2)^{-1/2} \Psi^{-1/2} \widetilde{\delta} \widetilde{\mathcal{C}}) \geq \mathbb{C}_{\alpha}\}| \\
	& \leq \sup_{\widetilde{\delta} \in \widetilde{\mathcal{D}}_{n,1}(\eps)}\bigg|\sup_{(a_1,a_2) \in \mathbb{A}(f_{s}(\widehat{D},\widehat{\gamma}(\beta_0)),\widehat{\gamma}(\beta_0))}\mathbb{E} 1\begin{Bmatrix}
		a_1 \mathcal{Z}_1^2 + a_2\left(\rho \mathcal{Z}_1 + (1-\rho^2)^{1/2}\mathcal{Z}_2(     (1-\rho^2)^{-1/2} \Psi^{-1/2} \widetilde{\delta} \widetilde{\mathcal{C}}) \right)^2 \\
		+ (1-a_1 - a_2)\mathcal{Z}_2^2(     (1-\rho^2)^{-1/2} \Psi^{-1/2} \widetilde{\delta} \widetilde{\mathcal{C}}) \geq \mathbb{C}_{\alpha}(a_1,a_2;\rho) 
	\end{Bmatrix} \\
	& -  \mathbb{E}1\{\mathcal{Z}_2^2(     (1-\rho^2)^{-1/2} \Psi^{-1/2} \widetilde{\delta} \widetilde{\mathcal{C}}) \geq \mathbb{C}_{\alpha}\}\bigg| + o_p(1),\\
	& \leq \sup_{\widetilde{\delta} \in \widetilde{\mathcal{D}}_{n,1}(\eps)}\bigg|\sup_{(a_1,a_2) \in \mathbb{A}_0}\mathbb{E} 1\begin{Bmatrix}
		a_1 \mathcal{Z}_1^2 + a_2\left(\rho \mathcal{Z}_1 + (1-\rho^2)^{1/2}\mathcal{Z}_2(     (1-\rho^2)^{-1/2} \Psi^{-1/2} \widetilde{\delta} \widetilde{\mathcal{C}}) \right)^2 \\
		+ (1-a_1 - a_2)\mathcal{Z}_2^2(     (1-\rho^2)^{-1/2} \Psi^{-1/2} \widetilde{\delta} \widetilde{\mathcal{C}}) \geq \mathbb{C}_{\alpha}(a_1,a_2;\rho) 
	\end{Bmatrix} \\
	& -  \mathbb{E}1\{\mathcal{Z}_2^2(     (1-\rho^2)^{-1/2} \Psi^{-1/2} \widetilde{\delta} \widetilde{\mathcal{C}}) \geq \mathbb{C}_{\alpha}\}\bigg| + o_p(1)= o_p(1),
\end{align*}
where the second inequality is due to the facts that $\underline{a}( f_{s}(\widehat{D},\widehat{\gamma}(\beta_0)),\widehat{\gamma}(\beta_0)) = o_p(1)$ under strong identification  and $\mathbb{E} 1\begin{Bmatrix}
	a_1 \mathcal{Z}_1^2 + a_2\left(\rho \mathcal{Z}_1 + (1-\rho^2)^{1/2}\mathcal{Z}_2(     (1-\rho^2)^{-1/2} \Psi^{-1/2} \widetilde{\delta} \widetilde{\mathcal{C}}) \right)^2 \\
	+ (1-a_1 - a_2)\mathcal{Z}_2^2(     (1-\rho^2)^{-1/2} \Psi^{-1/2} \widetilde{\delta} \widetilde{\mathcal{C}}) \geq \mathbb{C}_{\alpha}(a_1,a_2;\rho) 
\end{Bmatrix}$ is continuous at $a_1 = 0$ uniformly over $|\widetilde{\delta}| \leq M_1(\eps)$. Therefore, we have 
\begin{align}
	& \sup_{(a_1,a_2) \in \mathbb{A}(f_{s}(\widehat{D},\widehat{\gamma}(\beta_0)),\widehat{\gamma}(\beta_0)), \widetilde{\delta} \in \widetilde{\mathcal{D}}_{n,1}(\eps)} \biggl|Q_n(a_1,a_2,\widetilde{\delta}) - Q(a_1,a_2,\widetilde{\delta})  \biggr| \convP 0.
	\label{eq:Q1}
\end{align}

Next, we consider the case in which  $\widetilde{\delta} \in \widetilde{\mathcal{D}}_{n,2}(\eps)$. We have
\begin{align*}
	& \phi_{a_1,a_2,s}(d_n\widetilde{\delta},\widehat{D},\widehat{\gamma}(\beta_0)) \\
	& = 1\begin{Bmatrix}
		a_1 \mathcal{Z}_1^2(\widehat{C}_{1}(d_n\widetilde{\delta})f_s(\widehat{D},\widehat{\gamma}(\beta_0))) \\
		+ a_2\left(\widehat{\rho}(\beta_0)\mathcal{Z}_1(\widehat{C}_{1}(d_n\widetilde{\delta})f_s(\widehat{D},\widehat{\gamma}(\beta_0))) + (1-\widehat{\rho}^2(\beta_0))^{1/2}\mathcal{Z}_2(\widehat{C}_{2}(d_n\widetilde{\delta})f_s(\widehat{D},\widehat{\gamma}(\beta_0)))  \right)^2 \\
		+ (1-a_1-a_2)\mathcal{Z}_2^2(\widehat{C}_{2}(d_n\widetilde{\delta})f_s(\widehat{D},\widehat{\gamma}(\beta_0)))  \geq \mathbb{C}_{\alpha}(a_1,a_2;\widehat{\rho}(\beta_0))
	\end{Bmatrix} \\
	& \geq  1\begin{Bmatrix}
		\underline{a}(f_{s}(\widehat{D},\widehat{\gamma}(\beta_0)),\widehat{\gamma}(\beta_0)) \mathcal{Z}_1^2(\widehat{C}_{1}(d_n\widetilde{\delta})f_s(\widehat{D},\widehat{\gamma}(\beta_0))) \geq \mathbb{C}_{\alpha,\max}(\widehat{\rho}(\beta_0))
	\end{Bmatrix}.
\end{align*}

By \eqref{eq:mathcalE_1}, on $\mathcal{E}_n(\eps)$, there exists a constant $c>0$ such that 
\begin{align*}
	& \widehat{C}_1^2(d_n \widetilde{\delta})(d_nf_s(\widehat{D},\widehat{\gamma}(\beta_0)))^2  \\
	& = \frac{\widehat{\Phi}_1^{-1}(\beta_0) (d_n \widetilde{\delta})^4(d_nf_s(\widehat{D},\widehat{\gamma}(\beta_0)))^2 }{\left[1- (d_n^2\widetilde{\delta}^2, d_n\widetilde{\delta}) \left(\begin{pmatrix}
			\widehat{\Phi}_1(\beta_0) & \widehat{\Phi}_{12}(\beta_0) \\
			\widehat{\Phi}_{12}(\beta_0) & \widehat{\Psi}(\beta_0)
		\end{pmatrix}^{-1} \begin{pmatrix}
			\widehat{\Phi}_{13}(\beta_0) \\
			\widehat{\tau}(\beta_0)
		\end{pmatrix}\right) \right]^{2}} \\
	& \geq \frac{(\Phi_1(\beta_0)+c\eps)^{-1}(1-\eps)^4(|\Delta_*|-c\eps)^4 (\widetilde{\mathcal{C}}^2 - c\eps )}{c_{\mathcal{B}} + c\eps } \geq c
\end{align*}
and 
\begin{align*}
	& \underline{a}(f_{s}(\widehat{D},\widehat{\gamma}(\beta_0)),\widehat{\gamma}(\beta_0)) \widehat{C}_1^2(d_n \widetilde{\delta})f_s^2(\widehat{D},\widehat{\gamma}(\beta_0)) \\
	& \geq \frac{1.1 \mathbb{C}_{\alpha,\max}(\widehat{\rho}(\beta_0)) \widehat{\Phi}_1(\beta_0) \widehat{c}_{\mathcal{B}}(\beta_0) }{ \widehat{\Delta}_*^4(\beta_0) d_n^2f_{s}^2(\widehat{D},\widehat{\gamma}(\beta_0))} \widehat{C}_1^2(d_n \widetilde{\delta})(d_nf_s(\widehat{D},\widehat{\gamma}(\beta_0)))^2 \\
	&\geq \frac{1.1\mathbb{C}_{\alpha,\max}(\widehat{\rho}(\beta_0)) (\Phi_1 -c\eps)(c_{\mathcal{B}} - c\eps) }{(|\Delta_*|+c\eps)^4 (\widetilde{\mathcal{C}}^2 + c\eps)}  \frac{(\Phi_1(\beta_0)+c\eps)^{-1}(1-\eps)^4(|\Delta_*|-c\eps)^4 (\widetilde{\mathcal{C}}^2 - c\eps )}{c_{\mathcal{B}} + c\eps }\\
	& \geq (1.1-c\eps)\mathbb{C}_{\alpha,\max}(\widehat{\rho}(\beta_0)),
\end{align*}
where the last inequality holds because $\eps$ can be arbitrarily small. This means, on $\mathcal{E}_n(\eps)$ and when $\widetilde{\delta} \in \widetilde{\mathcal{D}}_{n,2}(\eps)$, 
\begin{align*}
	\mathbb{E}^*\phi_{a_1,a_2,s}(d_n\widetilde{\delta},\widehat{D},\widehat{\gamma}(\beta_0))  \geq  \mathbb{P}^*(o_p(1) + (1.1-c\eps)\mathbb{C}_{\alpha,\max}(\widehat{\rho}(\beta_0)) \geq \mathbb{C}_{\alpha,\max}(\widehat{\rho}(\beta_0)) ) \rightarrow 1.
\end{align*}
As $\mathbb{P}(\mathcal{E}_n(\eps))\rightarrow 1$, we have 
\begin{align*}
	\sup_{(a_1,a_2) \in \mathbb{A}(f_{s}(\widehat{D},\widehat{\gamma}(\beta_0)),\widehat{\gamma}(\beta_0)),\widetilde{\delta} \in \widetilde{\mathcal{D}}_{n,2}(\eps) }\left[1-\mathbb{E}^*\phi_{a_1,a_2,s}(d_n\widetilde{\delta},\widehat{D},\widehat{\gamma}(\beta_0))\right]  \convP 0,
\end{align*}
and thus, 
\begin{align}
	&\sup_{(a_1,a_2) \in \mathbb{A}(f_{s}(\widehat{D},\widehat{\gamma}(\beta_0)),\widehat{\gamma}(\beta_0)), \widetilde{\delta} \in \widetilde{\mathcal{D}}_{n,2}(\eps)} \left[\mathcal{P}_{d_n\widetilde{\delta},s}(\widehat{D},\widehat{\gamma}(\beta_0)) -     \mathbb{E}^*\phi_{a_1,a_2,s}(d_n\widetilde{\delta},\widehat{D},\widehat{\gamma}(\beta_0))\right] \notag \\
	&\leq \sup_{(a_1,a_2) \in \mathbb{A}(f_{s}(\widehat{D},\widehat{\gamma}(\beta_0)),\widehat{\gamma}(\beta_0)),\widetilde{\delta} \in \widetilde{\mathcal{D}}_{n,2}(\eps) }\left[1-\mathbb{E}^*\phi_{a_1,a_2,s}(d_n\widetilde{\delta},\widehat{D},\widehat{\gamma}(\beta_0))\right] \convP 0.
	\label{eq:PD2}
\end{align}
Furthermore, note that $a_1+a_2 \leq \overline{a} <1$ and when $\widetilde{\delta} \in \widetilde{\mathcal{D}}_{n,2}(\eps)$, on $\mathcal{E}_n(\eps)$, \eqref{eq:mathcalE_1} implies $\widetilde{\delta}^2 \rightarrow \infty$. Therefore, we have  
\begin{align*}
	& a_1 \mathcal{Z}_1^2 + a_2\left(\rho \mathcal{Z}_1 + (1-\rho^2)^{1/2}\mathcal{Z}_2(     (1-\rho^2)^{-1/2} \Psi^{-1/2} \widetilde{\delta} \widetilde{\mathcal{C}}) \right)^2 
	+ (1-a_1 - a_2)\mathcal{Z}_2^2(     (1-\rho^2)^{-1/2} \Psi^{-1/2} \widetilde{\delta} \widetilde{\mathcal{C}}) \\
	& \geq (1 - \overline{a}) \mathcal{Z}_2^2(     (1-\rho^2)^{-1/2} \Psi^{-1/2} \widetilde{\delta} \widetilde{\mathcal{C}}) =\frac{(1-\overline{a})\widetilde{\delta}^2 \widetilde{\mathcal{C}}^2}{(1-\rho^2)\Psi}(1+o_p(1))  \rightarrow \infty, 
\end{align*}
which further implies 
\begin{align*}
	\sup_{(a_1,a_2) \in \mathbb{A}(f_{s}(\widehat{D},\widehat{\gamma}(\beta_0)),\widehat{\gamma}(\beta_0)),\widetilde{\delta} \in \widetilde{\mathcal{D}}_{n,2}(\eps) }\left[1- \mathbb{E} 1 \begin{Bmatrix}
		a_1 \mathcal{Z}_1^2 + a_2\left(\rho \mathcal{Z}_1 + (1-\rho^2)^{1/2}\mathcal{Z}_2(     (1-\rho^2)^{-1/2} \Psi^{-1/2} \widetilde{\delta} \widetilde{\mathcal{C}}) \right)^2 \\
		+ (1-a_1 - a_2)\mathcal{Z}_2^2(     (1-\rho^2)^{-1/2} \Psi^{-1/2} \widetilde{\delta} \widetilde{\mathcal{C}}) \geq \mathbb{C}_{\alpha}(a_1,a_2;\rho)
	\end{Bmatrix}\right] \convP 0
\end{align*}
and
\begin{align}
	& \sup_{(a_1,a_2) \in \mathbb{A}(f_{s}(\widehat{D},\widehat{\gamma}(\beta_0)),\widehat{\gamma}(\beta_0)),\widetilde{\delta} \in \widetilde{\mathcal{D}}_{n,2}(\eps) }\bigg[
	\mathbb{E}1\{\mathcal{Z}_2^2(     (1-\rho^2)^{-1/2} \Psi^{-1/2} \widetilde{\delta} \widetilde{\mathcal{C}}) \geq \mathbb{C}_{\alpha}\} \notag \\
	& - \mathbb{E} 1\begin{Bmatrix}
		a_1 \mathcal{Z}_1^2 + a_2\left(\rho \mathcal{Z}_1 + (1-\rho^2)^{1/2}\mathcal{Z}_2(     (1-\rho^2)^{-1/2} \Psi^{-1/2} \widetilde{\delta} \widetilde{\mathcal{C}}) \right)^2 \\
		+ (1-a_1 - a_2)\mathcal{Z}_2^2(     (1-\rho^2)^{-1/2} \Psi^{-1/2} \widetilde{\delta} \widetilde{\mathcal{C}}) \geq \mathbb{C}_{\alpha}(a_1,a_2;\rho) 
	\end{Bmatrix} \bigg] \convP 0.
	\label{eq:ED2}
\end{align}

Combining \eqref{eq:PD2} and \eqref{eq:ED2}, we have
\begin{align}
	& \sup_{(a_1,a_2) \in \mathbb{A}(f_{s}(\widehat{D},\widehat{\gamma}(\beta_0)),\widehat{\gamma}(\beta_0)), \widetilde{\delta} \in \widetilde{\mathcal{D}}_{n,2}(\eps)} \biggl|Q_n(a_1,a_2,\widetilde{\delta}) - Q(a_1,a_2,\widetilde{\delta}) \biggr| \rightarrow 0.
	\label{eq:Q2}
\end{align}

Last, we consider the case in which $\widetilde{\delta} \in \widetilde{\mathcal{D}}_{n,3}(\eps)$. On $\mathcal{E}_n(\eps)$, \eqref{eq:mathcalE_1} implies 
\begin{align*}
	& \widehat{C}_2^2(d_n \widetilde{\delta})f_s^2(\widehat{D},\widehat{\gamma}(\beta_0)) \\
	&  = \frac{ \widetilde{\delta}^2(1-\frac{d_n\widetilde{\delta} }{\widehat{\Delta}_*(\beta_0)})^2}{(1-\widehat{\rho}^2(\beta_0)) \widehat{\Psi}(\beta_0)}\frac{d^2_nf^2_s(\widehat{D},\widehat{\gamma}(\beta_0))}{\left[1- (d_n^2\widetilde{\delta}^2, d_n\widetilde{\delta}) \left(\begin{pmatrix}
			\widehat{\Phi}_1(\beta_0) & \widehat{\Phi}_{12}(\beta_0) \\
			\widehat{\Phi}_{12}(\beta_0) & \widehat{\Psi}(\beta_0)
		\end{pmatrix}^{-1} \begin{pmatrix}
			\widehat{\Phi}_{13}(\beta_0) \\
			\widehat{\tau}(\beta_0)
		\end{pmatrix}\right) \right]^{2}} \\
	& \geq \frac{(1-c\eps)M_1^2(\eps)\eps^2 (\widetilde{\mathcal{C}}^2-c\eps) }{(1-\rho^2) \Psi c_{\mathcal{B}}} \\ 
	& \geq \frac{M_1^2(\eps)\eps^2 \widetilde{\mathcal{C}}^2 }{2(1-\rho^2) \Psi c_{\mathcal{B}}}, 
\end{align*}
where the second inequality holds when $\eps$ is sufficiently small. In this case, 
\begin{align*}
	\mathbb{E}^*\phi_{a_1,a_2,s}(d_n\widetilde{\delta},\widehat{D},\widehat{\gamma}(\beta_0)) & \geq  \mathbb{P}^*( (1-\overline{a}) \mathcal{Z}_2^2(\widehat{C}_2(d_n \widetilde{\delta})f_s(\widehat{D},\widehat{\gamma}(\beta_0))) \geq \mathbb{C}_{\alpha,\max}(\widehat{\rho}(\beta_0)) ) \\
	& \geq \mathbb{P}^*\left( (1-\overline{a}) \mathcal{Z}_2^2\left(\frac{M_1(\eps)\eps |\widetilde{\mathcal{C}}| }{(2(1-\rho^2) \Psi c_{\mathcal{B}})^{1/2}}\right) \geq \mathbb{C}_{\alpha,\max}(\widehat{\rho}(\beta_0))\right) \\
	& \geq \mathbb{P}^*\left( (1-\overline{a}) \mathcal{Z}_2^2\left(\frac{M_1(\eps)\eps |\widetilde{\mathcal{C}}| }{(2(1-\rho^2) \Psi c_{\mathcal{B}})^{1/2}}\right) \geq \mathbb{C}_{\alpha,\max}(\rho)+c\eps\right) - \eps \geq 1-2\eps,    
\end{align*}
where the second inequality is by the fact that the CDF (survival function) of $\mathcal{Z}^2(\lambda)$ is monotone decreasing (increasing) in $|\lambda|$ and the last equality is by the definition of $M_1(\eps)$ in (\ref{eq:def_M_1_eps})
and the fact that $\mathbb{C}_{\alpha,\max}(\widehat{\rho}(\beta_0)) \convP \mathbb{C}_{\alpha,\max}(\rho)$ . This implies, on $\mathcal{E}_n(\eps)$, 
\begin{align}
	\sup_{(a_1,a_2) \in \mathbb{A}(f_{s}(\widehat{D},\widehat{\gamma}(\beta_0)),\widehat{\gamma}(\beta_0)), \widetilde{\delta} \in \widetilde{\mathcal{D}}_{n,3}(\eps)} \left[\mathcal{P}_{d_n\widetilde{\delta},s}(\widehat{D},\widehat{\gamma}(\beta_0)) -     \mathbb{E}^*\phi_{a_1,a_2,s}(d_n\widetilde{\delta},\widehat{D},\widehat{\gamma}(\beta_0))\right] \leq 2\eps.
	\label{eq:PD3}
\end{align}

In addition, we note that $(1-\rho^2)^{-1} \Psi^{-1} \widetilde{\delta}^2 \widetilde{\mathcal{C}}^2$ satisfies 
\begin{align*}
	(1-\rho^2)^{-1} \Psi^{-1} \widetilde{\delta}^2 \widetilde{\mathcal{C}}^2 \geq \frac{M_1^2(\eps)\eps^2 \widetilde{\mathcal{C}}^2 }{2(1-\rho^2) \Psi c_{\mathcal{B}}}, 
\end{align*}
where we use the facts that $\widetilde{\delta}^2 \geq M_1^2(\eps)$, $c_{\mathcal{B}} \geq 1$, and $\eps <1$. Therefore, by the same argument, we have 
\begin{align*}
	&\mathbb{E} 1\begin{Bmatrix}
		a_1 \mathcal{Z}_1^2 + a_2\left(\rho \mathcal{Z}_1 + (1-\rho^2)^{1/2}\mathcal{Z}_2(     (1-\rho^2)^{-1/2} \Psi^{-1/2} \widetilde{\delta} \widetilde{\mathcal{C}}) \right)^2 \\
		+ (1-a_1 - a_2)\mathcal{Z}_2^2(     (1-\rho^2)^{-1/2} \Psi^{-1/2} \widetilde{\delta} \widetilde{\mathcal{C}}) \geq \mathbb{C}_{\alpha}(a_1,a_2;\rho) 
	\end{Bmatrix} \geq 1-\eps
\end{align*}
and
\begin{align}
	& \sup_{(a_1,a_2) \in \mathbb{A}(f_{s}(\widehat{D},\widehat{\gamma}(\beta_0)),\widehat{\gamma}(\beta_0)),\widetilde{\delta} \in \widetilde{\mathcal{D}}_{n,3}(\eps) }\bigg[
	\mathbb{E}1\{\mathcal{Z}_2^2(     (1-\rho^2)^{-1/2} \Psi^{-1/2} \widetilde{\delta} \widetilde{\mathcal{C}}) \geq \mathbb{C}_{\alpha}\} \notag \\
	& - \mathbb{E} 1\begin{Bmatrix}
		a_1 \mathcal{Z}_1^2 + a_2\left(\rho \mathcal{Z}_1 + (1-\rho^2)^{1/2}\mathcal{Z}_2(     (1-\rho^2)^{-1/2} \Psi^{-1/2} \widetilde{\delta} \widetilde{\mathcal{C}}) \right)^2 \\
		+ (1-a_1 - a_2)\mathcal{Z}_2^2(     (1-\rho^2)^{-1/2} \Psi^{-1/2} \widetilde{\delta} \widetilde{\mathcal{C}}) \geq \mathbb{C}_{\alpha}(a_1,a_2;\rho) 
	\end{Bmatrix} \bigg]  \leq \eps.
	\label{eq:ED3}
\end{align}
Combining \eqref{eq:PD3} and \eqref{eq:ED3}, we have, on $\mathcal{E}_n(\eps)$, 
\begin{align}
	& \sup_{(a_1,a_2) \in \mathbb{A}(f_{s}(\widehat{D},\widehat{\gamma}(\beta_0)),\widehat{\gamma}(\beta_0)), \widetilde{\delta} \in \widetilde{\mathcal{D}}_{n,3}(\eps)} \biggl|Q_n(a_1,a_2,\widetilde{\delta}) - Q(a_1,a_2,\widetilde{\delta})  \biggr| \leq 3\eps.     
	\label{eq:Q3}
\end{align}

Combining \eqref{eq:Q1}, \eqref{eq:Q2}, and \eqref{eq:Q3}, we have
\begin{align*}
	& \mathbb{P}\left( \sup_{(a_1,a_2) \in \mathbb{A}(f_{s}(\widehat{D},\widehat{\gamma}(\beta_0)),\widehat{\gamma}(\beta_0)), \widetilde{\delta} \in \widetilde{\mathcal{D}}_{n}}|Q_n(a_1,a_2,\widetilde{\delta}) - Q(a_1,a_2,\widetilde{\delta})| > 5\eps\right) \\
	& \leq \mathbb{P}\left( \sup_{(a_1,a_2) \in \mathbb{A}(f_{s}(\widehat{D},\widehat{\gamma}(\beta_0)),\widehat{\gamma}(\beta_0)), \widetilde{\delta} \in \widetilde{\mathcal{D}}_{n,1}(\eps)}|Q_n(a_1,a_2,\widetilde{\delta}) - Q(a_1,a_2,\widetilde{\delta})| > \eps,\mathcal{E}_n(\eps)\right) \\
	& + \mathbb{P}\left( \sup_{(a_1,a_2) \in \mathbb{A}(f_{s}(\widehat{D},\widehat{\gamma}(\beta_0)),\widehat{\gamma}(\beta_0)), \widetilde{\delta} \in \widetilde{\mathcal{D}}_{n,2}(\eps)}|Q_n(a_1,a_2,\widetilde{\delta}) - Q(a_1,a_2,\widetilde{\delta})| > \eps,\mathcal{E}_n(\eps)\right) \\
	& + \mathbb{P}\left( \sup_{(a_1,a_2) \in \mathbb{A}(f_{s}(\widehat{D},\widehat{\gamma}(\beta_0)),\widehat{\gamma}(\beta_0)), \widetilde{\delta} \in \widetilde{\mathcal{D}}_{n,3}(\eps)}|Q_n(a_1,a_2,\widetilde{\delta}) - Q(a_1,a_2,\widetilde{\delta})| > 3\eps,\mathcal{E}_n(\eps)\right) + \mathbb{P}\left(\mathcal{E}_n^c(\eps)\right) \\ 
	& \leq o(1) + \eps.
\end{align*}
Since $\eps$ is arbitrary, we have 
\begin{align*}
	\omega_n \equiv \sup_{(a_1,a_2) \in \mathbb{A}(f_{s}(\widehat{D},\widehat{\gamma}(\beta_0)),\widehat{\gamma}(\beta_0)), \widetilde{\delta} \in \widetilde{\mathcal{D}}_{n}}|Q_n(a_1,a_2,\widetilde{\delta}) - Q(a_1,a_2,\widetilde{\delta})| \convP 0.     
\end{align*}

Then we have
\begin{align*}
	0 & \leq  \sup_{\widetilde{\delta} \in \widetilde{\mathcal{D}}_{n}}   Q_n(\underline{a}(f_{s}(\widehat{D},\widehat{\gamma}(\beta_0)),\widehat{\gamma}(\beta_0)),0,\widetilde{\delta})-\sup_{\widetilde{\delta} \in \widetilde{\mathcal{D}}_{n}}Q_n(\mathcal{A}_s(\widehat{D},\widehat{\gamma}(\beta_0)),\widetilde{\delta}) \\
	& \leq   \sup_{\widetilde{\delta} \in \widetilde{\mathcal{D}}_{n}}   Q(\underline{a}(f_{s}(\widehat{D},\widehat{\gamma}(\beta_0)),\widehat{\gamma}(\beta_0)),0,\widetilde{\delta})-\sup_{\widetilde{\delta} \in \widetilde{\mathcal{D}}_{n}}Q(\mathcal{A}_s(\widehat{D},\widehat{\gamma}(\beta_0)),\widetilde{\delta}) + 2\omega_n \\
	& = o_p(1) - \sup_{\widetilde{\delta} \in \widetilde{\mathcal{D}}_{n}}Q(\mathcal{A}_s(\widehat{D},\widehat{\gamma}(\beta_0)),\widetilde{\delta}) + 2\omega_n,
\end{align*}
where the equality holds because (1) $\sup_{\widetilde{\delta} \in \Re}   Q(a_1,0,\widetilde{\delta})$ is continuous at $a_1 = 0$ as shown in the proof of I.\citet[Theorem 5]{Andrews(2016)}, (2) $\underline{a}(f_{s}(\widehat{D},\widehat{\gamma}(\beta_0)),\widehat{\gamma}(\beta_0)) = o_p(1)$ under strong identification, and (3) $\sup_{\widetilde{\delta} \in \Re}   Q(0,0,\widetilde{\delta}) = 0$ by construction. 

Furthermore, we have  
\begin{align*}
	Q(a_1,a_2,\widetilde{\delta}) & =  \mathbb{E}1\{\mathcal{Z}_2^2(     (1-\rho^2)^{-1/2} \Psi^{-1/2} \widetilde{\delta} \widetilde{\mathcal{C}}) \geq \mathbb{C}_{\alpha}\} \notag \\
	& - \mathbb{E} 1\begin{Bmatrix}
		a_1 \mathcal{Z}_1^2 + a_2\left(\rho \mathcal{Z}_1 + (1-\rho^2)^{1/2}\mathcal{Z}_2(     (1-\rho^2)^{-1/2} \Psi^{-1/2} \widetilde{\delta} \widetilde{\mathcal{C}}) \right)^2 \\
		+ (1-a_1 - a_2)\mathcal{Z}_2^2(     (1-\rho^2)^{-1/2} \Psi^{-1/2} \widetilde{\delta} \widetilde{\mathcal{C}}) \geq \mathbb{C}_{\alpha}(a_1,a_2;\rho) 
	\end{Bmatrix} \\
	& =  \mathbb{E}1\{\mathcal{Z}_2^2(     (1-\rho^2)^{-1/2} \Psi^{-1/2} \widetilde{\delta} \widetilde{\mathcal{C}}) \geq \mathbb{C}_{\alpha}\} \notag \\
	& - \mathbb{E} 1\begin{Bmatrix}
		(a_1+a_2\rho^2) \mathcal{Z}_1^2 + a_2\rho (1-\rho^2)^{1/2}\mathcal{Z}_1\mathcal{Z}_2(     (1-\rho^2)^{-1/2} \Psi^{-1/2} \widetilde{\delta} \widetilde{\mathcal{C}})  \\
		+ (1-a_1 - a_2\rho^2)\mathcal{Z}_2^2(     (1-\rho^2)^{-1/2} \Psi^{-1/2} \widetilde{\delta} \widetilde{\mathcal{C}}) \geq \mathbb{C}_{\alpha}(a_1,a_2;\rho) 
	\end{Bmatrix} 
\end{align*}
Note that $a_1 = 0$ and $a_2\rho = 0$ if and only if $a_1 + a_2\rho^2 = 0$, given that $a_1$ and $a_2$ are nonnegative. Therefore, Theorem \ref{thm:admissible}(ii) implies, for any constant $C>0$, there exists a constant $c>0$ such that
$$ \inf_{(a_1,a_2) \in \mathbb{A}_0, a_1 + a_2\rho^2\geq C }\sup_{\widetilde{\delta} \in \widetilde{\mathcal{D}}_{n} }Q(a_1,a_2,\widetilde{\delta})\geq c>0.$$

Therefore, 
\begin{align*}
	\mathbb{P}\left(\mathcal{A}_{1,s}(\widehat{D},\widehat{\gamma}(\beta_0)) + \mathcal{A}_{2,s}(\widehat{D},\widehat{\gamma}(\beta_0))\rho^2  \geq C >0 \right) \leq \mathbb{P}\left( c \leq o_p(1) + 2\omega_n \right) \rightarrow 0.
\end{align*}
This implies $\mathcal{A}_{1,s}(\widehat{D},\widehat{\gamma}(\beta_0)) \convP 0$ and $\mathcal{A}_{2,s}(\widehat{D},\widehat{\gamma}(\beta_0))\rho \convP 0$. 

To see the optimality result, note that
\begin{align*}
	(\widehat{\phi}_{\mathcal{A}_s(\widehat{D},\widehat{\gamma}(\beta_0))},\phi(AR(\beta_0),LM(\beta_0))) \convD (1\{\N_2^{*2} \geq \mathbb{C}_{\alpha}\},\phi(\N_1,\N_2)),    
\end{align*}
where $(\N_1,\N_2)$ is defined above Theorem \ref{thm:strongid} and $\N_2^* = (1-\rho^2)^{-1/2}(\N_2 - \rho \N_1)$. Then, the result holds by Theorem \ref{thm:admissible}(ii).  

\section{Proof of Theorem \ref{thm:uniform_size}}
\label{sec:uniform_size_pf}

We prove the result that $ \limsup_{n \rightarrow \infty} \sup_{\lambda_n \in \Lambda_n} \mathbb{E}_{\lambda}(\widehat{\phi}_{\mathcal{A}_s(\widehat{D},\widehat{\gamma}(\beta_0))}) = \alpha$. The other one can be proved in the same manner. Throughout the proof, we are under the null, i.e., $\beta_0 = \beta$. We start by proving the result for the full sequence $\{ n\}$, rather than a subsequence $\{n_k\}$ of $\{ n \}$. Then, we note that the same proof goes through with $n_k$ in place of $n$. 

We consider two cases: sequences $\lambda_n$ for which 
$\mathcal{C}_n$ converges to a constant and those for which it diverges to infinity. 
First, let us consider the case where $\mathcal{C}_n \rightarrow \widetilde{\mathcal{C}}$ for some fixed constant $\widetilde{\mathcal{C}} \in \Re$. For this case, it is established in Theorem \ref{thm:weakid} that under $\beta_0 = \beta$, 
\begin{align*}
	(AR^2(\beta_0), LM^{*2}(\beta_0),\mathcal{A}_s(\widehat{D},\widehat{\gamma}(\beta_0))) \convD (\mathcal{Z}^2_1, \mathcal{Z}^2_2, \mathcal{A}_s(D,\gamma)),  
\end{align*}
where the two normal random variables are independent from each other and independent of $D$, 
and furthermore (by letting $h(\cdot)$ in Theorem \ref{thm:weakid} be an identity function),
\begin{align*}
	\lim_{n \rightarrow \infty }\mathbb{E}_{\lambda_{n}}(\widehat{\phi}_{\mathcal{A}_s(\widehat{D},\widehat{\gamma}(\beta_0))}) = \alpha.
\end{align*}
Second, let us consider the case where $\mathcal{C}_n$ diverges to infinity.
Then, by Theorem \ref{thm:strongid}, we have
\begin{align*}
	\lim_{n \rightarrow \infty }\mathbb{E}_{\lambda_{n}}(\widehat{\phi}_{\mathcal{A}_s(\widehat{D},\widehat{\gamma}(\beta_0))}) = \mathbb{P}(\mathcal{Z}^2_2 \geq \mathbb{C}_\alpha) = \alpha.
\end{align*}
To complete the proof, we note that the above argument verifies Assumption B$^*$ in \cite{ACG2020}
and then we can establish the result by using Corollary 2.1 in their paper.

\section{Proof of Theorem \ref{thm:strong_fixed}}
We consider strong identification with fixed alternatives. By construction, we have $\mathcal{A}_{1,s}(\widehat{D},\widehat{\gamma}(\beta_0)) \geq  \frac{1.1 \mathbb{C}_{\alpha,\max}(\widehat{\rho}(\beta_0)) \widehat{\Phi}_1(\beta_0) \widehat{c}_{\mathcal{B}}(\beta_0)}{ \widehat{\Delta}_*^4(\beta_0) f_{s}^2(\widehat{D},\widehat{\gamma}(\beta_0))}$. By Theorem \ref{thm:admissible}(iii), it suffices to show that, w.p.a.1,  
\begin{align*}
	\frac{1.1 \mathbb{C}_{\alpha,\max}(\widehat{\rho}(\beta_0)) \widehat{\Phi}_1(\beta_0) \widehat{c}_{\mathcal{B}}(\beta_0)}{ \widehat{\Delta}_*^4(\beta_0) f_{s}^2(\widehat{D},\widehat{\gamma}(\beta_0))}   \geq \frac{\tilde{q}\Psi^{2}(\beta_0)\rho^{4}(\beta_0) }{\mathcal{C}^2\Phi_1(\beta_0)},
\end{align*}
or equivalently,
\begin{align}
	&\frac{1.1 \mathbb{C}_{\alpha,\max}(\widehat{\rho}(\beta_0)) \widehat{\Phi}_1(\beta_0) \widehat{c}_{\mathcal{B}}(\beta_0)}{ \widehat{\Delta}_*^4(\beta_0)d_n^2 f_{s}^2(\widehat{D},\widehat{\gamma}(\beta_0))}    \geq \frac{\tilde{q}\Psi^{2}(\beta_0)\rho^{4}(\beta_0) }{\widetilde{\mathcal{C}}^2\Phi_1(\beta_0)} = \frac{\tilde{q}\Phi_1(\beta_0) }{\widetilde{\mathcal{C}}^2\Delta_*^4(\beta_0)},
	\label{eq:a_fixed}    
\end{align}
for some constant $\tilde{q} > \mathbb{C}_{\alpha,max}(\rho(\beta_0))$. Under strong identification and fixed alternatives, we have
\begin{align*}
	d_n  \widehat{D} = d_n   \left( Q_{X,X} - (Q_{e(\beta_0),e(\beta_0)},Q_{X,e(\beta_0)}) \begin{pmatrix}
		\widehat{\Phi}_1(\beta_0) & \widehat{\Phi}_{12}(\beta_0) \\
		\widehat{\Phi}_{12}(\beta_0) & \widehat{\Psi}(\beta_0)
	\end{pmatrix}^{-1} \begin{pmatrix}
		\widehat{\Phi}_{13}(\beta_0)\\
		\widehat{\tau}(\beta_0) 
	\end{pmatrix} \right) \\
	\convP \left[1- (\Delta^2, \Delta) \left(\begin{pmatrix}
		\Phi_1(\beta_0) & \Phi_{12}(\beta_0) \\
		\Phi_{12}(\beta_0) & \Psi(\beta_0)
	\end{pmatrix}^{-1} \begin{pmatrix}
		\Phi_{13}(\beta_0) \\
		\tau(\beta_0)
	\end{pmatrix}\right) \right] \widetilde{\mathcal{C}}.
\end{align*}
Therefore, we have
\begin{align*}
	d_n  f_{s}(\widehat{D},\widehat{\gamma}(\beta_0)) = d_n   \widehat{D} + o_p(1) \convP  \left[1- (\Delta^2, \Delta) \left(\begin{pmatrix}
		\Phi_1(\beta_0) & \Phi_{12}(\beta_0) \\
		\Phi_{12}(\beta_0) & \Psi(\beta_0)
	\end{pmatrix}^{-1} \begin{pmatrix}
		\Phi_{13}(\beta_0) \\
		\tau(\beta_0)
	\end{pmatrix}\right) \right] \widetilde{\mathcal{C}}
\end{align*}
for $s \in \{pp,krs\}$. This means for any $\eps>0$, w.p.a.1, 
\begin{align*}
	d^2_n  f^2_{s}(\widehat{D},\widehat{\gamma}(\beta_0))  \leq (c_{\mathcal{B}}(\beta_0) +\eps)\widetilde{\mathcal{C}}^2.  
\end{align*} 
In addition, we have $\widehat{c}_{\mathcal{B}}(\beta_0) \convP c_{\mathcal{B}}(\beta_0)\geq 1$, $\widehat{\Delta}_*(\beta_0) \convP \Delta_*(\beta_0)$, $\mathbb{C}_{\alpha,\max}(\widehat{\rho}(\beta_0)) \convP \mathbb{C}_{\alpha,\max}(\rho(\beta_0))$, and $\widehat{\Phi}_1(\beta_0) \convP \Phi_1(\beta_0)>0$, which imply 
$\widehat{c}_{\mathcal{B}}(\beta_0) \geq c_{\mathcal{B}}(\beta_0) -c\eps,$ $\widehat{\Phi}_1(\beta_0) \geq \Phi_1(\beta_0) - c\eps$, $\mathbb{C}_{\alpha,\max}(\widehat{\rho}(\beta_0)) \geq  \mathbb{C}_{\alpha,\max}(\rho(\beta_0)) - c\eps$, and $\widehat{\Delta}_*^4(\beta_0) \leq \Delta_*^4(\beta_0) + c\eps$,  w.p.a.1. Therefore, we have, w.p.a.1,  
\begin{align*}
	\frac{1.1 \mathbb{C}_{\alpha,\max}(\widehat{\rho}(\beta_0)) \widehat{\Phi}_1(\beta_0) \widehat{c}_{\mathcal{B}}(\beta_0)}{ \widehat{\Delta}_*^4(\beta_0)d_n^2 f_{s}^2(\widehat{D},\widehat{\gamma}(\beta_0))} & \geq \frac{1.1 (\mathbb{C}_{\alpha,\max}(\rho(\beta_0))-c\eps) ( c_{\mathcal{B}}(\beta_0) -c\eps)(\Phi_1(\beta_0) - c\eps)  }{(\Delta_*^4(\beta_0) + c\eps)(c_{\mathcal{B}}(\beta_0) +\eps)\widetilde{\mathcal{C}}^2} \\
	& \geq \frac{(1.1-c\eps)\mathbb{C}_{\alpha,\max}(\rho(\beta_0))\Phi_1(\beta_0)}{\Delta_*^4(\beta_0)\widetilde{\mathcal{C}}^2},
\end{align*}
where the second inequality holds because $\eps$ can be arbitrarily small. Then, we can let $\tilde{q}$ in \eqref{eq:a_fixed} be $(1.1-c\eps)\mathbb{C}_{\alpha,\max}(\rho(\beta_0))$ which is greater than $\mathbb{C}_{\alpha,\max}(\rho(\beta_0))$. This concludes the proof.

\section{Proof of Theorem \ref{thm:W}}
We first extend our notation. For $a_i \in \Re^{d_1 \times 1}$ and $b_j \in \Re^{d_2 \times 1}$, we write $Q_{a,b}$ as $\sum_{i \in [n]}\sum_{j \neq i}a_i P_{ij}b_j^\top/\sqrt{K}$. Let $\hat \gamma_e = (W^\top W)^{-1}(W^\top \tilde e)$ and $\hat \gamma_V = (W^\top W)^{-1}(W^\top \tilde V)$. Then, we have
$e_i = \tilde e_i - W_i^\top \hat \gamma_e$, $V_i = \tilde V_i - W_i^\top \hat \gamma_V$, and $X_i = \Pi_i + V_i = \Pi_i + \tilde V_i - W_i^\top \hat \gamma_V$. By Lemma \ref{lem:W}, we have 
\begin{align*}
	Q_{e,e} = Q_{\tilde e- W \hat \gamma_e, \tilde e-W \hat \gamma_e} = Q_{\tilde e,\tilde e}-2Q_{\tilde e,W} \hat \gamma_e + \hat \gamma_e^\top Q_{W,W} \hat \gamma_e = Q_{\tilde e,\tilde e}+o_P(1).
\end{align*}
In addition, let $\overline{X} = \Pi+ \tilde V$. Then,  we have $X = \overline{X} - W \hat \gamma_V$ and
\begin{align*}
	Q_{X,e} & = Q_{\overline{X}-W \hat \gamma_V, \tilde e - W \hat \gamma_e}\\
	&= Q_{\overline{X},\tilde e} - Q_{\tilde e, W} \hat \gamma_V - Q_{\overline{X},W} \hat \gamma_e + \hat \gamma_V^\top Q_{W,W} \hat \gamma_e \\
	&=Q_{\overline{X},\tilde e} - Q_{\overline{X},W} \hat \gamma_e +o_P(1) \\
	&=Q_{\overline{X},\tilde e} - Q_{\Pi,W} \hat \gamma_e +o_P(1) \\
	&=Q_{\overline{X},\tilde e} + \sum_{i \in [n]} \Pi_i P_{ii}W_i^\top \hat \gamma_e/\sqrt{K} +o_P(1),
\end{align*}
where the last equality holds because 
\begin{align*}
	Q_{\Pi,W} = \sum_{i \in [n]} \Pi_i (\sum_{j \neq i}P_{ij}W_j^\top)/\sqrt{K} = - \sum_{i \in [n]} \Pi_i P_{ii}W_i^\top/\sqrt{K}.
\end{align*}
Denote $G_i = (\sum_{i \in [n]} \Pi_i P_{ii}W_i^\top)(\sum_{i \in [n]}W_iW_i^\top)^{-1}W_i$. Then, we have
\begin{align*}
	Q_{X,e} & = Q_{\tilde V,\tilde e} + Q_{\Pi,\tilde e} + \sum_{i \in [n]} G_i \tilde e_i/\sqrt{K} +o_P(1) \\
	& = \frac{\sum_{i \in [n]}\sum_{j\neq i} \tilde V_i P_{ij}\tilde e_j}{\sqrt{K}} + \sum_{i \in [n]} \frac{(G_i+\omega_i)}{\sqrt{K}} \tilde e_i +o_P(1),
\end{align*}
where $\omega_i = \sum_{j \neq i}P_{ij}\Pi_j$. 

Similarly, we have 
\begin{align*}
	Q_{X,X} & = Q_{\overline{X}-W\hat \gamma_V, \overline{X}-W\hat \gamma_V} \\
	& = Q_{\overline{X}, \overline{X}}- 2 Q_{\overline{X}, W}\hat \gamma_V +  \hat \gamma_V^\top Q_{W,W}\hat \gamma_V \\
	& = Q_{\Pi,\Pi} + 2 Q_{\Pi, \tilde V} + Q_{\tilde V, \tilde V}- 2 Q_{\Pi, W}\hat \gamma_V +o_P(1) \\
	& = Q_{\Pi,\Pi} + \frac{\sum_{i \in [n]}\sum_{j \neq i} \tilde V_i P_{ij} \tilde V_j}{\sqrt{K}} + 2\sum_{i \in [n]}\frac{\omega_i + G_i}{\sqrt{K}} \tilde V_i+o_P(1).
\end{align*}

Given $\{\tilde e_i,\tilde V_i\}_{i \in [n]}$ are independent, we can follow the same argument in the proof of \citet[Lemma 2]{Chao(2012)} and show the joint asymptotic normality of 
\begin{align*}
	\left(\frac{\sum_{i \in [n]}\sum_{j\neq i} \tilde e_i P_{ij}\tilde e_j}{\sqrt{K}},\frac{\sum_{i \in [n]}\sum_{j\neq i} \tilde V_i P_{ij}\tilde e_j}{\sqrt{K}},\frac{\sum_{i \in [n]}\sum_{j\neq i} \tilde V_i P_{ij}\tilde V_j}{\sqrt{K}}, \sum_{i \in [n]} \frac{(G_i+\omega_i)}{\sqrt{K}} \tilde e_i,\sum_{i \in [n]} \frac{(G_i+\omega_i)}{\sqrt{K}} \tilde V_i\right).
\end{align*}
In particular, we see that 
\begin{align*}
	Var\left(\sum_{i \in [n]} \frac{(G_i + \omega_i)\tilde e_i}{\sqrt{K}}\right) & = \sum_{i \in [n]} \frac{(G_i + \omega_i)^2 \tilde \sigma_i^2} {K} \\
	& \leq C \sum_{i \in [n]} \frac{(G_i + \omega_i)^2} {K} \\
	& \leq C \left[\frac{(\sum_{i \in [n]} \Pi_i P_{ii}W_i^\top)(\sum_{i \in [n]}W_iW_i^\top)^{-1} (\sum_{i \in [n]} \Pi_i P_{ii}W_i)}{K} + \frac{\Pi^\top \Pi}{K} \right] \\
	& \leq C \left[p_n^2 \frac{\Pi^\top \Pi}{K} + \frac{\Pi^\top \Pi}{K} \right] = O(1) 
\end{align*}
and the same result for $    Var(\sum_{i \in [n]} \frac{(G_i + \omega_i)\tilde V_i}{\sqrt{K}})$. 
This implies the joint asymptotic normality of 
\begin{align*}
	(Q_{e,e},Q_{X,e},Q_{X,X} - Q_{\Pi,\Pi}),
\end{align*}
and thus, verifying Assumption \ref{ass:weak_convergence}.

To see the second result in Theorem \ref{thm:W}, we note that 
\begin{align*}
	\mathbb{E}\left( \sum_{i \in [n]}G_i \tilde e_i/\sqrt{K}\right)^2 & \leq C \sum_{i \in [n]}G_i^2/K \\
	& = C (\sum_{i \in [n]} \Pi_i P_{ii}W_i^\top)(\sum_{i \in [n]}W_iW_i^\top)^{-1}(\sum_{i \in [n]} \Pi_i P_{ii}W_i)/K \\
	& \leq C \sum_{i \in [n]} \Pi_i^2 P_{ii}^2/K \\
	& \leq C \Pi^\top \Pi p_n^2/K.
\end{align*}

If $\Pi^\top \Pi p_n^2/K=o(1)$, then we have $\sum_{i \in [n]}G_i \tilde e_i/\sqrt{K} = o_P(1)$. Similarly, we can show that, if $\Pi^\top \Pi p_n^2/K=o(1)$, $\sum_{i \in [n]}G_i \tilde V_i/\sqrt{K} = o_P(1)$. These imply $Q_{\overline{X},W}\hat \gamma_e = o_P(1)$ and $Q_{\overline{X},W}\hat \gamma_V = o_P(1)$, which further imply that 
\begin{align*}
	Q_{X,e} = Q_{\overline{X},\tilde e} +o_P(1) \quad \text{and}    \quad Q_{X,X} = Q_{\overline{X},\overline{X}} +o_P(1).
\end{align*}

\section{Proof of Theorem \ref{thm:var1}}
\label{sec:var1_pf}
We focus on the consistency of $\widehat{\Phi}_1(\beta_0)$ and $\widehat{\Psi}(\beta_0)$. The consistency of the rest four estimators can be established in the same manner. We have $e_i(\beta_0) = e_i + \Delta X_i = V_i(\Delta) + \Delta \Pi_i$, where $V_i(\Delta) = e_i + \Delta V_i.$ Therefore, 
\begin{align*}
	\widehat{\Phi}_1(\beta_0) & =   \frac{2}{K}\sum_{i \in [n]}\sum_{j \neq i}P_{ij}^2 e_i^2(\beta_0)e_j^2(\beta_0) \\
	& =  \frac{2}{K}\sum_{i \in [n]}\sum_{j \neq i}P_{ij}^2 (\Delta^2 \Pi_i^2 + 2\Delta \Pi_i V_i(\Delta) + V_i^2(\Delta))(\Delta^2 \Pi_j^2 + 2\Delta \Pi_j U_j(\Delta) + U_j^2(\Delta)) \\
	& = \frac{2}{K}\sum_{i \in [n]}\sum_{j \neq i}P_{ij}^2 V_i^2(\Delta)U_j^2(\Delta) + \Delta \frac{4}{K}\sum_{i \in [n]}\sum_{j \neq i}P_{ij}^2(\Pi_i V_i(\Delta) U_j^2(\Delta) + \Pi_j U_j(\Delta) V_i^2(\Delta)) \\
	& + \Delta^2 \frac{2}{K}\sum_{i \in [n]}\sum_{j \neq i}P_{ij}^2(\Pi_i^2 U_j^2(\Delta) + \Pi_j^2 V_i^2(\Delta) + 4\Pi_i\Pi_j V_i(\Delta)U_j(\Delta)) \\
	& + \Delta^3\frac{4}{K}\sum_{i \in [n]}\sum_{j \neq i}P_{ij}^2(\Pi_i^2 \Pi_j U_j(\Delta) + \Pi_j^2 \Pi_i V_i(\Delta)) + \Delta^4\frac{2}{K}\sum_{i \in [n]}\sum_{j \neq i}P_{ij}^2\Pi_i^2 \Pi_j^2 \\
	& \equiv \sum_{l=0}^4\Delta^l T_l.
\end{align*}
We first note that $\frac{1}{K}\sum_{i \in [n]}\omega_i^2 \sigma_i^2 = o(1)$, $\frac{1}{K}\sum_{i \in [n]}\omega_i^2 \gamma_i = o(1)$, and $\frac{1}{K}\sum_{i \in [n]}\omega_i^2 \eta_i^2 = o(1)$.

To see this, note that
\begin{align*}
	\frac{1}{K}\sum_{i \in [n]}\omega_i^2 \sigma_i^2 & \leq \frac{C}{K}\sum_{i \in [n]}\omega_i^2 = \frac{C}{K} \sum_{i \in [n]} (P_i\Pi - P_{ii}\Pi_i)^2
	\\
	& \leq \frac{C}{K} (2\Pi^\top P^2 \Pi + 2\sum_{i \in [n]}P_{ii}^2 \Pi_i^2) 
	\leq C \frac{\Pi^\top \Pi}{K} = o(1),
\end{align*}
where the second and third inequalities are shown in the Proof of \citet[Lemma S1.4]{MS22}. 
The results for $\frac{1}{K}\sum_{i \in [n]}\omega_i^2 \gamma_i = o(1)$ and $\frac{1}{K}\sum_{i \in [n]}\omega_i^2 \eta_i^2 = o(1)$ can be established in the same manner. 

We first consider $T_0$. Denote $\xi_{ij} = V_i^2(\Delta)U_j^2(\Delta) - \mathbb{E}V_i^2(\Delta)U_j^2(\Delta)$. We want to show that 
$$\frac{1}{K}\sum_{i \in [n]}\sum_{j \neq i}P_{ij}^2 \xi_{ij} = o_p(1).$$ 
Note that 
\begin{align*}
	\mathbb{E}\left[\frac{1}{K}\sum_{i \in [n]}\sum_{j \neq i}P_{ij}^2 \xi_{ij}\right]^2 = 
	\frac{1}{K^2}\sum_{i \in [n]}\sum_{j \neq i}P_{ij}^4 \mathbb{E}\xi_{ij}^2 + \frac{4}{K^2}\sum_{i \in [n]}\sum_{j \neq i}\sum_{i' \neq i,j}P_{ij}^2P_{ii'}^2 \mathbb{E}\xi_{ij}\xi_{ii'}. 
\end{align*}
As both $\mathbb{E}\xi_{ij}^2$ and $|\mathbb{E}\xi_{ij}\xi_{ii'}|$ are bounded, we have 
\begin{align*}
	\frac{1}{K^2}\sum_{i \in [n]}\sum_{j \neq i}P_{ij}^4 \mathbb{E}\xi_{ij}^2 \leq \frac{C}{K^2}    \sum_{i \in [n]}\sum_{j \neq i} P_{ij}^2 \leq \frac{C}{K} = o(1)
\end{align*}
and 
\begin{align*}
	\left\vert \frac{1}{K^2}\sum_{i \in [n]}\sum_{j \neq i}\sum_{i' \neq i,j}P_{ij}^2P_{ii'}^2 \mathbb{E}\xi_{ij}\xi_{ii'}    \right\vert \leq \frac{C}{K^2}\sum_{i \in [n]}\sum_{j \neq i}\sum_{i' \neq i,j}P_{ij}^2P_{ii'}^2 \leq \frac{C}{K^2}\sum_{i \in [n]}\sum_{j \neq i}P_{ij}^2P_{ii} = o(1).
\end{align*}
Therefore, we have
\begin{align*}
	T_0 & =  \frac{2}{K}\sum_{i \in [n]}\sum_{j \neq i}P_{ij}^2 \mathbb{E}(V_i^2(\Delta)U_j^2(\Delta)) + o_p(1) \\
	& = \Delta^4 \frac{2}{K}\sum_{i \in [n]}\sum_{j \neq i}P_{ij}^2\eta_i^2 \eta_j^2 + \Delta^3\frac{4}{K}\sum_{i \in [n]}\sum_{j \neq i}P_{ij}^2(\eta_i^2 \gamma_j + \eta_j^2 \gamma_i) + \Delta^2\frac{2}{K}\sum_{i \in [n]}\sum_{j \neq i}P_{ij}^2(\eta_i^2 \sigma_j^2 + \eta_j^2 \sigma_i^2 + 4\gamma_i \gamma_j) \\
	&+ \Delta\frac{4}{K}\sum_{i \in [n]}\sum_{j \neq i}P_{ij}^2 (\gamma_i\sigma_j^2 + \gamma_j\sigma_i^2) + \frac{2}{K}\sum_{i \in [n]}\sum_{j \neq i}P_{ij}^2 \sigma_i^2 \sigma_j^2 + o_p(1) \\
	& = \Phi_1(\beta_0) + o_p(1).
\end{align*}

By the same argument above, we have 
\begin{align*}
	T_1 = \mathbb{E}T_1 + o_p(1) =o_p(1)    
\end{align*}
because $\mathbb{E}T_1 = 0$. Similarly, we have $\mathbb{E}T_3 = 0$ and $T_3 = o_p(1)$. Next, we have
\begin{align*}
	T_2 = \mathbb{E}T_2 + o_P(1) \leq \frac{C}{K}\sum_{i \in [n]}\sum_{j \neq i} P_{ij}^2\Pi_i^2 +o_p(1)\leq \frac{C p_n\Pi^\top \Pi}{K} +o_p(1) = o_p(1).
\end{align*}
Last, we have 
\begin{align*}
	T_4 \leq \frac{C}{K}\sum_{i \in [n]}\sum_{j \neq i} P_{ij}^2\Pi_i^2 = o(1),
\end{align*}
where the first inequality is by $\max_{i\in [n]}|\Pi_i|<C$. This implies 
\begin{align*}
	\widehat{\Phi}_1(\beta_0) - \Phi_1(\beta_0) = o_p(1). 
\end{align*}

Next, we consider the consistency of $\widehat{\Psi}(\beta_0)$. By the similar argument above, we have 
\begin{align}
	& \frac{1}{K} \sum_{i \in [n]}\sum_{j \neq i}P_{ij}^2 X_ie_i(\beta_0) X_je_j(\beta_0)) \notag \\
	& = \frac{1}{K} \sum_{i \in [n]}\sum_{j \neq i}P_{ij}^2 \Pi_ie_i(\beta_0) \Pi_je_j(\beta_0)) + \frac{1}{K} \sum_{i \in [n]}\sum_{j \neq i}P_{ij}^2 \Pi_ie_i(\beta_0) V_je_j(\beta_0)) \notag \\
	& + \frac{1}{K} \sum_{i \in [n]}\sum_{j \neq i}P_{ij}^2 V_ie_i(\beta_0) \Pi_je_j(\beta_0)) + \frac{1}{K} \sum_{i \in [n]}\sum_{j \neq i}P_{ij}^2 V_ie_i(\beta_0) V_je_j(\beta_0)) \notag \\
	& = \frac{1}{K} \sum_{i \in [n]}\sum_{j \neq i}P_{ij}^2 (\gamma_i+\Delta \eta_i^2)(\gamma_j+\Delta \eta_j^2) + o_p(1).
	\label{eq:psihat_1}
\end{align}

In addition, we have 
\begin{align}
	& \frac{1}{K}\sum_{i \in [n]}(\sum_{j \neq i}P_{ij}X_j)^2 e_i^2(\beta_0) \notag \\
	& = \frac{1}{K}\sum_{i \in [n]}(\omega_i + \sum_{j \neq i}P_{ij}V_j)^2 e_i^2(\beta_0) \notag     \\
	& = \frac{1}{K}\sum_{i \in [n]} \omega_i^2 \mathbb{E}e_i^2(\beta_0) +  \frac{1}{K}\sum_{i \in [n]}\sum_{j \neq i}P_{ij}^2 \eta_j^2 \mathbb{E}e_i^2(\beta_0) + o_p(1) \notag \\
	& = \frac{1}{K}\sum_{i \in [n]}\sum_{j \neq i}P_{ij}^2 \eta_j^2 (\sigma_i^2 + 2 \gamma_i\Delta + \Delta^2 \eta_i^2) + o_p(1),
	\label{eq:psihat_2}
\end{align}
where the second equality is due to \citet[proof of statement (a) in Lemma S3.2]{MS22}, and the third equality is due to $\frac{1}{K}\sum_{i \in [n]}\omega_i^2\sigma_i^2=o(1)$. In the next section, we show the same results hold under Assumption \ref{ass:reg1}. Combining \eqref{eq:psihat_1} and \eqref{eq:psihat_2}, we have 
\begin{align*}
	\widehat{\Psi}(\beta_0) & =  \frac{1}{K} \sum_{i \in [n]}\sum_{j \neq i}P_{ij}^2 (\gamma_i+\Delta \eta_i^2)(\gamma_j+\Delta \eta_j^2) + \frac{1}{K}\sum_{i \in [n]}\sum_{j \neq i}P_{ij}^2 \eta_j^2 (\sigma_i^2 + 2 \gamma_i\Delta + \Delta^2 \eta_i^2) + o_p(1)   \\
	& =  \frac{1}{K} \sum_{i \in [n]}\sum_{j \neq i}P_{ij}^2(\gamma_i \gamma_j + \sigma_i^2 \eta_j^2) + \frac{4\Delta}{K} \sum_{i \in [n]}\sum_{j \neq i}P_{ij}^2 \eta_i^2 \gamma_j + \frac{2\Delta^2}{K} \sum_{i \in [n]}\sum_{j \neq i}P_{ij}^2 \eta_i^2\eta_j^2 + o_p(1) \\
	& = \Psi(\beta_0) + o_p(1). 
\end{align*}

\section{Proof of Theorem \ref{thm:var2}}
\label{sec:var2_pf}
Given Lemma \ref{lem:SA}, Lemmas 2 and 3 in \cite{MS22} hold under Assumptions \ref{ass:K} and \ref{ass:reg2}. Therefore, \citet[Theorem 3]{MS22} shows that 
\begin{align*}
	\widehat{\Phi}_1(\beta_0) - \frac{2}{K}\sum_{i \in [n]}\sum_{j \neq i}P_{ij}^2 \mathbb{E}V_i^2(\Delta) \mathbb{E}U_j^2(\Delta) = o_p(1).  
\end{align*}
In addition, the proof of Theorem \ref{thm:var1} shows that 
\begin{align*}
	\frac{2}{K}\sum_{i \in [n]}\sum_{j \neq i}P_{ij}^2 \mathbb{E}V_i^2(\Delta) \mathbb{E}U_j^2(\Delta) = \Phi_1(\beta_0) + o(1),    
\end{align*}
which implies the consistency of $\widehat{\Phi}_1(\beta_0)$.

Similarly, given Lemma \ref{lem:SA}, Lemma S3.1 in \cite{MS22} holds under Assumptions \ref{ass:K} and \ref{ass:reg2}, so that the consistency of $\widehat{\Upsilon}$ to $\Upsilon$ is also shown by using their argument. 
In addition, we use the same argument in the proof of \citet[Theorem 5]{MS22} to show that 
\begin{align*}
	\widehat{\Psi}(\beta_0) & =  \left\{ \frac{1}{K}\sum_{i\in [n]}(\sum_{j \neq i}P_{ij}X_j)^2 \frac{e_i M_i e}{M_{ii}} + \frac{1}{K}\sum_{i \in [n]}\sum_{j \neq i}\widetilde{P}_{ij}^2M_i X e_i M_jXe_j \right\}   \\
	& + \Delta \left\{ \frac{1}{K}\sum_{i\in [n]}(\sum_{j \neq i}P_{ij}X_j)^2 \left(\frac{e_i M_i X}{M_{ii}} + \frac{X_i M_i e}{M_{ii}}\right) + \frac{2}{K}\sum_{i \in [n]}\sum_{j \neq i}\widetilde{P}_{ij}^2M_i X e_i M_jX X_j \right\} \\
	& + \Delta^2  \left\{ \frac{1}{K}\sum_{i\in [n]}(\sum_{j \neq i}P_{ij}X_j)^2 \frac{X_i M_i X}{M_{ii}}  + \frac{1}{K}\sum_{i \in [n]}\sum_{j \neq i}\widetilde{P}_{ij}^2M_i X X_i M_jX X_j \right\} \\
	& = \Psi + 2\Delta \tau + \Delta^2 \Upsilon + o_p(1) = \Psi(\beta_0) + o_p(1),
\end{align*}
where the second equality also follows from Lemma S3.1 in \cite{MS22}.

Next for $\widehat{\Phi}_{12}(\beta_0)$, we have 
\begin{align*}
	& \frac{1}{K}\sum_{i \in [n]}\sum_{j \neq i}\widetilde{P}_{ij}^2 M_jX e_j(\beta_0) e_i(\beta_0)M_i e(\beta_0) \\
	& = \frac{1}{K}\sum_{i \in [n]}\sum_{j \neq i}\widetilde{P}_{ij}^2 M_jX e_j e_iM_i e \\
	& + \Delta  \frac{1}{K}\sum_{i \in [n]}\sum_{j \neq i}\widetilde{P}_{ij}^2\left(M_j X X_j e_i M_ie + M_jX e_j X_i M_i e + M_jX e_j e_i M_i X \right) \\
	& + \Delta^2 \frac{1}{K}\sum_{i \in [n]}\sum_{j \neq i}\widetilde{P}_{ij}^2\left(M_jX X_j X_i M_i e + M_jX X_j e_iM_i X + M_jX e_j X_i M_iX\right) \\
	& + \Delta^3 \frac{1}{K}\sum_{i \in [n]}\sum_{j \neq i}\widetilde{P}_{ij}^2 M_jX X_j X_i M_iX. 
\end{align*}

Note that $ \frac{1}{K}\sum_{i \in [n]}\sum_{j \neq i}\widetilde{P}_{ij}^2 M_jX e_j e_iM_i e =     \frac{1}{K}\sum_{i \in [n]}\sum_{j \neq i}\widetilde{P}_{ij}^2 (M_jV + \lambda_i) e_j e_i M_i e$, where 
$\lambda_i = M_i \Pi$. Then, by Lemma \ref{lem:SA} and Lemma 3 of \cite{MS22}, 
\begin{align*}
	\frac{1}{K}\sum_{i \in [n]}\sum_{j \neq i}\widetilde{P}_{ij}^2 M_jX e_j e_iM_i e 
	-     \frac{1}{K}\sum_{i \in [n]}\sum_{j \neq i}\widetilde{P}_{ij}^2 M_jV e_j e_i M_i e = o_p(1).  
\end{align*}
Furthermore, by Lemma \ref{lem:SA} and Lemma 2 of \cite{MS22}, 
\begin{align*}
	\frac{1}{K}\sum_{i \in [n]}\sum_{j \neq i}\widetilde{P}_{ij}^2 M_jV e_j e_i M_i e -  \frac{1}{K}\sum_{i \in [n]}\sum_{j \neq i} P_{ij}^2 \gamma_j \sigma_i^2  = o_p(1). 
\end{align*}
By using similar arguments, we find that 
\begin{align*}
	& \frac{1}{K}\sum_{i \in [n]}\sum_{j \neq i}\widetilde{P}_{ij}^2M_j X X_j e_i M_ie = \frac{1}{K}\sum_{i \in [n]}\sum_{j \neq i} P_{ij}^2 \eta_j^2 \sigma_i^2 + o_p(1), \\
	& \frac{1}{K}\sum_{i \in [n]}\sum_{j \neq i}\widetilde{P}_{ij}^2M_jX e_j X_i M_i e = 
	\frac{1}{K}\sum_{i \in [n]}\sum_{j \neq i} P_{ij}^2 \gamma_j \gamma_i + o_p(1), \\
	& \frac{1}{K}\sum_{i \in [n]}\sum_{j \neq i}\widetilde{P}_{ij}^2M_jX e_j e_i M_i X = 
	\frac{1}{K}\sum_{i \in [n]}\sum_{j \neq i} P_{ij}^2 \gamma_j \gamma_i + o_p(1), \\
	& \frac{1}{K}\sum_{i \in [n]}\sum_{j \neq i}\widetilde{P}_{ij}^2M_jX X_j X_i M_i e  = 
	\frac{1}{K}\sum_{i \in [n]}\sum_{j \neq i} P_{ij}^2 \eta_j^2 \gamma_i + o_p(1),\\
	& \frac{1}{K}\sum_{i \in [n]}\sum_{j \neq i}\widetilde{P}_{ij}^2 M_jX X_j e_iM_i X =  
	\frac{1}{K}\sum_{i \in [n]}\sum_{j \neq i} P_{ij}^2 \eta_j^2 \gamma_i + o_p(1),\\
	& \frac{1}{K}\sum_{i \in [n]}\sum_{j \neq i}\widetilde{P}_{ij}^2 M_jX e_j X_i M_iX =  
	\frac{1}{K}\sum_{i \in [n]}\sum_{j \neq i} P_{ij}^2 \gamma_j \eta_i^2 + o_p(1),\\
	& \frac{1}{K}\sum_{i \in [n]}\sum_{j \neq i}\widetilde{P}_{ij}^2  M_jX X_j X_i M_iX  = 
	\frac{1}{K}\sum_{i \in [n]}\sum_{j \neq i} P_{ij}^2 \eta_j^2 \eta_i^2 + o_p(1). 
\end{align*}
Putting these results together, we obtain 
\begin{align*}
	\widehat{\Phi}_{12}(\beta_0) = \Phi_{12} + \Delta(2\Psi + \Phi_{13}) + 3 \Delta^2 \tau + \Delta^3 \Upsilon + o_p(1) = \Phi_{12}(\beta_0) + o_p(1). 
\end{align*}

We use similar arguments to prove the results for $\widehat{\Psi}_{13}(\beta_0)$ and $\widehat{\tau}(\beta_0)$. 
For $\widehat{\Phi}_{13}(\beta_0)$, notice that 
\begin{align*}
	& \frac{1}{K}\sum_{i \in [n]}\sum_{j \neq i}\widetilde{P}_{ij}^2 M_i X e_i(\beta_0) M_j Xe_j(\beta_0) \\
	& = \frac{1}{K}\sum_{i \in [n]}\sum_{j \neq i}\widetilde{P}_{ij}^2 M_i X e_i M_j X e_j \\
	& + \Delta \frac{1}{K}\sum_{i \in [n]}\sum_{j \neq i}\widetilde{P}_{ij}^2 
	(M_i X e_i M_j X X_j + M_i X X_i M_j X e_j) \\
	& + \Delta^2 \frac{1}{K}\sum_{i \in [n]}\sum_{j \neq i}\widetilde{P}_{ij}^2 M_i X X_i M_j X X_j \\
	& = \frac{1}{K}\sum_{i \in [n]}\sum_{j \neq i}P_{ij}^2 \gamma_i \gamma_j 
	+ \Delta \frac{1}{K}\sum_{i \in [n]}\sum_{j \neq i}P_{ij}^2 ( \gamma_i \eta_j^2 + \eta_i^2 \gamma_j) + \Delta^2 \frac{1}{K}\sum_{i \in [n]}\sum_{j \neq i}P_{ij}^2 \eta_i^2 \eta_j^2 + o_p(1),  
\end{align*}
which implies that 
\begin{align*}
	\widehat{\Phi}_{13}(\beta_0) = \Phi_{13} + 2 \Delta \tau + \Delta^2 \Upsilon + o_p(1) = 
	\Phi_{13}(\beta_0) + o_p(1). 
\end{align*}

Finally, for $\widehat{\tau}(\beta_0)$, notice that 
\begin{align*}
	& \frac{1}{K}\sum_{i \in [n]}\sum_{j \neq i}\widetilde{P}_{ij}^2 X_i M_i X M_j Xe_j(\beta_0) = \frac{1}{K}\sum_{i \in [n]}\sum_{j \neq i} P_{ij}^2 \eta_i^2 \gamma_j + \frac{1}{K}\sum_{i \in [n]}\sum_{j \neq i} P_{ij}^2 \eta_i^2 \eta_j^2 \Delta + o_p(1), \\
	& \frac{1}{K}\sum_{i \in [n]} ( \sum_{j \neq i} P_{ij}X_j)^2\left(\frac{e_i(\beta_0)M_i X}{2M_{ii}}+\frac{X_iM_i e(\beta_0)}{2M_{ii}}\right) = \frac{1}{K}\sum_{i \in [n]}\sum_{j \neq i} P_{ij}^2 \eta_i^2 \gamma_j + \frac{1}{K}\sum_{i \in [n]}\sum_{j \neq i} P_{ij}^2 \eta_i^2 \eta_j^2 \Delta + o_p(1),
\end{align*}
which implies that 
\begin{align*}
	\widehat{\tau}(\beta_0) = \tau + \Delta \Upsilon + o_p(1) = \tau(\beta_0) + o_p(1).
\end{align*}
This completes the proof of the theorem.

\section{Proof of Lemma \ref{lem:SA}}
\label{sec:lem_SA}
Let $p_n = \max_i P_{ii}$. We first give some useful bounds, which is similar to Lemma S1.4 in \cite{MS22}:
\begin{align*}
	& \sum_{i \in [n]}\omega_i^2 = \sum_{i \in [n]} (P_i\Pi - P_{ii}\Pi_i)^2 
	\leq 2 \Pi'P^2\Pi + 2 \sum_{i \in [n]}P_{ii}^2\Pi^2
	\leq C \Pi^\top \Pi, \\
	& \max_{i \in [n] }\omega_i^2  = \max_{i \in [n]}(\sum_{j \neq i}P_{ij}\Pi_j)^2\leq     \max_{i \in [n]}(\sum_{j \neq i}P_{ij}^2) \Pi^\top \Pi \leq p_n \Pi^\top \Pi,
\end{align*}
which imply 
\begin{align*}
	\sum_{i \in [n]}\omega_i^4 \leq \max_{i \in [n] }\omega_i^2 (\sum_{i \in [n] }\omega_i^2) \leq C p_n (\Pi^\top \Pi)^2. 
\end{align*}






First, we show that \citet[Lemma S2.1]{MS22} hold under our conditions following the lines of argument in their proof. More specifically, we notice that to show 
$\Delta^2|\mathbb{E}A_2| = o(1)$, where $A_2$ is defined in the proof of \citet[Lemma S2.1]{MS22}, it suffices to show the following terms are o(1): 
\begin{align*}
	& \frac{C\Delta^2}{K}\sum_{i \in [n]}\sum_{j \neq i}P_{ij}^2 |\lambda_i||\Pi_j| \leq \frac{C\Delta^2}{K}\left(\sum_{i \in [n]}P_{ii}\lambda_i^2\right)^{1/2}\left(\sum_{j \in [n]}P_{jj}\Pi_j^2\right)^{1/2} 
	\leq \frac{C\Delta^2}{K}p_n\left(\lambda^{\top}\lambda\right)^{1/2}\left(\Pi^{\top}\Pi\right)^{1/2} \\
	& \leq \frac{C\Delta^2}{K^{3/2}}p_n\left(\Pi^{\top}\Pi\right) = o(1) \; \text{by $\lambda^{\top}\lambda \leq C\frac{\Pi^{\top}\Pi}{K}$,}
	\\
	& \frac{C\Delta^2}{K}\sum_{i \in [n]}\sum_{j \neq i}P_{ij}^2|\Pi_i||\Pi_j| 
	\leq \frac{C\Delta^2}{K} \left(\sum_{i \in [n]}P_{ii}\Pi_i^2\right)^{1/2}\left(\sum_{j \in [n]}P_{jj}\Pi_j^2\right)^{1/2} \leq \frac{C\Delta^2}{K} p_n \left(\Pi^{\top}\Pi\right) =o(1).  
\end{align*}
Then, we prove the variance of $\Delta^2A_2 = o(1)$ by showing that 
\begin{align*}
	& \frac{C \Delta^4}{K^2}\sum_{i \in [n]}\sum_{j \in [n]}P_{ij}^4 \lambda_i^2 \lambda_j^2 
	\leq \frac{C\Delta^4}{K^2} p_n^2 \left( \lambda^{\top}\lambda\right)^2 \leq \frac{C\Delta^4}{K^2} p_n^2 \left( \frac{\Pi^{\top}\Pi}{K}\right)^2 = o(1) \;
	\text{by $P_{ij}^2 \leq P_{ii}$}, \\
	& \frac{C \Delta^4}{K^2} \left( \sum_{i \in [n]}\lambda_i^2 \left( \sum_{j \in [n]}P_{ij}^2\right)\Pi^{\top}\Pi + \lambda^{\top}\lambda \left( \sum_{j \in [n]}P_{jj}|\Pi_j|\right)^2\right) 
	\leq \frac{C \Delta^4}{K^2} \left( p_n (\lambda^{\top}\lambda)(\Pi^{\top}\Pi) + (\lambda^{\top}\lambda)(p_nK)(\Pi^{\top}\Pi) \right) 
	\\
	& \leq \frac{C \Delta^4 }{K^3} \left( p_n (\Pi^{\top}\Pi)^2 + p_n K (\Pi^{\top}\Pi)^2 \right) = o(1) \; \text{by $\sum_{j \in [n]}P_{jj}^2 \leq p_nK$}, \\
	& \frac{C\Delta^4}{K^2} \left( \sum_{i \in [n]}\sum_{j \in [n]}P_{ij}^2|\Pi_i\Pi_j|\right)^2
	\leq \frac{C\Delta^4}{K^2} \left(\sum_{i \in [n]}P_{ii}\Pi_i^2\right)\left(\sum_{j \in [n]}P_{jj}\Pi_j^2\right) \leq \frac{C\Delta^4}{K^2}p_n^2(\Pi^{\top}\Pi)^2 =o(1),
\end{align*}
and 
\begin{align*}
	&\frac{C \Delta^4}{K^2} \sum_{j \in [n]}\sum_{k \in [n]} \left( \sum_{i \in [n]} P_{ij}^2 | \lambda_i \Pi_iM_{jk}|\right)^2 =     \frac{C \Delta^4}{K^2} \sum_{j \in [n]}\sum_{k \neq j} \left( \sum_{i \in [n]} P_{ij}^2 | \lambda_i \Pi_iM_{jk}|\right)^2 +     \frac{C \Delta^4}{K^2} \sum_{j \in [n]} \left( \sum_{i \in [n]} P_{ij}^2 | \lambda_i \Pi_iM_{jj}|\right)^2 \\
	& \leq     \frac{C\Delta^4}{K^2} \sum_{j \in [n]}\sum_{k \neq j} M_{jk}^2 \left(\sum_{i \in [n]}P_{ii}\lambda_i^2\right)\left(\sum_{i\in[n]}P_{ii}\Pi_i^2\right) +    \frac{C \Delta^4}{K^2} \sum_{j \in [n]} \left( \sum_{i \in [n]} P_{ij}^2 | \lambda_i|\right)^2 \\
	& \leq  \frac{C\Delta^4}{K^2} K p_n^2 (\lambda^{\top}\lambda)(\Pi^{\top}\Pi) +  \frac{C \Delta^4}{K^2} \sum_{j \in [n]} \left( \sum_{i \in [n]} P_{ij}^4\right) \lambda^\top \lambda \\
	& \leq \frac{C\Delta^4 K p_n^2(\Pi^\top \Pi)^2}{K^2} + \frac{C\Delta^4 p_n K \Pi^\top \Pi}{K^2} = o(1)~
	\text{by $\sum_{j \in [n]}\sum_{k \neq j} M_{jk}^2 = \sum_{j \in [n]}\sum_{k \neq j} P_{jk}^2 \leq K$ and $P_{ij}^2 \leq P_{ii} \leq p_n$}.
\end{align*}

Second, we show that \citet[ Lemma S2.2]{MS22} holds under our conditions. 
Notice that $|\Delta \mathbb{E} A_1| = o(1)$ by 
\begin{align*}
	& \frac{C|\Delta|}{K} \sum_{i \in [n]}\sum_{j \neq i} P_{ij}^2 |\Pi_i| 
	\leq \frac{C|\Delta|}{K} \left( \sum_{i \in [n]} P_{ii}^2 \right)^{1/2} (\Pi^{\top}\Pi)^{1/2}
	\leq \frac{C|\Delta|}{K} (p_n K)^{1/2} (\Pi^{\top}\Pi)^{1/2} = o(1), 
\end{align*}
Then, we show that the variance of $\Delta A_1$  is $o(1)$ by showing the following terms are $o(1)$:
\begin{align*}
	& \frac{C\Delta^2}{K^2} \left( \sum_{i \in [n]} \left(\sum_{j \in [n]}P_{ij}^2\right)\lambda_i^2 + \sum_{i \in [n]}\sum_{j \in [n]}P_{ij}^2 |\lambda_i||\lambda_j| \right) \leq 
	\frac{C\Delta^2}{K^2} \left( p_n (\lambda^{\top}\lambda) + p_n (\lambda^{\top}\lambda) \right) = o(1),\\
	& \frac{C\Delta^2}{K^2} \left( \sum_{i \in [n]}\sum_{j \in [n]} P_{ij}^4 (\lambda_i^2 + |\lambda_i||\lambda_j|) + \sum_{i \in [n]}\sum_{j \in [n]}P_{ij}^2\lambda_i^2 \right) 
	\leq \frac{C\Delta^2}{K^2} \left( p_n^2(\lambda^{\top}\lambda) + p_n^2(\lambda^{\top}\lambda) + p_n(\lambda^{\top}\lambda)\right) = o(1),\\
	& \frac{C\Delta^2}{K^2} \left(\sum_{i \in [n]}\sum_{k \in [n]}P_{ik}^2|\lambda_i||\lambda_k| + \sum_{i \in [n]}\sum_{j \in [n]}P_{ij}^2|\lambda_i||\lambda_j| \right) 
	\leq \frac{C\Delta^2}{K^2} \left( p_n (\lambda^{\top}\lambda) + p_n (\lambda^{\top}\lambda) \right)=o(1), \\
	& \frac{C\Delta^2}{K^2} \sum_{j \in [n]} \left(\sum_{i \in [n]} P_{ij}^2|\lambda_i|\right)^2 \leq \frac{C\Delta^2}{K^2} \sum_{j \in [n]} \left( \sum_{i \in [n]}P_{ij}^4\right)(\lambda^{\top}\lambda)
	\leq \frac{C\Delta^2}{K^2} (p_n K) (\lambda^{\top}\lambda) =o(1), \\
	& \frac{C\Delta^2}{K^2}\sum_{j \in [n]}\left( \sum_{i \in [n]}P_{ij}^2|\Pi_i|\right)^2 \leq
	\frac{C\Delta^2}{K^2} \sum_{j \in [n]} \left( \sum_{i \in [n]}P_{ij}^4\right)(\Pi^{\top}\Pi) \leq \frac{C\Delta^2}{K^2} (p_n K) (\Pi^{\top}\Pi) = o(1), \\
	& \frac{C\Delta^2}{K^2} \sum_{j \in [n]}\sum_{k \in [n]} \left( \sum_{i \in [n]}P_{ij}^2|\Pi_iM_{ik}M_{jk}|\right)^2 \\
	& =  \frac{C\Delta^2}{K^2} \sum_{j \in [n]}\sum_{k \neq j} \left( \sum_{i \in [n]}P_{ij}^2|\Pi_iM_{ik}M_{jk}|\right)^2 + \frac{C\Delta^2}{K^2} \sum_{j \in [n]} \left( \sum_{i \in [n]}P_{ij}^2|\Pi_iM_{ij}M_{jj}|\right)^2 \\
	& \leq \frac{C\Delta^2}{K^2} \sum_{j \in [n]}\sum_{k \neq j} M_{jk}^2 \left(\sum_{i \in [n]}P_{ij}^4\right)\Pi^\top \Pi + \frac{C\Delta^2}{K^2} \left(\sum_{j \in [n]}\sum_{i \in [n]}P_{ij}^4\right)\Pi^\top \Pi \\
	& \leq \frac{C\Delta^2}{K^2} K p_n^2 (\Pi^{\top}\Pi) + \frac{C\Delta^2}{K^2} K p_n (\Pi^{\top}\Pi) =o(1), \\
	& \frac{C\Delta^2}{K^2} \sum_{j \in [n]}\sum_{k \in [n]} \left( \sum_{i \in [n]}P_{ij}^2|\Pi_iM_{ik}M_{jk}| \right) \left( \sum_{i \in [n]}P_{ik}^2|\Pi_iM_{ij}M_{jk}| \right) \leq \frac{C\Delta^2}{K^2} K p_n (\Pi^{\top}\Pi) =o(1), \\
	& \frac{C\Delta^2}{K^2}\left( \sum_{i \in [n]}\sum_{j \in [n]} P_{ij}^2 |\Pi_i| \right)^2 \leq 
	\frac{C\Delta^2}{K^2} (p_nK) (\Pi^{\top}\Pi) =o(1). 
\end{align*}

Then, to show that \citet[Lemma 3]{MS22} holds under our conditions, we show the following terms are o(1):
\begin{align*}
	& \frac{C}{K}\sum_{i \in [n]}\sum_{j \neq i} P_{ij}^2 | \Pi_i \lambda_i \Pi_j \lambda_j | \leq 
	\frac{C}{K} \left( \sum_{i \in [n]}\sum_{j \in [n]} P_{ij}^2 \Pi_i^2\Pi_j^2\right)^{1/2}     \left( \sum_{i \in [n]}\sum_{j \in [n]} P_{ij}^2 \lambda_i^2\lambda_j^2\right)^{1/2} \\
	& \leq \frac{C}{K} p_n \left(\Pi^{\top}\Pi\right) \left( \lambda^{\top}\lambda\right) \leq \frac{C}{K^2} p_n \left(\Pi^{\top}\Pi\right)^2 = o(1), \\
	& \frac{C}{K^2} \sum_{j \in [n]} \left(\sum_{i \in [n]}P_{ij}^2|\Pi_i||\lambda_i|\right)^2\lambda_j^2  \leq 
	\frac{C}{K^2} \sum_{j \in [n]}\left(p_n \sum_{i \in [n]}|\Pi_i||\lambda_j|\right)^2\lambda_j^2 \leq \frac{C}{K^2}p_n^2 \left(\Pi^{\top}\Pi\right)\left(\frac{\Pi^{\top}\Pi}{K}\right)^2 = o(1), \\
	& \frac{C}{K^2} \sum_{i \in [n]}\sum_{i' \in [n]}\sum_{j \in [n]}\sum_{j' \in [n]}P_{ij}^2|\Pi_i\lambda_i\Pi_j|P_{i'j'}^2|\Pi_{i'}\lambda_{i'}\Pi_{j'}|\sum_{k \in [n]}|M_{jk}M_{j'k}| \\
	& \leq \frac{C}{K^2}\left(\sum_{i\in[n]}\sum_{j\in[n]}P_{ij}^2\Pi_i^2\lambda_i^2 \right)\left(\sum_{i\in[n]}\sum_{j\in[n]}P_{ij}^2 \Pi_j^2\right) \leq \frac{C}{K^2} p^2_n (\Pi^{\top}\Pi) (\lambda^{\top}\lambda) 
	\leq \frac{C}{K^3}p_n^2\left(\Pi^{\top}\Pi\right)^2 = o(1), 
\end{align*}
where $\sum_{k \in [n]}|M_{jk}M_{j'k}|\leq 1$ by \citet[Lemma S1.1(ii)]{MS22}.

Now we show that \citet[Lemma S3.2 ]{MS22} holds under our conditions, i.e.,
\begin{align*}
	& (a) \quad \frac{1}{K} \sum_{i=1}^n (\omega_i + \sum_{j \neq i}P_{ij}V_j)^2V_i - \left( \frac{1}{K} \sum_{i=1}^n \omega_i^2 \mathbb{E}[V_i] + \frac{1}{K} \sum_{i,j\neq i} P_{ij}^2 \mathbb{E}[V_i]\eta_j^2\right) \convP 0, \\
	& (b) \quad \frac{1}{K}  \sum_{i=1}^n (\omega_i + \sum_{j \neq i}P_{ij}V_j)^2 \frac{\xi_{1,i}}{M_{ii}}\sum_{k \neq j} P_{ik} \xi_{2,k} \convP 0, \\
	& (c) \quad \frac{1}{K}  \sum_{i=1}^n (\omega_i + \sum_{j \neq i}P_{ij}V_j)^2 a_i \xi_{1,i} \convP 0, \\
	& (d) \quad \frac{1}{K}  \sum_{i=1}^n (\omega_i + \sum_{j \neq i}P_{ij}V_j)^2 \frac{a_i}{M_{ii}}\sum_{k \neq i}P_{ik}\xi_{1,k} - \frac{2}{K} \sum_{i=1}^n\sum_{j \neq i} P_{ij}^2 \omega_i \frac{a_i}{M_{ii}}\mathbb{E}[V_j\xi_{1,j}] \convP 0, \\
	& (e) \quad \frac{1}{K}  \sum_{i=1}^n (\omega_i + \sum_{j \neq i}P_{ij}V_j)^2 \Pi_i \frac{\lambda_i}{M_{ii}} \convP 0,
\end{align*}
where $\xi_{1,i}, \xi_{2,i}$ stay for either $e_i$ or $V_i$, $V_i$ stay for $e_i^2, e_iV_i,$ or $V_i^2$, and $a_i$ stay for either $\Pi_i$ or $\frac{\lambda_i}{M_{ii}}$. 

To prove statement (a), following the arguments in \cite{MS22}, we just need to show the following terms are $o(1)$:
\begin{align*}
	& \mathbb{E} \left[ \frac{1}{K} \sum_{i \in [n]} \omega_i^2 V_i \right]^2 \leq \frac{C}{K^2} \sum_{i \in [n]} \omega_i^4 \leq \frac{C}{K^2} \max_{i \in [n] }\omega_i^2 \left(\sum_{i \in [n] }\omega_i^2\right)
	\leq \frac{C}{K^2} 
	p_n
	\left( \Pi^{\top}\Pi \right)^2 = o(1), \\
	& \frac{C}{K} \sum_{i \in [n]} \sum_{j \neq i} P_{ij}^2 (\omega_i^2 + |\omega_i||\omega_j|) 
	\leq \frac{C}{K}\left( \sum_{i \in [n]}P_{ii}\omega_i^2 +  \left(\sum_{i \in [n]}P_{ii}\omega_i^2\right)^{1/2} \left(\sum_{j \in [n]}P_{jj}\omega_j^2\right)^{1/2}\right)
	\leq \frac{C}{K} 
	p_n
	\Pi^{\top}\Pi = o(1), 
\end{align*}
where we have used $\max_{i\in[n]}\omega_i^2 \leq p_n\Pi^{\top}\Pi$, 
$\sum_{i \in [n]} \omega_i^2 \leq C \Pi^{\top}\Pi$, 
and \citet[Lemma S1.3(b)]{MS22}. 

To prove statement (b), we show that 
\begin{align*}
	& \frac{C}{K^2} \sum_{i \in [n]} \sum_{j \neq i}(P_{ij}^2 \omega_i^4 + P_{ij}^2 w_i^2 w_j^2 + P_{ij}^4 w_i^2 + P_{ij}^4 |\omega_i\omega_j|) \\
	& \leq \frac{C}{K^2} \left( p_n \sum_{i \in [n]}\omega_i^4 + \left(\sum_{i \in [n]} P_{ii} \omega_i^4\right)^{1/2}\left(\sum_{j \in [n]} P_{jj} \omega_j^4\right)^{1/2} + \sum_{i \in [n]} P_{ii} \omega_i^2 p_n + 
	p_n\left(\sum_{i \in [n]} P_{ii} \omega_i^2\right)^{1/2}\left(\sum_{j \in [n]} P_{jj} \omega_j^2\right)^{1/2}
	\right) \\
	& \leq \frac{C}{K^2} \left(p^{2}_n (\Pi^{\top}\Pi)^2 + p^{2}_n (\Pi^{\top}\Pi)^2 + p^{2}_n 
	(\Pi^{\top}\Pi) + p^{2}_n 
	(\Pi^{\top}\Pi) \right) = o(1), \\
	& \frac{C}{K^2} \left( \sum_{i \in [n]}\omega_i^2 + \sum_{i \in [n]}\sum_{j \in [n]} P_{ij}^2 |\omega_i \omega_j| \right)  
	\leq \frac{C}{K^2} \left( 
	\Pi^{\top}\Pi + 
	p_n
	\Pi^{\top}\Pi \right) = o(1), 
\end{align*}
where we have used 
$\sum_{i\in[n]}\omega_i^2 \leq C \Pi^{\top}\Pi$ 
and $\sum_{i\in[n]} \omega_i^4 \leq C p_n (\Pi^{\top}\Pi)^2$. 

To prove statement (c), we show that, for $a_i = \Pi_i$ or $\lambda_i/M_{ii}$, 
\begin{align*}
	& \frac{C}{K^2}\left( \sum_{i \in [n]} P^2_{ii} a_i^2 + \sum_{i \in [n]}\sum_{j \in [n]} P_{ij}^2 \left|a_ia_j\right| \right) 
	\leq \frac{C}{K^2} \left( p_n^2 a^{\top}a + p_n a^{\top}a \right) = o(1), \\
	& \frac{C}{K^2} \sum_{i \in [n]} \omega_i^4 \frac{\lambda^2_i}{M^2_{ii}} 
	\leq \frac{C}{K^2} \left(\max_{i \in [n]}\omega_i^2\right)^2 \sum_{i \in [n]}\lambda_i^2
	\leq C p_n^2 \left( \frac{\Pi^{\top}\Pi}{K} \right)^3 = o(1), \\
	& \frac{C}{K^2} \sum_{i\in[n]}\omega_i^4\Pi_i^2 \leq \frac{C}{K^2}\sum_{i\in[n]}\omega_i^4
	\leq \frac{C}{K^2}p_n\left(\Pi^{\top}\Pi\right)^2 = o(1), \; \text{where we have used $\max_{i\in[n]}|\Pi_i|\leq C$},
	\\
	&  \frac{C}{K^2} \sum_{i \in [n]}\sum_{j \neq i} P_{ij}^4 \left( a_i^2 + \left|a_i\right|\left|a_j\right|\right) 
	\leq \frac{C}{K^2} \left( p^2_n a^{\top}a + p^2_n a^{\top}a \right) = o(1), \\
	& \frac{C}{K^2} \sum_{i\in[n]}\sum_{j \neq i} P_{ij}^2(\omega_i^2a_i^2 + |\omega_ia_i||\omega_ja_j|) \leq \frac{C}{K^2} \left(p_n^2(\Pi^{\top}\Pi)(a^{\top}a) + p_n(\Pi^{\top}\Pi)(a^{\top}a)\right) =o(1).
\end{align*}
To prove statement (d), we first show that 
\begin{align*}
	& \frac{C}{K^2} \left( \left(\sum_{i \in [n]}\omega_i^2\left|a_i\right|\right)^2 + \left(\sum_{i \in [n]} \left|\omega_i a_i \right|\right)^2\right) = o(1).
\end{align*}
In particular, when $a_i = \Pi_i$, we have
\begin{align*}
	& \frac{C}{K^2} \left( \left(\sum_{i \in [n]}\omega_i^2\left|\Pi_i\right|\right)^2 + \left(\sum_{i \in [n]} \left|\omega_i \Pi_i \right|\right)^2\right) 
	\leq \frac{C}{K^2} \left( \left(\sum_{i \in [n]}\omega_i^2\right)^2 + \left(\sum_{i \in [n]} \left|\omega_i \Pi_i\right|\right)^2\right) \\
	&  \leq \frac{C}{K^2} \left( \left(\Pi^\top \Pi \right)^2 + \left(\sum_{i \in [n]} \omega_i^2\right)\left(\Pi^{\top}\Pi\right)\right) 
	\leq \frac{C}{K^2} \left( (\Pi^{\top}\Pi)^2  +  (\Pi^{\top}\Pi)^2 \right) = o(1),
\end{align*}
When $a_i = \frac{\lambda_i}{M_{ii}}$, we have 
\begin{align*}
	& \frac{C}{K^2} \left( \left(\sum_{i \in [n]}\omega_i^2\left| \frac{\lambda_i}{M_{ii}} \right|\right)^2 + \left(\sum_{i \in [n]} \left|\omega_i \frac{\lambda_i}{M_{ii}} \right|\right)^2\right) 
	\leq \frac{C}{K^2} \left( \left(\sum_{i \in [n]}\omega_i^4\right)(\lambda^{\top}\lambda) + \left(\sum_{i \in [n]}\omega_i^2\right)(\lambda^{\top}\lambda) \right)
	\\
	& \leq \frac{C}{K^2} \left( p_n (\Pi^{\top}\Pi)^2(\lambda^{\top}\lambda)
	+  (\Pi^{\top}\Pi)(\lambda^{\top}\lambda)\right) = o(1).
\end{align*}

Furthermore, we can show that 
\begin{align*}
	& \frac{C}{K^2} \left( \sum_{i \in [n]}|\omega_ia_i| \right)^2 \leq 
	\frac{C}{K^2} (\Pi^{\top}\Pi)(a^{\top}a) = o(1), \\
	& \frac{C}{K} \sum_{i \in [n]} P_{ii} \left| a_i \right| 
	\leq \frac{C}{K} \left( \sum_{i \in [n]} P_{ii}^2 \right)^{1/2} \left( a^{\top}a \right)^{1/2}
	\leq \frac{C}{K} (p_n K)^{1/2} \left( a^{\top}a \right)^{1/2} = o(1),\\
	& \frac{C}{K^2} \left( \sum_{i \in [n]} P_{ii} \left| a_i \right| \right)^2 
	\leq \frac{C}{K^2} \left(\sum_{i \in [n]}P_{ii}^2\right) \left( a^{\top}a \right)
	\leq \frac{C}{K^2} p_n K \left( a^{\top}a \right) = o(1).
\end{align*}

To prove statement (e), we show that 
\begin{align*}
	&     \left| \frac{C}{K} \sum_{i \in [n]} \omega_i^2 \Pi_i \frac{\lambda_i}{M_{ii}} \right|
	\leq  \frac{C}{K} \sum_{i \in [n]} \omega_i^2 \left| \frac{\lambda_i}{M_{ii}} \right|
	\leq \frac{C}{K} \left(\sum_{i\in[n]} \omega_i^4 \right)^{1/2} \left( \lambda^{\top} \lambda \right)^{1/2} 
	\leq \frac{C}{K} p_n^{1/2} (\Pi^\top \Pi) (\lambda^{\top} \lambda)^{1/2} = o(1), \\
	& \frac{C}{K^2} \sum_{j \in [n]} \left( \sum_{i \neq j} P_{ij} \omega_i \Pi_i \frac{\lambda_i}{M_{ii}} \right)^2 
	\leq \frac{C}{K^2} \sum_{j \in [n]} \left( \sum_{i \neq j} |P_{ij}| |\omega_i| |\lambda_i| \right)^2 \leq \frac{C}{K^2} \sum_{j \in [n]} \left(\sum_{i \neq j}\omega_i^2\right)\left( \sum_{i \neq j} P_{ij}^2\lambda_i^2 \right) \\
	& \leq \frac{C K p_n^{1/2} \Pi^\top \Pi \lambda^\top \lambda}{K^2} = o(1), \\
	& \frac{C}{K^2} \sum_{j \in [n]} \left( \sum_{i \neq j} P_{ij}^2 \Pi_i \frac{\lambda_i}{M_{ii}} \right)^2 \leq \frac{C}{K^2} \sum_{j \in [n]} \left( \sum_{i \neq j} P_{ij}^2  |\lambda_i| \right)^2 \leq \frac{C K p_n \lambda^\top \lambda}{K^2} = o(1), \\
	& \frac{C}{K} \sum_{j \in [n]} \sum_{i \neq j} P_{ij}^2 \left|\Pi_i \frac{\lambda_i}{M_{ii}} \right| \leq \frac{C}{K}\sum_{i \in n}\sum_{j \in [n]}P_{ij}^2|\Pi_i\lambda_i|
	\leq \frac{C}{K} p_n (\Pi^{\top}\Pi)^{1/2}(\lambda^{\top}\lambda)^{1/2} = o(1), 
	\\
	& \frac{C}{K^2}\sum_{j \in [n]}\sum_{k \neq j} \left(\sum_{i \neq j,k} P^2_{ij}P^2_{ik} \Pi_i \frac{\lambda_i}{M_{ii}}\right)^2  
	\leq \frac{C}{K^2}\sum_{j \in [n]}\sum_{k \neq j} \left(\sum_{i \neq j,k} P^2_{ij}P^2_{ik} |\lambda_i|\right)^2 \\
	& \leq  \frac{C}{K^2}\left(\sum_{j \in [n]}\sum_{k \neq j} \sum_{i \neq j,k} P^4_{ij}P^4_{ik} \right) \lambda^\top \lambda \leq \frac{C p_n^3 K \lambda^\top \lambda}{K^2} = o(1),
\end{align*}
where we have used \citet[Lemma S1.1(ii)]{MS22}. 

Finally, we can show that \citet[Lemma S3.1]{MS22} also holds under our conditions by using similar arguments. We omit the details for brevity.

\section{Lemma \ref{lem:W} and Its Proof}
\label{sec:lem_W}
\begin{lem}
	Suppose assumptions in Theorem \ref{thm:W} hold. Then, we have
	\begin{align*}
		&\hat \gamma_{e} = O_P(n^{-1/2}), \quad     
		\hat \gamma_{V} = O_P(n^{-1/2}),
		\quad Q_{\tilde e, W} = O_P(1), \quad Q_{\tilde V, W} = O_P(1), \\
		& \hat \gamma_V^\top Q_{W,W^\top} \hat \gamma_V = o_P(1), \quad
		\hat \gamma_e^\top Q_{W,W^\top} \hat \gamma_e = o_P(1),
		\quad \text{and} \quad \hat \gamma_e^\top Q_{W,W^\top} \hat \gamma_V = o_P(1).
	\end{align*}
	\label{lem:W}
\end{lem}

\begin{proof}
	We have $\hat \gamma_{e} = O_P(n^{-1/2})$ because $\mathbb{E}\tilde e_i = 0$ and $\text{mineig}(W^\top W/n) \geq c>0$. 
	Similarly, we have $\hat \gamma_{V} = O_P(n^{-1/2})$. 
	To see that $Q_{\tilde e, W} = O_P(1)$, 
	we note that $\mathbb{E}Q_{\tilde e, W} = 0$ and 
	\begin{align*}
		\mathbb{E} Q_{\tilde e, W} Q_{\tilde e, W}^{\top} 
		\leq C \sum_{i \in [n]}(\sum_{j \neq i}P_{ij}W_j)^\top(\sum_{j \neq i}P_{ij}W_j)/K = C \sum_{i \in [n]}P_{ii}^2 W_i^\top W_i/K \leq C,
	\end{align*}
	where we use the fact that $\sum_{j \neq i}P_{ij}W_j = - P_{ii}W_i$ since $P_{ij}$ is the $ij$-th element of $P=Z(Z^{\top}Z)^{-1}Z^{\top}$. Similarly, we have $Q_{\tilde V, W} = O_P(1).$
	
	To see $\hat \gamma_V^\top Q_{W,W^\top} \hat \gamma_V = o_P(1)$, we note that 
	\begin{align*}
		\left|\hat \gamma_V^\top Q_{W,W^\top} \hat \gamma_V\right| \leq \sum_{i \in [n]} (W_i^\top \hat \gamma_V)^2/\sqrt{K} = o_P(1),
	\end{align*}
	where we use the fact that $\sum_{i \in [n]}W_iW_i^{\top}/n = O_P(1)$ and $\hat{\gamma}_V = O_P(n^{-1/2})$, so that
	\begin{align*}
		\sum_{i \in [n]} (W_i^\top \hat \gamma_V)^2 = O_P(1).
	\end{align*}
	Similarly, we can show that 
	$$\hat \gamma_e^\top Q_{W,W^\top} \hat \gamma_e = o_P(1), \quad \text{and} \quad \hat \gamma_e^\top Q_{W,W^\top} \hat \gamma_V = o_P(1).$$
\end{proof}

\section{Comparison with HLIM Estimator under Strong Identification}
\label{sec:HLIM}
We consider the model in Section \ref{sec:W0} and the HLIM estimator proposed by \cite{Haus2012}. Specifically, \cite{Haus2012} estimate $(\beta,\gamma)$ by $(\hat \beta^{HLIM},\hat \gamma^{HLIM})$ defined as
\begin{align*}
	(\hat \beta^{HLIM},\hat \gamma^{HLIM}) = \argmin_{b,r} \mathcal{Q}(b,r),\quad     \mathcal{Q}(b,r) = \frac{\sum_{i \in [n]} \sum_{j \neq i} (\tilde Y_i - \tilde X_i b - W_i^\top r) \tilde P_{ij}(\tilde Y_i - \tilde X_i b - W_i^\top r) }{\sum_{i \in [n]}(\tilde Y_i - \tilde X_i b - W_i^\top r)^2},
\end{align*}
where $\tilde P_{ij}$ is the projection matrix constructed by $(W_i^\top,\tilde Z_i^\top)^\top$. Following \cite{Haus2012}, we let $\tilde \Pi_i = \mu_n \tilde \pi_i/\sqrt{n}$ such that $\sum_{i \in [n]} \tilde \pi_i^2/n \geq c>0$ for some constant $c$. As explained in the paper, under strong identification, we have $\mu_n^2/\sqrt{K} \rightarrow \infty$. In both cases considered in \citet[Assumption 6]{Haus2012}, the convergence rate can be unified as 
$\sqrt{K}/ \mu_n^2$.
Then, the Wald statistic can be written as 
\begin{align*}
	W_h(\beta_0) = \frac{\mu_n^2 (\hat \beta^{HLIM} - \beta_0)/\sqrt{K}}{\hat \Phi_h^{1/2}},
\end{align*}
where $\hat \Phi_h$ is a consistent estimator of $\Phi_h$, and $\Phi_h$ is the asymptotic variance of $\hat \beta^{HLIM}$. 
To study the behaviour of $W_h(\beta_0)$ under strong identification and local alternatives, we let $\beta_0$ denote the local alternative in the sense that $\beta_0 = \beta + \frac{\widetilde \Delta}{\mu^2_n /\sqrt{K}}$. We will provide the expression for $\Phi_h$ later. We also note that the notation in \cite{Haus2012} and our paper is different. Specifically, their $\delta_0$ is our $(\gamma^\top,\beta_0)^\top$, their $\hat \delta$ is our $((\hat \gamma^{HLIM})^\top, \hat \beta^{HLIM})^\top$, their $X_i$ is our $(W_i^\top,\tilde X_i)^\top$, their $Z_i$ is our $(W_i^\top,\tilde Z_i^\top)^\top$, and thus their projection matrix $P$ is our $\tilde P$, which is the one based on $(W_i^\top,\tilde Z_i^\top)^\top$. We use $P$ and $P_W$ to denote the projection matrices based on our $Z_i$ and $W_i$, respectively, where $Z_i = ([M_W]_{i\cdot} \tilde Z)^\top$, $[M_W]_{i\cdot}$ is the $i$th row of $M_W$, and $M_W = I_n - P_W$.  

Further denote $L$ as a matrix that selects the last element of 
$\hat \delta = ((\hat \gamma^{HLIM})^\top, \hat \beta^{HLIM})^\top$ 
and 
\begin{align*}
	S_n = \begin{pmatrix}
		I_{d} & 0\\
		\pi_x^\top & 1
	\end{pmatrix} \diag(\sqrt{n},\cdots,\sqrt{n},\mu_n),
\end{align*}
where $\pi_x = (W^\top W)^{-1} W^\top \tilde \Pi$ is the projection coefficient of $\tilde \Pi$ on $W$. 
Then, the corresponding definition of $\hat{D}(\delta_0)$ in \citet[p.235]{Haus2012} under our notation is as follows:
\begin{align*}
	\hat D(\delta_0) = \frac{\sum_{i \in [n]}\sum_{j \neq i}\left[\mathbb W_i \tilde P_{ij} \overline{e}_j(\beta_0) - \overline{e}_i(\beta_0)\tilde P_{ij} \overline{e}_j(\beta_0) \frac{\mathbb W^\top \overline{e}(\beta_0)}{\overline{e}^\top(\beta_0) \overline{e}(\beta_0)}\right]}{\sqrt{K}},
\end{align*}
where $\mathbb W_i = (W_i^\top, \tilde X_i)^\top$, $\mathbb W$ is a $n \times (d+1)$ matrix with its $i$th row being $\mathbb W_i^\top$ where $d$ is the dimension of $W_i$, and $\overline{e}_i(\beta_0) = \tilde e_j- \tilde X_j(\beta_0- \beta)$. In addition, we note that  $\overline{X}_i = \tilde X_i - W_i^\top \pi_x = \Pi_i +\tilde V_i$ as defined in Theorem \ref{thm:W}, 
$X_i = \overline{X}_i - W_i^\top \hat \gamma_V$, 
$\overline{e}_i(\beta_0) = e_i(\beta_0)+W_i^\top \hat \gamma_e - W_i^\top \hat \pi_x (\beta_0-\beta)$, where $ \pi_x = (W^\top W)^{-1}(W^\top \tilde \Pi)$, $\hat \pi_x = (W^\top W)^{-1}(W^\top \tilde X) = \pi_x + \hat \gamma_V$, $\hat \gamma_V = (W^{\top}W)^{-1}(W^{\top}\tilde V)$, and
$\hat \gamma_e = (W^{\top}W)^{-1}(W^{\top}\tilde e)$.
Further let $\overline{\delta}$ be between $\delta = (\gamma^\top,\beta)^\top$ and $\delta_0$.

Then, following the argument in the proof of \citet[Theorem 2]{Haus2012}, we have
\begin{align*}
	& (\mu^2_n/\sqrt{K})(\hat \beta^{HLIM} - \beta_0) \\
	& = (\mu^2_n/\sqrt{K})    L(\hat \delta - \delta_0) \\
	& = -(\mu^2_n/\sqrt{K})    L \left( \frac{\partial \hat D (\overline{\delta})}{\partial \delta}\right)^{-1} \hat D(\delta_0) \\
	& = -(\mu^2_n/\sqrt{K})    L (S_n^\top )^{-1}\left(S_n^{-1} \frac{\partial \hat D (\overline{\delta})}{\partial \delta} (S_n^\top )^{-1}\right)^{-1} S_n^{-1} \hat D(\delta_0) \\
	& = -(\mu^2_n/\sqrt{K}) (0,1/\mu_n)(H^{-1} + o_P(1)) \diag(1/\sqrt{n},\cdots,1/\sqrt{n},1/\mu_n) \begin{pmatrix}
		I_{d} & 0\\
		-\pi_x^\top & 1
	\end{pmatrix} \hat D(\delta_0) \\
	& = -\frac{\mu_n}{\sqrt{K}}\left(\begin{pmatrix}
		(H^{21}   + o_P(1))/\sqrt{n} - \pi_x^\top (H^{22}   + o_P(1))/\mu_n, (H^{22}   + o_P(1))/\mu_n 
	\end{pmatrix} \right) \hat D (\delta_0) \\
	& = (H^{22}+o_P(1))(-\pi_x^\top, 1)\hat D (\delta_0)/\sqrt{K} \\
	& = (H^{22}+o_P(1)) \frac{\sum_{i \in [n]}\sum_{j \neq i}\left[\overline{X}_i\tilde P_{ij} \overline{e}_j(\beta_0) - \overline{e}_i(\beta_0)\tilde P_{ij} \overline{e}_j(\beta_0) \frac{\overline{X}^\top \overline{e}(\beta_0)}{\overline{e}^\top(\beta_0) \overline{e}(\beta_0)}\right]}{\sqrt{K}},   
\end{align*}
where by \citet[Lemma A7]{Haus2012}, $S_n^{-1} \frac{\partial \hat D (\overline{\delta})}{\partial \delta} (S_n^\top )^{-1} \convP H$, and we denote $H^{-1} = \begin{pmatrix}
	H^{11} & H^{12} \\
	H^{21} & H^{22}
\end{pmatrix}$.

Following the same argument in the proof of Lemma \ref{lem:W}, we can show that 
\begin{align*}
	& \frac{\sum_{i \in [n]}\sum_{j \neq i} \hat {\gamma}_V^\top 
		W_i\tilde P_{ij} \overline{e}_j(\beta_0)}{\sqrt{K}} = o_P(1), \quad \frac{\sum_{i \in [n]}\sum_{j \neq i} \overline{X}_i\tilde P_{ij} W_i^\top (\hat \gamma_e - \hat \pi_x (\beta_0-\beta))}{\sqrt{K}} = o_P(1) \\
	& \frac{\sum_{i \in [n]}\sum_{j \neq i} \overline{e}_i(\beta_0)\tilde P_{ij} W_i^\top (\hat \gamma_e - \hat \pi_x (\beta_0-\beta))}{\sqrt{K}} = o_P(1), \quad \text{and} \\
	&\frac{\sum_{i \in [n]}\sum_{j \neq i} (\hat \gamma_e - \hat \pi_x (\beta_0-\beta))^\top W_i\tilde P_{ij} W_i^\top (\hat \gamma_e - \hat \pi_x (\beta_0-\beta))}{\sqrt{K}} = o_P(1).
\end{align*}
In addition, we have $\overline{X}^\top \overline{e}(\beta_0)/\overline{e}^\top(\beta_0)\overline{e}(\beta_0) \convP \tilde \rho$. Then, we have
\begin{align*}
	\mu_n^2 (\hat \beta^{HLIM} - \beta_0)/\sqrt{K} = H^{22} \frac{\sum_{i \in [n]}\sum_{j \neq i}\left[X_i\tilde P_{ij} e_j(\beta_0) - e_i(\beta_0)\tilde P_{ij} e_j(\beta_0) \tilde \rho\right]}{\sqrt{K}} + o_P(1).  
\end{align*}
Because $X^\top W = 0$ and $e^\top W = 0$, we have $X^\top \tilde P e(\beta_0) = X^\top P e(\beta_0)$ and $e(\beta_0)^\top \tilde P e(\beta_0) = e(\beta_0)^\top P e(\beta_0)$. Therefore, we have
\begin{align*}
	\frac{\sum_{i \in [n]}\sum_{j \neq i}X_i\tilde P_{ij} e_j(\beta_0)}{\sqrt{K}} & = \frac{X^\top P e(\beta_0) - \sum_{i \in [n]}X_i\tilde P_{ii} e_i(\beta_0) }{\sqrt{K}} \\
	& = \frac{\sum_{i \in [n]} \sum_{j \neq i} X_i P_{ij} e_j(\beta_0) + \sum_{i \in [n]} X_i e_i(\beta_0) (P_{ii} - \tilde P_{ii})}{\sqrt{K}} \\
	& = Q_{X,e(\beta_0)} - \frac{\sum_{i \in [n]} X_i e_i(\beta_0) P_{W,ii}}{\sqrt{K}} \\
	& = Q_{X,e(\beta_0)} +o_P(1),
\end{align*}
where we use the facts that $\tilde P_{ii} = P_{ii} + P_{W,ii}$ and 
\begin{align*}
	\sum_{i \in [n]} X_i e_i(\beta_0) P_{W,ii} = \frac{1}{n}     \sum_{i \in [n]} X_i e_i(\beta_0) W_i^\top \left(W^\top W/n\right)^{-1}W_i = O_P(1).
\end{align*}
Similarly, we have
\begin{align*}
	\frac{\sum_{i \in [n]}\sum_{j \neq i}e_i(\beta_0)\tilde P_{ij} e_j(\beta_0)}{\sqrt{K}}  = Q_{e(\beta_0),e(\beta_0)} +o_P(1),
\end{align*}
and thus, 
\begin{align*}
	\mu_n^2 (\hat \beta^{HLIM} - \beta_0)/\sqrt{K} = H^{22}(Q_{X,e(\beta_0)} - \tilde \rho Q_{e(\beta_0),e(\beta_0)}) + o_P(1). 
\end{align*}
In order for the HLIM based Wald test to have a pivotal standard normal distribution in the limit, the asymptotic variance $\Phi_h$ must be 
\begin{align*}
	\Phi_h = (H^{22})^2 (\Psi- 2\tilde \rho \Phi_{12} + \tilde \rho^2 \Phi_1),
\end{align*}
which means the Wald statistic satisfies $W_h(\beta) = \frac{Q_{X,e(\beta_0)} - \tilde \rho Q_{e(\beta_0),e(\beta_0)}}{(\Psi- 2\tilde \rho \Phi_{12} + \tilde \rho^2 \Phi_1)^{1/2}} + o_P(1)$.

\section{Additional Simulation Results}
\label{sec:add_sim}
\subsection{Additional Simulation Results Based on the Limit Problem}\label{sec: further_simu_limit}
In this section, we present further simulation results for the power behavior of tests under the limit problem described in Section \ref{sec:limit}. 

For Figures \ref{further_limit_fig1}--\ref{fig_01_2_rho09}, all the settings remain the same as those in Section \ref{sec:sim1} in the main paper except we use alternative values of the tuning parameters
for (\ref{eq:a_underline}). Specifically, for the values of $p_1$ and $p_2$ in 
\begin{align*}
	\underline{a}(\mu_D,\gamma(\beta_0)) = \min\left(p_1, \frac{p_2 \mathbb{C}_{\alpha,\max}(\rho(\beta_0)) \Phi_1(\beta_0) c_{\mathcal{B}}(\beta_0) }{ \Delta_*^4(\beta_0) \mu_D^2} \right),
\end{align*}
we use 
$(p_1, p_2) = (0.01, 1.5), (0.01, 2)$, 
$(0.001,1.1)$, $(0.001, 1.5)$, $(0.001, 2)$,
$(0.1,1.1)$, $(0.1, 1.5)$, or $(0.1, 2)$,
instead of $(0.01, 1.1)$ in Section \ref{sec:sim}. 
Specifically, 
Figures \ref{further_limit_fig1}--\ref{further_limit_fig4} report the results for $(0.01, 1.5)$, 
Figures \ref{further_limit_fig5}--\ref{further_limit_fig8} report those for $(0.01, 2)$, 
Figures \ref{fig_0001_11_rho02}--\ref{fig_0001_11_rho09} report those for $(0.001, 1.1)$,
Figures \ref{fig_0001_15_rho02}--\ref{fig_0001_15_rho09} report those for $(0.001, 1.5)$,
Figures \ref{fig_0001_2_rho02}--\ref{fig_0001_2_rho09} report those for $(0.001, 2)$,
Figures \ref{fig_01_11_rho02}--\ref{fig_01_11_rho09} report those for $(0.1, 1.1)$,
Figures \ref{further_limit_fig9}--\ref{further_limit_fig12} report those for $(0.1,1.5)$,  
and Figures \ref{fig_01_2_rho02}--\ref{fig_01_2_rho09} report those for $(0.1, 2)$,
respectively. 
We find the results are very similar to those
reported in the main paper. 

Furthermore, Figures \ref{further_limit_fig13}--\ref{further_limit_fig16} present the power curves in the cases with stronger identification ($\mathcal C=9$ or $12$). The overall patterns are very similar to those for $\mathcal C=6$. For Figures \ref{further_limit_fig13}--\ref{further_limit_fig16}, the tuning parameters are set as $(p_1, p_2) = (0.01, 1.1)$, which are same as those in Section \ref{sec:sim} of the main text.
The results for other values of $p_1$ and $p_2$ remain very similar and thus are omitted for brevity.

\begin{figure}[H]
	\centering
	\includegraphics[width=0.9\textwidth,height = 5.85cm]{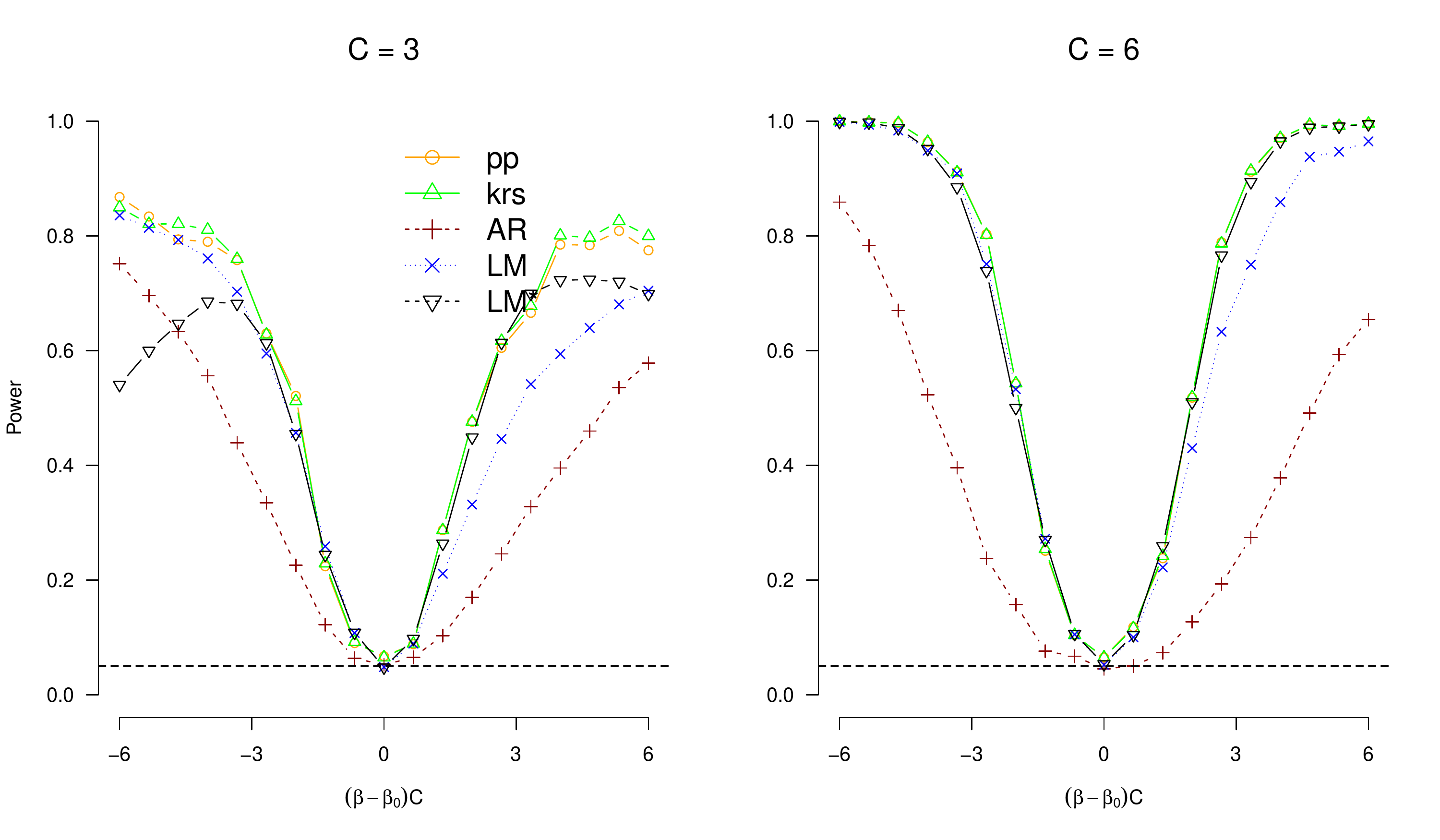}
	\caption{Power Curve for $\rho = 0.2$, $p_1=0.01$, and $p_2=1.5$}
	\label{further_limit_fig1}
\end{figure}

\begin{figure}[H]
	\centering
	\includegraphics[width=0.9\textwidth,height = 5.85cm]{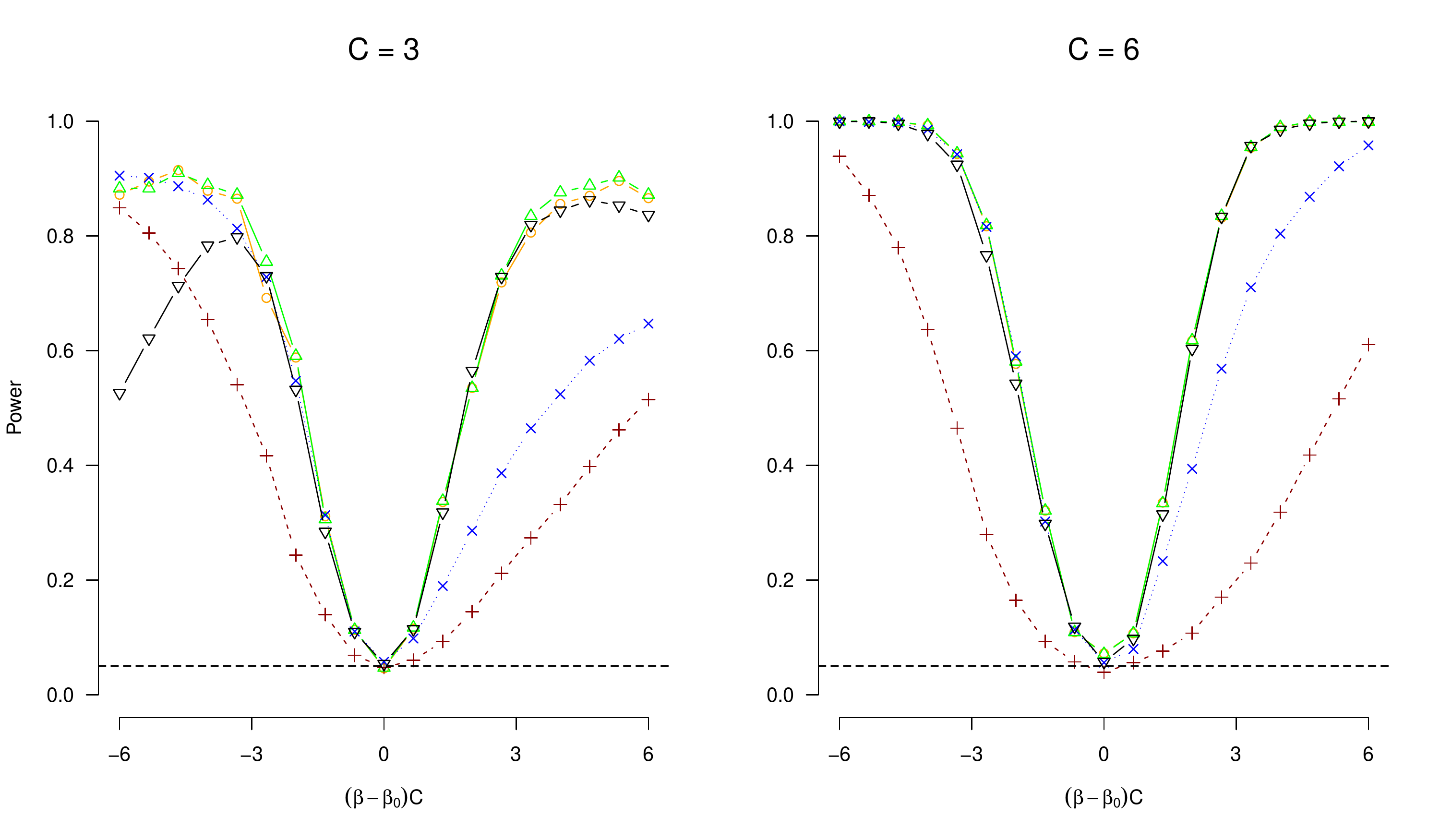}
	\caption{Power Curve for $\rho=0.4$, $p_1=0.01$, and $p_2=1.5$}
	\label{further_limit_fig2}
\end{figure}

\begin{figure}[H]
	\centering
	\includegraphics[width=0.9\textwidth,height = 5.85cm]{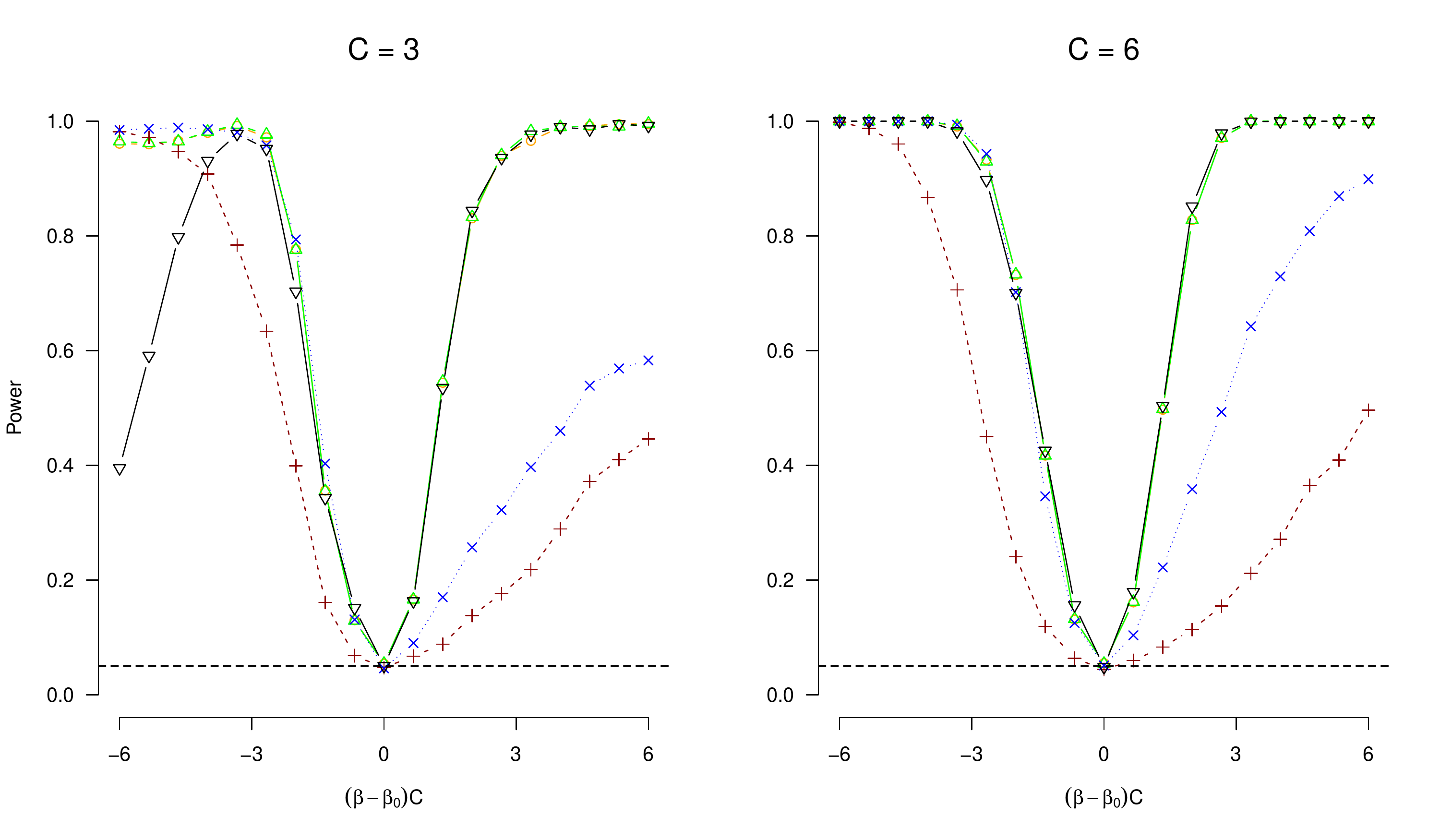}
	\caption{Power Curve for $\rho=0.7$, $p_1=0.01$, and $p_2=1.5$}
	\label{further_limit_fig3}
\end{figure}

\begin{figure}[H]
	\centering
	\includegraphics[width=0.9\textwidth,height = 5.85cm]{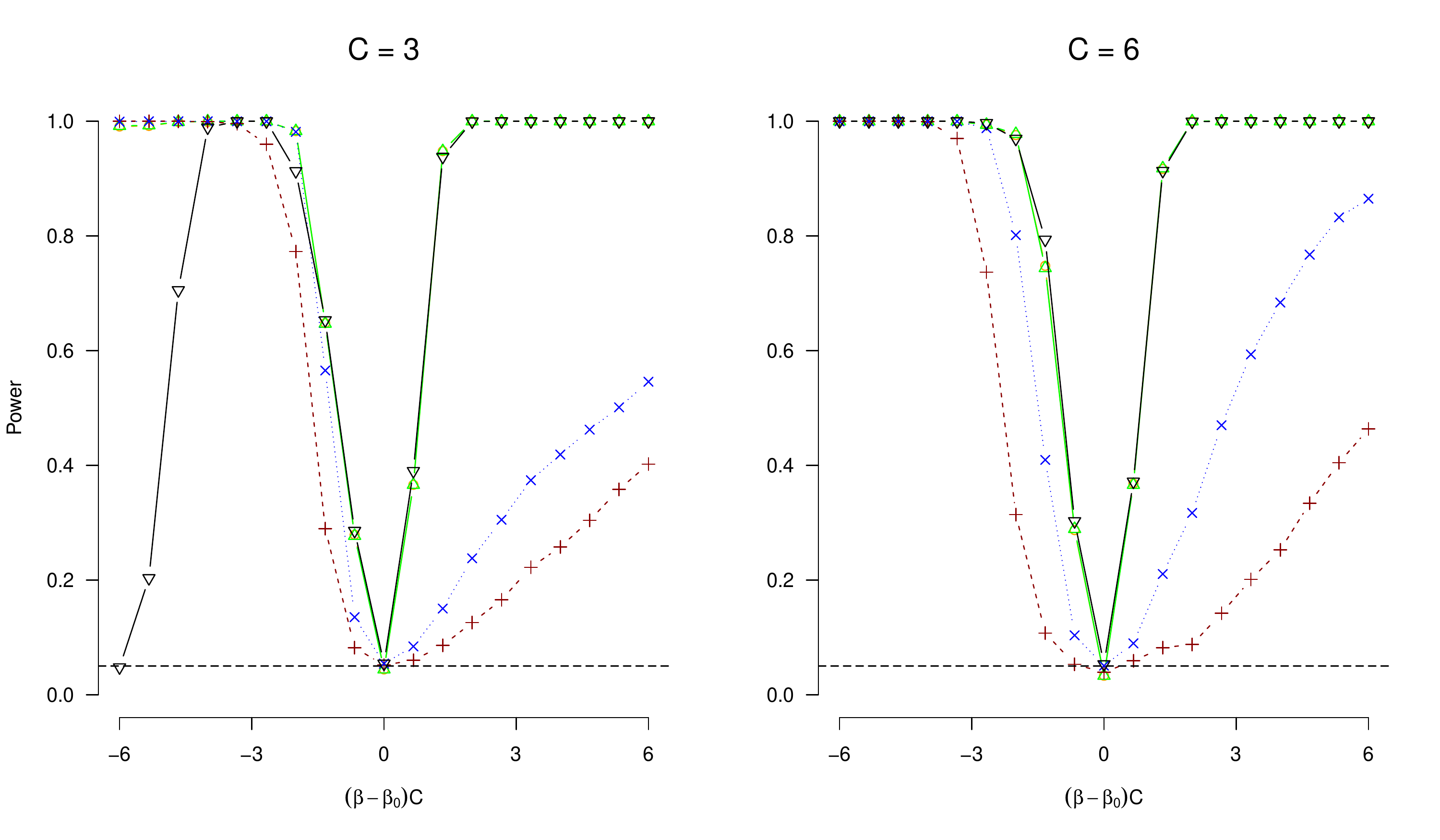}
	\caption{Power Curve for $\rho=0.9$, $p_1=0.01$, and $p_2=1.5$}
	\label{further_limit_fig4}
\end{figure}

\begin{figure}[H]
	\centering
	\includegraphics[width=0.9\textwidth,height = 5.85cm]{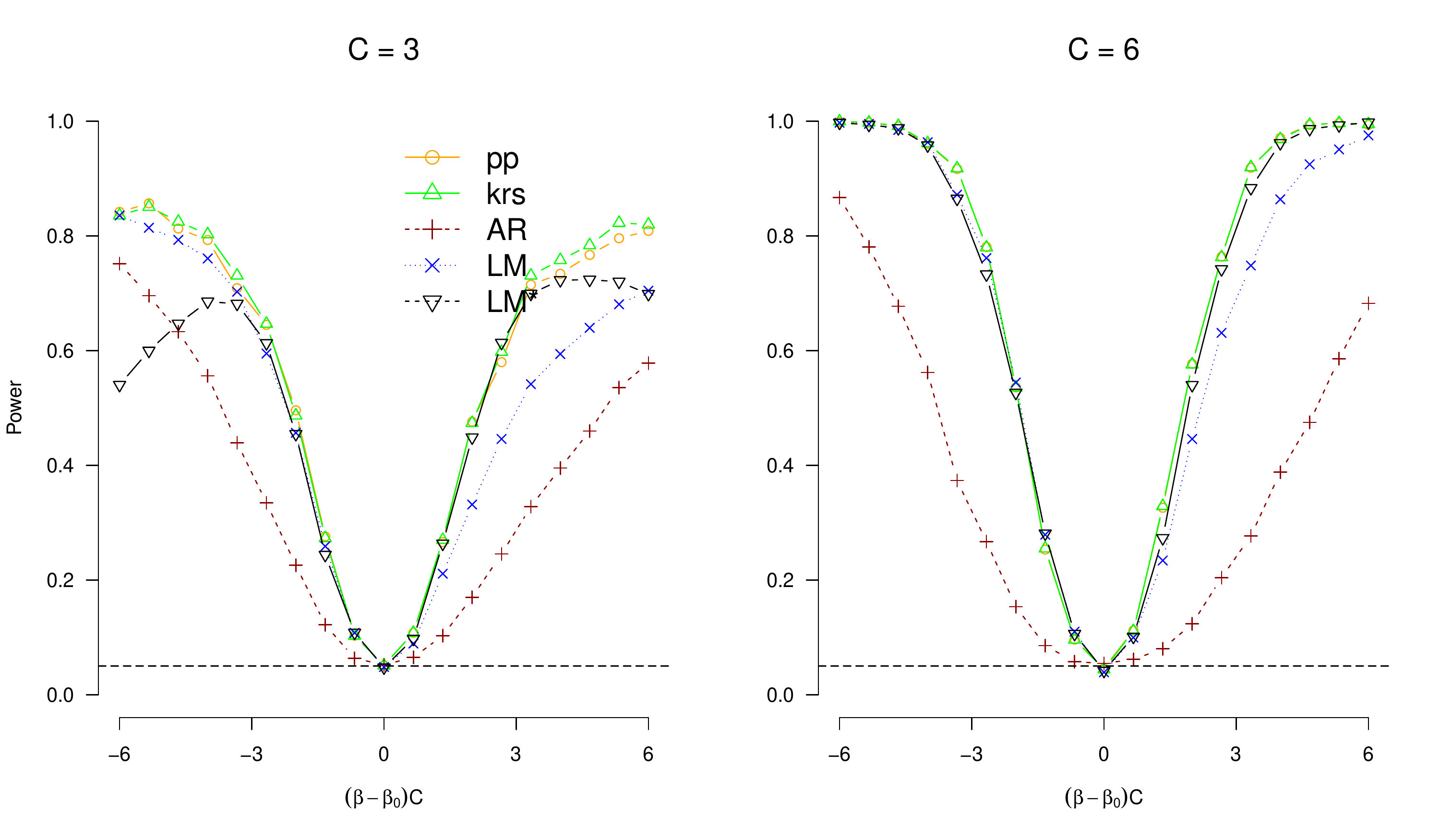}
	\caption{Power Curve for $\rho = 0.2$, $p_1=0.01$, and $p_2=2$}
	\label{further_limit_fig5}
\end{figure}

\begin{figure}[H]
	\centering
	\includegraphics[width=0.9\textwidth,height = 5.85cm]{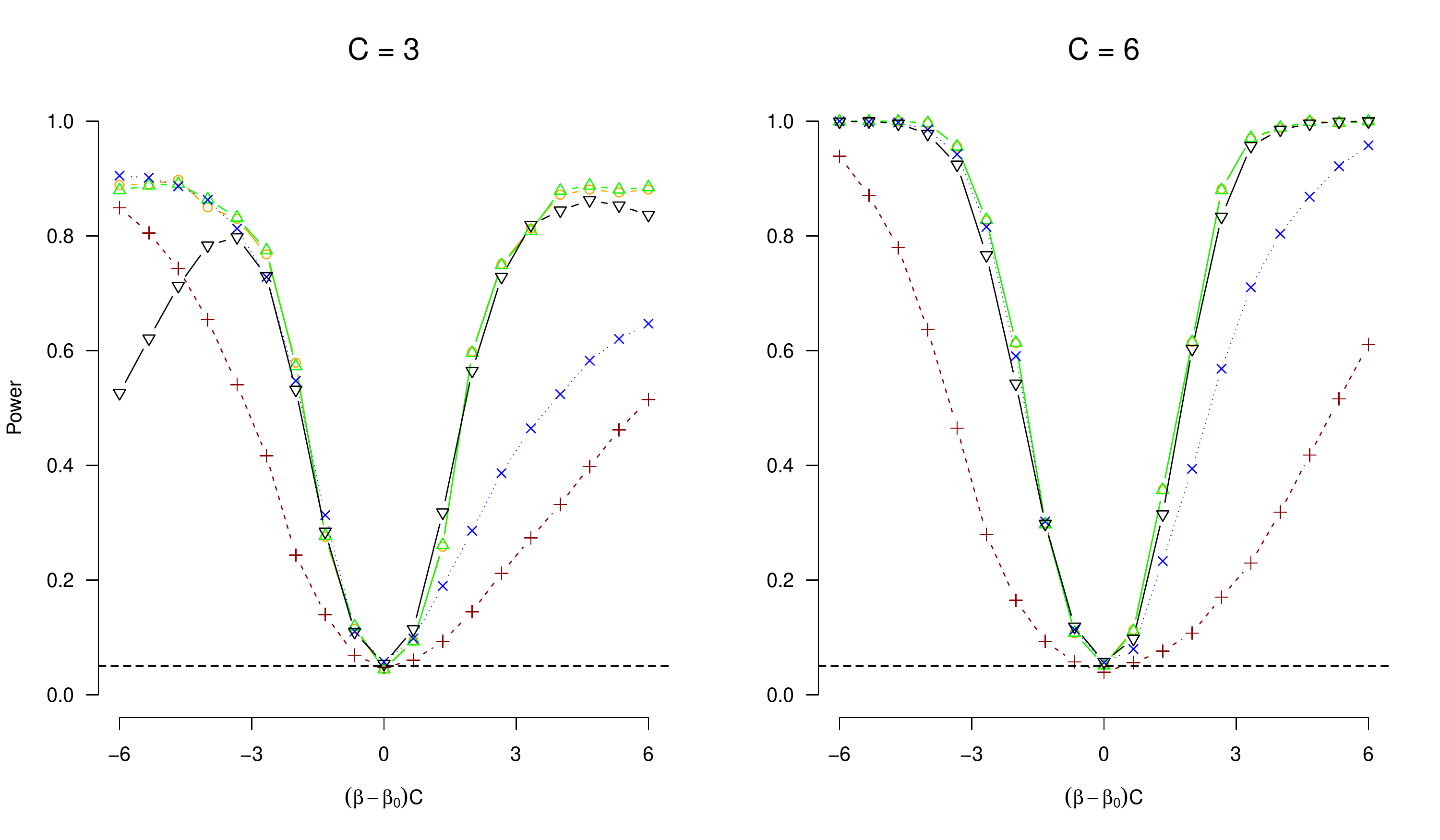}
	\caption{Power Curve for $\rho=0.4$, $p_1=0.01$, and $p_2=2$}
	\label{further_limit_fig6}
\end{figure}

\begin{figure}[H]
	\centering
	\includegraphics[width=0.9\textwidth,height = 5.85cm]{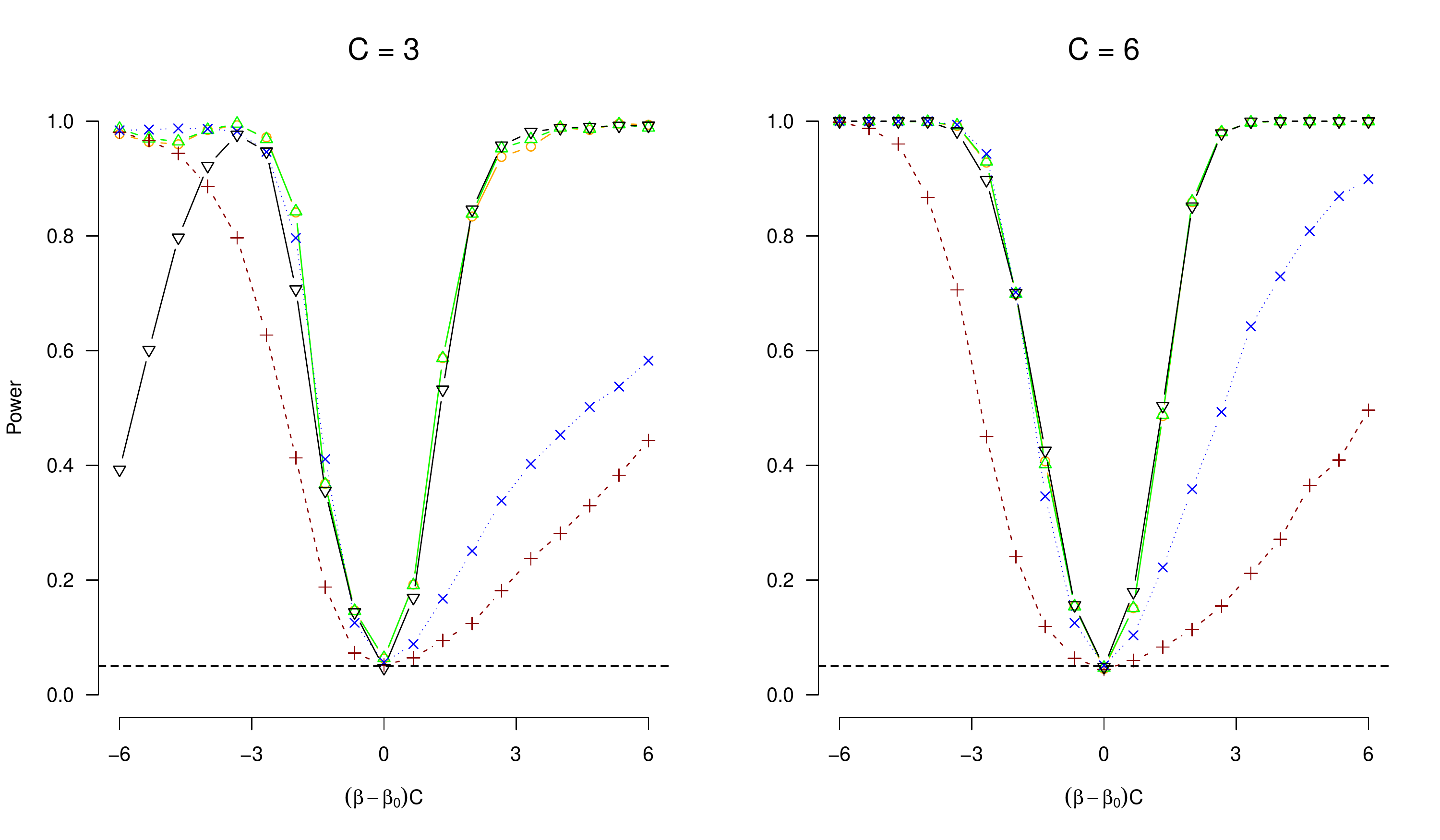}
	\caption{Power Curve for $\rho=0.7$, $p_1=0.01$, and $p_2=2$}
	\label{further_limit_fig7}
\end{figure}

\begin{figure}[H]
	\centering
	\includegraphics[width=0.9\textwidth,height = 5.85cm]{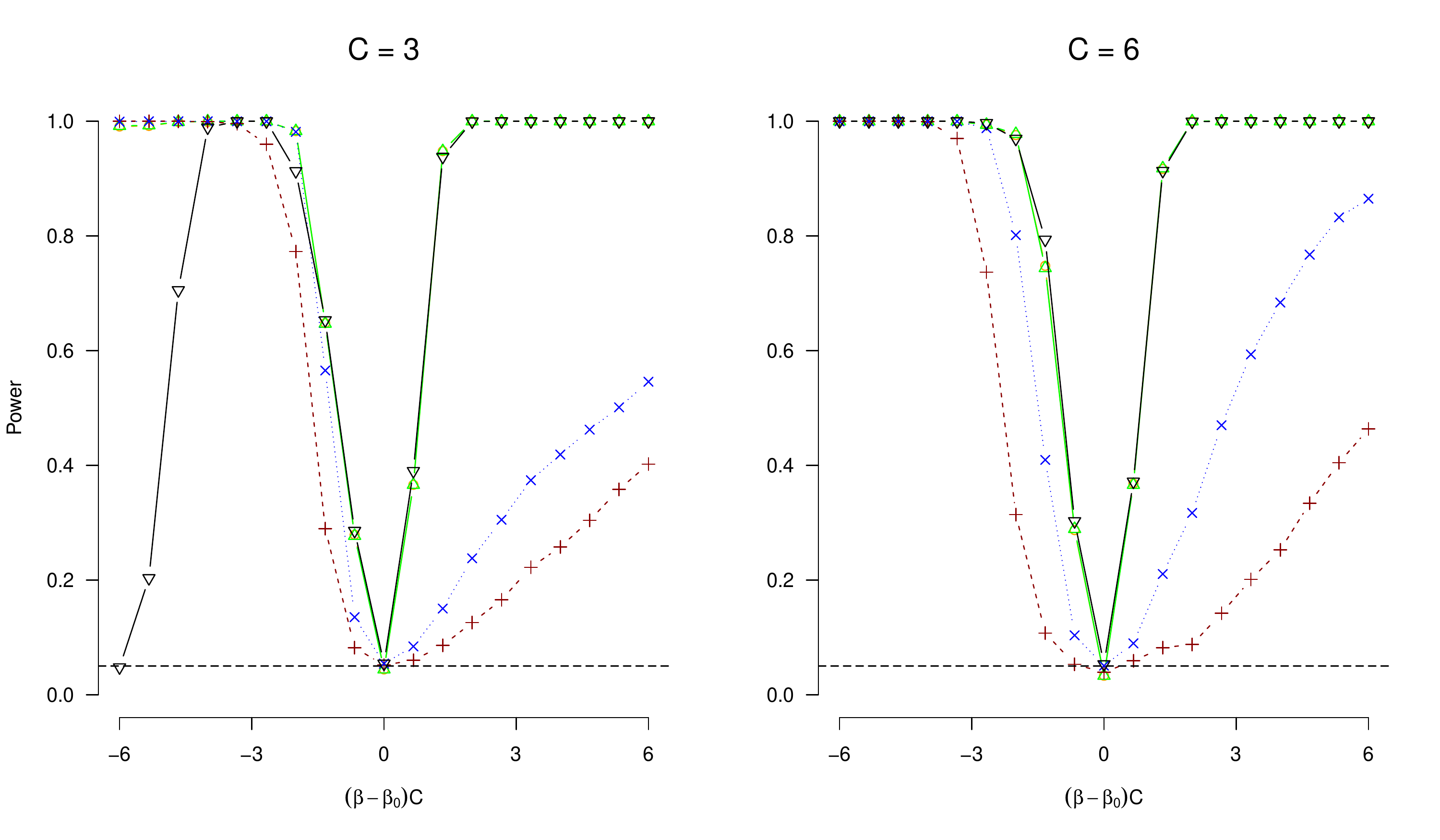}
	\caption{Power Curve for $\rho=0.9$, $p_1=0.01$, and $p_2=2$}
	\label{further_limit_fig8}
\end{figure}

\begin{figure}[H]
	\centering
	\includegraphics[width=0.9\textwidth,height = 5.85cm]{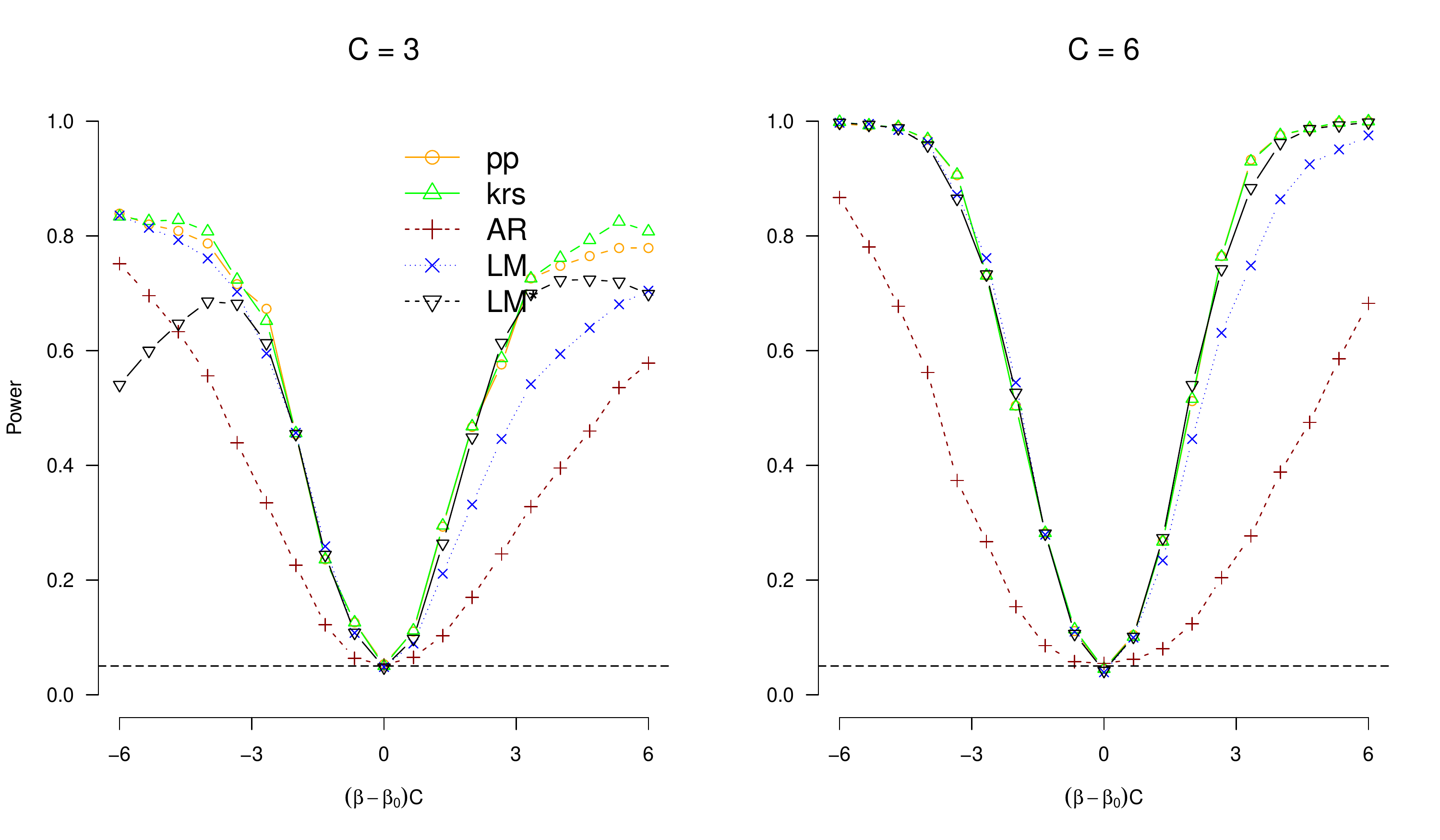}
	\caption{Power Curve for $\rho = 0.2$, $p_1=0.001$, and $p_2=1.1$}
	\label{fig_0001_11_rho02}
\end{figure}

\begin{figure}[H]
	\centering
	\includegraphics[width=0.9\textwidth,height = 5.85cm]{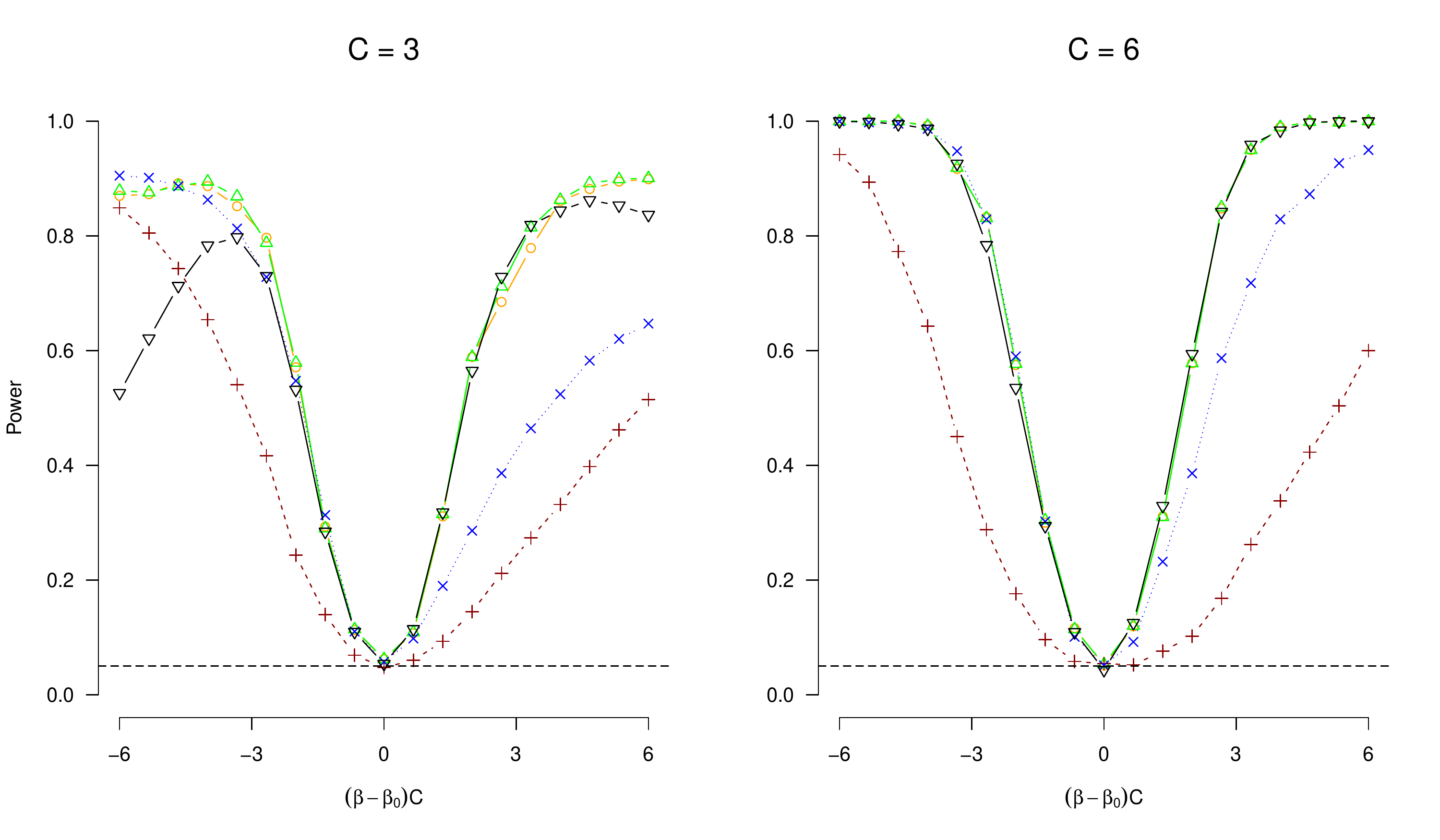}
	\caption{Power Curve for $\rho=0.4$, $p_1=0.001$, and $p_2=1.1$}
	\label{fig_0001_11_rho04}
\end{figure}

\begin{figure}[H]
	\centering
	\includegraphics[width=0.9\textwidth,height = 5.85cm]{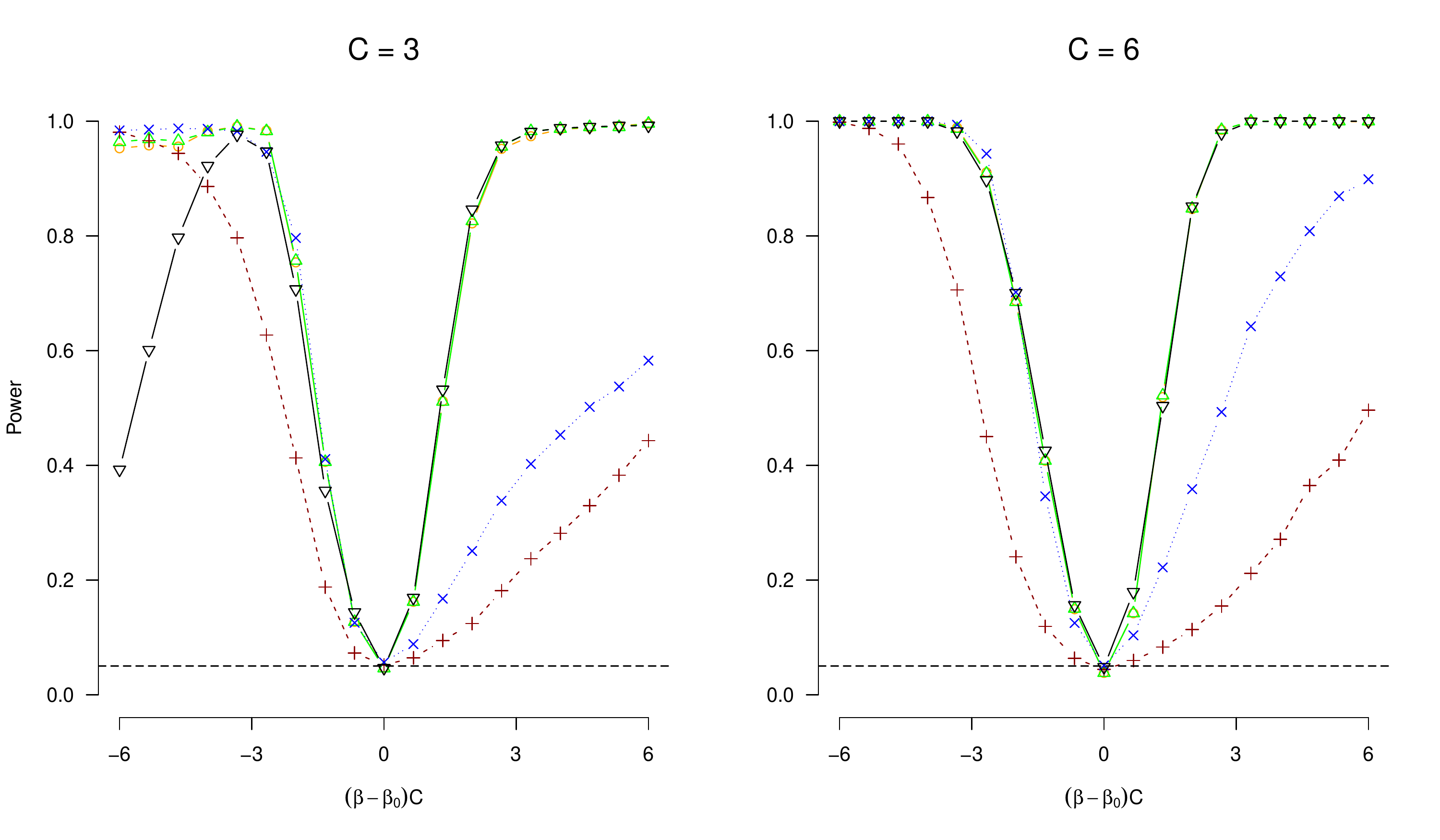}
	\caption{Power Curve for $\rho=0.7$, $p_1=0.001$, and $p_2=1.1$}
	\label{fig_0001_11_rho07}
\end{figure}

\begin{figure}[H]
	\centering
	\includegraphics[width=0.9\textwidth,height = 5.85cm]{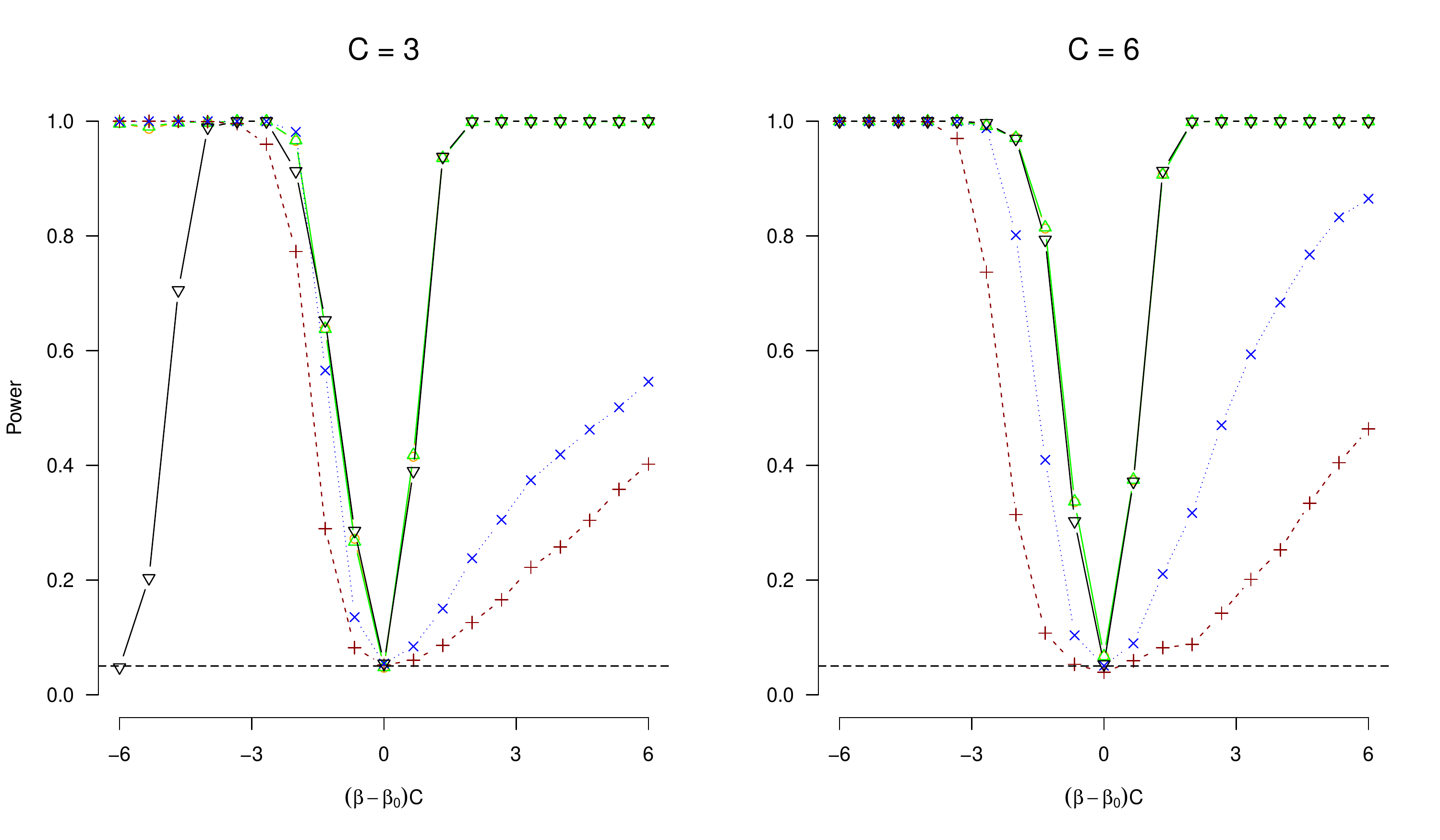}
	\caption{Power Curve for $\rho=0.9$, $p_1=0.001$, and $p_2=1.1$}
	\label{fig_0001_11_rho09}
\end{figure}

\begin{figure}[H]
	\centering
	\includegraphics[width=0.9\textwidth,height = 5.85cm]{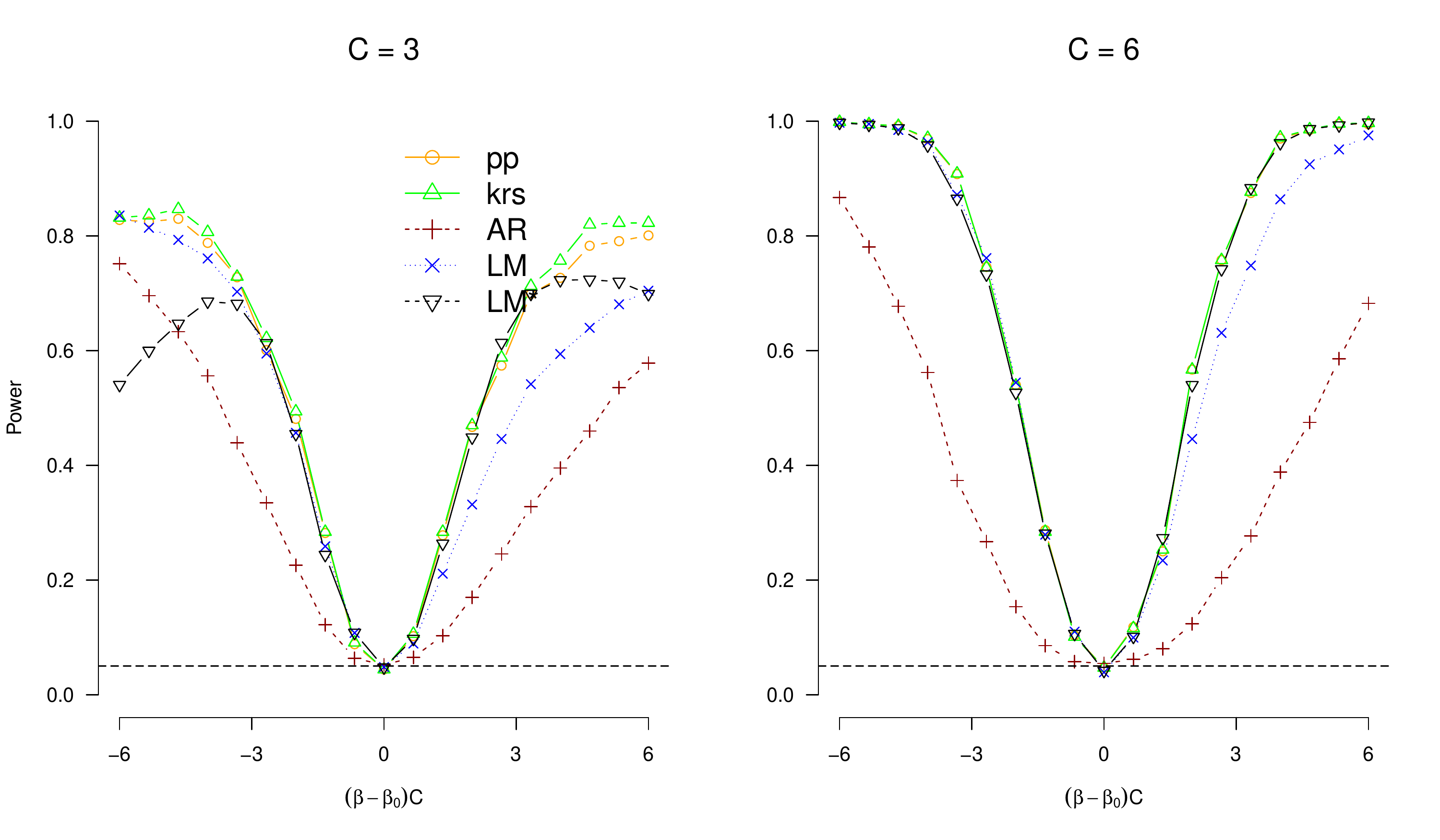}
	\caption{Power Curve for $\rho = 0.2$, $p_1=0.001$, and $p_2=1.5$}
	\label{fig_0001_15_rho02}
\end{figure}

\begin{figure}[H]
	\centering
	\includegraphics[width=0.9\textwidth,height = 5.85cm]{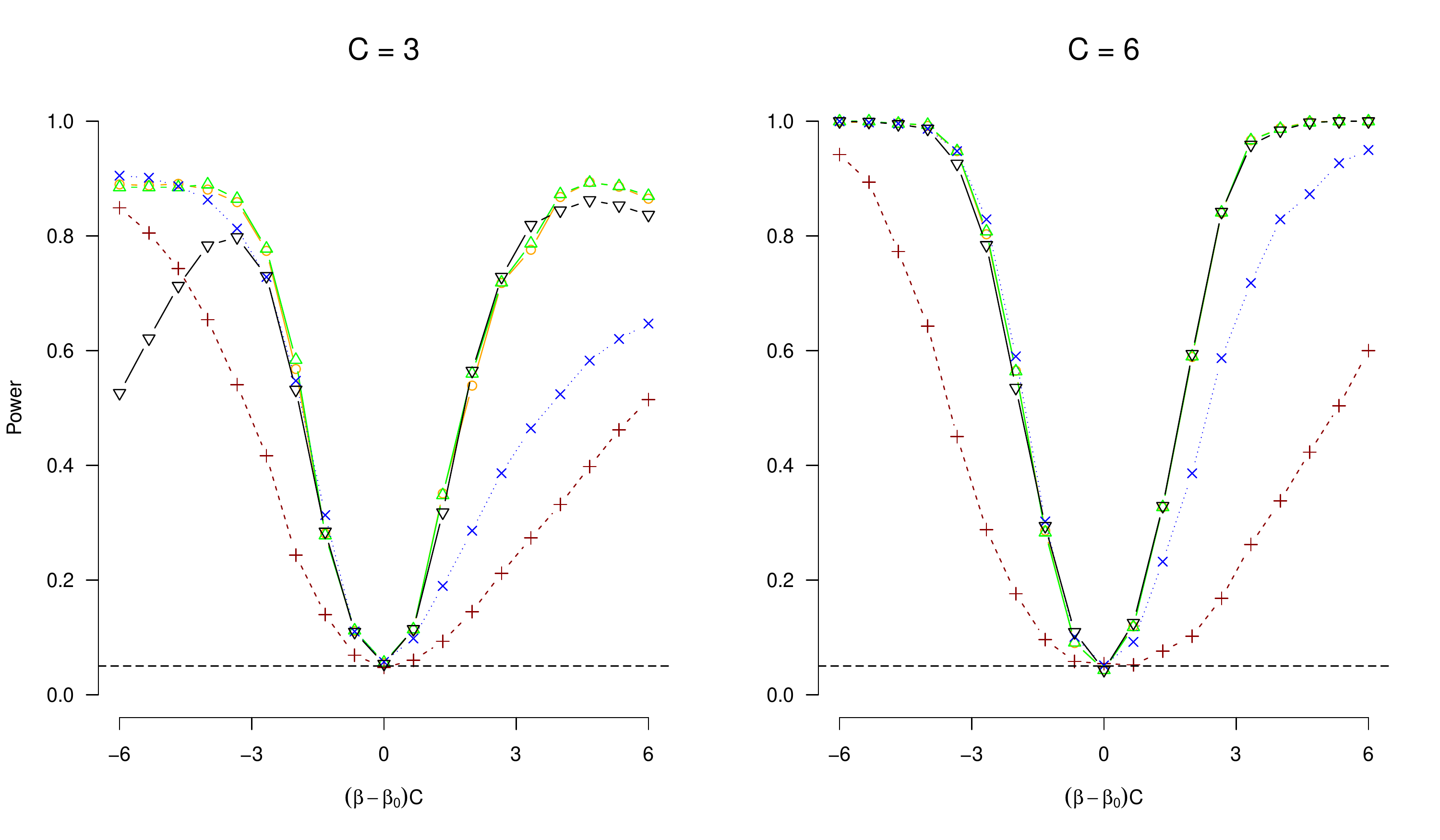}
	\caption{Power Curve for $\rho=0.4$, $p_1=0.001$, and $p_2=1.5$}
	\label{fig_0001_15_rho04}
\end{figure}

\begin{figure}[H]
	\centering
	\includegraphics[width=0.9\textwidth,height = 5.85cm]{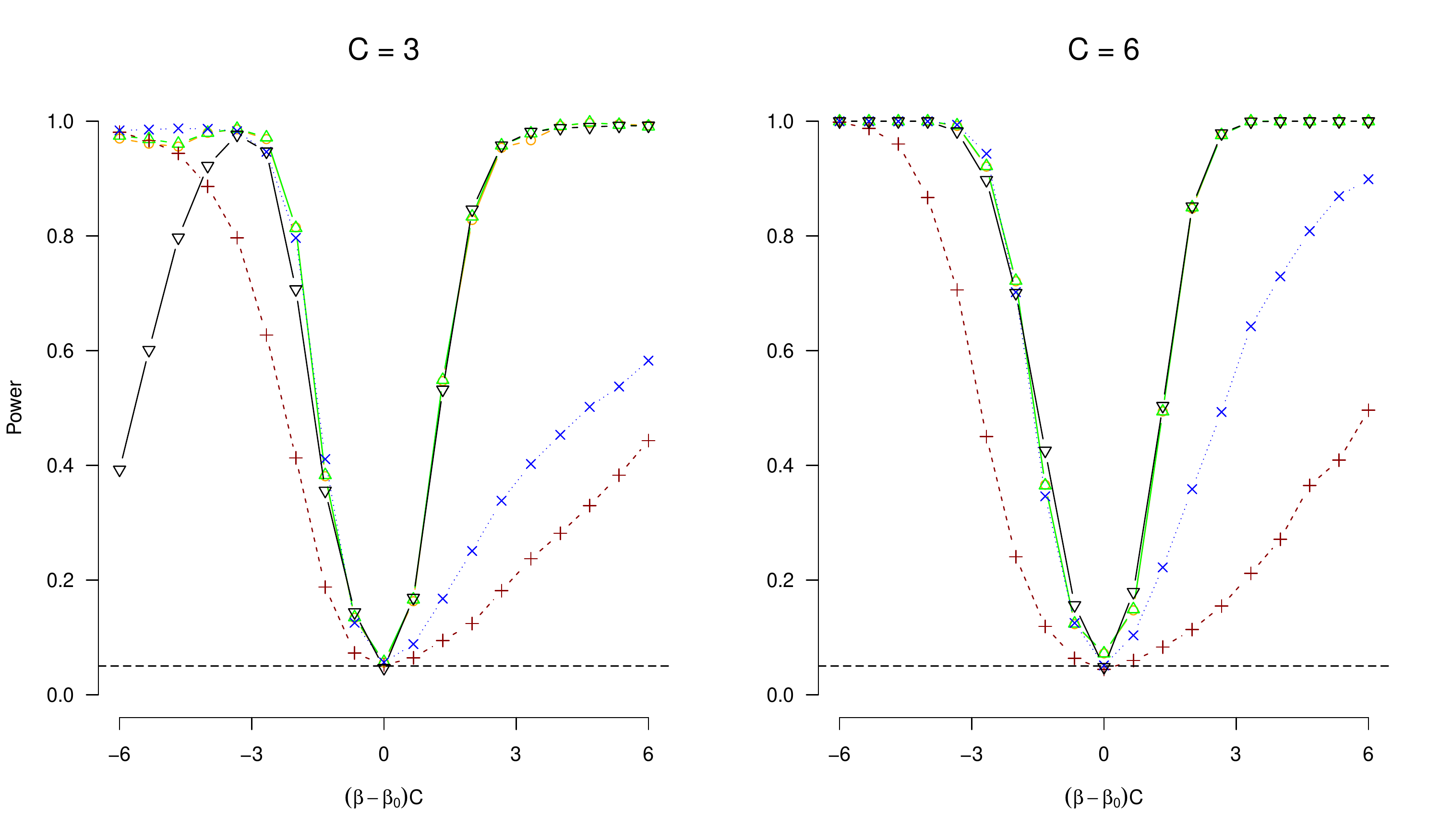}
	\caption{Power Curve for $\rho=0.7$, $p_1=0.001$, and $p_2=1.5$}
	\label{fig_0001_15_rho07}
\end{figure}

\begin{figure}[H]
	\centering
	\includegraphics[width=0.9\textwidth,height = 5.85cm]{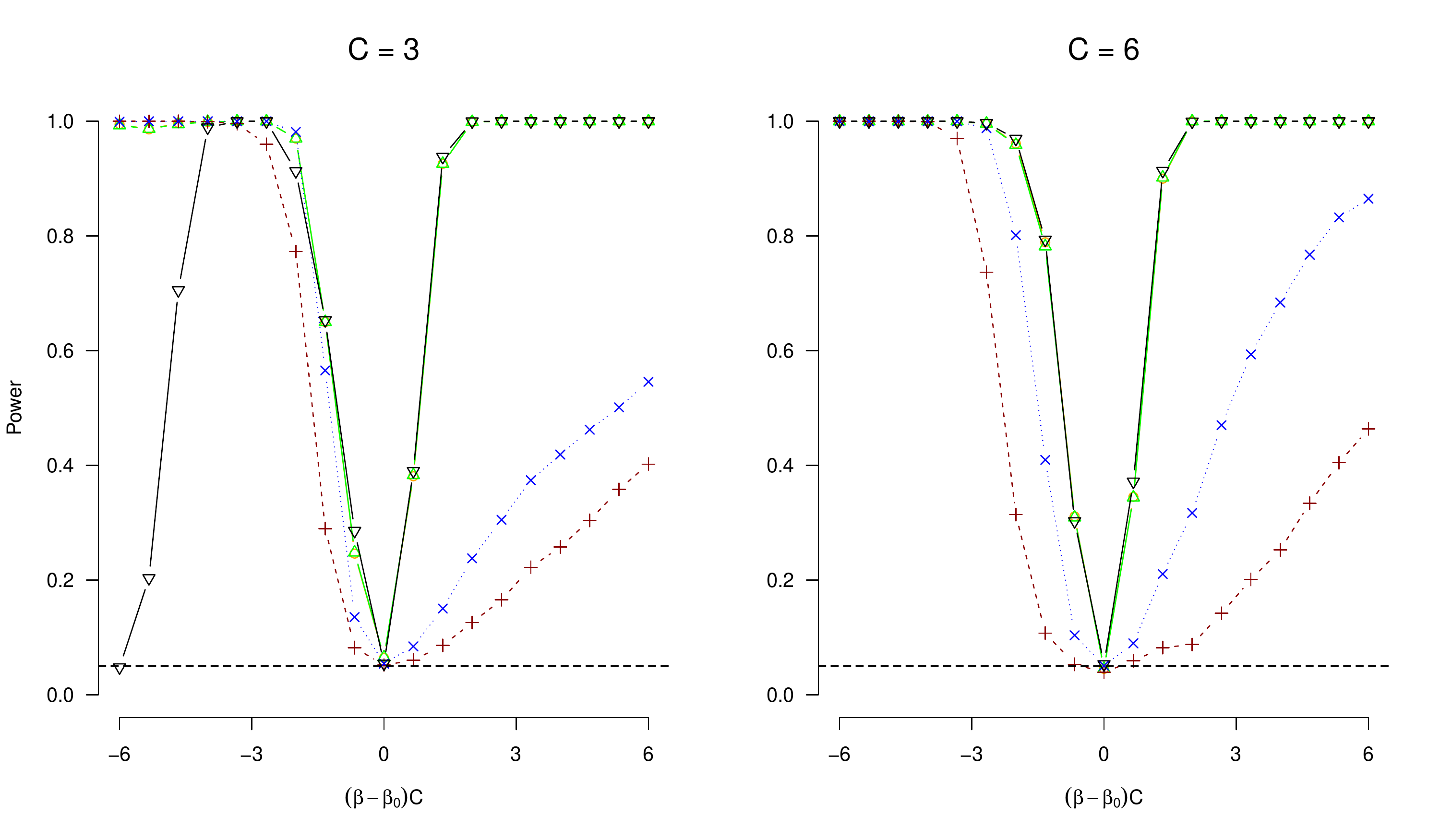}
	\caption{Power Curve for $\rho=0.9$, $p_1=0.001$, and $p_2=1.5$}
	\label{fig_0001_15_rho09}
\end{figure}

\begin{figure}[H]
	\centering
	\includegraphics[width=0.9\textwidth,height = 5.85cm]{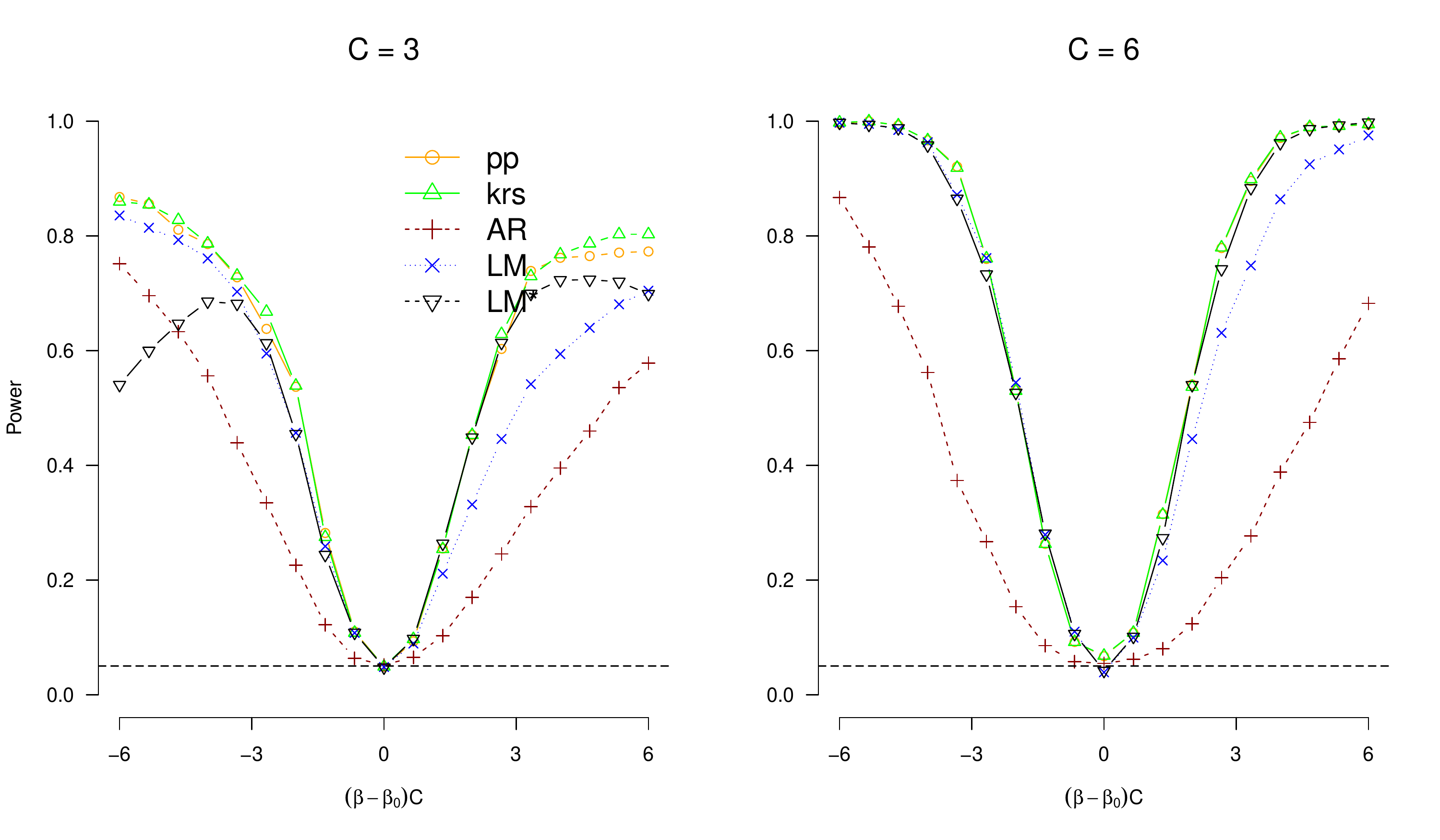}
	\caption{Power Curve for $\rho = 0.2$, $p_1=0.001$, and $p_2=2$}
	\label{fig_0001_2_rho02}
\end{figure}

\begin{figure}[H]
	\centering
	\includegraphics[width=0.9\textwidth,height = 5.85cm]{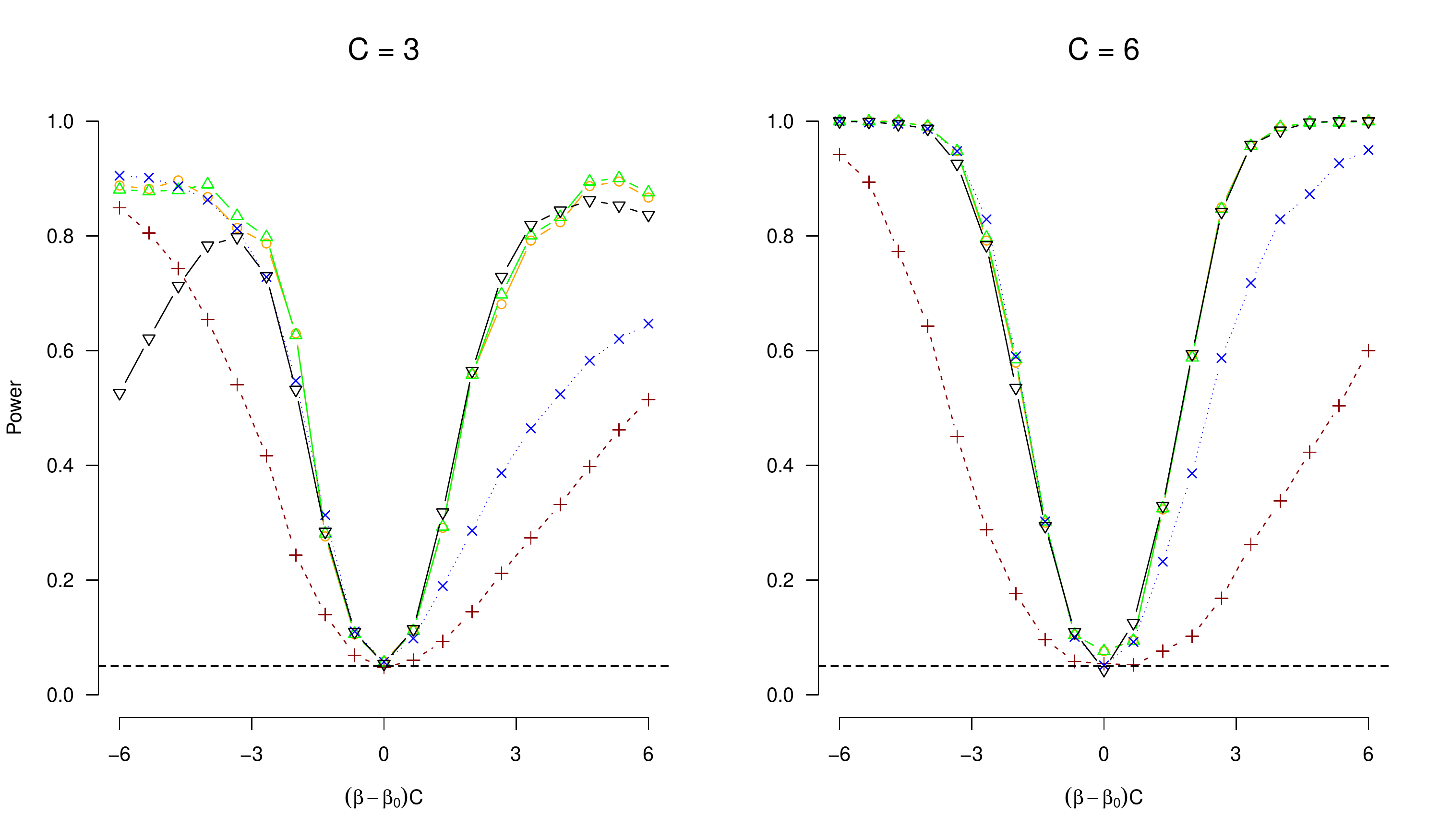}
	\caption{Power Curve for $\rho=0.4$, $p_1=0.001$, and $p_2=2$}
	\label{fig_0001_2_rho04}
\end{figure}

\begin{figure}[H]
	\centering
	\includegraphics[width=0.9\textwidth,height = 5.85cm]{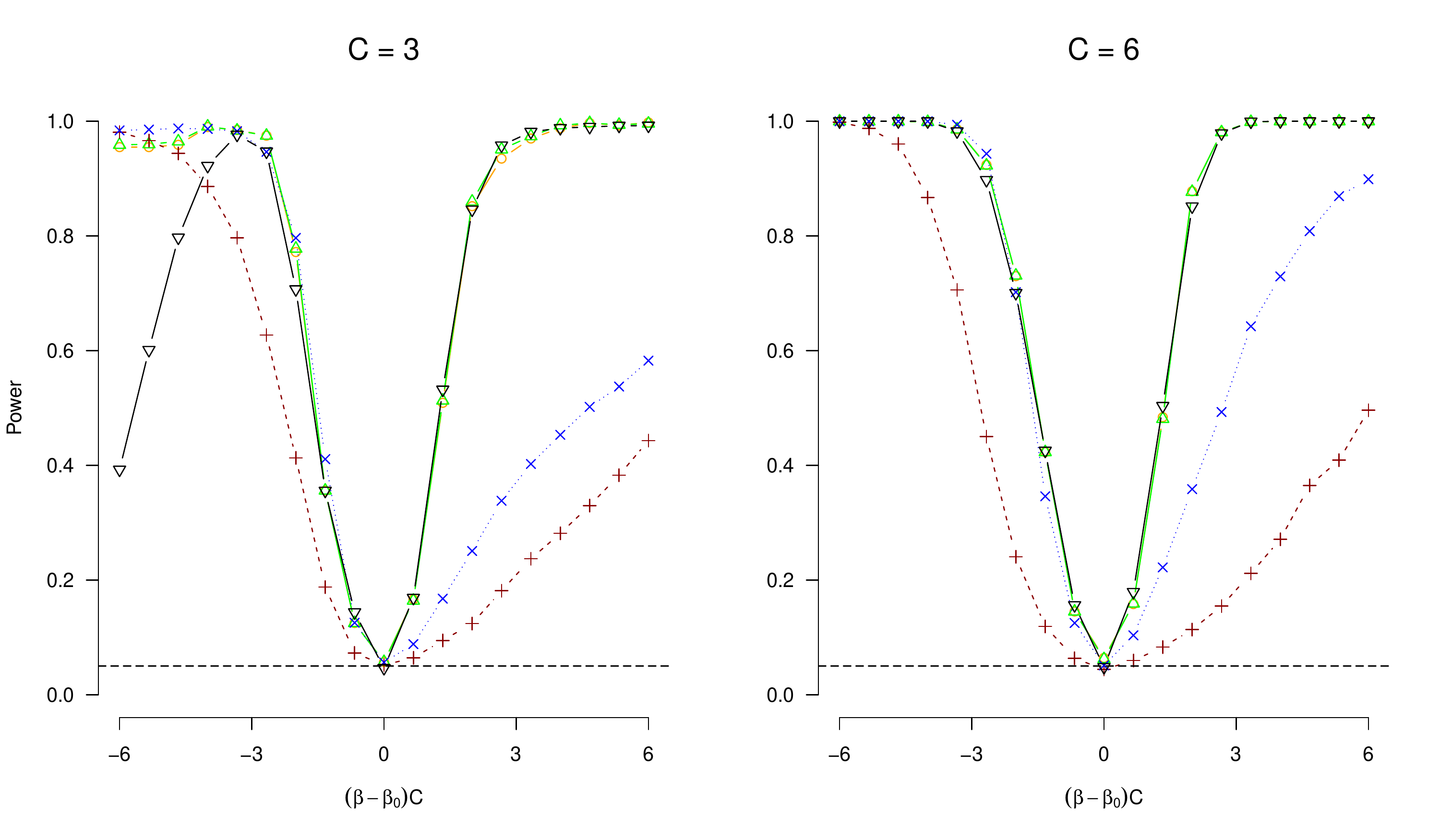}
	\caption{Power Curve for $\rho=0.7$, $p_1=0.001$, and $p_2=2$}
	\label{fig_0001_2_rho07}
\end{figure}

\begin{figure}[H]
	\centering
	\includegraphics[width=0.9\textwidth,height = 5.85cm]{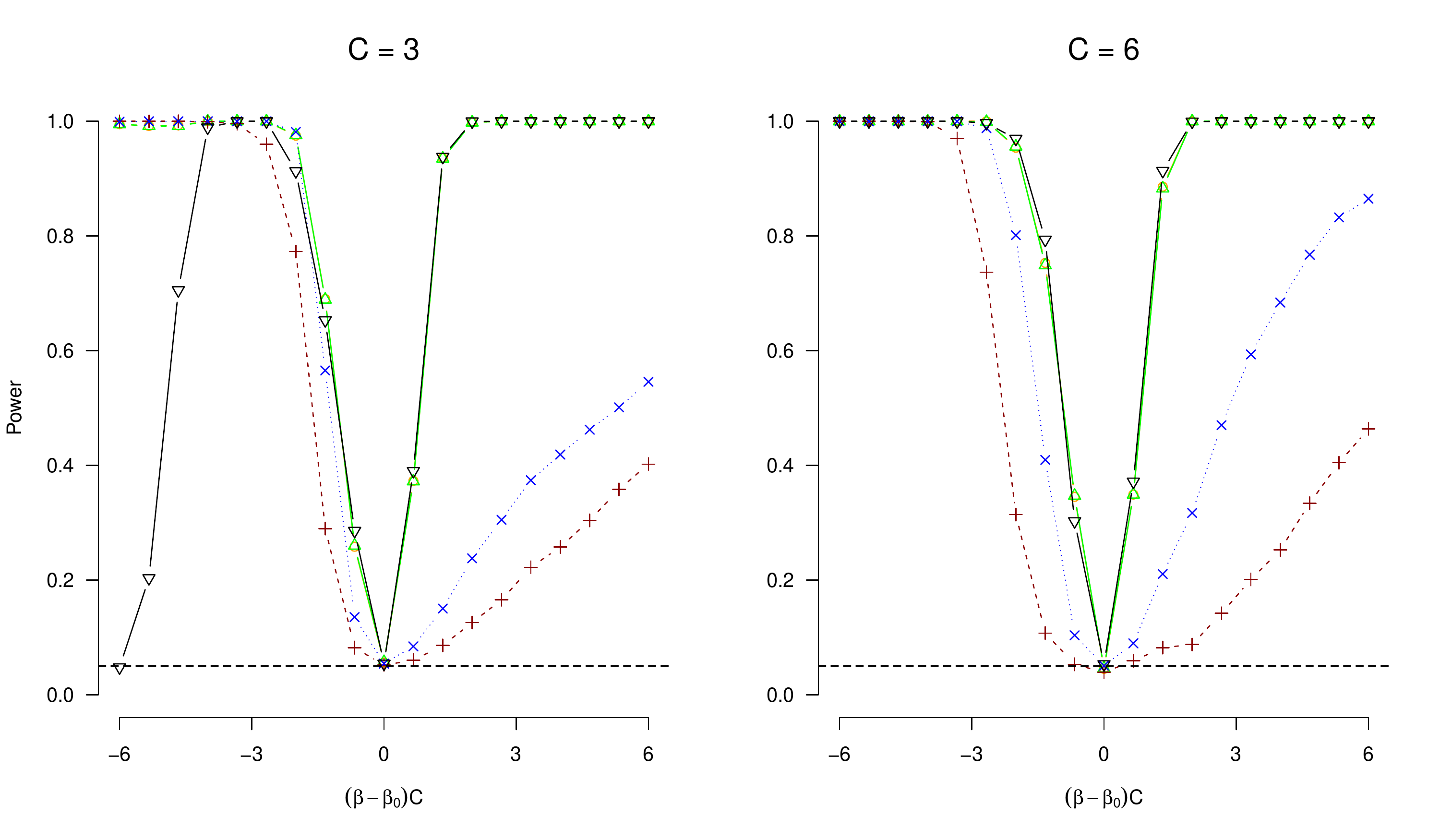}
	\caption{Power Curve for $\rho=0.9$, $p_1=0.001$, and $p_2=2$}
	\label{fig_0001_2_rho09}
\end{figure}

\begin{figure}[H]
	\centering
	\includegraphics[width=0.9\textwidth,height = 5.85cm]{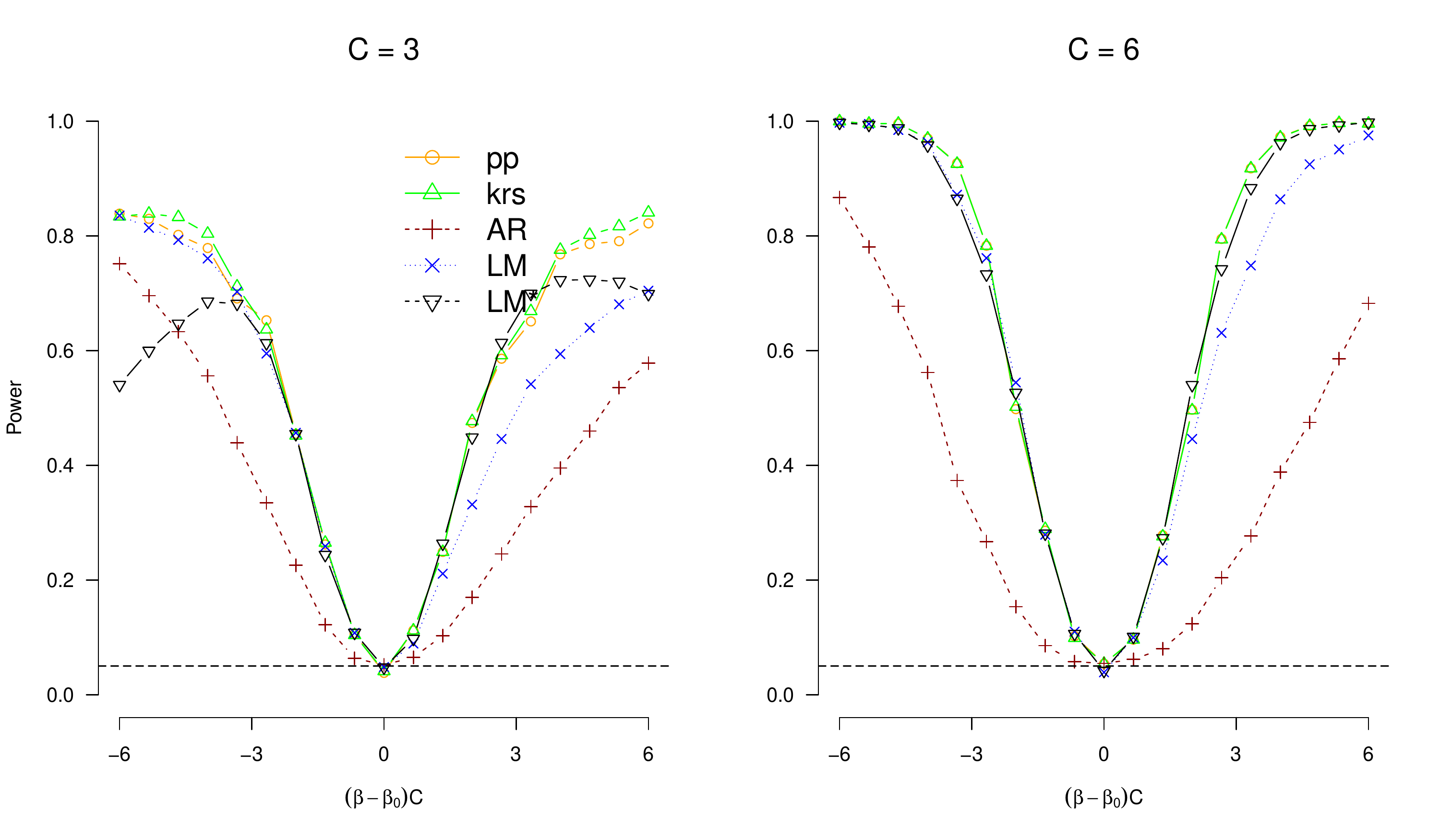}
	\caption{Power Curve for $\rho = 0.2$, $p_1=0.1$, and $p_2=1.1$}
	\label{fig_01_11_rho02}
\end{figure}

\begin{figure}[H]
	\centering
	\includegraphics[width=0.9\textwidth,height = 5.85cm]{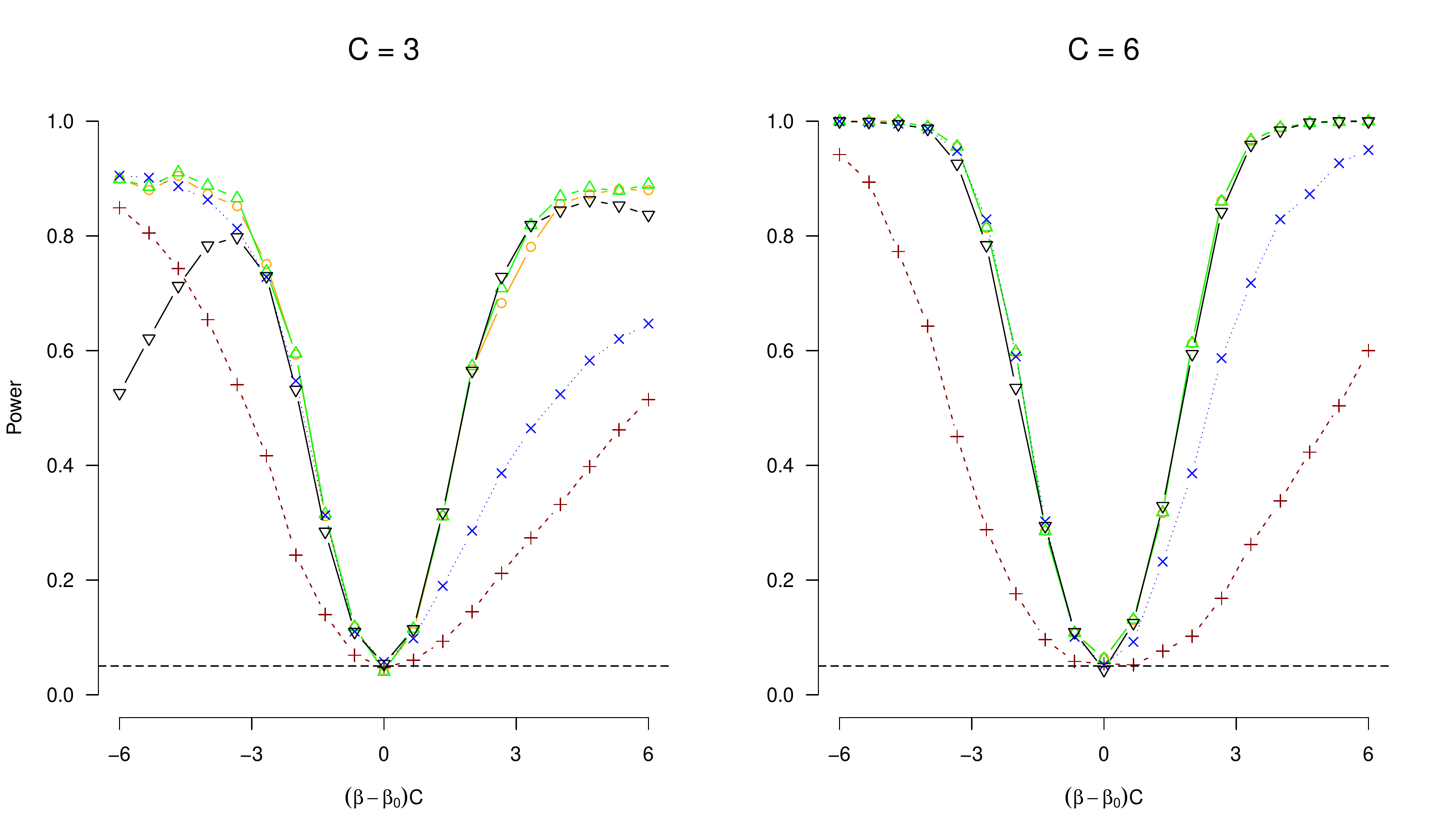}
	\caption{Power Curve for $\rho=0.4$, $p_1=0.1$, and $p_2=1.1$}
	\label{fig_01_11_rho04}
\end{figure}

\begin{figure}[H]
	\centering
	\includegraphics[width=0.9\textwidth,height = 5.85cm]{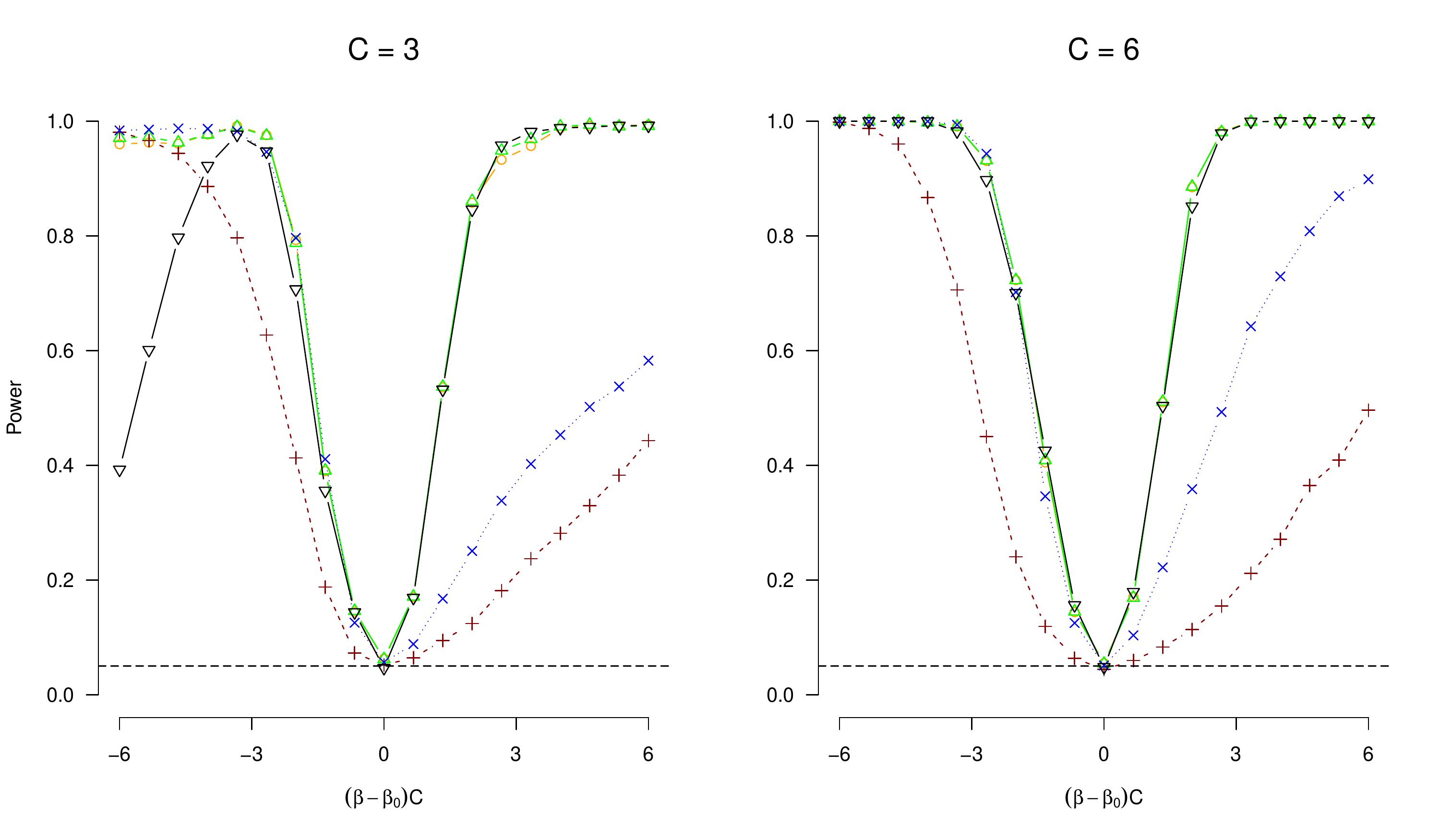}
	\caption{Power Curve for $\rho=0.7$, $p_1=0.1$, and $p_2=1.1$}
	\label{fig_01_11_rho07}
\end{figure}

\begin{figure}[H]
	\centering
	\includegraphics[width=0.9\textwidth,height = 5.85cm]{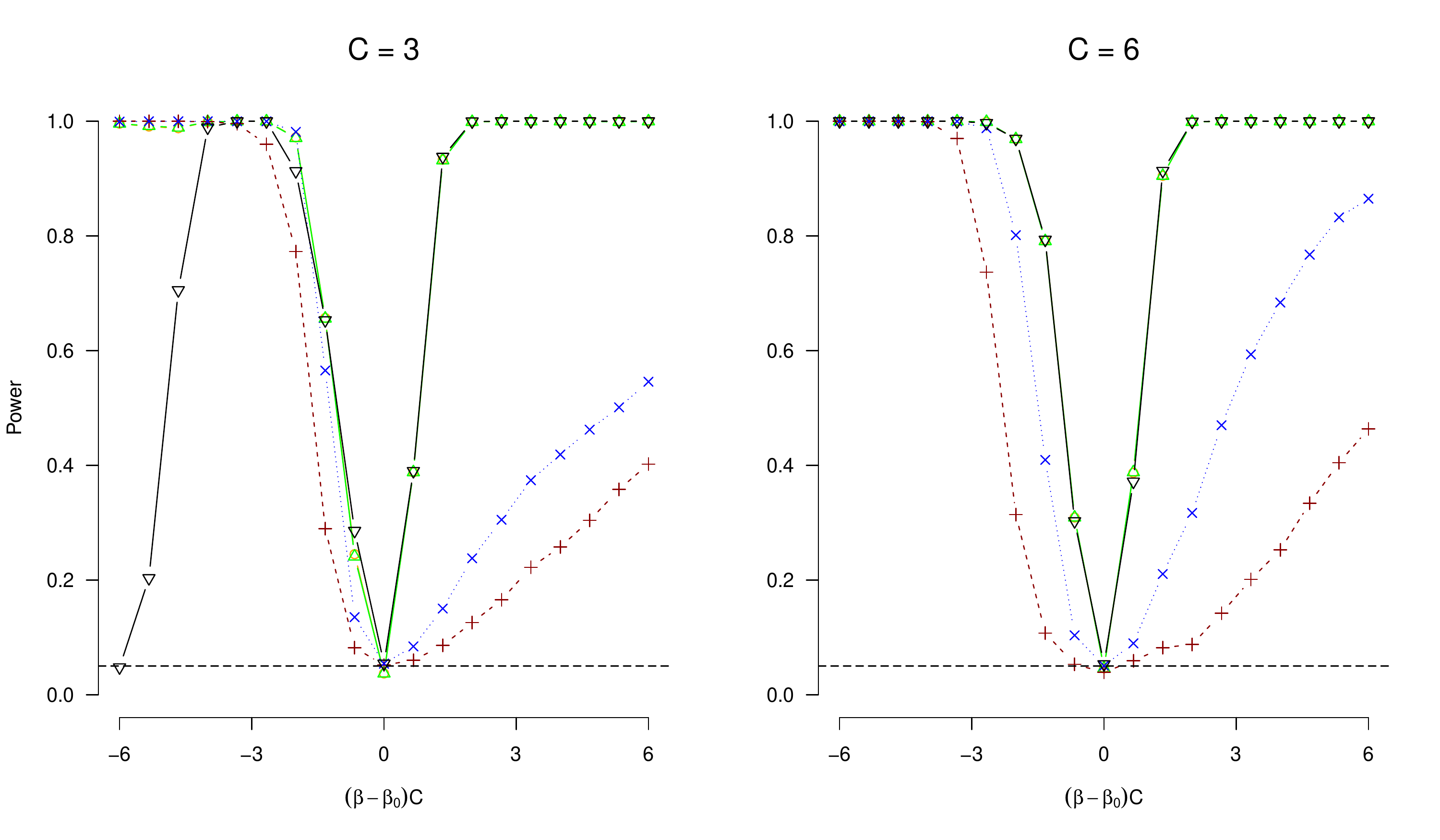}
	\caption{Power Curve for $\rho=0.9$, $p_1=0.1$, and $p_2=1.1$}
	\label{fig_01_11_rho09}
\end{figure}

\begin{figure}[H]
	\centering
	\includegraphics[width=0.9\textwidth,height = 5.85cm]{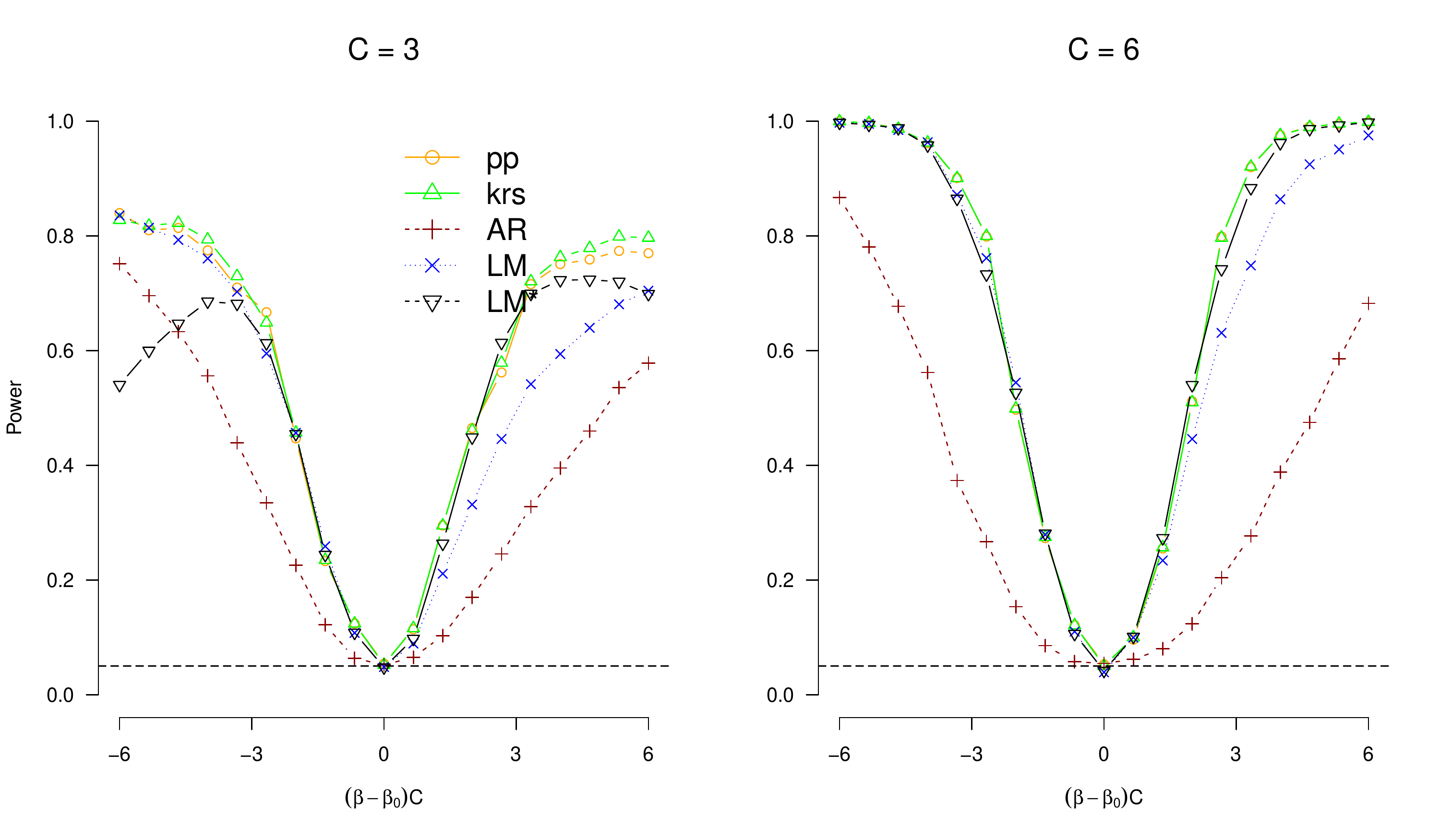}
	\caption{Power Curve for $\rho = 0.2$, $p_1=0.1$, and $p_2=1.5$}
	\label{further_limit_fig9}
\end{figure}

\begin{figure}[H]
	\centering
	\includegraphics[width=0.9\textwidth,height = 5.85cm]{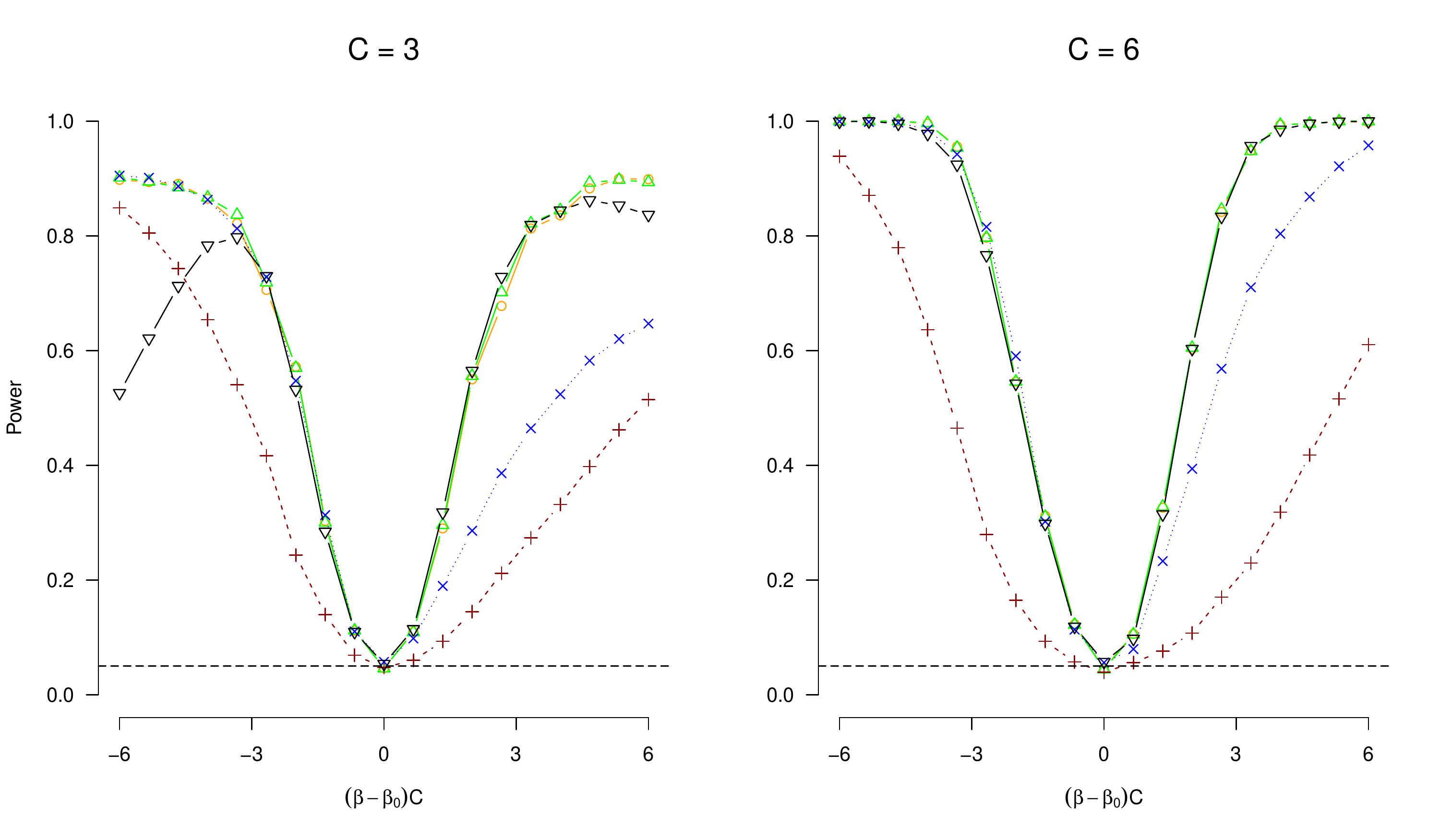}
	\caption{Power Curve for $\rho=0.4$, $p_1=0.1$, and $p_2=1.5$}
	\label{further_limit_fig10}
\end{figure}

\begin{figure}[H]
	\centering
	\includegraphics[width=0.9\textwidth,height = 5.85cm]{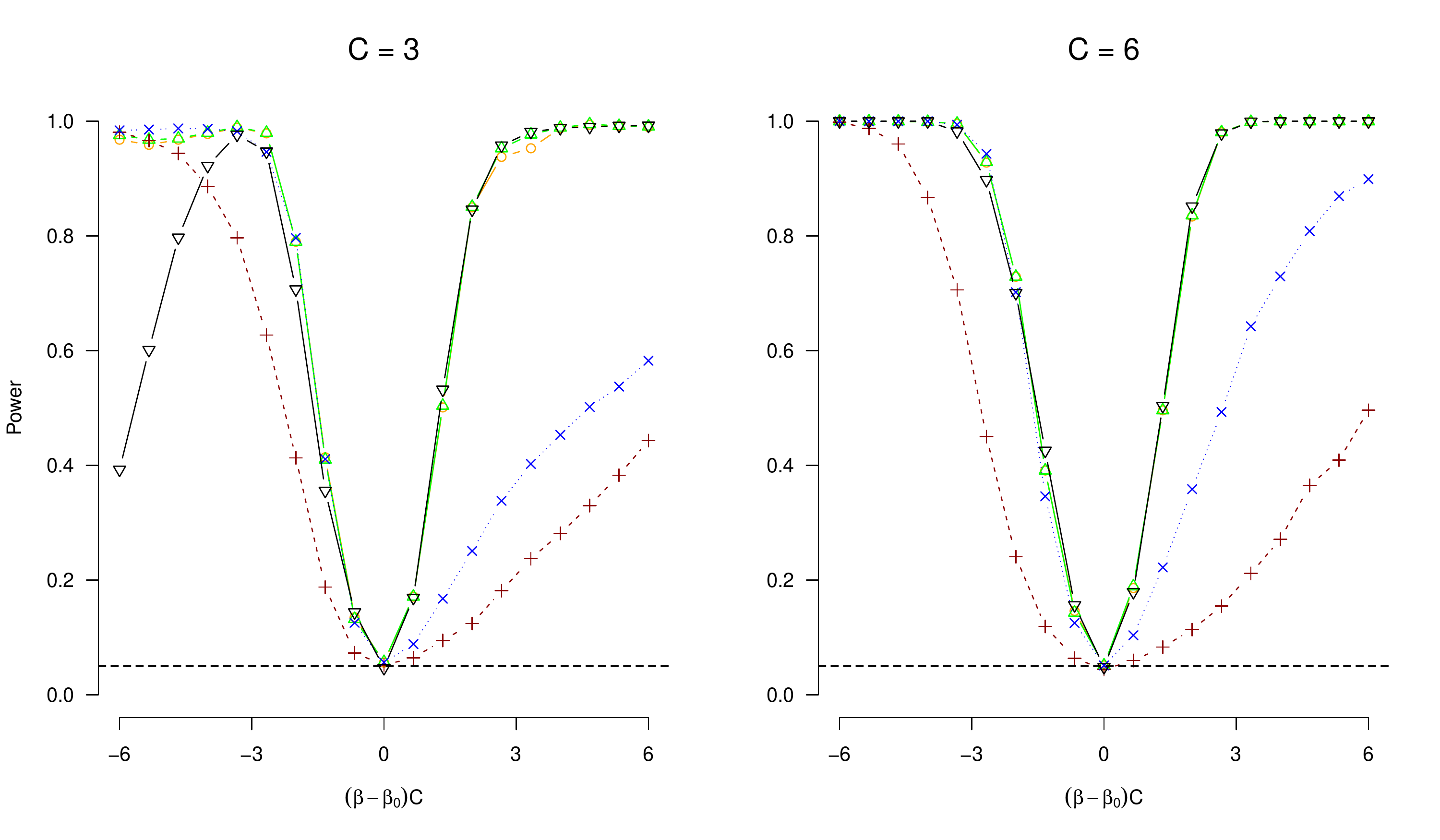}
	\caption{Power Curve for $\rho=0.7$, $p_1=0.1$, and $p_2=1.5$}
	\label{further_limit_fig11}
\end{figure}

\begin{figure}[H]
	\centering
	\includegraphics[width=0.9\textwidth,height = 5.85cm]{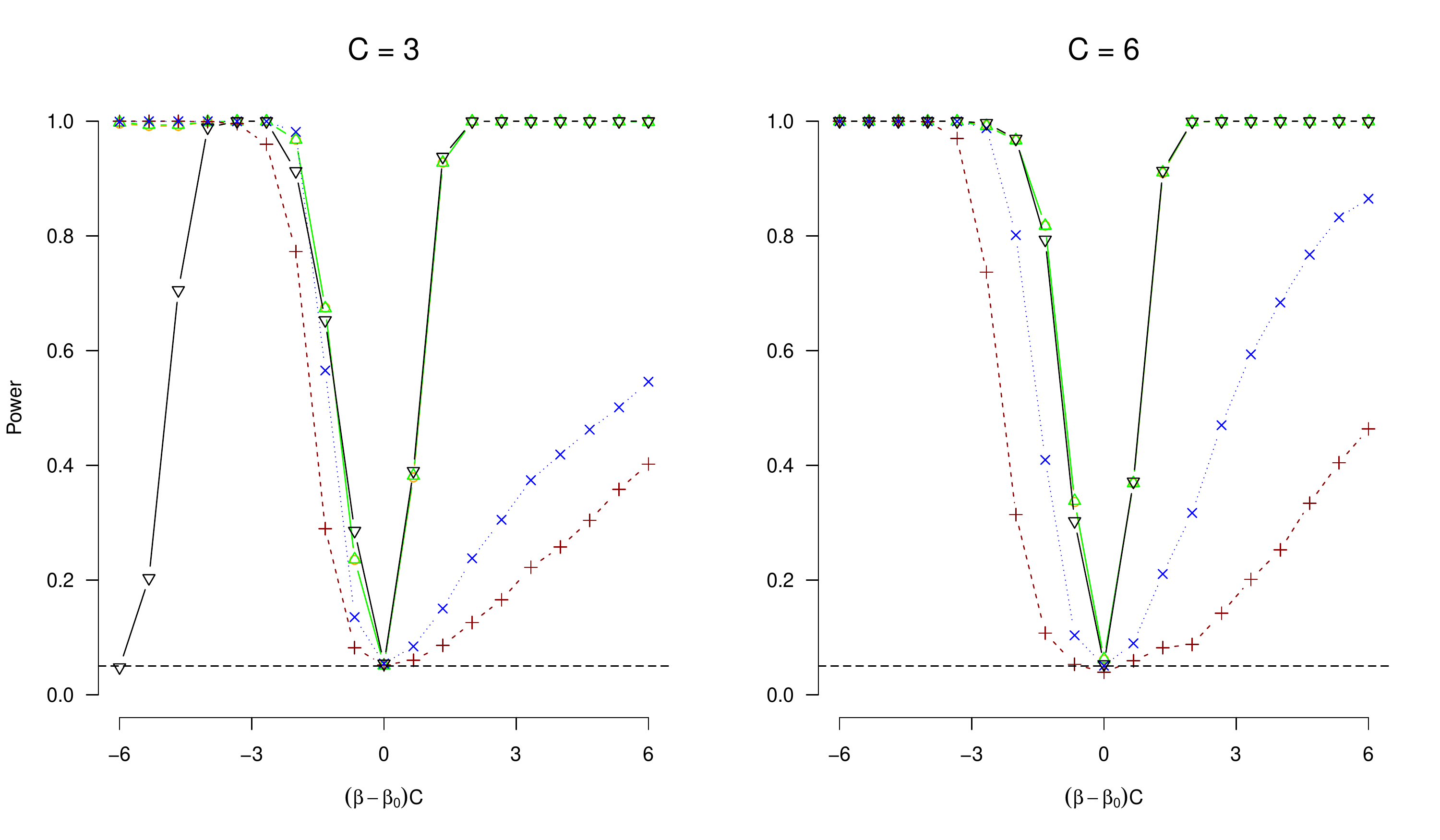}
	\caption{Power Curve for $\rho=0.9$, $p_1=0.1$, and $p_2=1.5$}
	\label{further_limit_fig12}
\end{figure}

\begin{figure}[H]
	\centering
	\includegraphics[width=0.9\textwidth,height = 5.85cm]{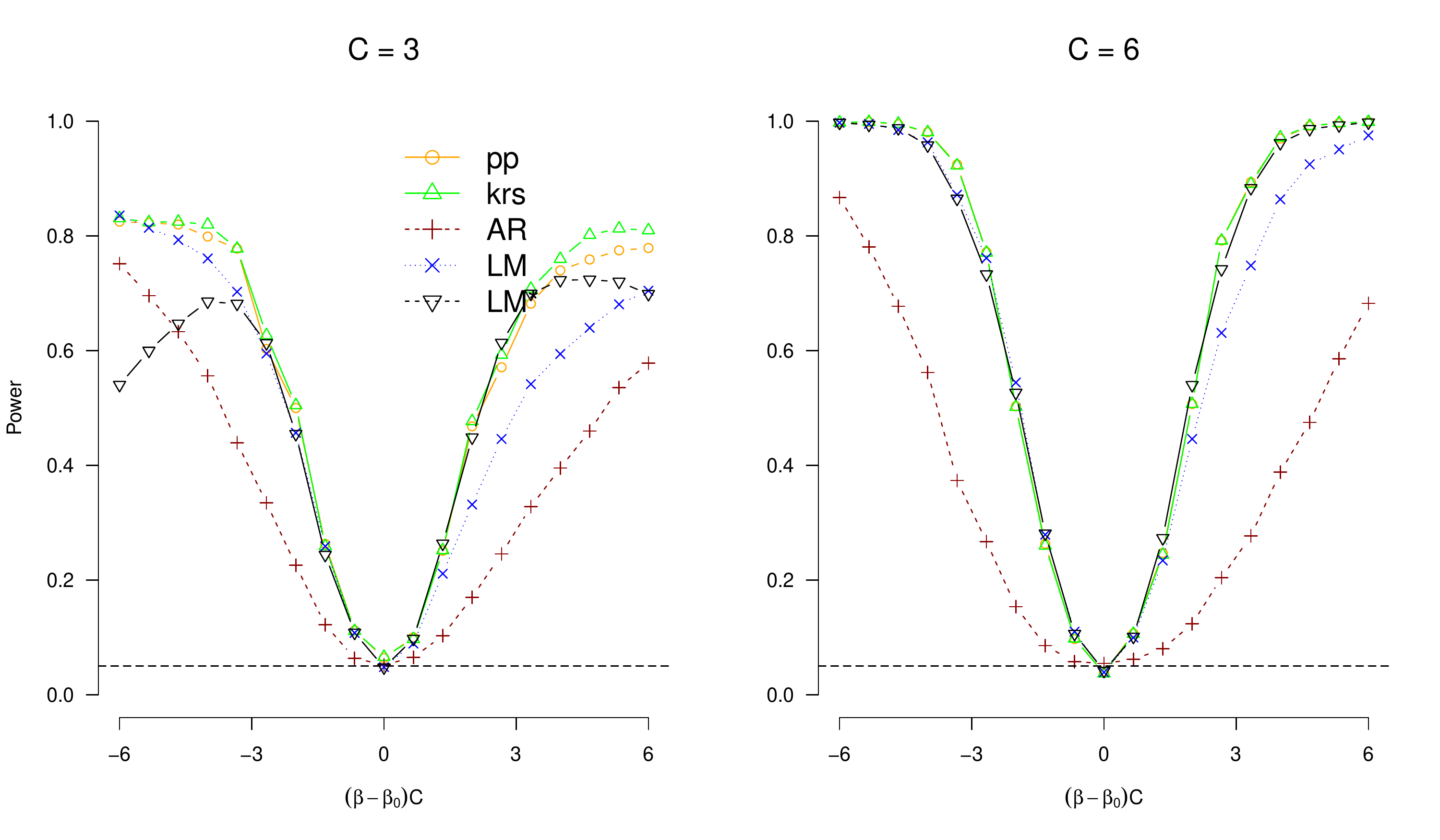}
	\caption{Power Curve for $\rho = 0.2$, $p_1=0.1$, and $p_2=2$}
	\label{fig_01_2_rho02}
\end{figure}

\begin{figure}[H]
	\centering
	\includegraphics[width=0.9\textwidth,height = 5.85cm]{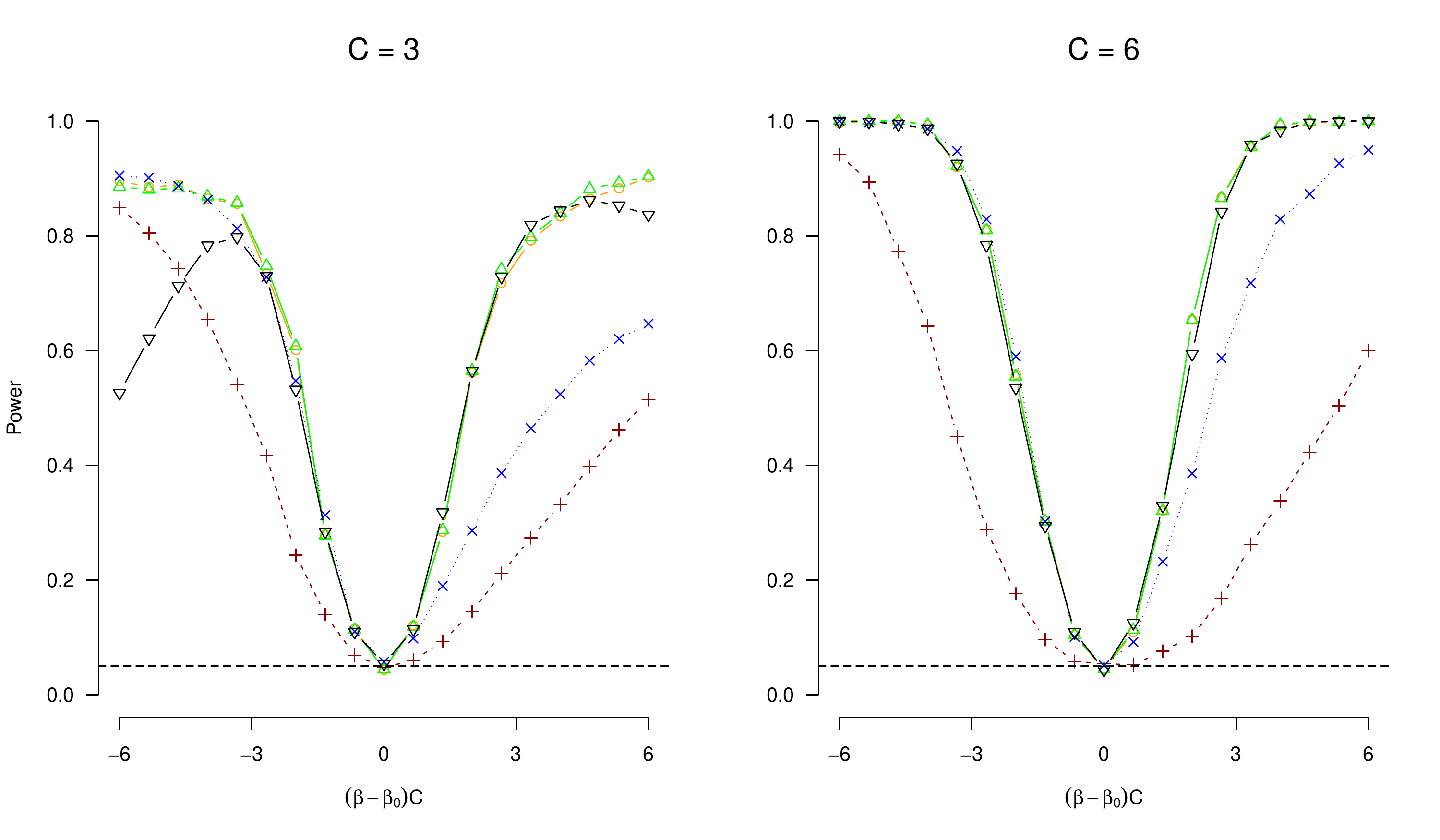}
	\caption{Power Curve for $\rho=0.4$, $p_1=0.1$, and $p_2=2$}
	\label{fig_01_2_rho04}
\end{figure}

\begin{figure}[H]
	\centering
	\includegraphics[width=0.9\textwidth,height = 5.85cm]{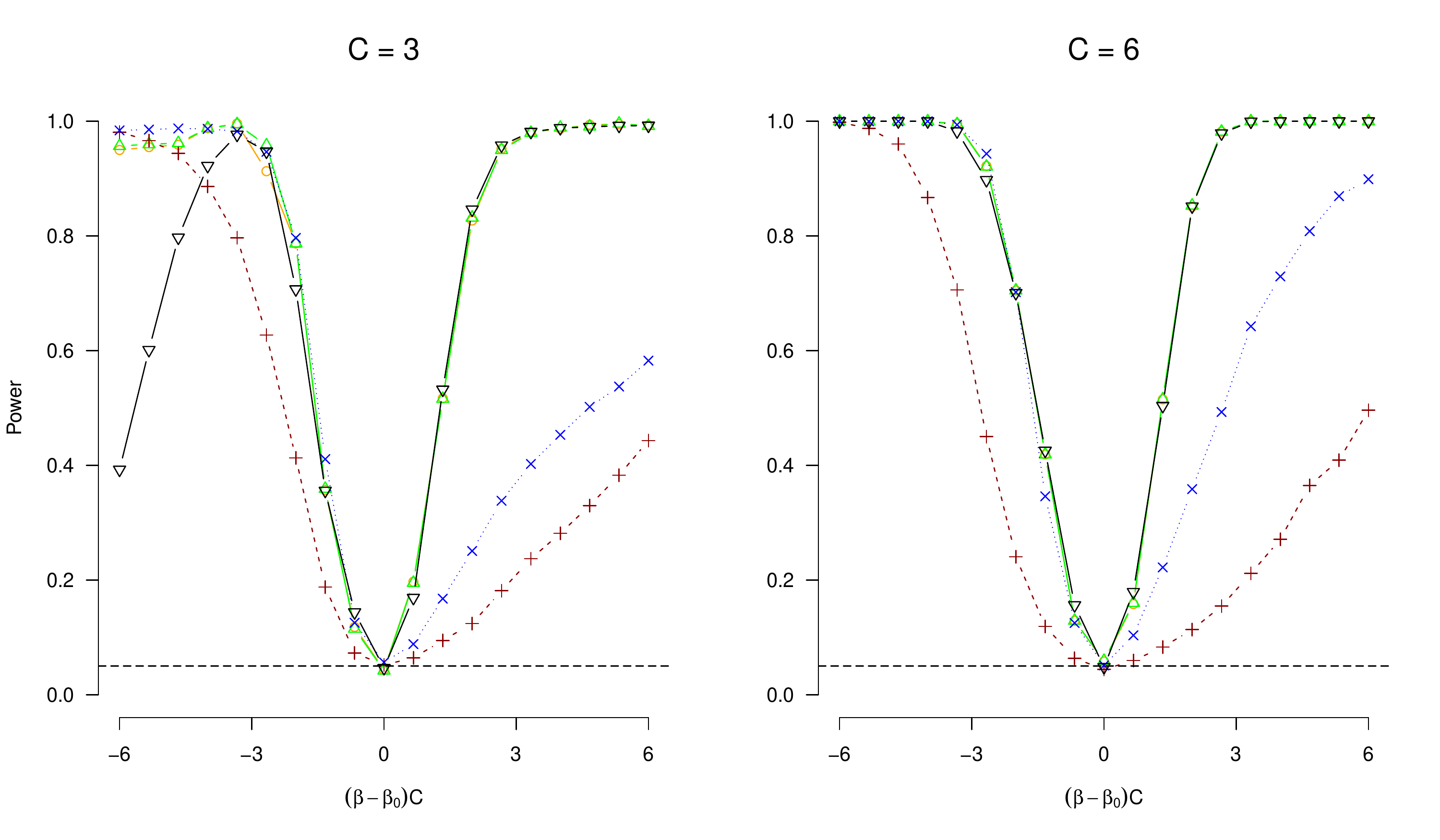}
	\caption{Power Curve for $\rho=0.7$, $p_1=0.1$, and $p_2=2$}
	\label{fig_01_2_rho07}
\end{figure}

\begin{figure}[H]
	\centering
	\includegraphics[width=0.9\textwidth,height = 5.85cm]{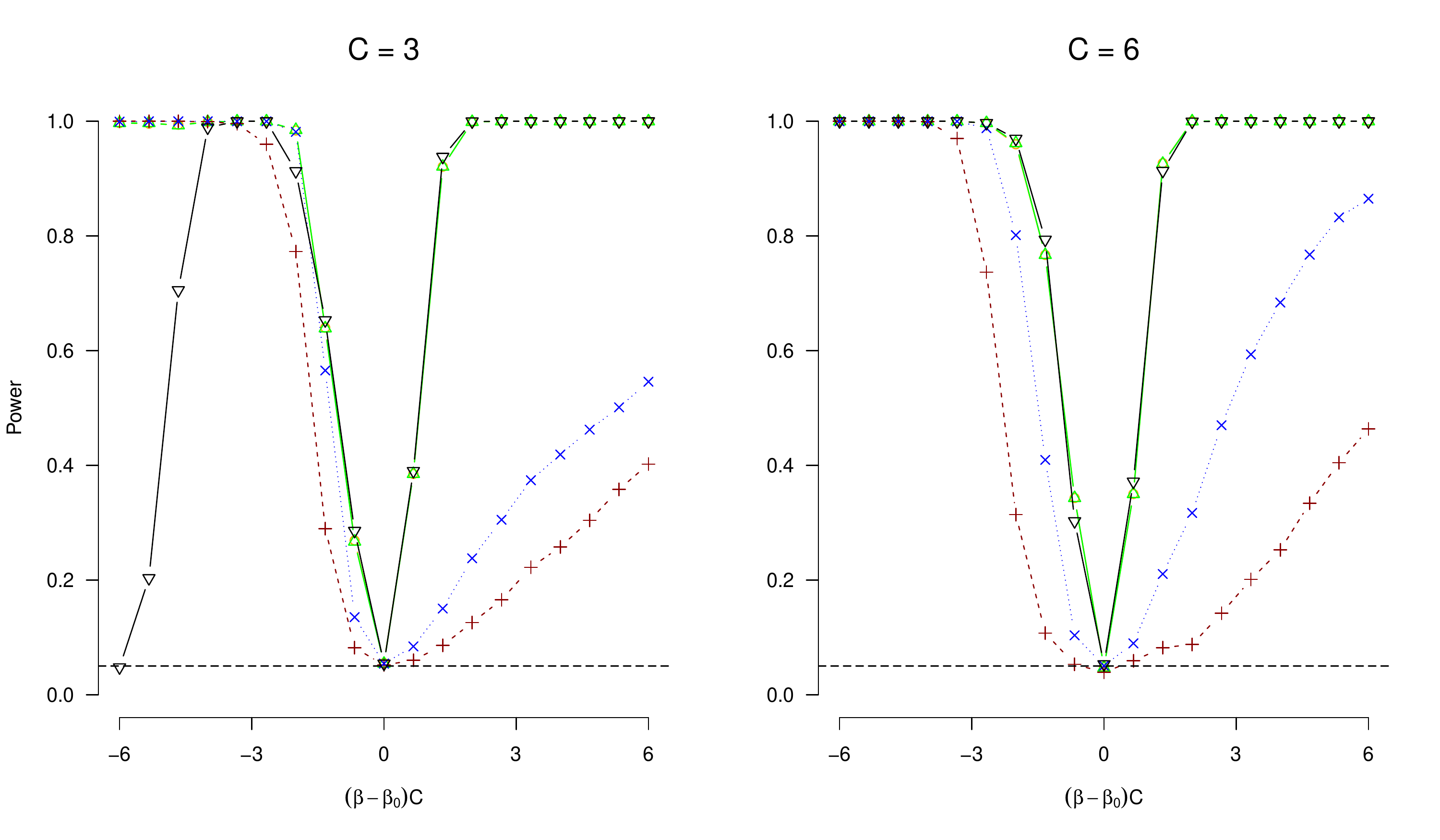}
	\caption{Power Curve for $\rho=0.9$, $p_1=0.1$, and $p_2=2$}
	\label{fig_01_2_rho09}
\end{figure}

\begin{figure}[H]
	\centering
	\includegraphics[width=0.9\textwidth,height = 5.85cm]{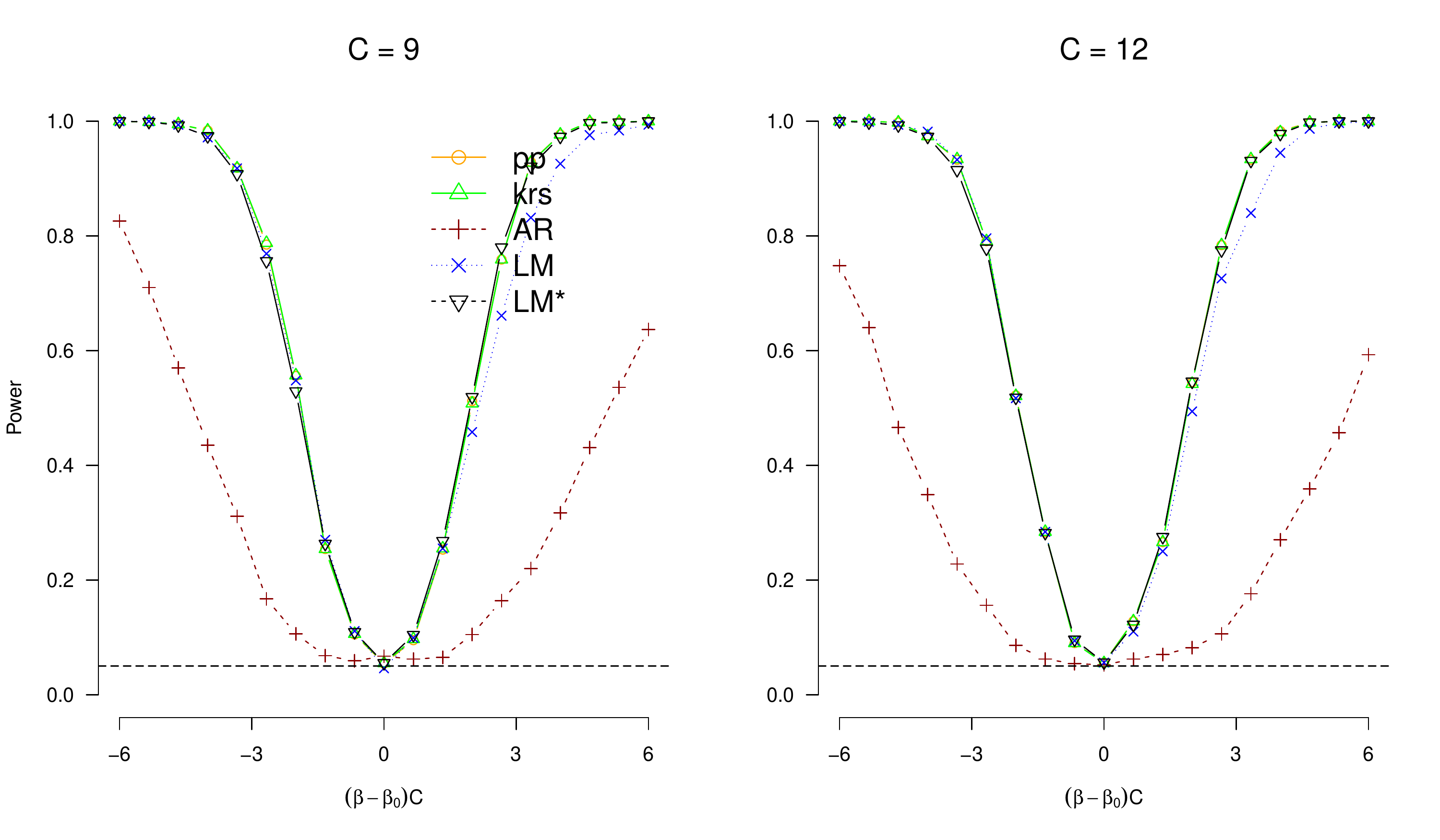}
	\caption{Power Curve for $\rho = 0.2$ with $C=9$ or $12$}
	\label{further_limit_fig13}
\end{figure}

\begin{figure}[H]
	\centering
	\includegraphics[width=0.9\textwidth,height = 5.85cm]{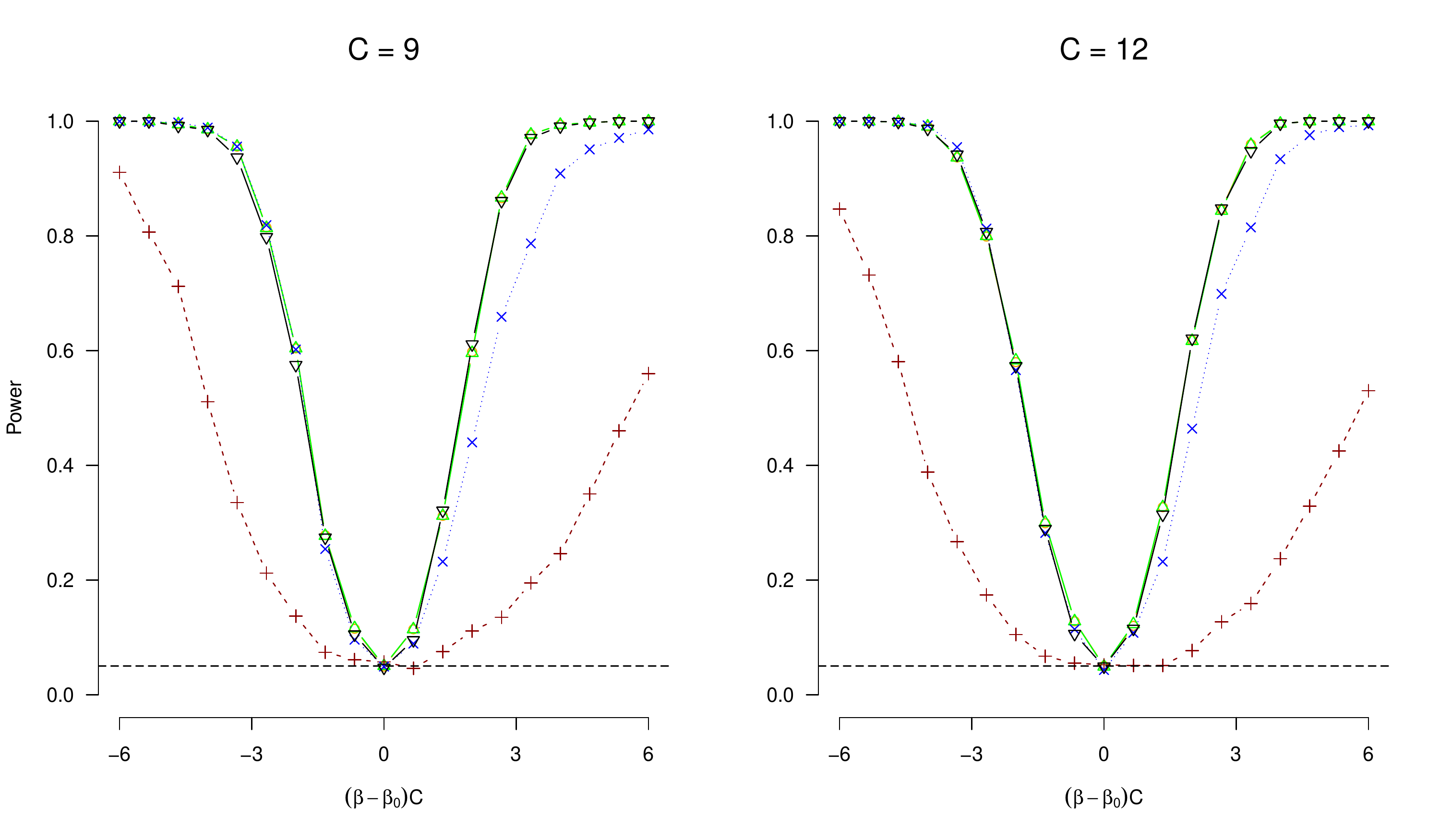}
	\caption{Power Curve for $\rho=0.4$ with $C=9$ or $12$}
	\label{further_limit_fig14}
\end{figure}

\begin{figure}[H]
	\centering
	\includegraphics[width=0.9\textwidth,height = 5.85cm]{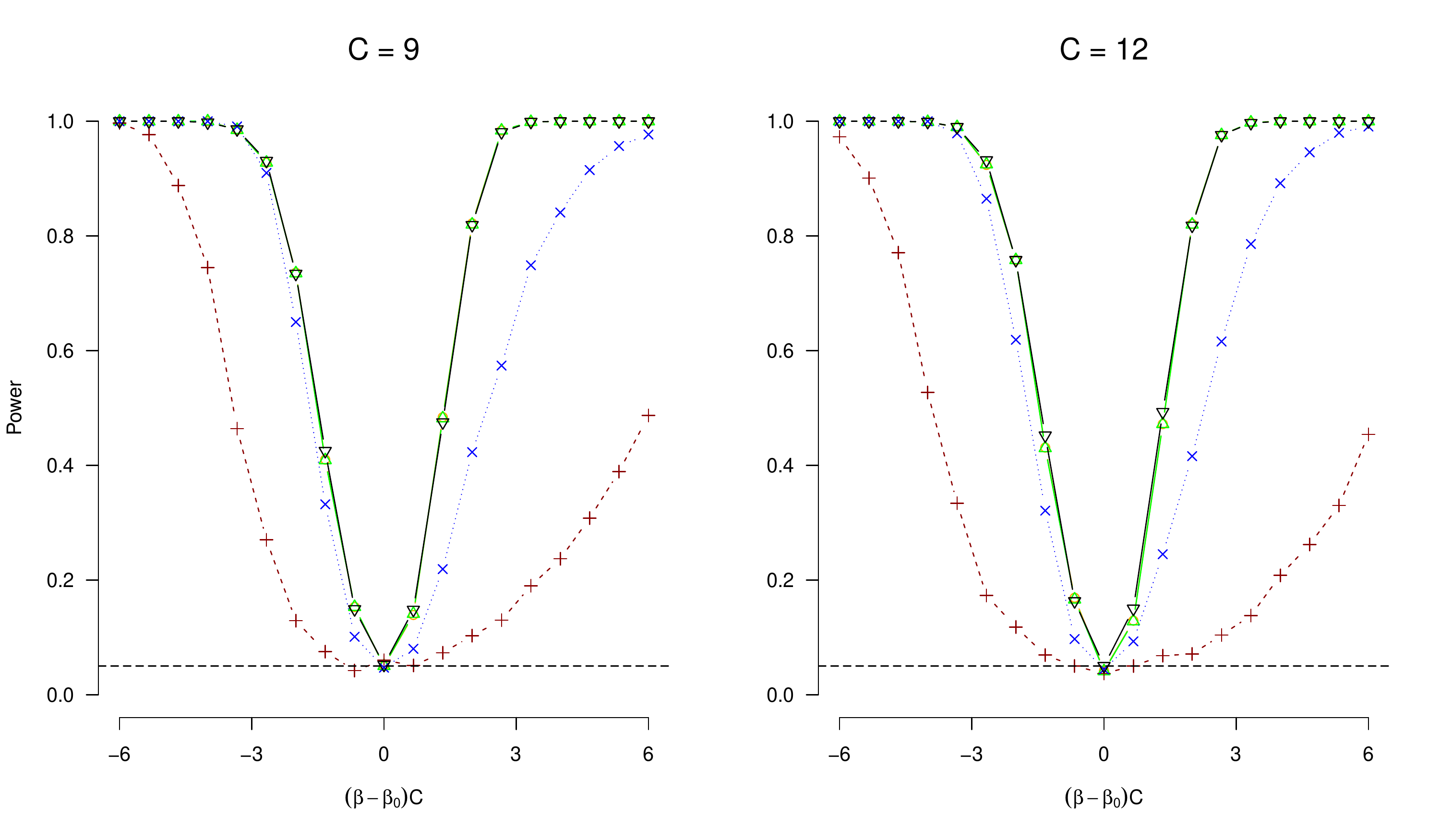}
	\caption{Power Curve for $\rho=0.7$ with $C=9$ or $12$}
	\label{further_limit_fig15}
\end{figure}

\begin{figure}[H]
	\centering
	\includegraphics[width=0.9\textwidth,height = 5.85cm]{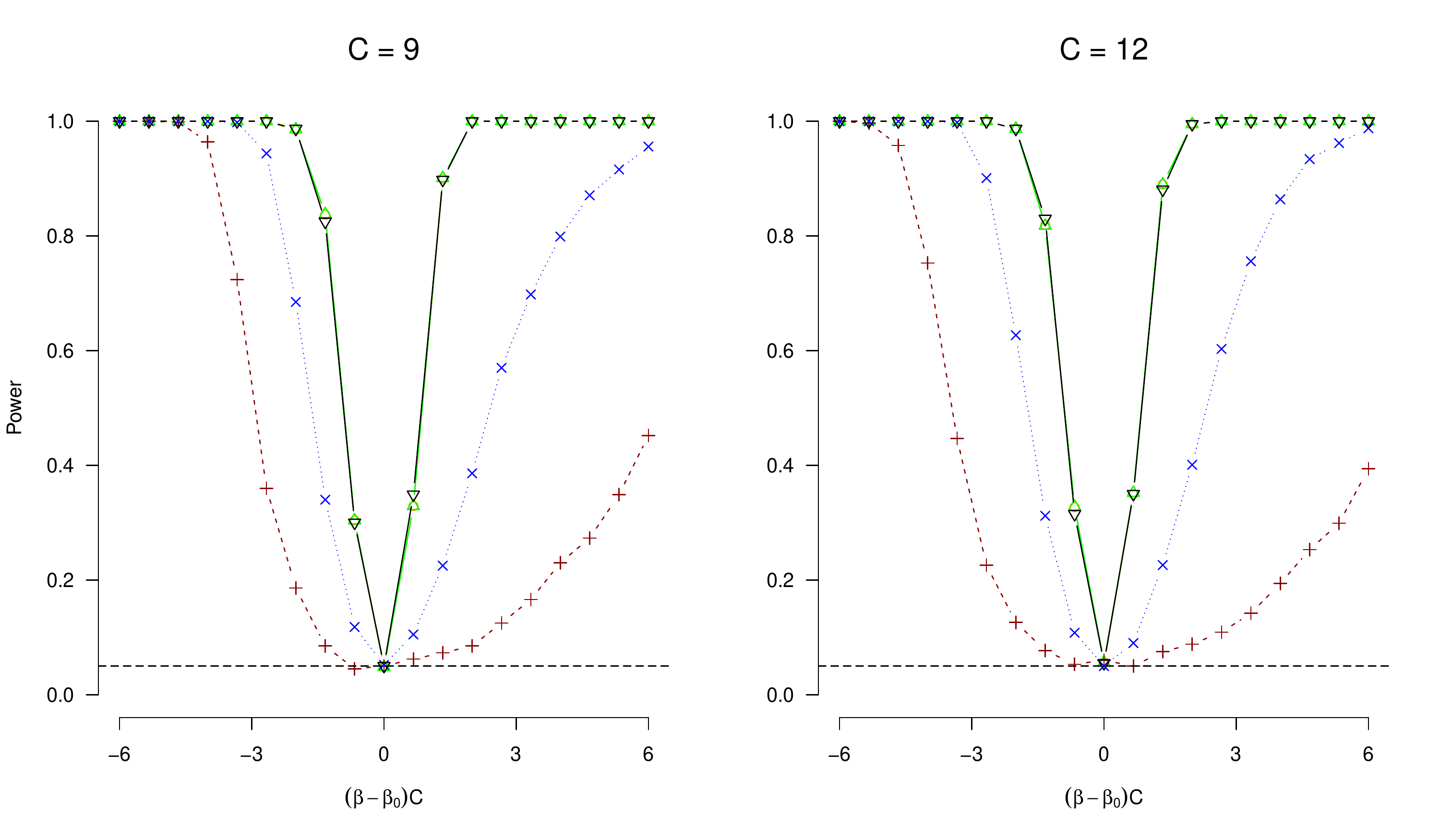}
	\caption{Power Curve for $\rho=0.9$ with $C=9$ or $12$}
	\label{further_limit_fig16}
\end{figure}

\subsection{Additional Simulation Results Based on the Calibrated Data}
\label{sec:add_sim_2}
We run two sets of robustness checks. For the first set, we retained the parameter space of $\mathcal{B} = [-0.5,0.5]$ and used 16 grid-points in total over this space instead of 31 grid-points used in the main text. As in the previous section, we vary over $(p_1,p_2)$ equals $(0.001,1.1)$, $(0.001,1.5)$, $(0.001,2)$, $(0.01,1.5)$, $(0.01,2)$, $(0.1,1.1)$, $(0.1,1.5)$, and $(0.1,2)$. Figures \ref{further_limit_DGP1_fig1}--\ref{further_limit_DGP1_fig8} are results for DGP 1, while Figures \ref{further_limit_DGP2_fig1}--\ref{further_limit_DGP2_fig8} are results for DGP 2. We find that our results are very similar to the main text's specification, i.e. $(p_1,p_2) = (0.01,1.1)$.

\vspace{5mm}
For the second set of robustness checks, we fix $(p_1,p_2) = (0.01,1.1)$ as in the main text and vary the parameter space as $\mathcal{B}_2 = [-0.25,0.25]$ and $\mathcal{B}_3 = [-1,1]$ over 21 equally-sized grid-points. This is done in order to capture the null of $H_0: \beta = 0.1$. DGP 1 is reported in Figures \ref{further_limit_DGP1_fig9} and \ref{further_limit_DGP1_fig10}, while DGP 2 is reported in Figures \ref{further_limit_DGP2_fig9} and \ref{further_limit_DGP2_fig10}.

\begin{figure}[H]
	\centering
	\includegraphics[width=0.9\textwidth,height = 5.85cm]{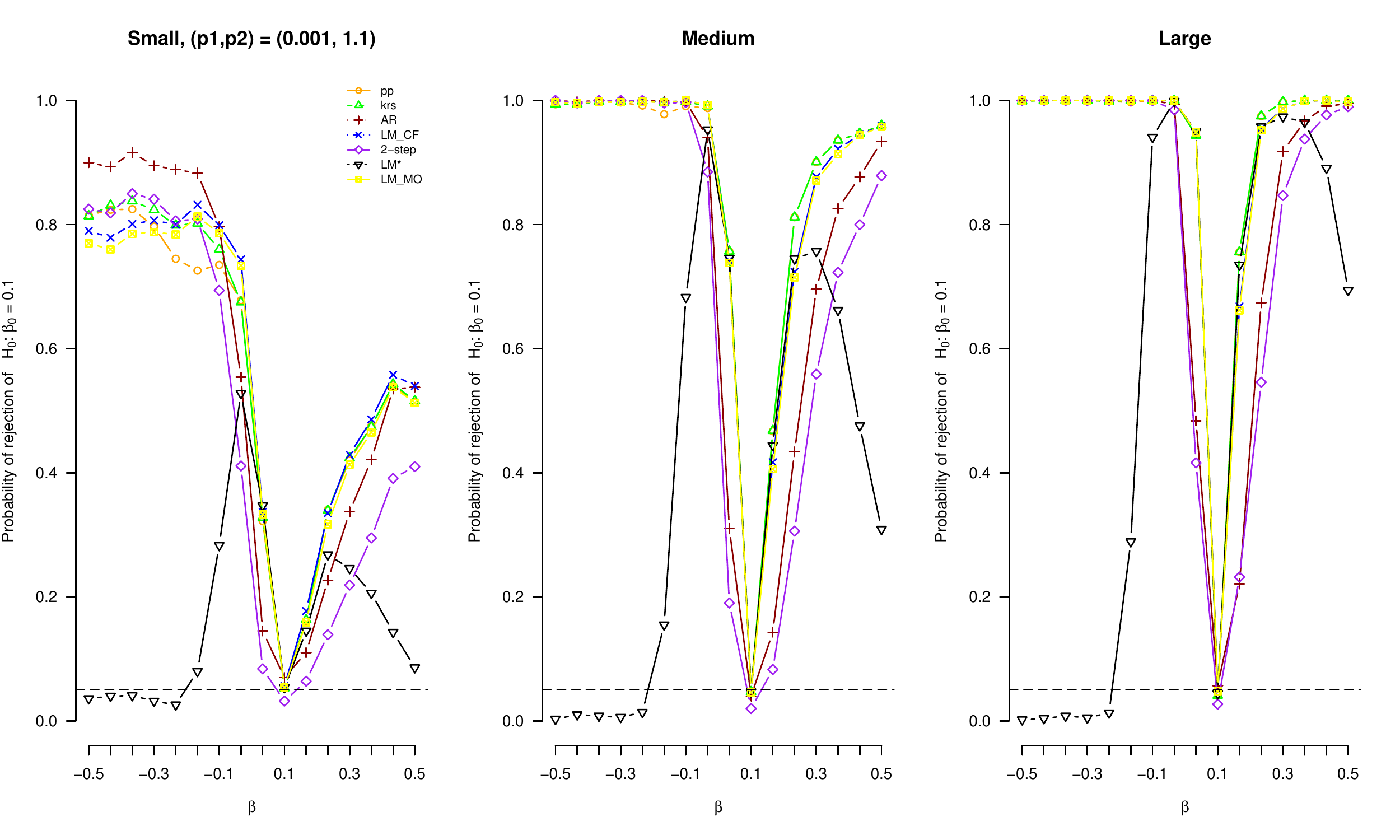}
	\caption{Power Curve for DGP 1 with $(p_1,p_2) = (0.001,1.1)$}
	\label{further_limit_DGP1_fig1}
\end{figure}

\begin{figure}[H]
	\centering
	\includegraphics[width=0.9\textwidth,height = 5.85cm]{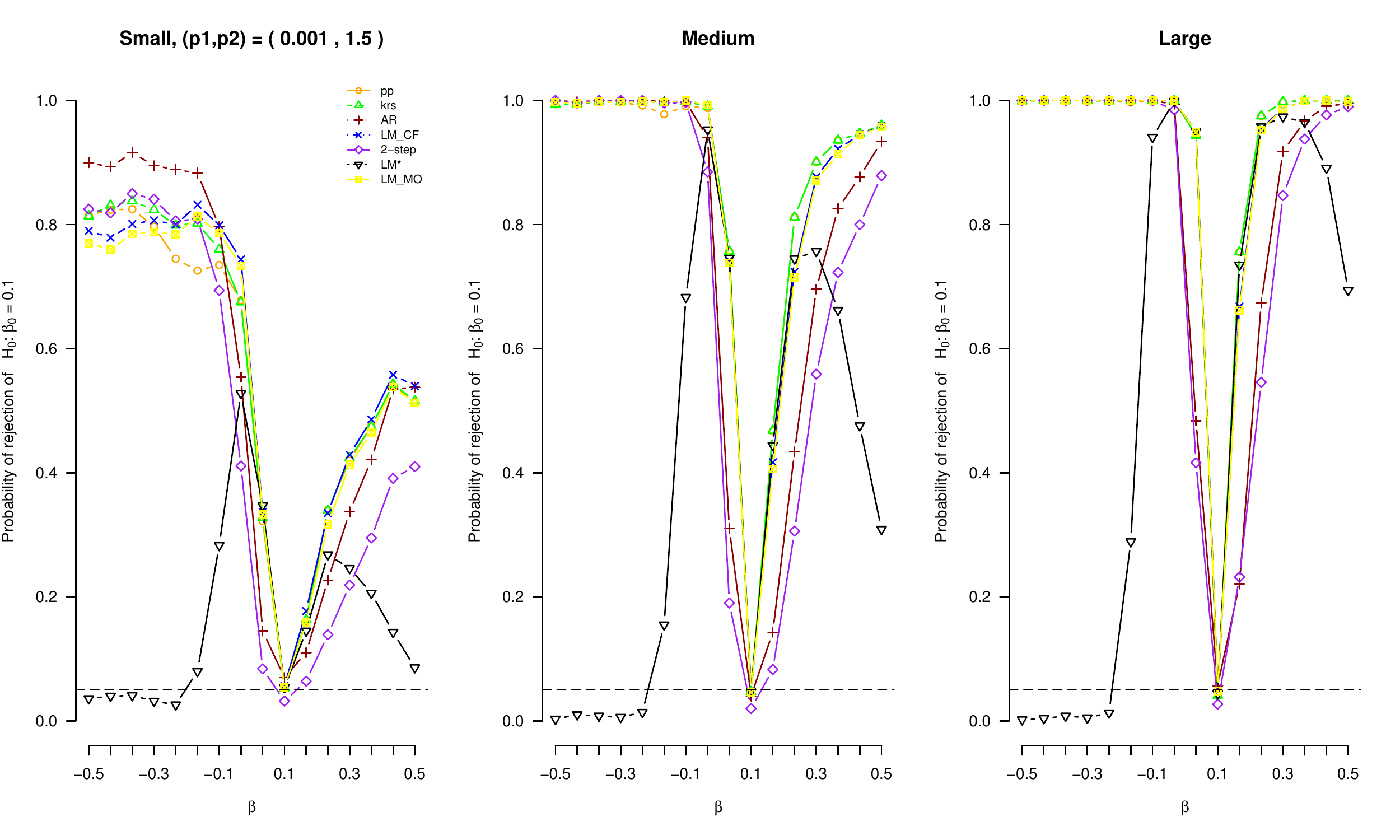}
	\caption{Power Curve for DGP 1 with $(p_1,p_2) = (0.001,1.5)$}
	\label{further_limit_DGP1_fig2}
\end{figure}

\begin{figure}[H]
	\centering
	\includegraphics[width=0.9\textwidth,height = 5.85cm]{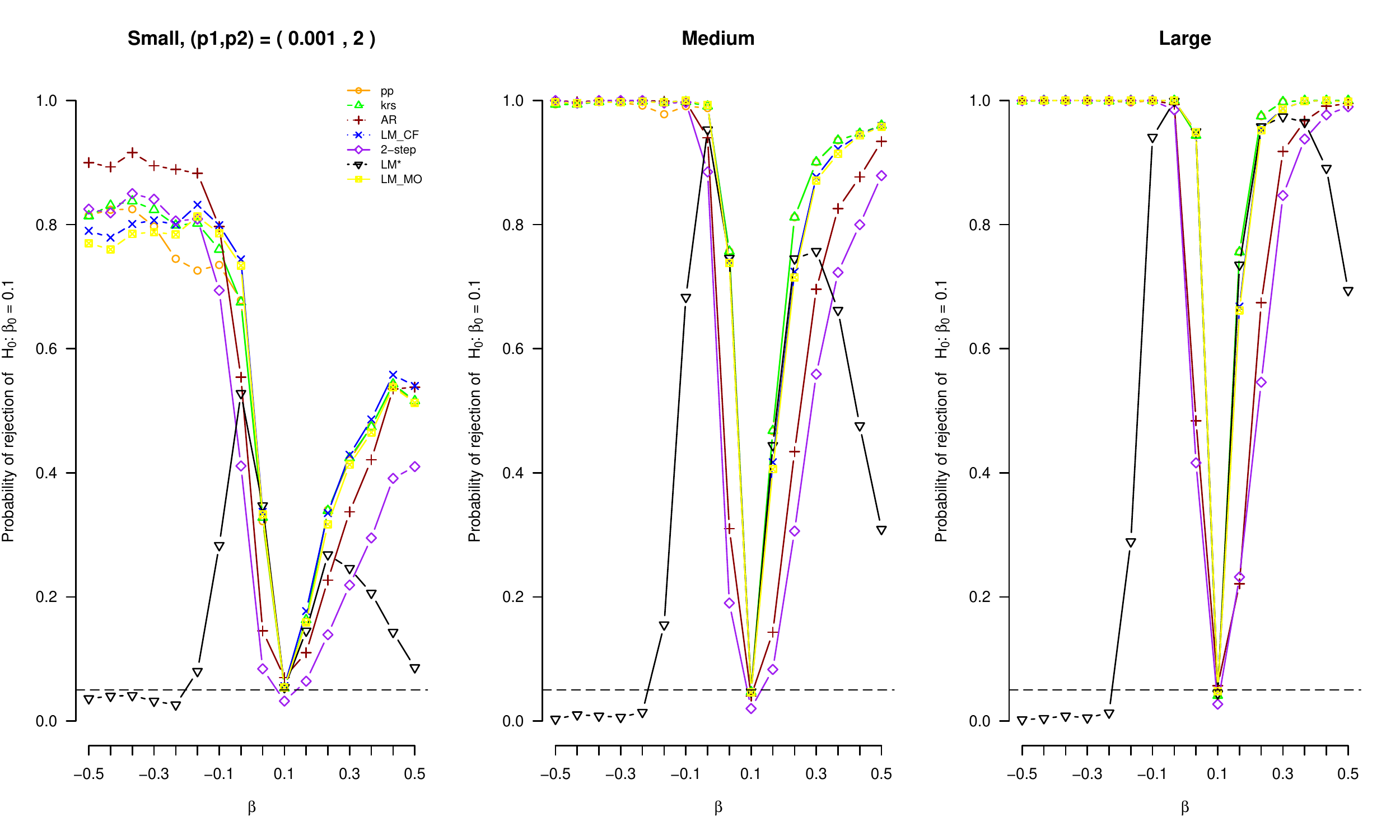}
	\caption{Power Curve for DGP 1 with $(p_1,p_2) = (0.001,2)$}
	\label{further_limit_DGP1_fig3}
\end{figure}

\begin{figure}[H]
	\centering
	\includegraphics[width=0.9\textwidth,height = 5.85cm]{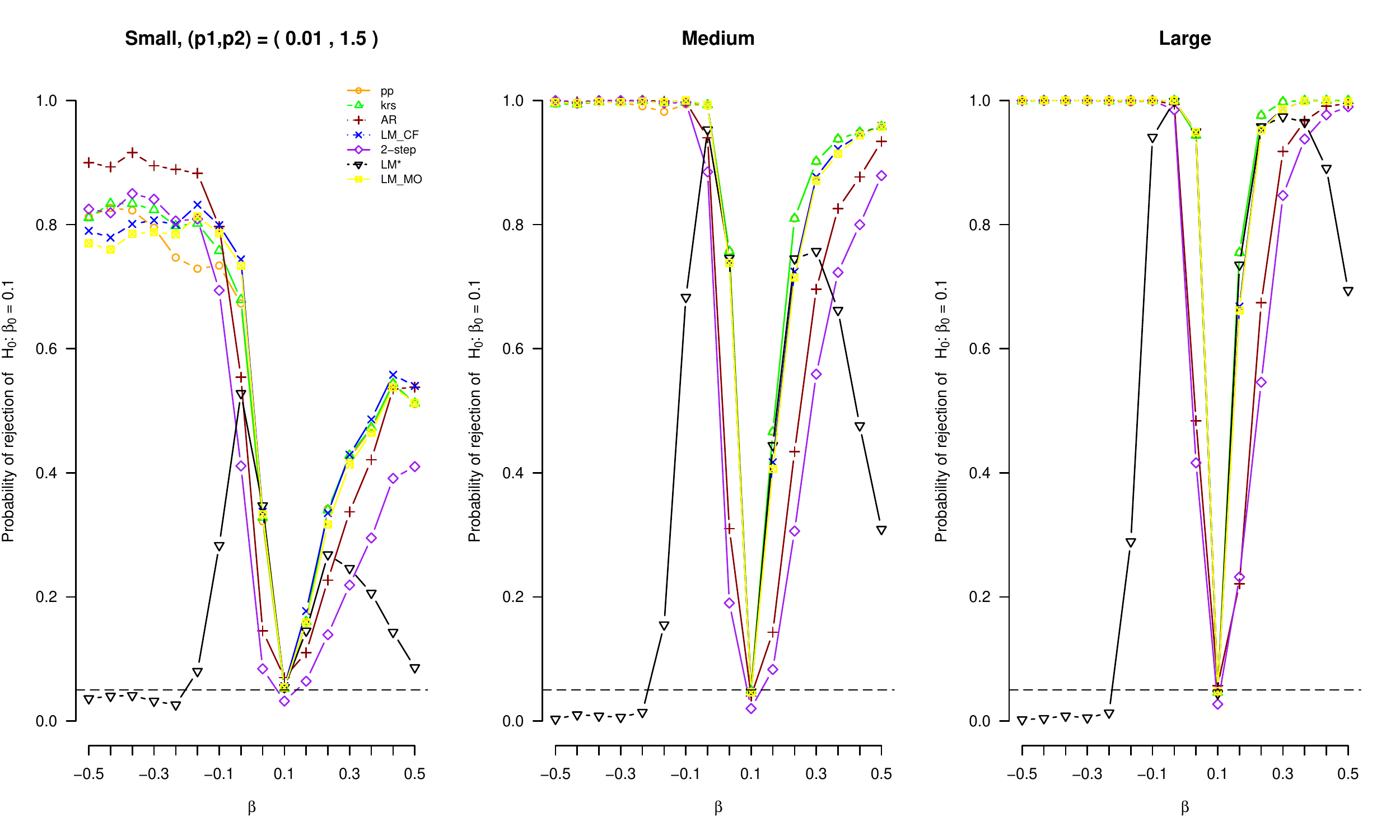}
	\caption{Power Curve for DGP 1 with $(p_1,p_2) = (0.01,1.5)$}
	\label{further_limit_DGP1_fig4}
\end{figure}

\begin{figure}[H]
	\centering
	\includegraphics[width=0.9\textwidth,height = 5.85cm]{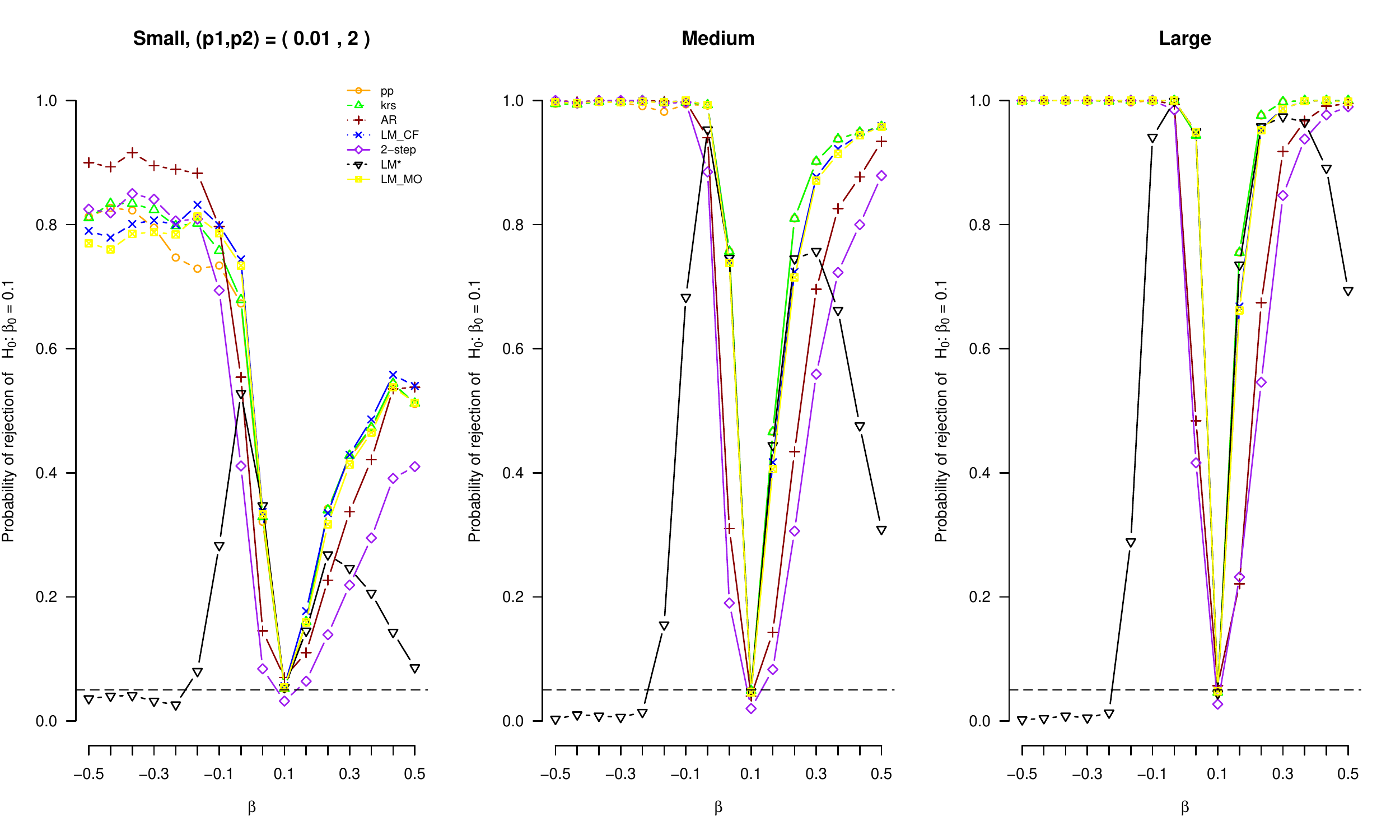}
	\caption{Power Curve for DGP 1 with $(p_1,p_2) = (0.01,2)$}
	\label{further_limit_DGP1_fig5}
\end{figure}

\begin{figure}[H]
	\centering
	\includegraphics[width=0.9\textwidth,height = 5.85cm]{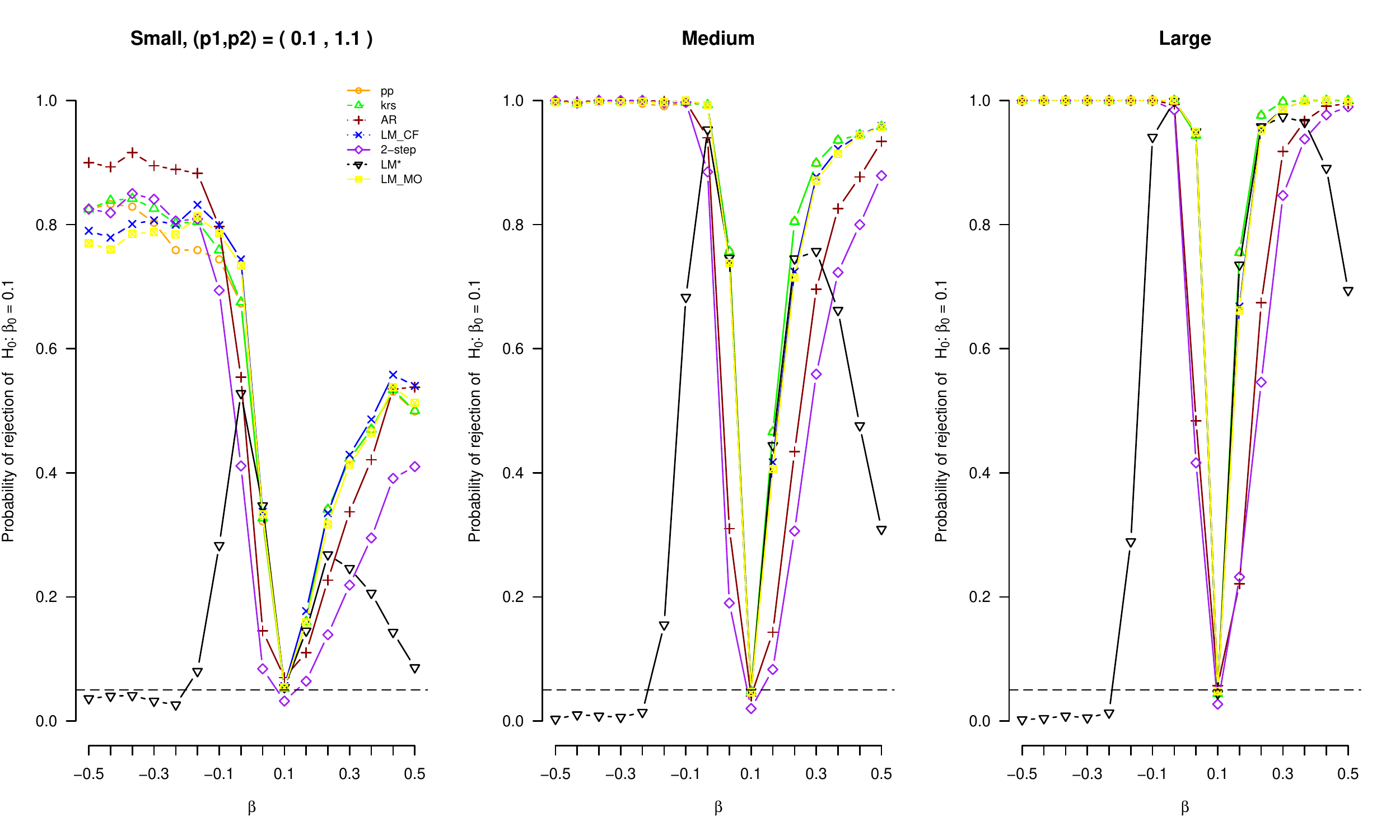}
	\caption{Power Curve for DGP 1 with $(p_1,p_2) = (0.1,1.1)$}
	\label{further_limit_DGP1_fig6}
\end{figure}

\begin{figure}[H]
	\centering
	\includegraphics[width=0.9\textwidth,height = 5.85cm]{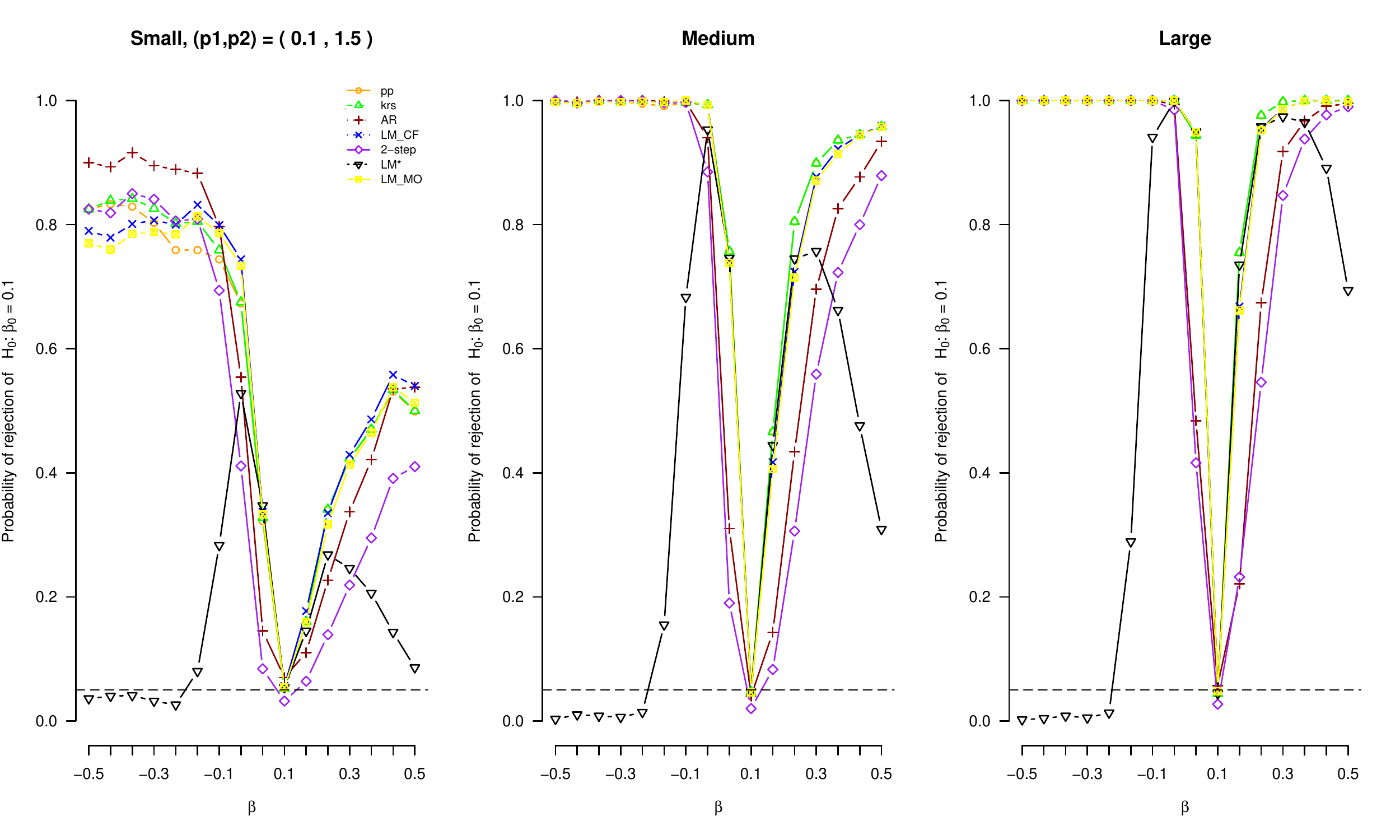}
	\caption{Power Curve for DGP 1 with $(p_1,p_2) = (0.1,1.5)$}
	\label{further_limit_DGP1_fig7}
\end{figure}

\begin{figure}[H]
	\centering
	\includegraphics[width=0.9\textwidth,height = 5.85cm]{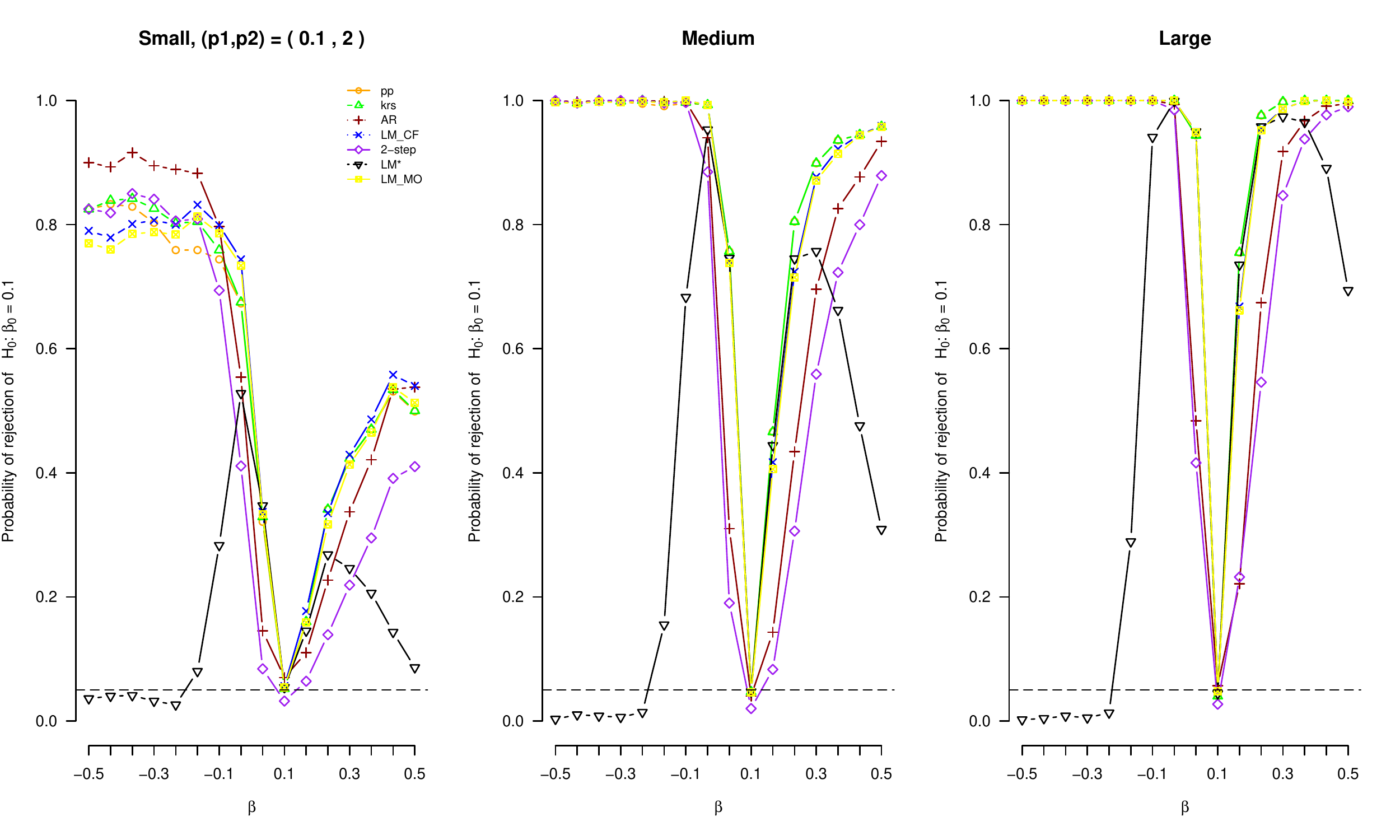}
	\caption{Power Curve for DGP 1 with $(p_1,p_2) = (0.1,2)$}
	\label{further_limit_DGP1_fig8}
\end{figure}

\begin{figure}[H]
	\centering
	\includegraphics[width=0.9\textwidth,height = 5.85cm]{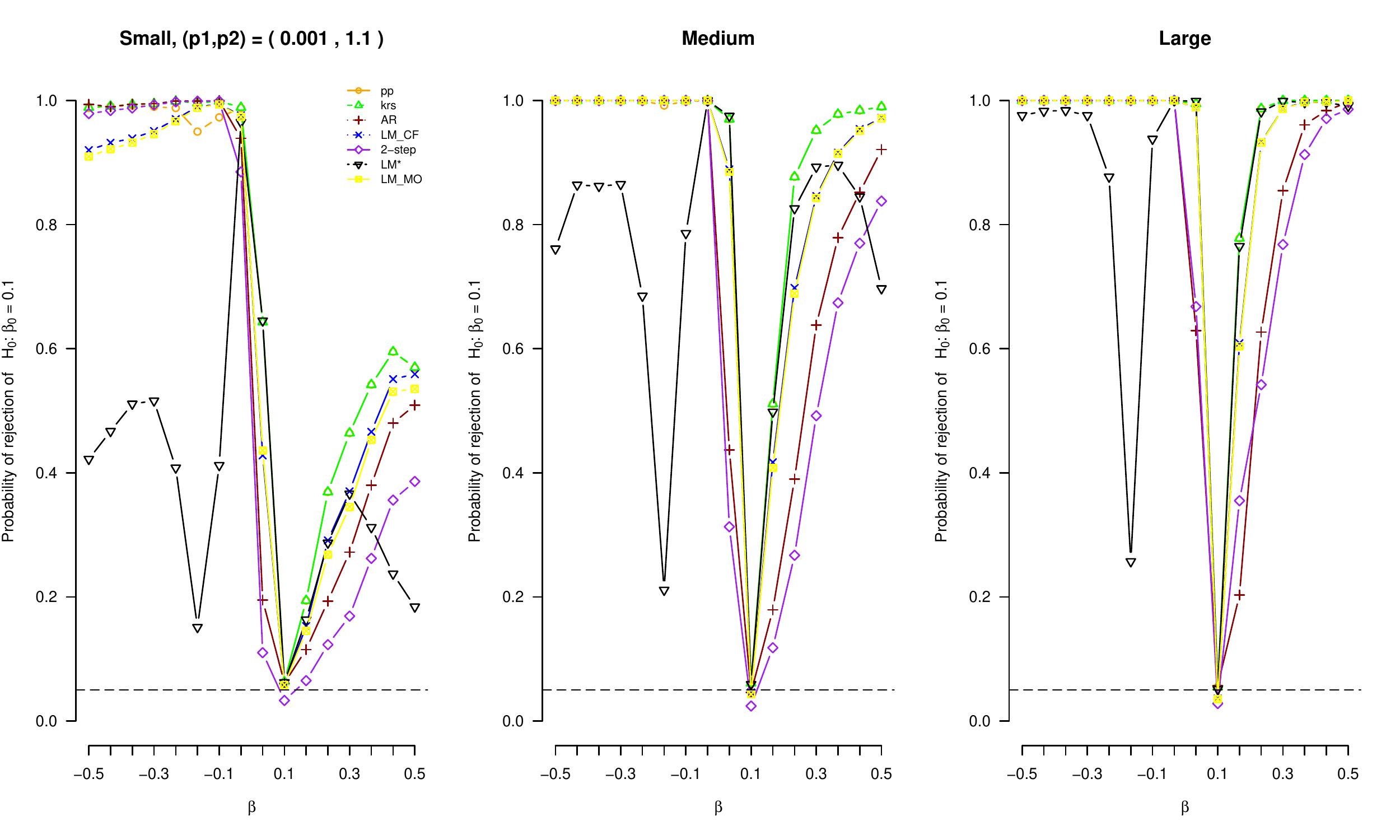}
	\caption{Power Curve for DGP 2 with $(p_1,p_2) = (0.001,1.1)$}
	\label{further_limit_DGP2_fig1}
\end{figure}

\begin{figure}[H]
	\centering
	\includegraphics[width=0.9\textwidth,height = 5.85cm]{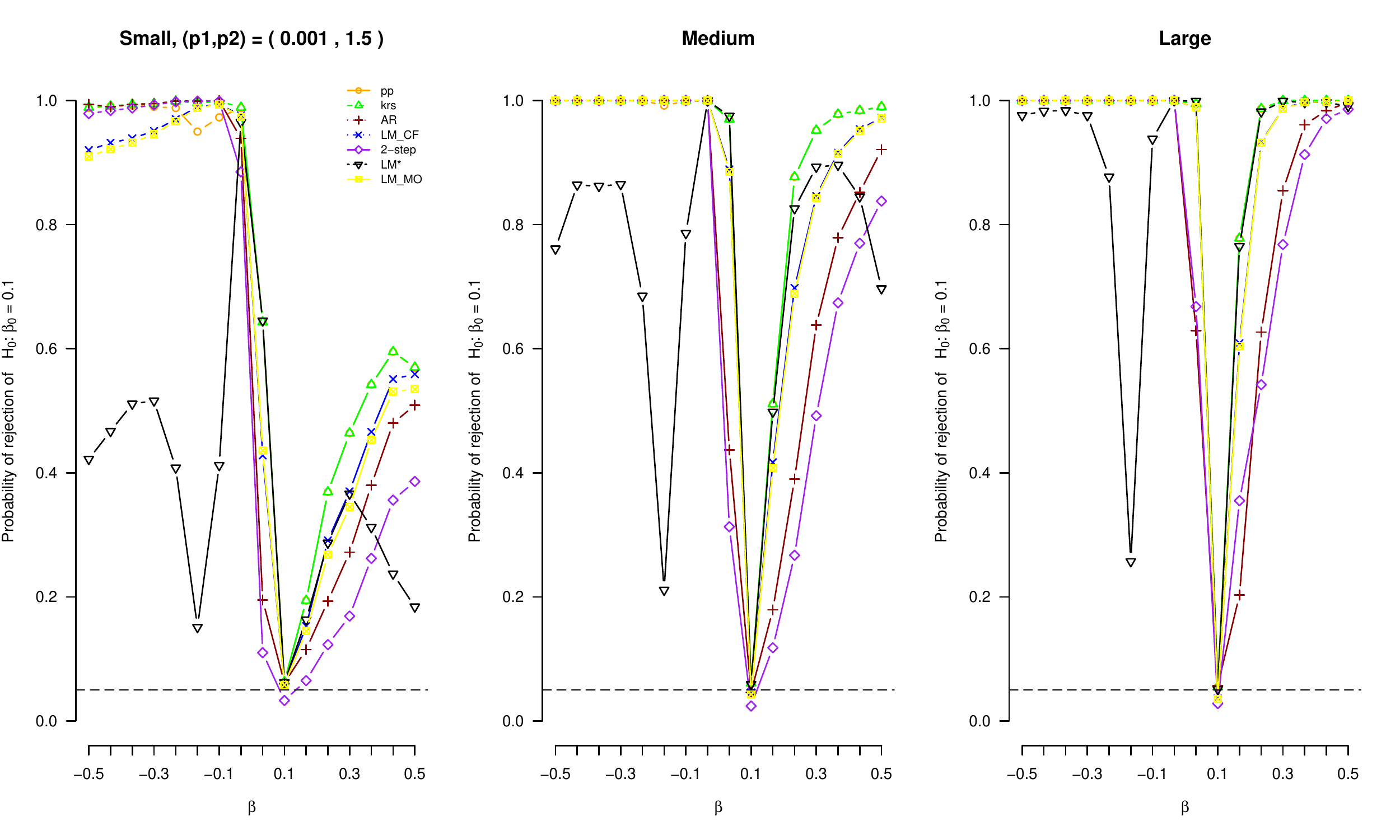}
	\caption{Power Curve for DGP 2 with $(p_1,p_2) = (0.001,1.5)$}
	\label{further_limit_DGP2_fig2}
\end{figure}

\begin{figure}[H]
	\centering
	\includegraphics[width=0.9\textwidth,height = 5.85cm]{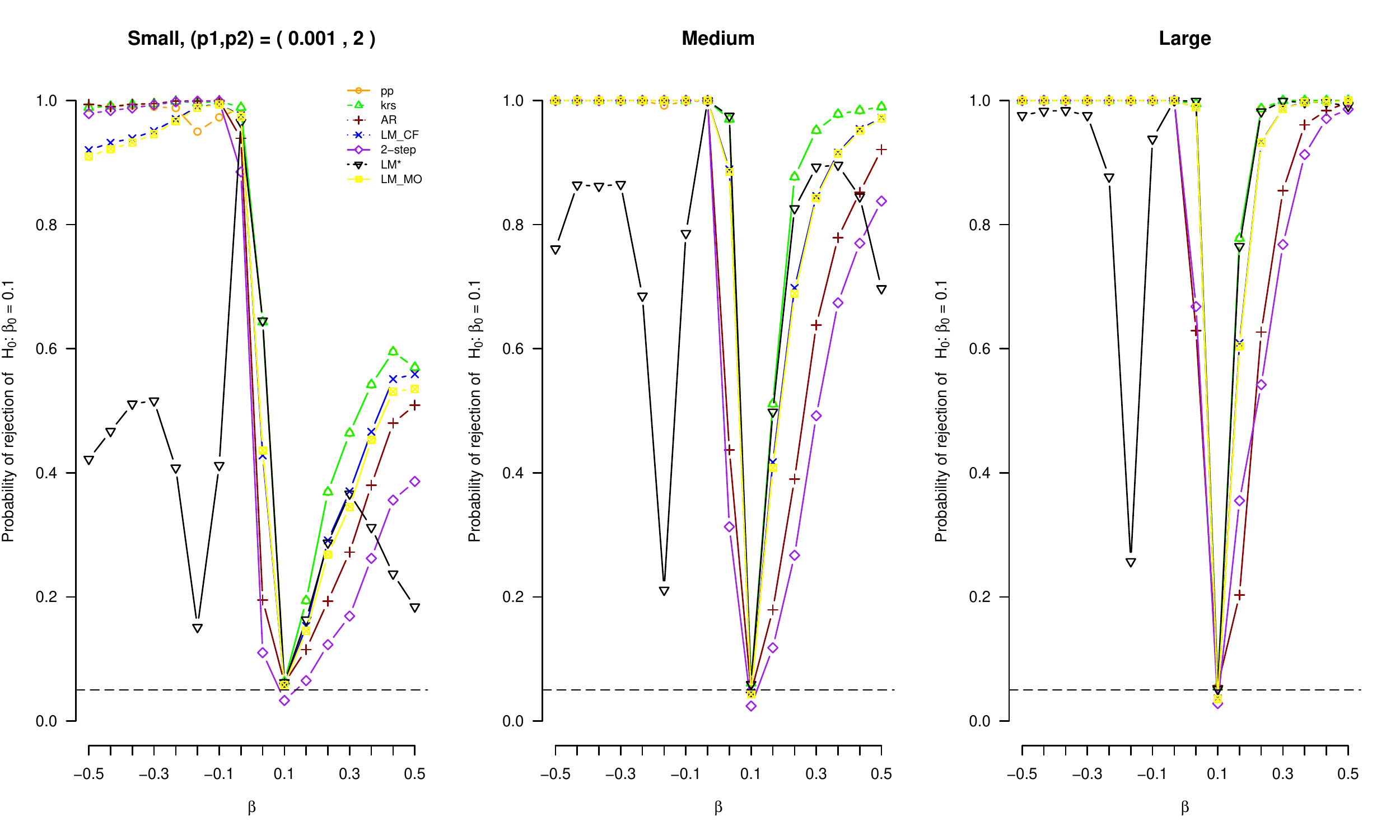}
	\caption{Power Curve for DGP 2 with $(p_1,p_2) = (0.001,2)$}
	\label{further_limit_DGP2_fig3}
\end{figure}

\begin{figure}[H]
	\centering
	\includegraphics[width=0.9\textwidth,height = 5.85cm]{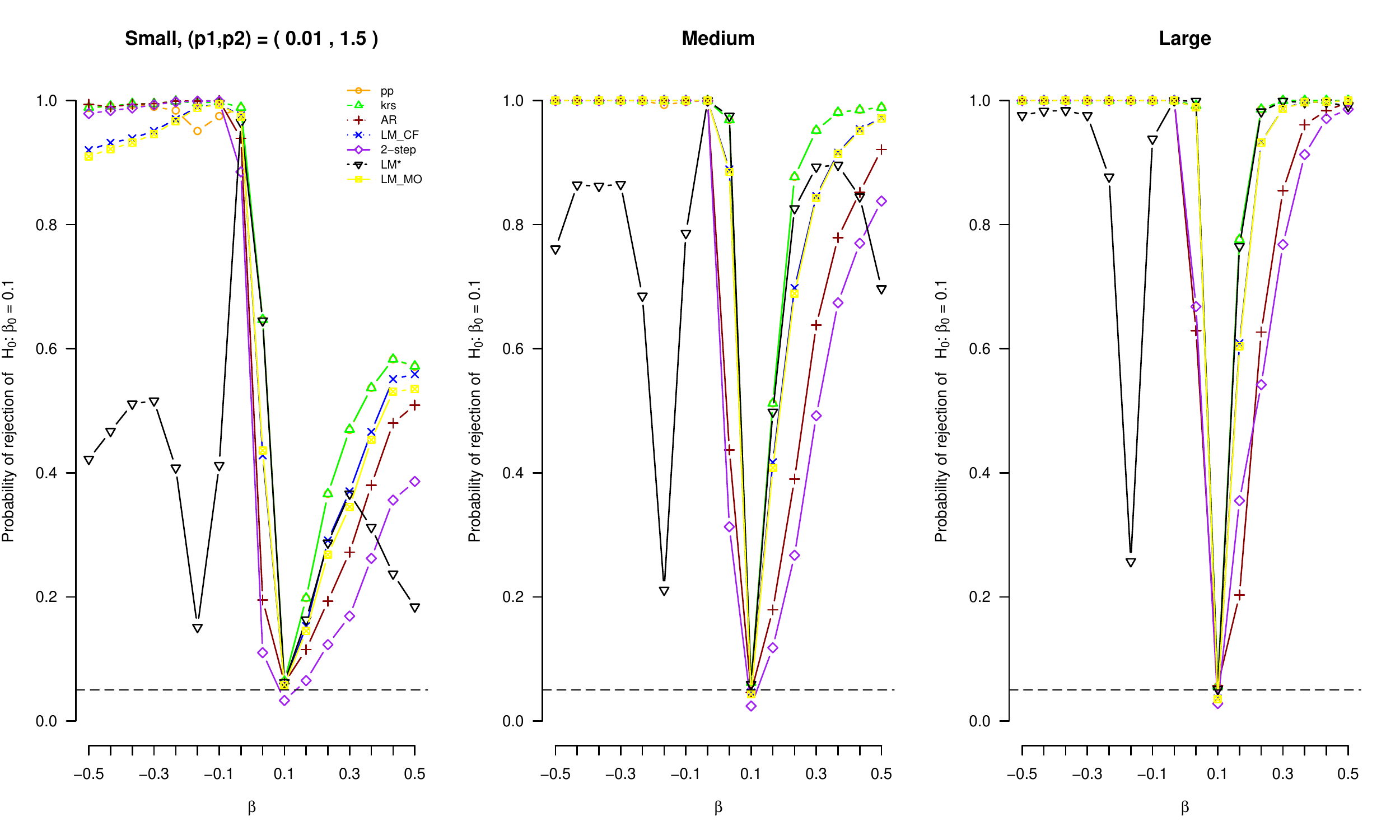}
	\caption{Power Curve for DGP 2 with $(p_1,p_2) = (0.01,1.5)$}
	\label{further_limit_DGP2_fig4}
\end{figure}

\begin{figure}[H]
	\centering
	\includegraphics[width=0.9\textwidth,height = 5.85cm]{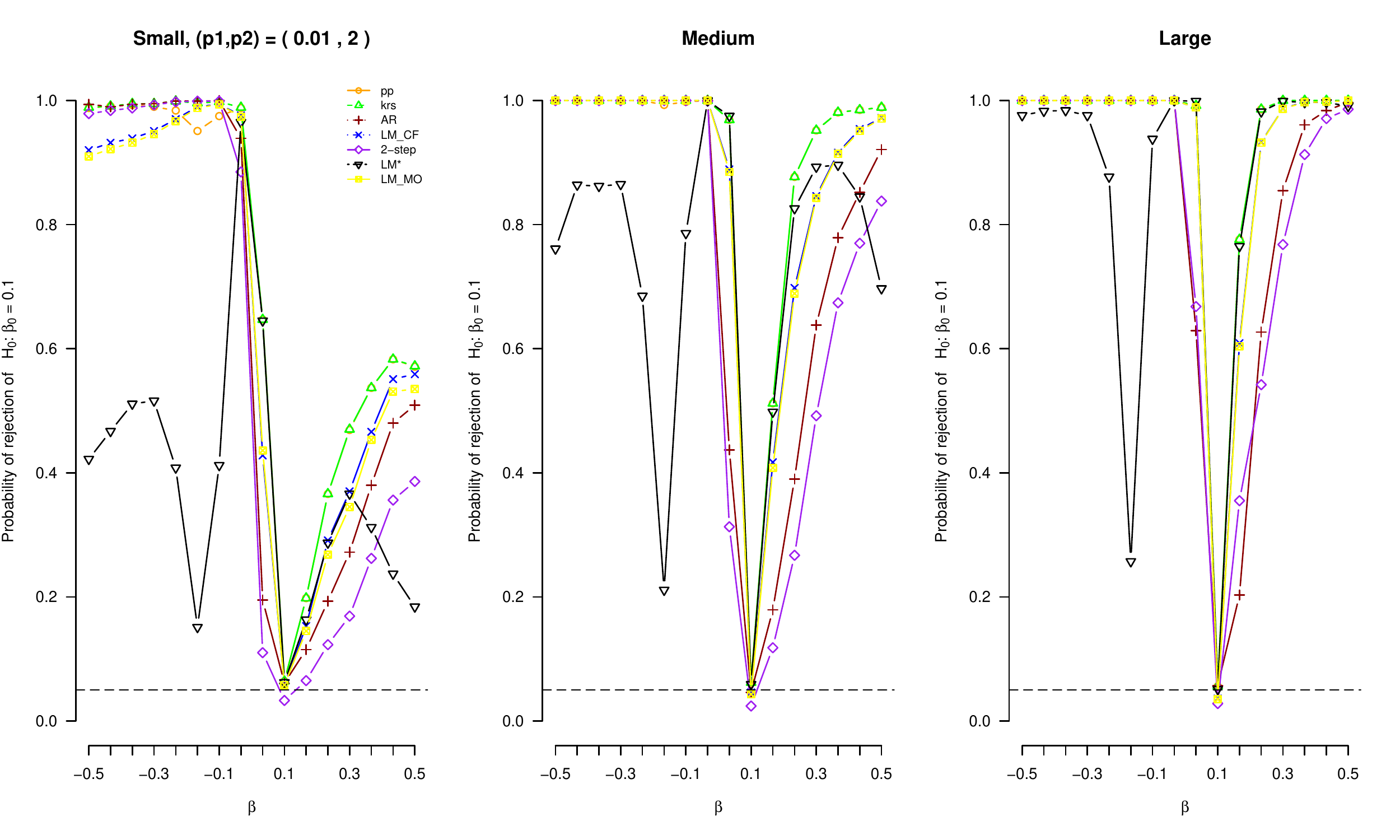}
	\caption{Power Curve for DGP 2 with $(p_1,p_2) = (0.01,2)$}
	\label{further_limit_DGP2_fig5}
\end{figure}

\begin{figure}[H]
	\centering
	\includegraphics[width=0.9\textwidth,height = 5.85cm]{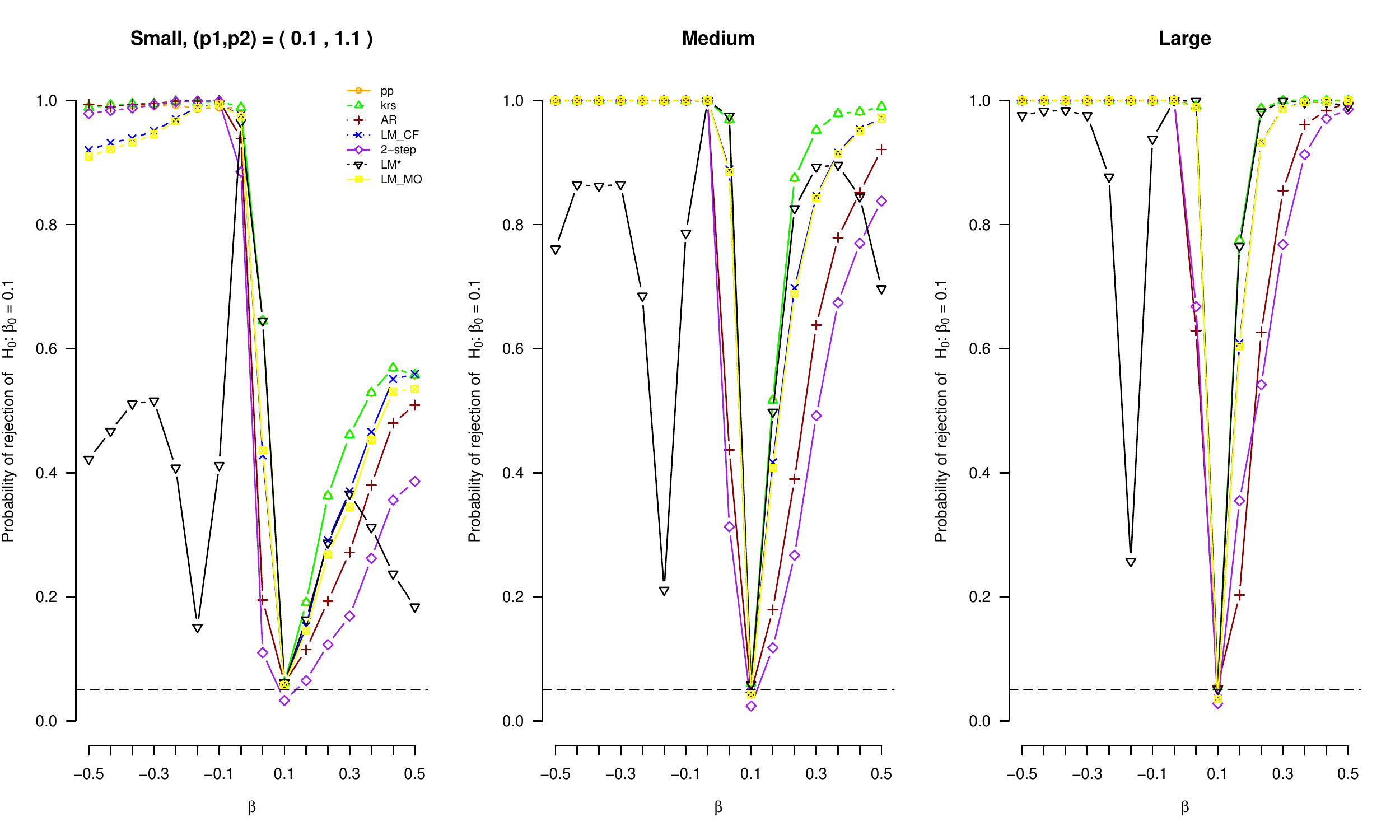}
	\caption{Power Curve for DGP 2 with $(p_1,p_2) = (0.1,1.1)$}
	\label{further_limit_DGP2_fig6}
\end{figure}

\begin{figure}[H]
	\centering
	\includegraphics[width=0.9\textwidth,height = 5.85cm]{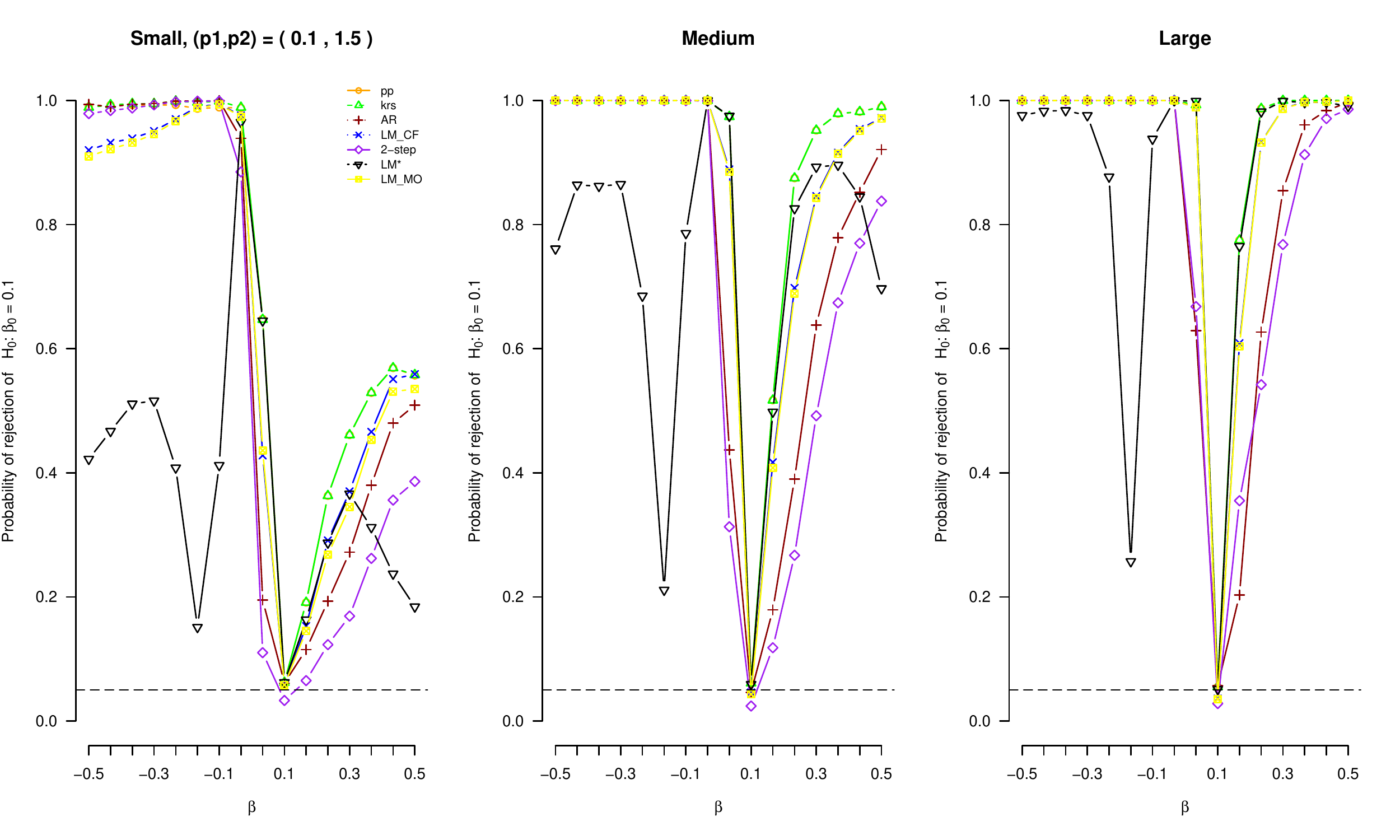}
	\caption{Power Curve for DGP 2 with $(p_1,p_2) = (0.1,1.5)$}
	\label{further_limit_DGP2_fig7}
\end{figure}

\begin{figure}[H]
	\centering
	\includegraphics[width=0.9\textwidth,height = 5.85cm]{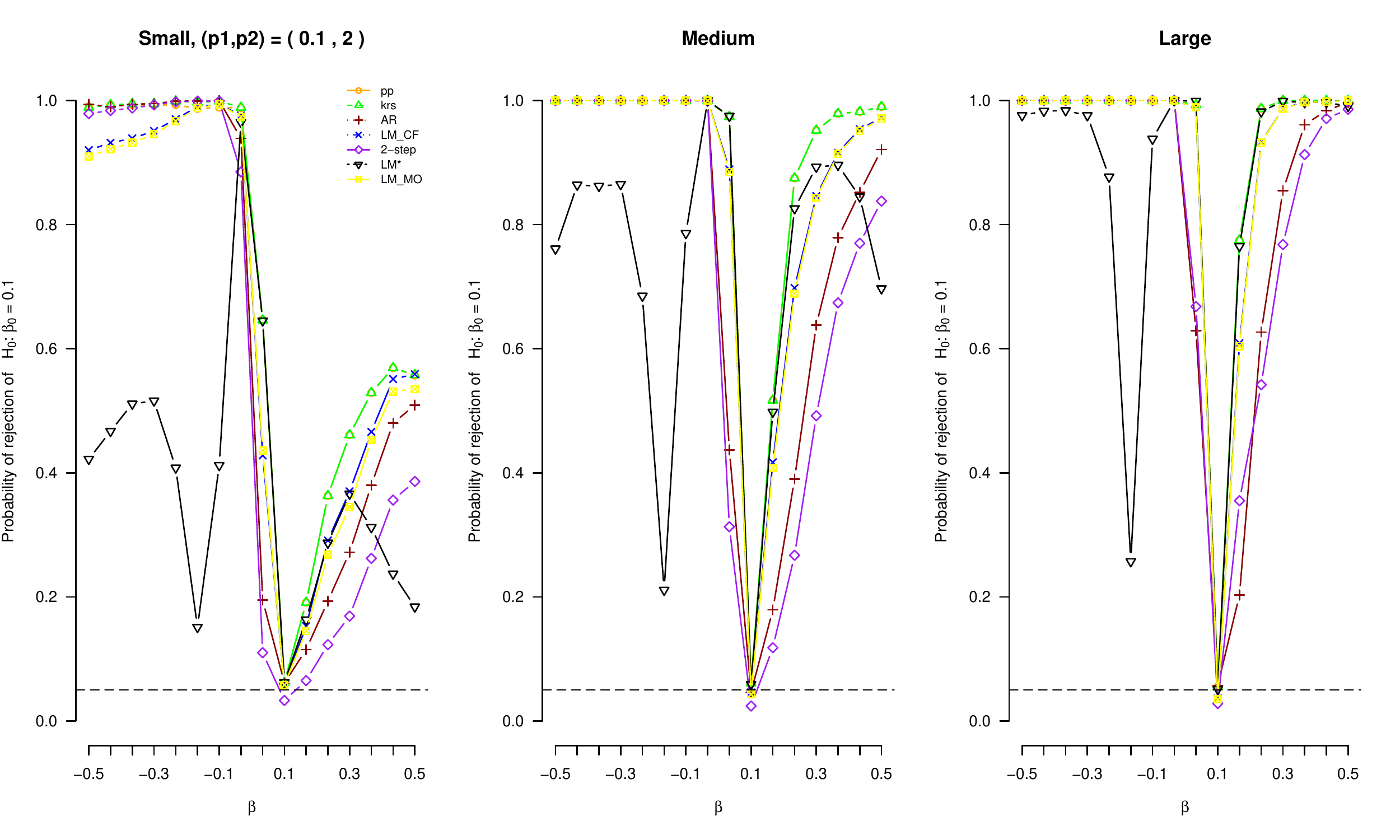}
	\caption{Power Curve for DGP 2 with $(p_1,p_2) = (0.1,2)$}
	\label{further_limit_DGP2_fig8}
\end{figure}

\begin{figure}[H]
	\centering
	\includegraphics[width=0.9\textwidth,height = 5.85cm]{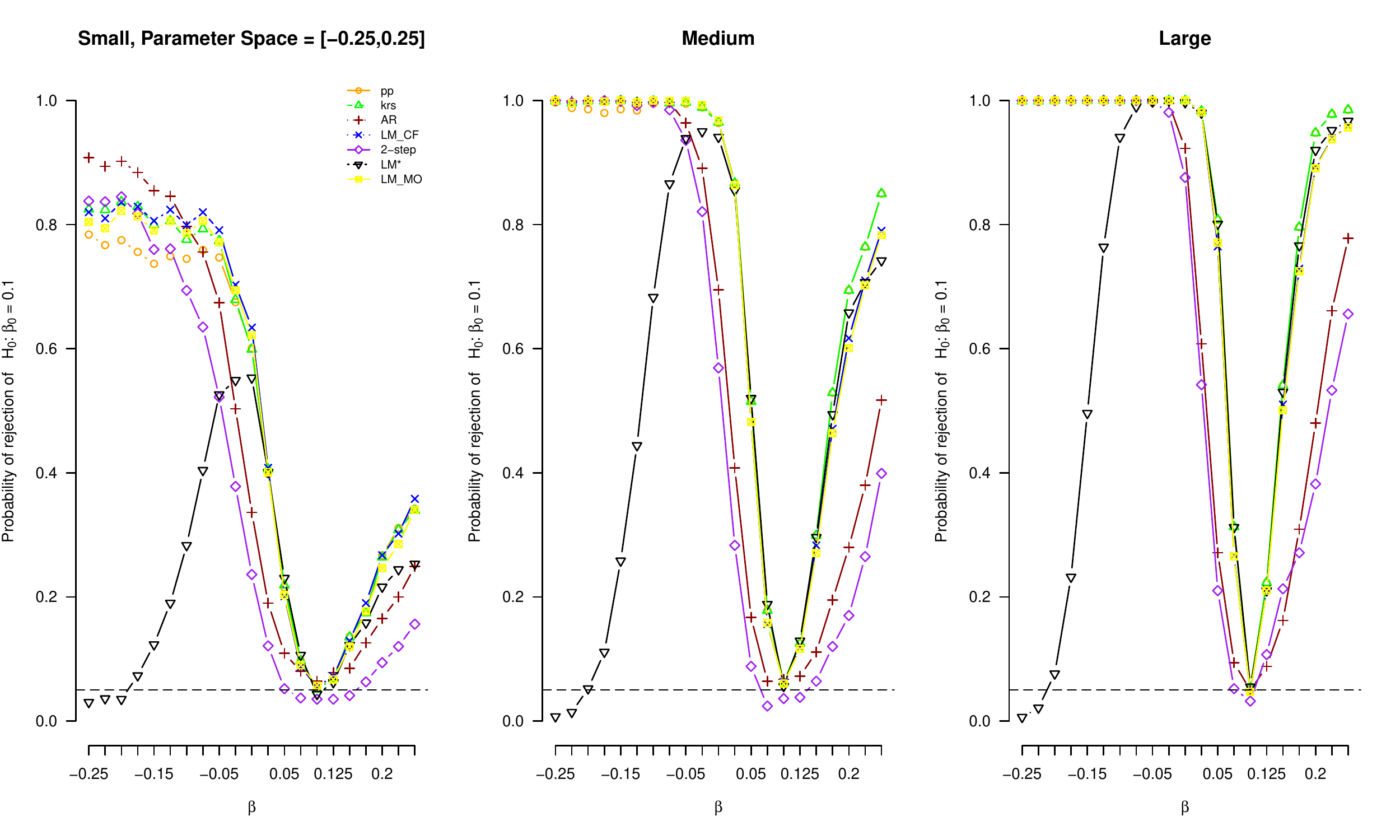}
	\caption{Power Curve for DGP 1 with Parameter Space = $[-0.25,0.25]$}
	\label{further_limit_DGP1_fig9}
\end{figure}

\begin{figure}[H]
	\centering
	\includegraphics[width=0.9\textwidth,height = 5.85cm]{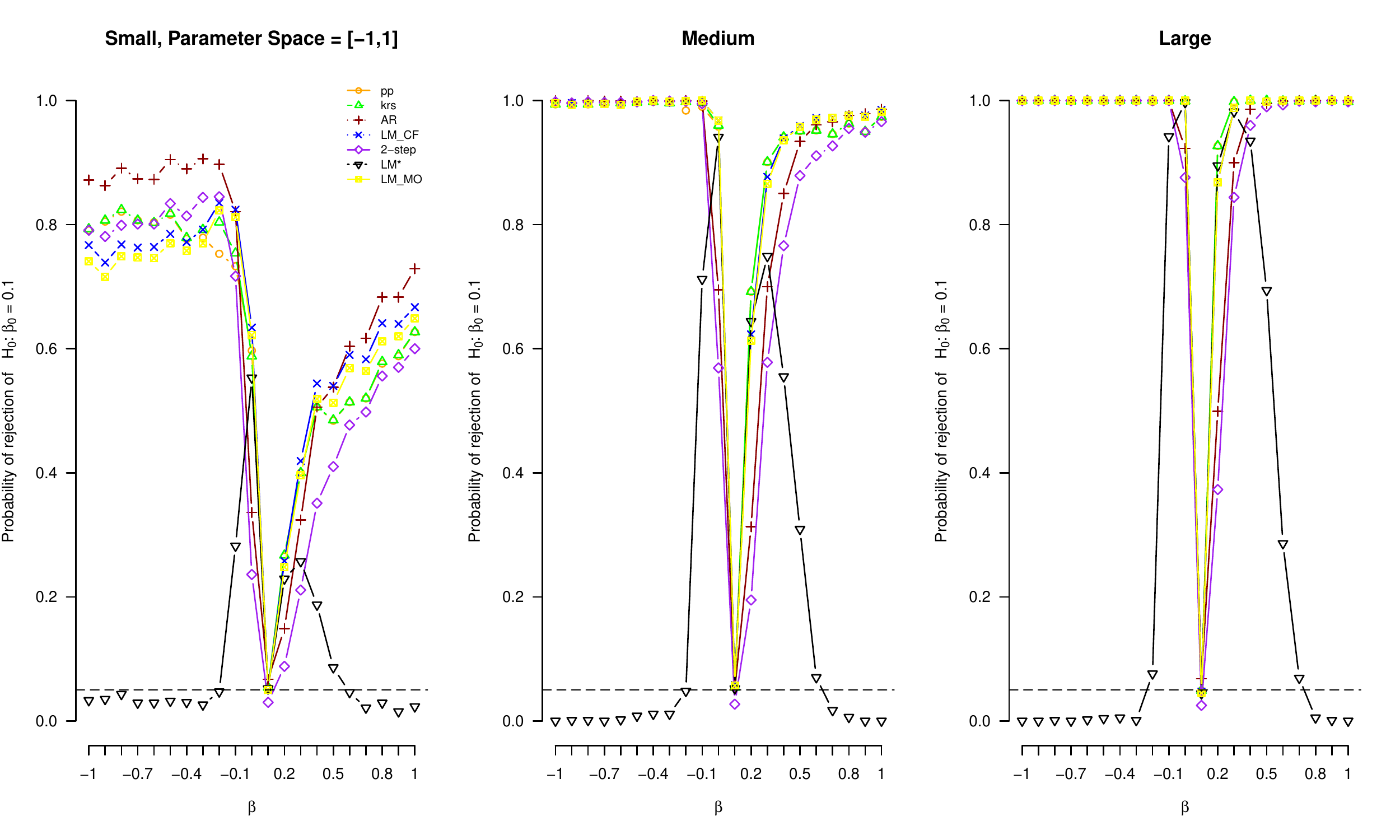}
	\caption{Power Curve for DGP 1 with Parameter Space = $[-1,1]$}
	\label{further_limit_DGP1_fig10}
\end{figure}

\begin{figure}[H]
	\centering
	\includegraphics[width=0.9\textwidth,height = 5.85cm]{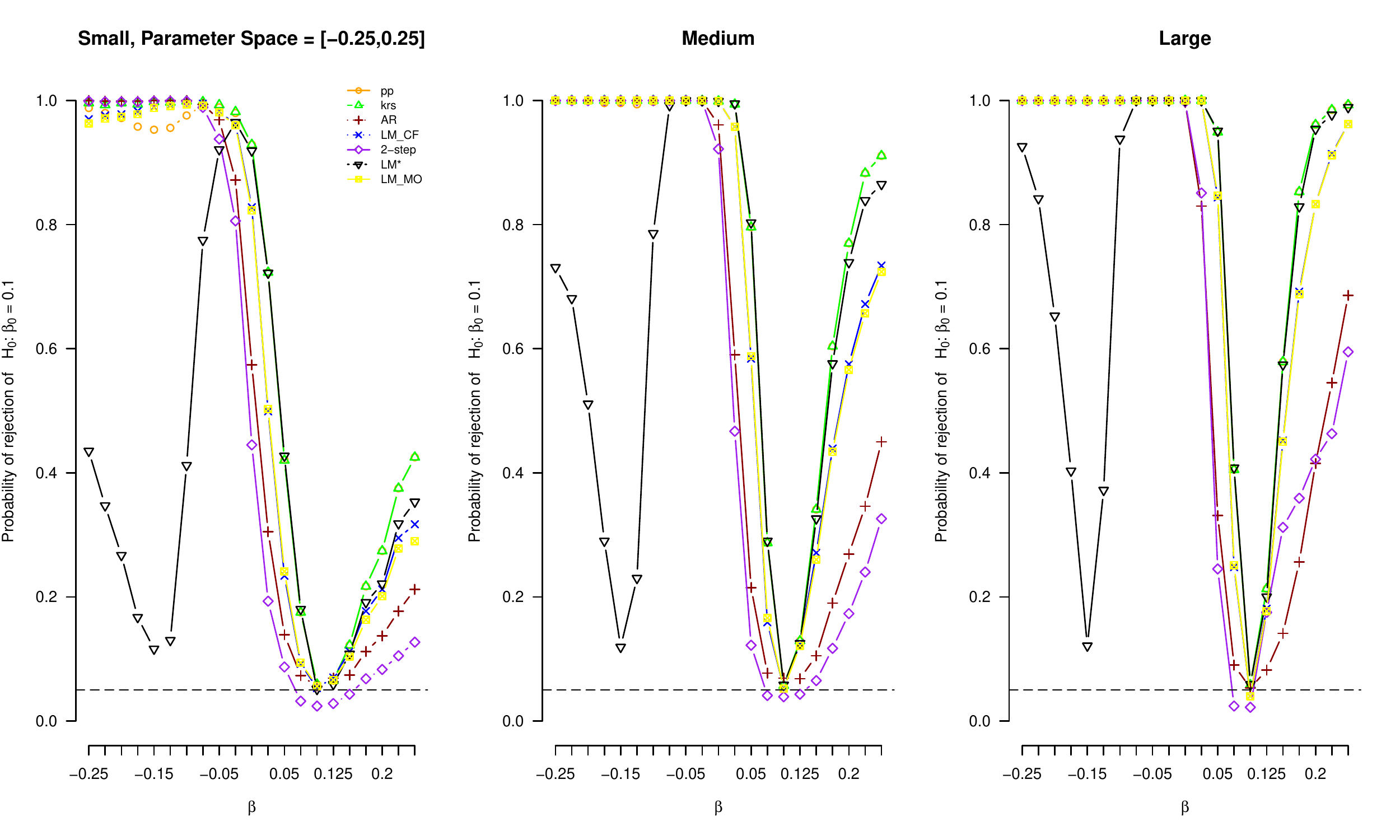}
	\caption{Power Curve for DGP 2 with Parameter Space = $[-0.25,0.25]$}
	\label{further_limit_DGP2_fig9}
\end{figure}

\begin{figure}[H]
	\centering
	\includegraphics[width=0.9\textwidth,height = 5.85cm]{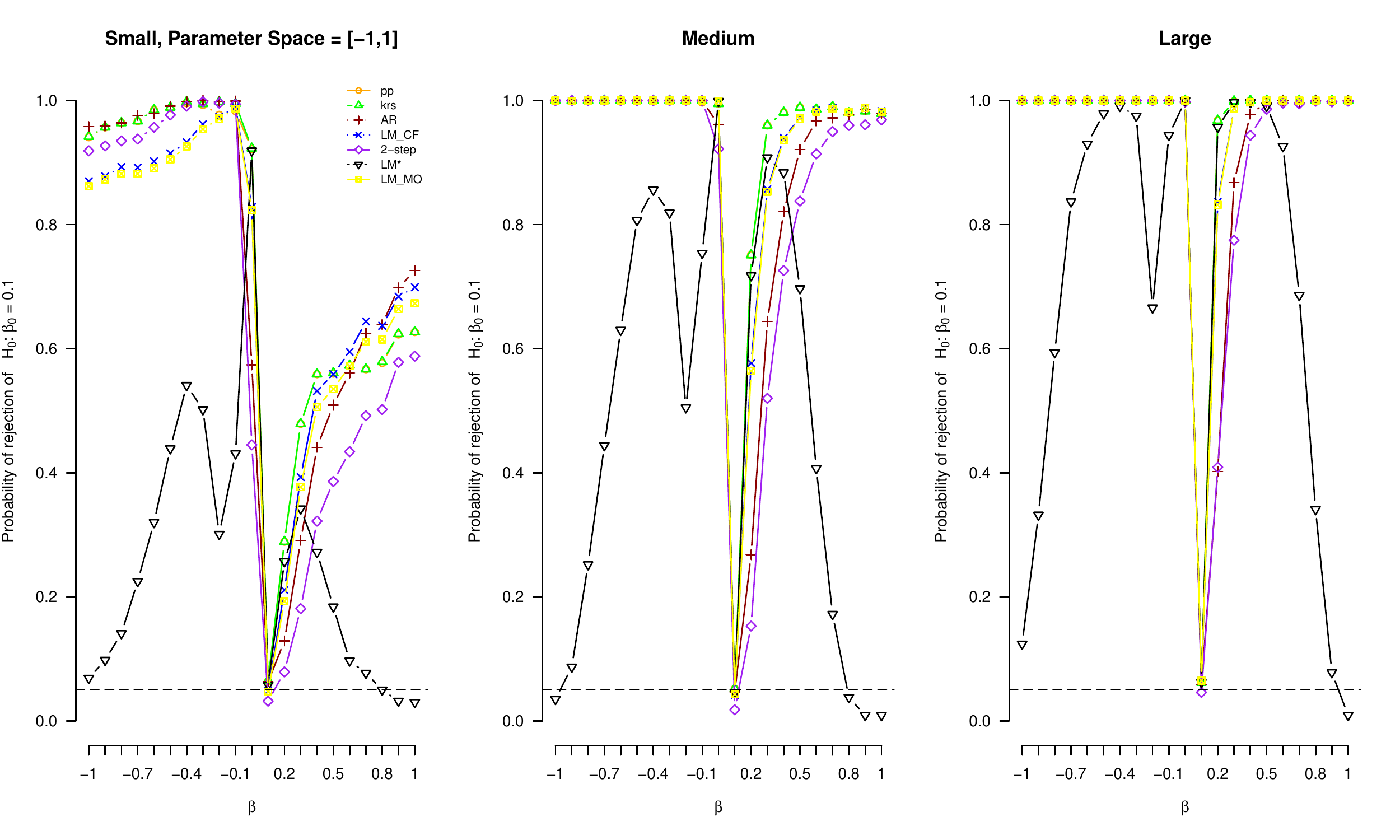}
	\caption{Power Curve for DGP 2 with Parameter Space = $[-1,1]$}
	\label{further_limit_DGP2_fig10}
\end{figure}

\section{Additional Results for the Empirical Application}
\label{sec:add_app}
For the first set of robustness check, we ran 1001 equal-spaced grid-points from parameter space $\mathcal{B} = [-0.5,0.5]$ (step size $=0.001$) over the 9 different variations of $(p_1,p_2)$, which we furnish in Table \ref{add:table1}. The first row is the specification used in the main text, $(p_1,p_2) = (0.01,1.1)$. We do not include `jackknife AR', `jackknife LM', `JIVE-t' and `Two-step' since variations of $(p_1,p_2)$ will not affect the result of those methods. We find that our results are similar to the main text.

\begin{table}[H]
	\adjustbox{max width=\textwidth}{%
		\centering
		\begin{tabular}{c| c c c c c c c } 
			$(p_1,p_2)$-values &   pp with 180 IVs & krs with 180 IVs & pp with 1530 IVs  & krs with 1530 IVs  \\  
			& (5\%) & (5\%) &  (5\%) & (5\%)    \\ [0.5ex] 
			\hline
			(0.01,1.1) & [0.067,0.128] & [0.067,0.128] & [0.037,0.133] & [0.037,0.133] \\
			\hline
			(0.001,1.1) &  [0.072,0.127] & [0.072,0.127] & [0.041,0.132] & [0.041,0.132]  \\ 
			\hline
			(0.001,1.5) & [0.067,0.127] & [0.067,0.127] & [0.038,0.132] & [0.038,0.132]  \\ 
			\hline
			(0.001,2) & [0.066,0.128] & [0.066,0.128] & [0.039,0.133] & [0.039,0.133] \\  
			\hline
			(0.01,1.5) & [0.067,0.127] & [0.067,0.127] & [0.04,0.134] & [0.04,0.134]  \\ 
			\hline
			(0.01,2) & [0.071,0.125] & [0.071,0.125] & [0.041,0.133] & [0.041,0.133]  \\ 
			\hline
			(0.1,1.1) & [0.069,0.126] & [0.069,0.126] & [0.037,0.132] & [0.037,0.132]  \\ 
			\hline
			(0.1,1.5) & [0.072,0.126] & [0.072,0.126] & [0.044,0.132] & [0.044,0.132]  \\ 
			\hline
			(0.1,2) & [0.069,0.127] & [0.069,0.127] & [0.035,0.132] & [0.035,0.132]  \\ 
			\hline
		\end{tabular}
	}
	\caption{Confidence Intervals under different values of $(p_1,p_2)$ with Parameter Space $\mathcal{B}$}
	\label{add:table1}
\end{table}

For the second set of robustness checks, we consider two different parameter spaces, namely $\mathcal{B}_2=[-1,1]$ and $\mathcal{B}_3=[-0.25,0.25]$. Both parameter spaces have 1001 equal-spaced grid-points, and we have retained the values $(p_1,p_2) = (0.01,1.1)$ as in our main text. Table \ref{add:table2} reports the results. Overall, these additional robustness checks show that the results reported in our main text are reliable and hold for different parameter spaces.

\begin{table}[H]
	\adjustbox{max width=\textwidth}{%
		\centering
		\begin{tabular}{c| c c c c c c c } 
			Parameter Space &  pp with 180 IVs & krs with 180 IVs  & pp with 1530 IVs & krs with 1530 IVs \\  
			& (5\%) & (5\%) &  (5\%) & (5\%) &    \\ [0.5ex] 
			\hline
			$\mathcal{B}$ &  [0.067,0.128] & [0.067,0.128] & [0.037,0.133] & [0.037,0.133]  \\ 
			\hline
			$\mathcal{B}_2$ &   [0.068,0.124] & [0.068,0.124] & [0.042,0.134] & [0.042,0.134]  \\ 
			\hline
			$\mathcal{B}_3$ & [0.07,0.1275] & [0.07,0.1275] & [0.037,0.1335] & [0.037,0.1335]\\ [1ex] 
			\hline
		\end{tabular}
	}
	\caption{Confidence Intervals under $(p_1,p_2) = (0.01,1.1)$ with varying Parameter Space $\mathcal{B}_2$ and $\mathcal{B}_3$}
	\label{add:table2}
\end{table}


	\bibliographystyle{chicago}
	\bibliography{Biblio_boot_few_clusters}

\end{document}